\documentclass[10pt, a4paper]{article}
\usepackage[utf8]{inputenc}
\usepackage{amsmath, amsthm}
\usepackage{amsfonts}
\usepackage{amssymb}
\usepackage[normalem]{ulem}

\usepackage{bbm}
\usepackage{hyperref}	
\usepackage{epigraph}	
\usepackage[hang,flushmargin]{footmisc}
\newtheorem{theo}{Theorem}
\newtheorem{open}{Open problem}

\usepackage{blindtext}
\usepackage[left=1.7cm,right=1.7cm,top=1.6cm,bottom=1.6cm]{geometry}

\usepackage{authblk}
\usepackage{enumerate}

\theoremstyle{plain}
\newtheorem{thm}{Theorem}[section]

\theoremstyle{remark}
\newtheorem{rmk}{Remark}[section]

\theoremstyle{definition}

\newcommand{\ep}{\epsilon} 
\newcommand{\RR}{\mathbb{R}}

\theoremstyle{plain}

\newtheorem{conj}{Conjecture}

\theoremstyle{remark}

\theoremstyle{definition}

\usepackage{bm}

\usepackage{color}

\usepackage{float}

\addtocounter{tocdepth}{-2}
\usepackage{graphicx}
\usepackage{authblk}
 
\numberwithin{equation}{section}
\usepackage[german,english]{babel} 
\usepackage{booktabs} 
\usepackage{comment} 
\usepackage[utf8]{inputenc} 
\usepackage[T1]{fontenc} 

\usepackage{bigints,bbm,slashed,mathtools,amssymb,amsmath,amsfonts,amsthm} 
\usepackage{physics}
\usepackage{cancel}
\usepackage{url}
\usepackage{mathrsfs}
\usepackage{enumerate}
\usepackage{tikz} 
\usepackage{tikz-cd}
\usetikzlibrary{babel}
\usetikzlibrary{positioning,arrows} 
\usetikzlibrary{decorations.pathreplacing} 
\usepackage{tikzsymbols} 
\usetikzlibrary{babel, patterns}
\usepackage{blindtext} 

\usepackage{geometry} 

\usepackage{setspace} 
\usepackage[T1]{fontenc} 
\usepackage[utf8]{inputenc} 

\usepackage{XCharter} 

\usepackage{fancyhdr} 
\usepackage[natbib=true, citestyle=numeric, sorting=nyt, url=false, doi=false, giveninits=true, isbn=false, maxbibnames=10]{biblatex}
\usepackage{csquotes}

\numberwithin{equation}{section}
\theoremstyle{definition}
\newtheorem{theorem}{Theorem}[section]

\newtheorem{proposition}[theorem]{Proposition}
\newtheorem{lemma}[theorem]{Lemma}
\newtheorem{corollary}[theorem]{Corollary}

\newtheorem{fact}[theorem]{Fact}

\newtheorem{innercustomgeneric}{\customgenericname}
\providecommand{\customgenericname}{}
\newcommand{\newcustomtheorem}[2]{%
\newenvironment{#1}[1]
{%
\renewcommand\customgenericname{#2}%
\renewcommand\theinnercustomgeneric{##1}%
\innercustomgeneric
}
{\endinnercustomgeneric}
}

\newcustomtheorem{customthm}{Theorem}
\newcustomtheorem{customlemma}{Lemma}

\theoremstyle{remark}
\newtheorem*{remark}{Remark}
\newcommand{\red}[1]{{\color{black} #1}}
\newcommand{\blue}[1]{{\color{black} #1}}

\newcommand{\R}{\mathbb{R}}
\newcommand{\Z}{\mathbb{Z}}
\newcommand{\N}{\mathbb{N}}

\newcommand{\CC}{\mathbb{C}}

\DeclareMathOperator{\sgn}{sgn}

\title{Kasner \blue{bounces} and fluctuating collapse\\  inside hairy black holes with charged matter }
\author{Warren Li and Maxime Van de Moortel}
\begin{document}

\maketitle

\begin{abstract}

We study the interior of black holes in the presence of charged scalar hair of small amplitude $\epsilon$ on the event horizon and show their terminal boundary is a crushing Kasner-like singularity. These spacetimes are spherically symmetric, spatially homogeneous and they differ significantly from 
the hairy black holes with uncharged matter previously studied in \emph{[M. Van de Moortel, Violent nonlinear collapse inside charged hairy black holes, Arch. Rational. Mech. Anal., 248, 89, 2024]} in that the electric field is dynamical \blue{and subject to the backreaction of charged matter}. We  prove \blue{this charged backreaction causes drastically different dynamics compared to the uncharged case} that ultimately impact  the formation of the spacelike singularity\blue{, exhibiting novel phenomena such as}

\begin{itemize}
\item \underline{\blue{Collapsed} oscillations}:\ oscillatory growth of the scalar hair, nonlinearly induced by the collapse
\item \blue{A} \underline{fluctuating collapse}:\ \red{The} final Kasner exponents' dependency in $\ep$ is via an expression of the form \\ $|\sin\left(\omega_0 \cdot \ep^{-2}+ O(\log (\ep^{-1}))\right)|$.  
\item \blue{A} \underline{Kasner \blue{bounce}}:\ a transition from an unstable Kasner metric to a different stable Kasner metric
\end{itemize}
The Kasner \blue{bounce} occurring in our spacetime is reminiscent of the celebrated BKL scenario in  cosmology.

We additionally propose a construction  indicating the relevance of the above phenomena -- including Kasner \blue{bounces} -- to spacelike singularities inside more general \blue{(asymptotically flat)} black holes, beyond the hairy case.

\blue{While our result applies to all values of $\Lambda \in \RR$, in the $\Lambda<0$ case, our spacetime  corresponds to the interior region of a charged asymptotically Anti-de-Sitter stationary black hole, also known as a \emph{holographic superconductor} in high-energy physics, and whose exterior region was rigorously constructed in the recent mathematical work [W. Zheng,  \emph{Asymptotically Anti-de Sitter Spherically Symmetric Hairy Black Holes}, arXiv.2410.04758]}.
\end{abstract}

\setcounter{tocdepth}{2}
\tableofcontents

\section{Introduction}

We are interested in the following two fundamental problems in astrophysics and cosmology: \begin{enumerate}[A.]
\item What does the interior of a black hole look like, and how strong is the (potential) singularity within it?\label{PbI}
\item How does the universe behave near its initial time, and is there a ``Big Bang'' singularity? \label{PbII}
\end{enumerate} As it turns out, these two questions are intimately connected, loosely speaking because a black hole interior's terminal boundary corresponds to the time-reverse of an initial-time singularity (at least locally).

In the present manuscript, we analyze a class of spatially homogeneous singular spacetimes, with the goal to shed some light on Problem~\ref{PbI} as our underlying motivation. We will also emphasize the deep connections between  Problem~\ref{PbI} and Problem~\ref{PbII}: The common theme to both problems are \emph{spacelike singularities} and whether they are stable dynamically. A very important example of such a spacelike singularity at $\{\tau=0\} $ is given by the so-called Kasner metrics \cite{Kasner}, i.e.~spatially homogeneous (but \blue{potentially} anisotropic) spacetimes of the form\footnote{\blue{Note that in the original work of Kasner \cite{Kasner}, \eqref{Kasner} is a solution of the Einstein vacuum equations and satisfies the additional condition  $p_1^2+ p_{2}^2+p_{3}^2=1$. \eqref{Kasner}, however, refers to a  Kasner solution in the presence of matter where this constraint is relaxed, see Section~\ref{cosmo.intro}.}} \begin{equation}\label{Kasner}
g_{K} = -d\tau^2 + \tau^{2 p_1} dx_1^2 + \tau^{2p_{2}} dx_2^2 + \tau^{2p_{3}} d x_3^2, \text{ with } p_1+ p_{2}+p_{3}=1.
\end{equation} 
The main conjectured dynamics of spacelike singularities (in (3+1)-dimensional vacuum, or for gravity coupled to a  reasonable matter model), corresponding respectively to Problem~\ref{PbI} and Problem~\ref{PbII} are as follows:

\begin{conj}[Spacelike singularity conjecture, \cite{KerrStab,Kommemi,JonathanICM,r0}] \label{conj.BH}
In the setting of gravitational collapse (one-ended  data, left picture in Figure~\ref{Fig1}), the terminal boundary of a generic asymptotically flat black hole  consists of \begin{itemize}
\item  a null piece emanating from infinity $i^+$ -- the \emph{Cauchy horizon} $\mathcal{CH}_{i^+}$, as depicted in Figure~\ref{Fig1}.
\item a non-empty \underline{spacelike singularity} $\mathcal{S}$, such that, for all $p \in \mathcal{S}$, the causal past of $p$ has relatively compact intersection with the initial data hypersurface $\Sigma$.
\end{itemize}  
\end{conj} 

\begin{figure}[ht]

\begin{center}

\includegraphics[width=160 mm, height=50 mm]{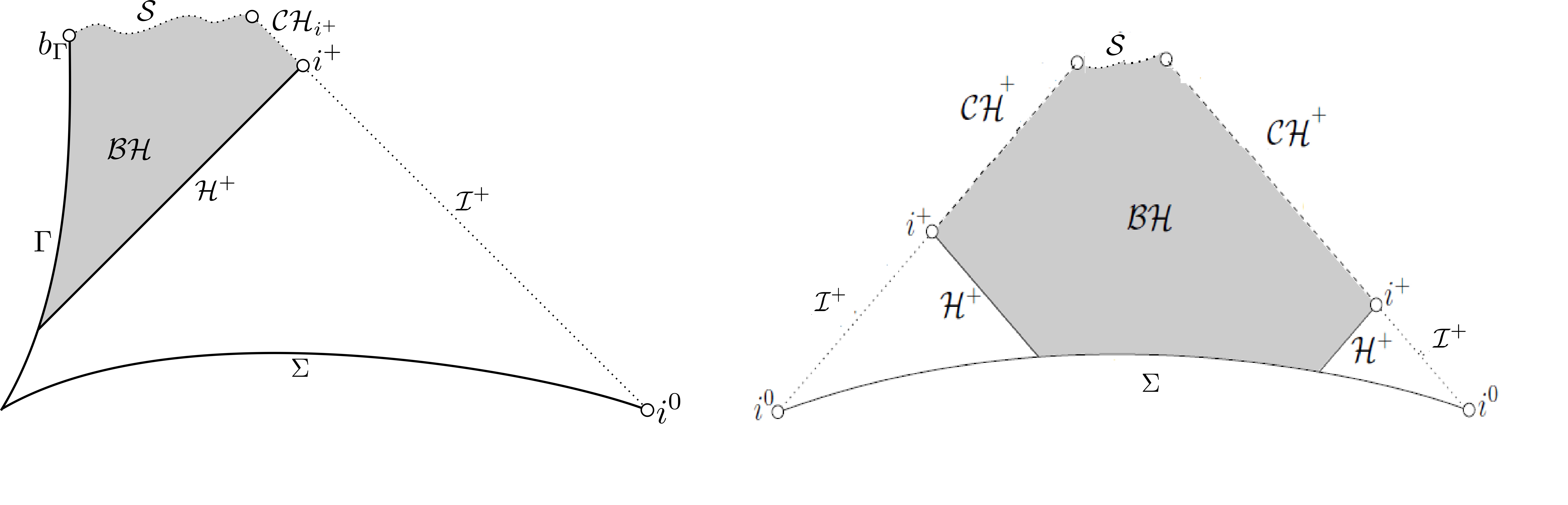}

\end{center}

\caption{Penrose diagram of a (non-hairy) black hole interior with a Cauchy horizon $\mathcal{CH}_{i^+}$ and a spacelike singularity $\mathcal{S}$. Left: one-ended black hole (gravitational collapse case). Right: two-ended black hole.}
\label{Fig1}
\end{figure}
\noindent The following conjecture regroups a series of heuristics \cite{BK1, BKL1, BKL2, KL} and somewhat imprecise statements regarding the typical behavior near spacelike singularities, and is often termed the ``BKL scenario''. \blue{In what follows, we imprecisely use the term ``Kasner regime'' to denote a region of spacetime in which the metric resembles \eqref{Kasner}, in a sense which will be made more precise in Section~\ref{cosmo.intro} (see also \eqref{Kasner.like} and the surrounding discussion).}

\begin{conj}\label{conj.cosmo}[BKL proposal.] The dynamics near a generic spacelike, initial cosmological singularity $\mathcal{S}$, once restricted to a region sufficiently close to $\mathcal{S}$, are described as \blue{follows}. \begin{enumerate}
\item \label{BKL1} \underline{Asymptotically velocity term
dominated behavior}. The causal future $J^{+}(p)$ of any given point $p \in \mathcal{S}$ on the singularity is well-described by a nonlinear system of ODEs, once one is sufficiently close to $\mathcal{S}$. Solutions to these ODEs resemble a sequence of Kasner-like regimes, which may be stable or unstable.
\item\label{BKL2bii} \underline{Kasner \blue{bounces}}. Any unstable Kasner regime transitions towards a (stable or unstable) different Kasner \blue{regime}.
\end{enumerate}
\begin{enumerate}
\item[3.a.] \label{BKL2a} \label{monotone.conj} (For stiff matter models only) \underline{Monotonic regime}. There are finitely many Kasner \blue{bounces}:  ultimately the spacetime only approaches a single (stable) Kasner metric with monotonic dynamics.

\item[3.b.] \label{BKL2b} \label{osc.conj} (For vacuum or non-stiff matter)  \underline{Chaotic regime}.  There are infinitely many Kasner \blue{bounces} between unstable Kasner-like regimes in any generic $J^{+}(p)$. 

\end{enumerate}

\end{conj} Here, a stiff matter model is either a stiff fluid or scalar field, see \cite{AnderssonRendall,BK1, BKL2} for a discussion. \\\indent Our objective is to study a class of \textit{spatially homogeneous solutions} 
of a stiff-matter model: the Einstein--Maxwell--Klein--Gordon equations, in which  a \emph{charged} scalar field $\phi$ is coupled to electromagnetism and gravity. \begin{align}\label{E1} & Ric_{\mu \nu}(g)- \frac{1}{2}R(g)g_{\mu \nu}+ \Lambda g_{\mu \nu}= \mathbb{T}^{EM}_{\mu \nu}+  \mathbb{T}^{KG}_{\mu \nu} , \\ & \label{E2} \mathbb{T}^{EM}_{\mu \nu}=2\left(g^{\alpha \beta}F _{\alpha \nu}F_{\beta \mu }-\frac{1}{4}F^{\alpha \beta}F_{\alpha \beta}g_{\mu \nu}\right), \hskip 7 mm \nabla^{\mu}F_{\mu \nu}=  iq_{0}\left( \frac{ \phi \overline{D_{\nu}\phi} -\overline{\phi} D_{\nu}\phi}{2}\right) , \; F=dA, \\ &  \label{E3}  \mathbb{T}^{KG}_{\mu \nu}= 2\left( \Re(D _{\mu}\phi \overline{D _{\nu}\phi}) -\frac{1}{2}(g^{\alpha \beta} D _{\alpha}\phi \overline{D _{\beta}\phi} + m ^{2}|\phi|^2  )g_{\mu \nu} \right),  \hskip 7 mm D_{\mu}= \nabla_{\mu}+ iq_0 A_{\mu}, \\ &\label{E5} g^{\mu \nu} D_{\mu} D_{\nu}\phi = m ^{2} \phi,  \end{align} with $\Lambda \in \RR$ the cosmological constant, $m^2 \in \mathbb{R}$  the Klein--Gordon mass, and $q_0 \neq 0$ the scalar field charge.

Our study is based on the evolution of initial data posed on bifurcate characteristic hypersurfaces $\mathcal{H}_L$ and $\mathcal{H}_R$ emulating the event horizons of a two-ended black hole, and we show they lead to a spacelike singularity. 
Our main Theorem~\ref{thm.intro} will also provide a precise near-singularity behavior in terms of one or two Kasner metrics of the form \eqref{Kasner}. In cosmological terms, the spacetimes we construct in Theorem~\ref{thm.intro} are so-called Kantowski-Sachs metrics, namely spatially homogeneous, but anisotropic cosmological spacetimes with spatial topology $\RR\times \mathbb{S}^2$.

We shall \blue{also} view  each of our constructed spacetimes as the interior region of a so-called \emph{hairy black hole} (see Section~\ref{holo.intro}), namely a stationary black hole with non-trivial matter fields (the ``hair(s)'') on the horizon (see \cite{Review2,Review3,Review1} and references therein for a review of various types of hairy black holes). \blue{While our result holds\footnote{\blue{We also note the existence of exponentially growing modes for charged scalar fields on Reissner--Nordstr\"om-de Sitter and Kerr--Newman--de Sitter obtained in \cite{BessetHafner}, which might suggest compatibility with the existence of hairy black hole solutions in the $\Lambda>0$ setting.}} for any $\Lambda\in \RR$, in the $\Lambda<0$ case, a large class of asymptotically Anti-de-Sitter (AdS) hairy black holes  has recently been constructed \cite{Weihao1,Weihao2} and their interior region corresponds to the spacetime we construct in our Theorem~\ref{thm.intro} (see Section~\ref{holo.intro}). Although} we only expect to be able to construct an exterior region to our spacetime in the asymptotically AdS case, we nonetheless name our spacetimes hairy black holes for all choices of $\Lambda \in \R$.

Furthermore, \blue{for $\Lambda=0$}, our construction from Theorem~\ref{thm.intro} has bearings on the interior of  (non-hairy) \emph{asymptotically flat black holes} as well, as we explain in Section~\ref{non-hairy.intro}. To summarize: the domain of dependence property allows to consider the black hole interior region independently from the black hole exterior (see Figure~\ref{construction}); thus the repercussions of Theorem~\ref{thm.intro} extend significantly beyond asymptotically AdS hairy black holes.



\subsection{Rough version of our main result} \label{roughthm.intro}

Before explaining the relevance to Conjectures~\ref{conj.BH} and \ref{conj.cosmo} of these novel hairy black hole interiors that we construct, we will first give a rough (but detailed) version of our main result immediately below. To this effect, we recall the well-known interior region of the Reissner--Nordstr\"{o}m-(dS/AdS) black hole (see also Section~\ref{sub:reissner_nordstrom}).$$g_{RN} = - \left ( 1 - \frac{2M}{r} + \frac{\mathbf{e}^2}{r^2} - \frac{\Lambda r^2}{3}\right ) dt^2 + \left ( 1 - \frac{2M}{r} + \frac{\mathbf{e}^2}{r^2} - \frac{\Lambda r^2}{3} \right )^{-1} dr^2 + r^2 d \sigma_{\mathbb{S}^2}$$ with parameters $(M, \mathbf{e},\Lambda)$. The hairy black hole constructed in the following Theorem~\ref{thm.intro} has initial data (\ref{rough.data1}), (\ref{rough.data2}) that are $O(\ep^2)$-perturbations of $g_{RN}$, with scalar hair of initial size $\ep$. 

\begin{theo} \textup{[Rough version]} \label{thm.intro}
Fix the following characteristic initial data on bifurcate event horizons $\mathcal{H}_L \cup \mathcal{H}_R$: \begin{align}
& \phi \equiv \epsilon,\label{rough.data1}\\ & g= g_{RN} \color{black} + O(\epsilon^2), \label{rough.data2}
\end{align} where $ g_{RN}$ is a Reissner--Nordstr\"{o}m-(dS/AdS) metric with sub-extremal parameters $(M, \mathbf{e},\Lambda)$ \blue{and $\mathbf{e} \neq 0$}. 

Define  $(\mathcal{M}=\RR\times (-\infty,s_{\infty}) \times \mathbb{S}^2,g,F,\phi)$ to be the maximal globally hyperbolic future development of this data. $(\mathcal{M}, g)$ is a spatially homogeneous, spherically symmetric spacetime, and we write $g$, $F$, and $\phi$ in a suitable gauge as
\begin{equation}\label{MGHD}
g= \Omega^2(s) \blue{[- ds^2+ dt^2]}   + r^2(s) d\sigma_{\mathbb{S}^2},\; F = \frac{Q(s)}{r^2(s)}\Omega^2(s) ds \wedge dt,\; \phi=\phi(s),
\end{equation} 
solving  \eqref{E1}--\eqref{E5}  with $q_0\neq 0$  and initial data given by \eqref{rough.data1},  \eqref{rough.data2} and $Q_{|\mathcal{H}_L \cup \mathcal{H}_R}  \equiv   \blue{\mathbf{e} \neq 0}$.


Let $\eta > 0$ be  sufficiently small. Then there exists  $\ep_0(M,\mathbf{e},\Lambda,m^2,q_0,\eta)>0$ and a set $E_{\eta} \subset (- \ep_0, \ep_0) \setminus \{0\}$, satisfying $\frac{|(-\delta, \delta) \setminus E_{\eta}|}{2  \delta} = O(\eta)$ for all $0 < \delta \leq \ep_0$, such that for all $\ep \in E_{\eta}$,  the spacetime $(\mathcal{M},g)$ terminates at a \textbf{spacelike singularity $\mathcal{S}=\{r=0\}$}, asymptotically described by a Kasner metric of positive exponents $(p_1,p_{2},p_{3}) \in (0,1)^3$.
The spacetime $(M, g)$ may be partitioned into several regions, as illustrated by the Penrose diagram of Figure~\ref{Penrose_simplified}, and  has the following features:
\begin{enumerate}
\item \label{I1} \underline{Almost formation of a Cauchy horizon}.\ \label{thm.a}In the early regions, $(g,F,\phi)$ are uniformly close to those of the Reissner--Nordstr\"{o}m-(dS/AdS) background, and  the scalar field is approximated by a linearly-oscillating profile: 
\begin{equation}\label{linear.osc}
\phi(s) = B(M, \mathbf{e},\Lambda,m^2,q_0)\cdot \ep \cdot e^{i\omega_{RN}(q_0,M,\mathbf{e},\Lambda) s} + \overline{B}(M,\mathbf{e},\Lambda,m^2,q_0) \cdot \ep \cdot e^{-i \omega_{RN}(q_0,M,\mathbf{e},\Lambda)  s}+ O(\ep^2),
\end{equation} 
where $B(M,\mathbf{e},\Lambda,m^2,q_0) \in \mathbb{C}\setminus\{0\}$ is a linear scattering parameter, and \red{$\omega_{RN}(M,\mathbf{e},\Lambda,q_0)= |q_0 \mathbf{e}| \cdot (\frac{1}{r_-}-\frac{1}{r_+}) > 0$.} Here $r_{\pm}(M,\mathbf{e},\Lambda)>0$ are respectively the radii of the event and Cauchy horizons of the background Reissner--Nordstr\"{o}m-(dS/AdS) metric as defined in Section~\ref{sub:reissner_nordstrom}.

\item \label{I2}\underline{\blue{Collapsed} oscillations}.\ \label{thm.b} The scalar field experiences \textbf{growing} oscillations while $r$ shrinks towards $0$ \blue{in the sense described in Section~\ref{osc.intro}.}  \blue{More precisely, there exists  $C(M,\mathbf{e}, \Lambda, m^2, q_0) \neq 0$ and $\omega_0(M,\mathbf{e},\Lambda,m^2,q_0)>0$ and 
\begin{equation} \label{alpha.def.intro}
\alpha(\epsilon)= C \cdot \sin( \omega_0\cdot \ep^{-2}+ O(\log (\ep^{-1})))+ O(\ep^2\log(\ep^{-1})) \quad \text{ as } |\epsilon|\rightarrow 0,
\end{equation}
so that, on the future boundary of the collapsed oscillations region}    \begin{equation}\label{Jo.intro}
\phi =  \alpha(\epsilon)  \text{ and }  r \approx \ep.
\end{equation}  Note that $\phi_{|r \approx \ep} =O(1)$ despite \blue{the $\phi$ initial} data being $O(\ep)$.

\item \label{I3} \underline{Formation of the first Kasner regime}.\ \label{thm.c}A Kasner regime starts developing with Kasner exponents \begin{equation}\label{first.Kasner}
p_1 = P(\alpha) =\frac{\alpha^2-1}{3+\alpha^2},\quad p_2=p_3=\frac{2}{3+\alpha^2},
\end{equation} \blue{where $p_1$ is the Kasner exponent in the $t$-direction and $p_2=p_3$ are the exponents in the $\mathbb{S}^2$ direction\blue{s} in \eqref{MGHD}.}
\item \label{I4} \label{thm.d}\underline{The final Kasner regime}. For any \blue{fixed}  $\sigma\in(0,1)$ , we introduce the following disjoint subsets of $E_{\eta}$. \begin{itemize}

\item (Non-\blue{bounce} case)	 $\ep\in  E_{\eta,\sigma}^{'\ \blue{Nbo}} \subset  E_{\eta}$ if $|\alpha(\ep)|\geq 1+\sigma$ (i.e.\  $p_1>0$). When $\ep \in E_{\eta,\sigma}^{'\ \blue{Nbo}} $,  the first Kasner regime is also the final  Kasner regime and continues all the way to $\mathcal{S}=\{r=0\}$ in a monotonic fashion.
\item (Kasner \blue{bounce} case)   $\ep\in E_{\eta,\sigma}^{'\ \blue{bo}} \subset  E_{\eta}$ if $ \eta \leq |\alpha(\ep)| \leq 1-\sigma $. We have $|  E_{\eta,\sigma}^{'\ \blue{bo}}| >0$ (in particular $E_{\eta,\sigma}^{'\ \blue{bo}}\neq \emptyset$).   When $\ep\in  E_{\eta,\sigma}^{'\ \blue{bo}} $, the above Kasner regime eventually transitions towards a (different) final Kasner regime with positive Kasner exponents \begin{equation}\label{last.Kasner}
p_1 = P\left(\frac{1}{\alpha}\right) =\frac{1-\alpha^2}{1+3\alpha^2},\quad p_{2}=p_{3}=\frac{2\alpha^2}{1+3\alpha^2}.
\end{equation}
The final Kasner regime then continues all the way to $\mathcal{S} = \{ r = 0 \}$ in a monotonic fashion.

\end{itemize}

\item \label{charge.retention}  \label{I5}\underline{Charge retention of the Kasner singularity}.\ \label{thm.e}The charge $Q(\blue{s})$ admits a limit $Q_{\infty} \neq 0$ as $r\rightarrow 0$. More precisely \begin{align}
& \label{retention1} Q_{\infty} = (1 - \blue{\nu}(M,\mathbf{e},\Lambda)) \cdot \mathbf{e}  + O(\epsilon) \text{ with } \blue{\nu}(M,\mathbf{e},\Lambda) \in \left( 0,\frac{1}{2} \right) ,\\ & \text{ and } \blue{\nu}(M,\mathbf{e},\Lambda)  \left\{
\begin{array}{ll}
= \frac{1}{4} & \mbox{if } \Lambda=0, \\
\in (0,\frac{1}{4}) & \mbox{if } \Lambda>0, \\
\in (\frac{1}{4},\frac{1}{2}) & \mbox{if } \Lambda<0.
\end{array}
\right.
\end{align}

\end{enumerate}

\end{theo}
\blue{In the above  Theorem~\ref{thm.intro}, a Kasner regime with Kasner exponents $(p_1,p_2,p_3=p_2)$ means the metric takes the following approximate spherically symmetric form, where $\chi, \mathcal{R}> 0$ are constants and $\mathcal{E}_{R}(\tau), \mathcal{E}_{R}(\tau)$ are  errors with a \red{small amplitude of size} $O(\epsilon^2)$, which moreover decay as $\tau \rightarrow 0$ \red{(see already \eqref{error.bound} in Theorem~\ref{maintheorem2})}:  \begin{equation}\label{Kasner.like}
g = -d\tau^2 + \chi \cdot (1+\mathcal{E}_{X}(\tau))\tau^{2p_1} dt^2 + \mathcal{R} \cdot (1+\mathcal{E}_{R}(\tau))   \tau^{2p_2}(r_-^2 d\sigma_{\mathbb{S}^2}),
	\end{equation}   \red{Metrics of the form \eqref{Kasner.like} are called ``Kasner-like''}, see Section~\ref{cosmo.intro}.} Theorem~\ref{thm.intro} is a rough version of our main results later stated as Theorem~\ref{maintheorem} and Theorem~\ref{maintheorem2}. The content of Theorem~\ref{maintheorem} covers the statements \ref{I1}, \ref{I2} and \ref{I5} in Theorem~\ref{thm.intro}, together with some preliminary estimates towards the statements \ref{I3}-\ref{I4}. Theorem~\ref{maintheorem2} is specifically dedicated to Kasner regimes (including the Kasner \blue{bounce} phenomenon) and covers the statements \ref{I3} and \ref{I4} in Theorem~\ref{thm.intro}.
    
\begin{rmk}\label{planar.remark}
While Theorem~\ref{thm.intro} is in the setting of spherically symmetric spatially homogeneous spacetimes (also known as Kantowski--Sachs cosmologies), it also applies to spatially homogeneous spacetimes with planar symmetry. While \eqref{E1}--\eqref{E5} take a slightly different form in the planar symmetric setting, our analysis still carries\footnote{\red{Note, however, that the only known planar symmetric analog of the Reissner--Nordstr\"om black hole is a solution of \eqref{E1}--\eqref{E5} with $\Lambda<0$, so the planar symmetric equivalent of Theorem~\ref{thm.intro} is restricted to the negative cosmological constant setting.}} through, and the conclusions of Theorem~\ref{thm.intro} hold mutatis mutandis, see Remark~\ref{planar.rmk}.
\end{rmk}

\begin{rmk}
The constants $\eta, \sigma > 0$, presumed small throughout this article, are present to ensure that $\alpha(\epsilon)$ is bounded away from $\{0,1\}$. Note that we do not obtain a statement for every sufficiently small $\epsilon$, but only for a set $E_{\eta}$, which features an $O(\eta)$ loss in the sense given above.  However, the methods used to prove Theorem~\ref{thm.intro} allow for $|1-\alpha| \ep^{-0.01}>1$ or  $|\alpha| \ep^{-0.01}>1$, with only minor adjustments. Understanding what happens when say  $|1-\alpha| \ll \ep^{2}$ or  $|\alpha| \ll \ep^{2}$, is an interesting open problem (see Section~\ref{cosmo.intro} and \ref{holo.intro}).
\end{rmk}

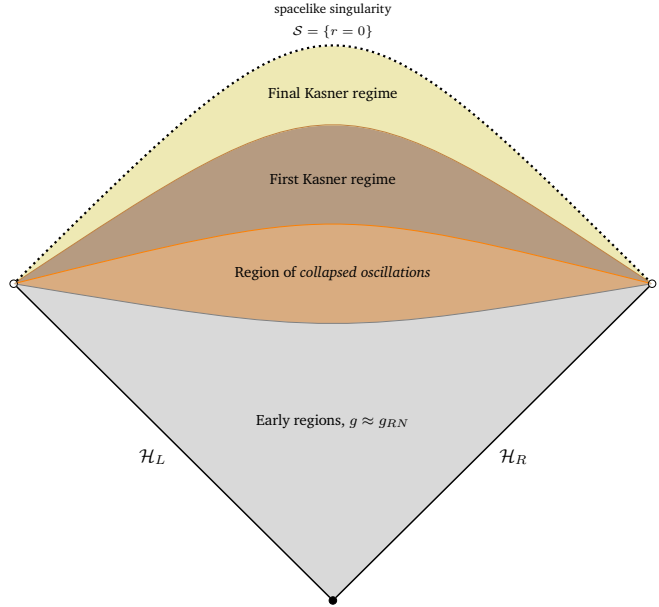
\begin{figure}
\centering
\scalebox{0.7}{
\begin{tikzpicture}
\path[fill=gray, opacity=0.3] (0, -6) -- (-6, 0)
.. controls (0, -1) .. (6, 0) -- (0, -6);
\path[fill=orange!70!black, opacity=0.5] (6, 0)
.. controls (0, 1.5) .. (-6, 0)
.. controls (0, -1) .. (6, 0);
\path[fill=brown!70!black, opacity=0.6] (6, 0)
.. controls (0, 1.5) .. (-6, 0)
.. controls (0, 4) .. (6, 0);
\path[fill=yellow!80!black, opacity=0.4] (6, 0)
.. controls (0, 4) .. (-6, 0)
.. controls (0, 6) .. (6, 0);

\node (p) at (0, -6) [circle, draw, inner sep=0.5mm, fill=black] {};
\node (r) at (6, 0) [circle, draw, inner sep=0.5mm] {};
\node (l) at (-6, 0) [circle, draw, inner sep=0.5mm] {};

\draw [thick] (p) -- (r)
node [midway, below right] {$\mathcal{H}_R$};
\draw [thick] (p) -- (l)
node [midway, below left] {$\mathcal{H}_L$};
\draw [color=gray]
(l) .. controls (0, -1) .. (r);
\draw [orange] (l) .. controls (0, 1.5) .. (r);
\draw [brown] (l) .. controls (0, 4) .. (r);
\draw [very thick, dotted] (l) .. controls (0, 6) .. (r)
node [midway, above, align=center] {\scriptsize spacelike singularity \\ \scriptsize $\mathcal{S} = \{ r = 0 \}$};

\node at (0, -2.6) {\footnotesize Early regions, $g \approx g_{RN}$};
\node at (0, +0.2) {\footnotesize Region of \textit{\blue{collapsed} oscillations}};
\node at (0, +1.95) {\footnotesize First Kasner regime};
\node at (0, +3.6) {\footnotesize Final Kasner regime};
\end{tikzpicture}}
\caption{Penrose diagram of the hairy black hole interior from Theorem \ref{thm.intro}. If $|\alpha| \geq 1 + \sigma > 1$, then the first Kasner regime matches the final Kasner regime and continues to $\{ r = 0 \}$. If $0 < \eta \leq |\alpha| \leq 1 - \sigma < 1$, then there is a Kasner \blue{bounce} between the first and final Kasner regimes. A more detailed breakdown is given in Figure~\ref{Penrose_detailed}.}
\label{Penrose_simplified}
\end{figure}

We will now discuss the relations between Theorem~\ref{thm.intro}, Problem~\ref{PbI} and   Problem~\ref{PbII}. The following paragraphs provide short summaries of the subsequent sections in the introduction. One of the main features is the existence of a \emph{Kasner \blue{bounce}}: we provide the first rigorous example of a spacetime in which \textbf{an unstable Kasner regime forms dynamically, and then disappears under \blue{the} Kasner \blue{bounce}}. 

\paragraph{Differences and connections with non-hairy  black holes} The hairy black holes of Theorem~\ref{thm.intro} and the uncharged-matter ones from \cite{VDM21} have a spacelike singularity $\mathcal{S}=\{r=0\}$ and no Cauchy horizon (see Section~\ref{non-hairy.intro}) and, in that, differ globally from the non-hairy black holes 
corresponding to Conjecture~\ref{conj.BH} (compare Figure~\ref{Fig1} and Figure~\ref{Penrose_simplified}). However, the domain of dependence property (in two-ended black holes\footnote{The differences between one-ended and two-ended black holes in Figure~\ref{Fig1} will be elaborated upon in Section~\ref{non-hairy.intro}. For now, let us note that the \emph{appearance of the spacelike singularity}, i.e.\ $\mathcal{S}\neq \emptyset$, crucially depends on whether the black hole is one-ended or two-ended \cite{r0}. However, numerics indicate that the quantitative behavior near the spacelike singularity is  similar in the one or two ended case.}) shows that the spacelike singularity $\mathcal{S}$  in non-hairy black holes can arise from initial data 
that is locally similar to the data of Theorem~\ref{thm.intro}, see Figure~\ref{construction}. Therefore, our new theorem on hairy black holes may also dictate the qualitative behavior of the spacelike singularity $\mathcal{S}$ inside the \emph{non-hairy black holes} of Conjecture~\ref{conj.BH} (see Section~\ref{non-hairy.intro} for a discussion). Lastly,  our proof opens the door to studying spacelike singularity  beyond the spatially homogeneous case in subsequent works -- notably in spherical symmetry -- see Open Problem~\ref{open1}.

\paragraph{Comparison with hairy black hole interiors for other matter models} 
The Kasner exponents' dependency on $\sin\left(\omega_0\cdot \ep^{-2}+O(\log (\ep^{-1}))\right)$ found in Theorem~\ref{thm.intro}  gives rise to fluctuations near $\epsilon=0$:  we term this phenomenon the \emph{fluctuating collapse}. 
The fluctuating collapse and Kasner \blue{bounces} from Theorem~\ref{thm.intro} contrast with the \emph{violent nonlinear collapse} of the \emph{uncharged} hairy black holes (\eqref{E1}-\eqref{E5} with $q_0=0$, also a stiff model) found in \cite{VDM21} (see Section~\ref{other.intro}). Recall from \cite{VDM21} that in the $q_0=0$ case, \emph{there is no Kasner \blue{bounce}}, and the final Kasner exponents and curvature are of the form $$ (p_1,p_{2},p_{3})=(1-O(\ep^2),\ O(\ep^2),\ O (\ep^2)),\  K(r) \approx r^{-O(\ep^{-2})}.$$  The name ``violent collapse'' in the $q_0=0$ comes from the $O(\ep^{-2})$ power in the blow-up rate of the curvature $K(r)$, which becomes \emph{more singular} as $\epsilon \rightarrow0$. This is an example of a singular limit, since $\ep=0$ corresponds to Reissner--Nordstr\"{o}m, while we can make sense of \blue{the} $\ep\rightarrow 0$ (weak) limit  in the appropriate region using  $$ \lim_{\ep\rightarrow 0} (p_1,p_{2},p_{3})(\ep)= (1,0,0) \text{,    which corresponds to a subset of Minkowski}.$$ In contrast, in  the $q_0 \neq 0$ case newly studied in our Theorem~\ref{thm.intro}, there is \emph{not even a weak limit} as $\ep\rightarrow 0$, since  $(p_1,p_{2},p_{3})(\ep)$ depends on $\sin(\omega_0 \cdot \ep^{-2} + O(\log (\ep^{-1})))$ as discussed above. Also, in the $q_0\neq 0$ case, any neighborhood of $\ep=0$ contains all (positive) Kasner exponents (away from the degenerate cases), whereas in the $q_0=0$ case of \cite{VDM21} a neighborhood of $\ep=0$ is mapped into a neighborhood of  $(p_1,p_{2},p_{3})=(1,0,0)$.

\blue{Lastly, we mention that previous numerics in the physics literature also predicted the presence of Kasner bounces in the interior of Einstein--Yang-Mills hairy black holes \cite{Galtsov1,Maison,Review1}. It is a very interesting open problem to compare these Kasner bounces with the one obtained in the setting of Theorem~\ref{thm.intro}.}



\paragraph{Collapsed oscillations and charge retention} The main new phenomenon driving the dynamics of Theorem~\ref{thm.intro} are what we term \emph{\blue{collapsed} oscillations}. Even at the linear level, a (spatially homogeneous) charged scalar field has (infinite) linear oscillations of the form \eqref{linear.osc} near the Cauchy horizon of the Reissner--Nordstr\"{o}m interior spacetime \cite{MoiChristoph}. In the nonlinear setting, such linear oscillations appear in some ``early regions'' in the dynamics (namely $\mathcal{EB}$ and $\mathcal{LB}$  in Figure~\ref{Penrose_detailed}). Subsequently, these oscillations start interacting nonlinearly with the collapse process (as $r$ gets closer to $0$), leading to a Bessel-function type behavior in terms of $\frac{r}{\ep}$. Because $\frac{r}{\ep}$ is a \emph{decreasing} function of time, this Bessel behavior leads to \emph{growth of the scalar field} which transitions from amplitude $\ep$ to amplitude $\alpha(\ep)=O(1)$. We will elaborate on \blue{collapsed} oscillations and their mechanism in Section~\ref{osc.intro}.

Another important question relating to Problem~\ref{PbI} is whether spacelike singularities can retain charge/angular momentum. This issue is puzzling, as the only explicit black hole solution with a spacelike singularity is the Schwarzschild interior, which is uncharged and non-rotating. In point~\ref{thm.e} of Theorem~\ref{thm.intro}, we exhibit a \emph{mechanism of discharge} of the black hole, passing from $\mathbf{e}$ at the event horizon $\mathcal{H} = \mathcal{H}_L \cup \mathcal{H}_R$ to $(1-\delta)\mathbf{e}$ at  the spacelike singularity $\mathcal{S}$, where $\delta \in (0,\frac{1}{2})$. Note that  the \emph{discharge is not complete} and  the spacelike singularity retains a non-zero final charge $(1-\delta)\mathbf{e}$. It is remarkable that, for $\Lambda=0$, the discharge ratio $\delta=\frac{1}{4}$ is independent of the black hole parameters.
To the best of our knowledge, the spacetime of Theorem~\ref{thm.intro} is the first example of a spacelike singularity retaining\footnote{We note that the charged hairy black holes with uncharged matter from \cite{VDM21} also had non-zero charge, but \blue{it is} not dynamical.} charge. The main mechanism of discharge and charge retention occurs at the same time as the \blue{collapsed} oscillations, and the charge varies little past the \blue{collapsed} oscillations region (see Figure~\ref{Fig1}) until the spacelike singularity, see Remark~\ref{retention.rmk}.

\paragraph{Kasner bounces} In Theorem~\ref{thm.intro}, we show that if $|\alpha(\ep)| = C \cdot |\sin(\omega_0 \ep^{-2} +O(\log (\ep^{-1})))|\blue{+O(\ep^2 \log(\ep^{-1}))} < 1$, then there will be a Kasner \blue{bounce}. For $(M,\mathbf{e},\Lambda,m^2, q_0)$ such that $C(M,\mathbf{e},\Lambda,m^2,q_0)>1$, the \blue{bounce} condition $C \cdot |\sin(\omega_0 \ep^{-2} +O(\log (\ep^{-1})))|<1$ holds for all  $\ep \in E_{\eta}^{\blue{bo}} \cap (0,\epsilon_0)$, where $E_{\eta}^{\blue{bo}} $ has non-zero measure: in fact  $|E_{\eta}^{\blue{bo}} \cap (0, \ep_0)|  \approx \frac{\epsilon_0}{C}$ if $C\gg 1$. On the other hand, if $C(M, \mathbf{e}, \Lambda, m^2, q_0) < 1$, then the \blue{bounce} condition always holds.  These Kasner \blue{bounces} have the following features: \begin{enumerate}
\item For  some range of proper time $ \{ \blue{  e^{- b^2(\alpha) \cdot \ep^{-2} } \ll \tau< \ep^{q^2(\alpha)}  }\} $ with $q(\alpha)>0,\ b(\alpha)>0$, the metric is uniformly close to a $(p_1,\frac{1-p_1}{2},\frac{1-p_1}{2})$ Kasner\blue{-(like) metric}, where $p_1(\ep) \approx \frac{\alpha^2-1}{ 3+\alpha^2 }<0$.  
\item In the smallest values of proper time from the singularity $\{ 0 <  \tau \ll  e^{- b^2(\alpha) \cdot \ep^{-2} }\}$, the metric is uniformly close to, and indeed converges as $\tau \rightarrow 0$ towards, a $(\acute{p}_1,\frac{1-\acute{p}_1}{2},\frac{1-\acute{p}_1}{2})$ Kasner\blue{-(like) metric}, where $\acute{p}_1(\ep) \approx \frac{1-\alpha^2}{ 1+3\alpha^2 }>0$.
\end{enumerate} 
\begin{rmk}
When writing (imprecisely) $ e^{- b^2(\alpha) \cdot \ep^{-2}} \ll \tau  $, we mean  $ \ep^{-N_1} \cdot e^{- b^2(\alpha) \cdot \ep^{-2}}< \tau$ for some $N_1>0$, and similarly $\tau \ll  e^{- b^2(\alpha) \cdot \ep^{-2}} $ means  $\tau < \ep^{N_2} \cdot e^{- b^2(\alpha) \cdot \ep^{-2}}$ for $N_2>0$, so that the transition region in between the two Kasner regimes has size $O(\log(\ep^{-1}))$ in terms of $\log(\tau^{-1})$, while $\log(\tau^{-1}) \approx \ep^{-2}$ at the times where the \blue{bounce} is occurring. We interpret this to mean that the \blue{bounce} occurs very quickly in terms of proper time $\tau$.
\end{rmk}

We note that the final (post-\blue{bounce}) Kasner regime has $\blue{(\acute{p}_1,\acute{p}_2,\acute{p}_3)\in(0,1)^3}$,   while the pre-\blue{bounce} Kasner has $p_1<0$. This is consistent with the early predictions of BKL \cite{BK1, BKL1, BKL2} that Kasner metrics with positive Kasner exponents are stable, while those with at least one negative exponent are unstable (see Section \ref{cosmo.intro} for an extended discussion).

\paragraph{Holographic superconductors in the AdS-CFT correspondence}

The  numerical study \cite{hartnolletal2021} already predicted the scenario of Theorem~\ref{thm.intro} \blue{ and the presence of Kasner bounces in this context}. Beyond the setting of Theorem~\ref{thm.intro} (which features zero, or one Kasner \blue{bounce} depending on $\ep$), \cite{hartnolletal2021} also discusses the possibility of having \emph{two} (or more) Kasner \blue{bounces}. The original motivation of \cite{hartnolletal2021} is related to a body of work on the physical significance of the hairy black hole (including its asymptotically AdS exterior region) of Theorem~\ref{thm.intro}, claimed to be the model for a \emph{holographic superconductor} in the context of the AdS-CFT theory  (see Section~\ref{holo.intro}).

\paragraph{Outline of the rest of the Introduction}

\begin{itemize}
\item In Section~\ref{non-hairy.intro}, we will first extend our discussion of Problem~\ref{PbI} and Conjecture~\ref{conj.BH}  and elaborate on the links between  non-hairy black holes and Theorem~\ref{thm.intro}. We will also review, in this context, the existing literature on the interior of black holes, and \blue{discuss some open}  problems.

\item In Section~\ref{other.intro}, we will compare our new hairy black holes with charged matter from Theorem~\ref{thm.intro} to hairy black holes arising from other matter models. We will, in particular, discuss the charged hairy black holes with uncharged matter from \cite{VDM21}, whose setting and model are very similar to that of Theorem~\ref{thm.intro}, but the late-time spacetime dynamics end up being very different.

\item In Section~\ref{osc.intro}, we will discuss  one of the two primary nonlinear mechanisms governing the dynamics of the spacetime of Theorem~\ref{thm.intro}: \emph{the \blue{collapsed} oscillations}, which lead to the growth of the scalar field. We will explain in particular how this phenomenon arises from the interaction between the linear oscillatory behavior for charged scalar fields at the Reissner--Nordstr\"{o}m Cauchy horizon on the one hand, and the tendency of the Einstein equations to form a spacelike singularity $\{r=0\}$ on the other hand.

\item In Section~\ref{cosmo.intro}, we will discuss the other primary nonlinear mechanism at play: the occurence of \emph{Kasner \blue{bounces}}. This phenomenon has previously been investigated in depth in the cosmological setting, and we will elaborate on the connection with the pre-existing literature regarding such phenomena.

\item In Section~\ref{holo.intro}, we will discuss the \emph{exterior region} corresponding to the black hole interior  of Theorem~\ref{thm.intro}, the most physically relevant case of which is asymptotically Anti-de-Sitter. We will emphasize the physical motivation in studying these spacetimes, called \emph{holographic superconductors}, which have been discovered and studied in the high-energy physics literature most notably in the context of the AdS-CFT correspondence.

\item In Section~\ref{outline.sec}, we give an outline of the paper and introduce the different regions of Figure~\ref{Penrose_detailed}.

\end{itemize}

\subsection{Differences and connections with non-hairy  black holes} \label{non-hairy.intro}

The spacetimes constructed in Theorem~\ref{thm.intro} are spatially homogeneous and we interpret them as the interior region of so-called ``hairy black holes''. The main distinctive feature of ``hairy black holes'' is the presence of so-called \emph{scalar hair}, meaning that the scalar field $\phi$ in \eqref{E5} \emph{does not decay on the event horizon} $\mathcal{H}^+$ and tends to a non-zero constant instead. (In the case of the hairy black holes of Theorem~\ref{thm.intro}, $\phi$ is identically equal to this constant on  $\mathcal{H} = \mathcal{H}_L \cup \mathcal{H}_R$, see \eqref{rough.data1}). 

\blue{We have already explained that it is not possible to obtain an asymptotically flat spherically symmetric hairy black hole exterior solution of \eqref{E1}--\eqref{E5} with $\Lambda=0$; however, as we will discuss in Section~\ref{holo.intro}, the construction of a such   hairy black hole exterior is possible for  asymptotically-AdS data, with $\Lambda < 0$ in \eqref{E1}.}
 \blue{Such a construction will necessarily involve a static  exterior region, while the corresponding interior is spatially homogeneous, see Section~\ref{prelim.section} for a proof of this claim. This is due to the fact that for a $t$-independent black hole, $t$ is a timelike coordinate in the exterior which, by definition, becomes a spacelike coordinate in the black hole interior.}

Nonetheless, one may also wish to consider asymptotically flat (where $\Lambda = 0$) spacetimes which are relevant to the study of astrophysical black holes. In this context, one anticipates that solutions of 
\eqref{E1}--\eqref{E5} with regular Cauchy data \emph{decay} towards a Reissner--Nordstr\"{o}m exterior solution, in particular $\phi$ tends to $0$ on $\mathcal{H}^+$ (in spherical symmetry, see \cite{PriceLaw,JonathanStabExt} for  \eqref{E1}--\eqref{E5} with $q_0=m^2=0$, and also \cite{Moi2} for small $|q_0\mathbf{e}|$ on a fixed  Reissner--Nordstr\"{o}m exterior). The resulting black holes thus feature $\phi$ decaying on  $\mathcal{H}^+$ at the following rate: for all $v>1$, \begin{equation}\label{decay.s}
|\phi|_{|\mathcal{H}+}(v)+  |D_v\phi|_{|\mathcal{H}+}(v) \lesssim v^{-s},
\end{equation} where $v$ is a standard Eddington--Finkenstein type advanced-time coordinate and $s>\frac{1}{2}$. We will call such black holes ``non-hairy'' in the sequel, to mark the contrast with the black holes from Theorem~\ref{thm.intro}.

The interior of the non-hairy black hole solving \eqref{E1}-\eqref{E5} in spherical symmetry was studied in \cite{MoiChristoph,Moi,MoiThesis,r0,Moi4} and their Penrose diagram was completely characterized (modulo issues related to locally naked singularities, see the second paragraph below). In this section, we will briefly explain these results and provide contrast with our hairy black hole interiors from Theorem~\ref{thm.intro}. We also comment that one can use our hairy black hole interiors as a tool to retrieve information on non-hairy black holes, connecting Theorem~\ref{thm.intro} to Conjecture~\ref{conj.BH}, see  the third paragraph below.

\paragraph{Local structure of the non-hairy black hole interior near infinity \texorpdfstring{$i^+$}{i+}}

We discuss the terminal boundary of the black hole interior. In the case of the hairy black holes of Theorem~\ref{thm.intro}, it is entirely spacelike $\mathcal{S}=\{r=0\}$. In contrast, for (spherically symmetric) non-hairy black holes solutions of, the terminal boundary is \emph{not} entirely spacelike and admits a null component near $i^+$ -- called the Cauchy horizon $\mathcal{CH}^+ \neq \emptyset$. This fact constitutes the most important difference between hairy and non-hairy black holes.

\begin{thm}[\cite{Moi}] \label{CH.thm} Consider regular spherically symmetric characteristic data on  $\mathcal{H}^+ \cup \underline{C}_{in}$, where $\mathcal{H}^+:= [1,+\infty)_v  \times \mathbb{S}^2$, converging to a sub-extremal Reissner--Nordstr\"{o}m exterior as \eqref{decay.s} . Then, restricting $\underline{C}_{in}$ to be sufficiently short, the future domain of dependence of $\mathcal{H}^+ \cup \underline{C}_{in}$  is bounded by a Cauchy horizon $\mathcal{CH}^+$, namely a null boundary emanating from $i^+$ and foliated by spheres of strictly positive area-radius $r$, as depicted in Figure~\ref{Fig.relax}.
\end{thm}

\begin{figure}

\begin{center}

\includegraphics[width= 58 mm, height=40
mm]{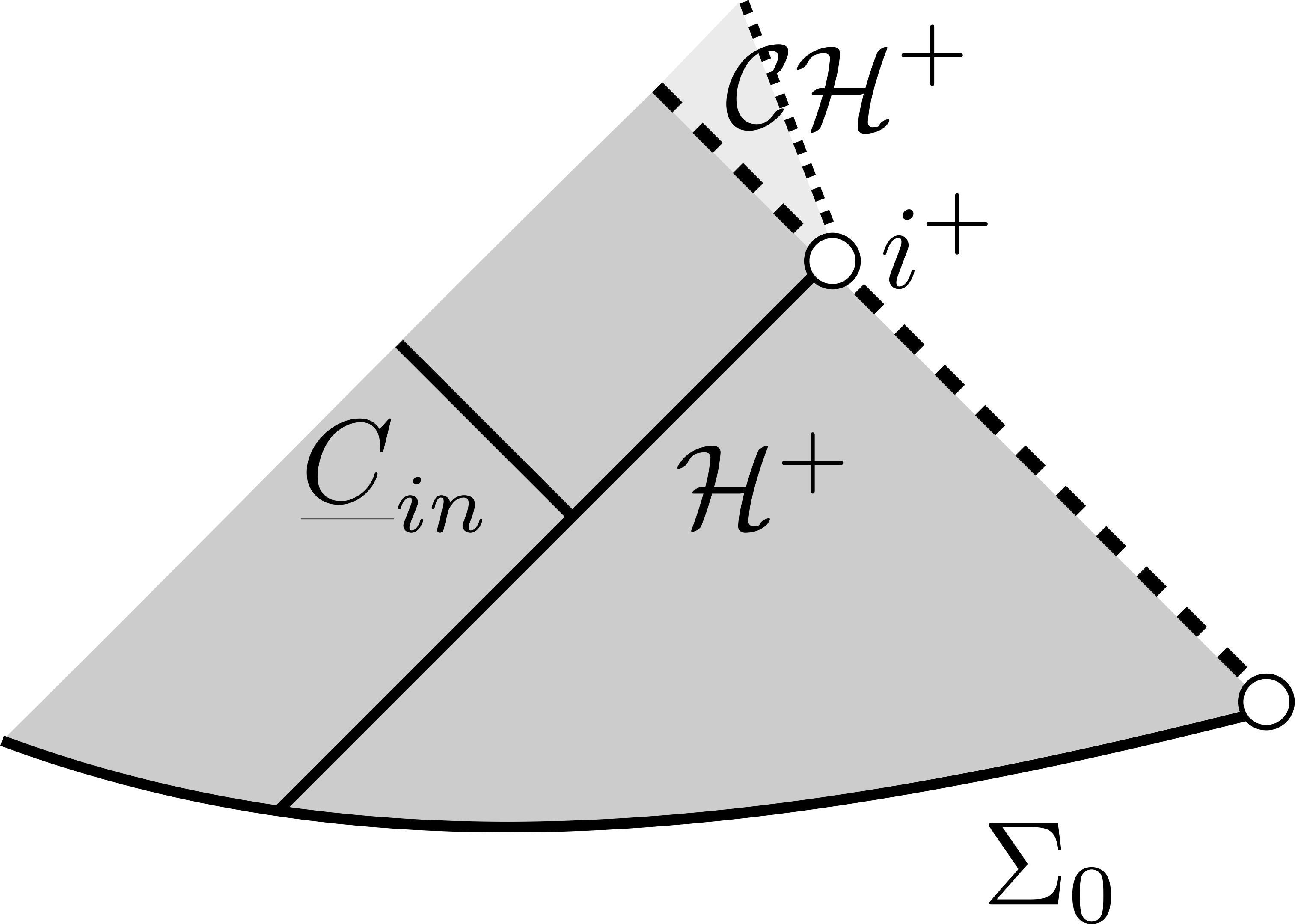}

\end{center}

\caption{Penrose diagram of the spacetime corresponding to Theorem~\ref{CH.thm}}
\label{Fig.relax}
\end{figure}

\begin{rmk}
We note that the presence of a Cauchy horizon $\mathcal{CH}^+ \neq \emptyset$ in the interior of dynamical black holes is not specific to spherical symmetry: for instance, it has been obtained for perturbations of Kerr for the Einstein vacuum equations (without symmetry) in \cite{KerrStab}. Whether Theorem~\ref{CH.thm} can be generalized to \eqref{E1}-\eqref{E5} without spherical symmetry is an interesting open problem, in view of the particularly slow decay of the form \eqref{decay.s} one has to assume in the presence of matter (see \cite{MoiChristoph2,JonathanICM,MoiThesis} for an extended discussion of slow decay).
\end{rmk}

\paragraph{Global structure of non-hairy black hole interiors}

Theorem~\ref{CH.thm} only gives information on a local region located near $i^+$ (see Figure~\ref{Fig.relax}). The \emph{global nature of the terminal boundary}, as it turns out (see Theorem~\ref{r0.thm}), depends on the topology of the initial data; we distinguish two important cases: \begin{enumerate}[a.]
\item The two-ended \blue{case} (topology of time-slices: $\mathbb{R}\times \mathbb{S}^2$). The maximally-extended Schwarzschild/Reissner-Nordstr\"{o}m/Kerr spacetimes, and our hairy black holes from Theorem~\ref{thm.intro} possess the two-ended topology.

\item\label{b} \blue{The} one-ended \blue{case} (topology of time-slices: $\mathbb{R}^3$) \blue{is the} topology suitable to studying  the gravitational collapse of a star into a black hole \cite{Christo1,Kommemi,JonathanICM,r0} (referred to as ``gravitational collapse'' for short).
\end{enumerate}
\red{In the one-ended case \ref{b}, the following  is known about the terminal boundary. }

\begin{thm}[Black hole interior in gravitational collapse, \cite{r0,Moi2,Moi4}] \label{r0.thm} We consider a one-ended black hole interior, under the assumptions of Theorem~\ref{CH.thm} and additional inverse-polynomial lower bounds on $\phi$ consistent with \eqref{decay.s}. Then, assuming the absence of locally naked singularities emanating from the center $\Gamma$, there is a (non-empty) spacelike singularity $\mathcal{S}=\{r=0\}$ and the Penrose diagram is given by the left-most of Figure~\ref{Fig1}.
\end{thm}
\begin{rmk} \label{loc.naked.rmk}
\blue{A} locally naked \blue{singularity} is an (outgoing) null boundary $\mathcal{CH}_{\Gamma}$ emanating from the center $\Gamma$. Assuming their absence in Theorem~\ref{r0.thm} is unavoidable, since examples of locally naked singularities have been constructed \cite{Christo.existence} for \eqref{E1}-\eqref{E5}. However, for \eqref{E1}-\eqref{E5} with $F\equiv 0$, such locally naked singularities are non-generic within spherical symmetry \cite{Christo2,Christo3} and one may conjecture the same statement in the more general situation where $F\neq 0$. See \cite{Kommemi,r0} for an extended discussion of this delicate issue.  
\end{rmk}

For two-ended black holes (irrelevant, however, to gravitational collapse), Theorem~\ref{r0.thm} is false   because small perturbations of Reissner--Nordstrom obeying \eqref{decay.s} feature no spacelike singularity \cite{nospacelike}. However, it is conjectured \cite{nospacelike,KerrDaf,Kommemi} that, even in the two-ended case, large perturbations would yield a spacelike singularity $\mathcal{S}=\{r=0\} \neq \emptyset$ and a Penrose diagram corresponding to the rightmost in Figure~\ref{Fig1}.

\begin{rmk}\label{domain.rmk}
Note that for a two-ended black hole as in the rightmost Penrose diagram of Figure~\ref{Fig1},  the causal past of any compact subset of $\mathcal{S}$ intersects the event horizon $\mathcal{H}^+$ on a set with compact closure. This observation will be important in the discussion of the next paragraph,  see Figure~\ref{construction}.
\end{rmk}

We conclude this section by mentioning previous works providing a pretty detailed characterization of spacelike singularities in spherical symmetry  \cite{DejanAn,AnZhang,Christo4,Christo1,Christo2,Christo3} in the uncharged case (i.e.\ \eqref{E1}-\eqref{E5} with $F\equiv 0$). 


\paragraph{Connections between hairy and non-hairy black holes and related open problems}

We come back to our original motivation: Conjecture~\ref{conj.BH} and understanding spacelike singularities inside black holes. Our goal is to construct a large class of (asymptotically flat) black holes with \emph{both} a spacelike singularity $\mathcal{S}=\{r=0\}$ and a null Cauchy horizon $\mathcal{CH}^+$, with precise quantitative information on (at least part of)  $\mathcal{S}=\{r=0\}$. 

\blue{It is an open problem to provide a quantitative description of the spacelike singularity inside any charged\footnote{Note, however, that it is possible to describe asymptotically Schwarzschild one-ended uncharged spherically symmetric black holes \cite{DejanAn}.}  or \red{rotating} black hole  (note indeed that Theorem~\ref{r0.thm} does not describe $\mathcal{S}$ quantitatively). In the two-ended case, however, we propose a construction of a black hole featuring both a null Cauchy horizon $\mathcal{CH}^+\neq \emptyset$, and a spacelike singularity $\mathcal{S}$ which is partially described  by the \red{terminal} boundary in Theorem~\ref{thm.intro}.}



 Parameterize the two event horizons by $\mathcal{H}_R=\{(-\infty,v),\ v \in \RR\}$, $\mathcal{H}_L=\{(u,-\infty),\ u \in \RR\}$. \begin{enumerate}[i.]
\item \label{G1} Fix $\phi_{|\mathcal{H}^+_R}(v)\equiv \ep$ for $v \leq A$, and $\phi_{|\mathcal{H}^+_L}(v)\equiv \ep$ for $u \leq A$, \blue{for some large $A$ depending on $\ep$}.  \blue{We evolve (non-uniquely)}  this characteristic data on $(\mathcal{H}_R \cap \{ v\leq A\})  \cup  (\mathcal{H}_L \cap \{ u\leq A\})$ towards the past   \blue{in the following fashion: \begin{itemize}
		\item Attach a small past-directed ingoing (respectively outgoing) cone $\underline{C}_R^+$ (respectively $C_L^+$)  to the future-endpoint of $\mathcal{H}_R \cap \{ v\leq A\})$  (respectively $(\mathcal{H}_L \cap \{ u\leq A\})$).
		\item By local  well-posedness, the  past domain of influence of $\mathcal{H}_R \cap \{ v\leq A\})\cup \underline{C}_R^+$ (respectively $\mathcal{H}_L \cap \{ y\leq A\}\cup {C}_L^+$) is the causal rectangle with past ingoing boundary $\underline{C}_R^-$  (respectively $C_L^{-}$).
		\item Evolving the initial data on $\underline{C}_R^- \cup C_L^{-}$ towards the past and appealing to local well-posedness again provides another smaller causal rectangle. This rectangle combines with the two previously-constructed rectangles into the (future) domain of influence of some $C_{out} \cup \underline{C}_{in}$.
\end{itemize}} 
We  obtain a (non-unique) solution to \eqref{E1}--\eqref{E5} up to the bifurcate null cones $C_{out} \cup \underline{C}_{in}$ (see Figure~\ref{construction}).

\item \label{G2}Extend $C_{out}$ (respectively $\underline{C}_{in}$) into an outgoing (respectively ingoing) cone which is asymptotically flat using a \blue{standard truncation} argument \blue{(see  \cite{Christo.existence,nakedigoryakov})}. The bifurcate null cones $\tilde{C}_{out} \cup \underline{\tilde{C}}_{in}$ thus obtained  intersect (what should be thought of as) future null infinity $\mathcal{I}^+=\mathcal{I}_L^+ \cup \mathcal{I}_R^+$, and $\phi_{|\tilde{C}_{out} \cup \underline{\tilde{C}}_{in}}$ decays towards $\mathcal{I}^+$.

\item \label{G3}Solve \emph{forward} for the above  characteristic data on $\tilde{C}_{out} \cup \underline{\tilde{C}}_{in}$. By the domain of dependence property, the  following spacetime region (consisting of the dark grey, orange, brown and yellow regions in Figure~\ref{construction})  $$\{ (u,v): u \leq A,\ v\leq A\}\cap \mathcal{D}^+((\mathcal{H}_R \cap \{ v\leq A\})  \cup  (\mathcal{H}_L \cap \{ u\leq A\})), \text{   where }  \mathcal{D}^+  \text{ denotes the domain of dependence }$$ is isometric to a  subset of the hairy black hole $g_{\ep}$ of Theorem~\ref{thm.intro} that contains a large portion of the spacelike singularity $\mathcal{S}$ (in green in Figure~\ref{construction}).  

\item \label{G4} Using that  $\phi_{|\tilde{C}_{out} \cup \underline{\tilde{C}}_{in}}$ decays towards $\mathcal{I}^+$,  prove that the decay condition \eqref{decay.s} is satisfied on the new event horizons  $\mathcal{H}^*_R$ and  $\mathcal{H}^*_L$ (note that the event horizons $\mathcal{H}^*_R$, $\mathcal{H}^*_L$ for the newly constructed black hole do not coincide with the original event horizons $\mathcal{H}_R$, $\mathcal{H}_L$ of the hairy black hole). As a consequence of Theorem~\ref{CH.thm},  obtain the existence of Cauchy horizons $\mathcal{CH}_L, \mathcal{CH}_R \neq  \emptyset$.	

\end{enumerate}

\begin{figure}[ht]
\centering

\scalebox{0.6}{
\begin{tikzpicture}
\path[fill=gray, opacity=0.5] (0, -6) -- (-6, 0)
.. controls (0, -1) .. (6, 0) -- (0, -6);
\path[fill=orange!70!black, opacity=0.5] (6, 0)
.. controls (0, 1.5) .. (-6, 0)
.. controls (0, -1) .. (6, 0);
\path[fill=brown!70!black, opacity=0.6] (6, 0)
.. controls (0, 1.5) .. (-6, 0)
.. controls (0, 4) .. (6, 0);
\path[fill=yellow!80!black, opacity=0.4] (6, 0)
.. controls (0, 4) .. (-6, 0)
.. controls (0, 6) .. (6, 0);
\draw [color=gray]
(l) .. controls (0, -1) .. (r);
\draw [orange] (l) .. controls (0, 1.5) .. (r);
\draw [brown] (l) .. controls (0, 4) .. (r);
\draw [very thick, dotted] (l) .. controls (0, 6) .. (r)
node [midway, above] {\scriptsize spacelike singularity};
\draw [thick, dashed] (p) -- (r)
node [midway, above left] {$\mathcal{H}_R \cap \{ v  \leq A \}$};
\draw [thick, dashed] (p) -- (l)
node [midway, above right] {$\mathcal{H}_L \cap \{ u \leq A \}$};
\path[fill=white] (0, -6.5) -- (6.75, 0.25) -- (3, 4)
.. controls (2, 4) and (1.5, 4.2) .. (0.5, 4.45)
-- (5.475, -0.525) -- (p) -- (-5.475, -0.525) -- (-0.5, 4.45)
.. controls (-1.5, 4.2) and (-2, 4) .. (-3, 4)
-- (-6.75, 0.25) -- (0, -6.5);
\path[fill=lightgray, opacity=0.3] (0, -6.5) -- (6.75, 0.25) -- (3, 4)
.. controls (2, 4) and (1.5, 4.2) .. (0.5, 4.45)
-- (5.475, -0.525) -- (p) -- (-5.475, -0.525) -- (-0.5, 4.45)
.. controls (-1.5, 4.2) and (-2, 4) .. (-3, 4)
-- (-6.75, 0.25) -- (0, -6.5);

\node (p) at (0, -6) [circle, draw, inner sep=0.5mm, fill=black] {};
\node (s) at (0.5, 4.45) [circle, draw, inner sep=0.5mm, fill=red, red] {};
\node (t) at (-0.5, 4.45) [circle, draw, inner sep=0.5mm, fill=red, red] {};
\node (c) at (0, -8) [circle, draw, inner sep=0.5mm, fill=black] {};
\node (r2) at (7.5, -0.5) [circle, draw, inner sep=0.5mm] {};
\node (l2) at (-7.5, -0.5) [circle, draw, inner sep=0.5mm] {};
\node (h) at (0, -6.5) [circle, draw, inner sep=0.5mm, fill=blue, blue] {};
\node (r3) at (6.75, 0.25) [circle, draw, inner sep=0.5mm, blue] {};
\node (l3) at (-6.75, 0.25) [circle, draw, inner sep=0.5mm, blue] {};
\node (chr) at (3, 4) [circle, draw, inner sep=0.5mm, black] {};
\node (chl) at (-3, 4) [circle, draw, inner sep=0.5mm, black] {};
\node at (r3) [anchor=south west, blue] {$i^+_R$};
\node at (l3) [anchor=south east, blue] {$i^+_L$};

\draw [thick, dashed] (p) -- (1, -7) node [midway, above right] {$\underline{C}_R^-$};
\draw [thick, dashed] (p) -- (-1, -7) node [midway, above left] {${C}_L^-$};
\draw (c) -- (7.5, -0.5) node [midway, below right] {$C_{out}$ extends to $\tilde{C}_{out}$};
\draw (c) -- (-7.5, -0.5) node [midway, below left] {$\underline{C}_{in}$ extends to $\tilde{\underline{C}}_{in}$};
\draw [very thick] (c) -- (6.475, -1.525) node [midway, below right] {};
\draw [very thick] (c) -- (-6.475, -1.525) node [midway, below left] {};
\draw [thick, dashed, blue] (r3) -- node [midway, below right] {$\mathcal{H}^*_R$} (h) -- node [midway, below left] {$\mathcal{H}^*_L$} (l3);
\draw [thick, dashed] (r3) -- (r2) node [midway, above right] {$\mathcal{I}^+_R$};
\draw [thick, dashed] (l3) -- (l2) node [midway, above left] {$\mathcal{I}^+_L$};
\draw [thick, dashed] (r3) -- (chr) node [midway, above right] {$\mathcal{CH}_R$};
\draw [thick, dashed] (l3) -- (chl) node [midway, above left] {$\mathcal{CH}_L$};
\draw [very thick, dotted] (s) .. controls (1.5, 4.2) and (2, 4) .. (chr);
\draw [very thick, dotted] (t) .. controls (-1.5, 4.2) and (-2, 4) .. (chl);
\draw [very thick, red] (s) -- (6.475, -1.525) node [pos=0.925, above right] {$\underline{C}_R^+$};
\draw [very thick, red] (t) -- (-6.475, -1.525) node [pos=0.925, above left] {$C_L^+$};
\draw [very thick, green] (s) .. controls (0, 4.55) .. (t);

\node at (0, -0.3) [align = center] {Region isometric \\ to the hairy \\ black hole metric \\ $g_{\ep}$ from Theorem~\ref{thm.intro}};

\end{tikzpicture}}
\caption{The proposed construction of a two-ended black hole with a spacelike singularity $\mathcal{S}$ via gluing argument. The union of the dark grey, orange, brown and yellow regions (including the green part of the spacelike singularity) is isometric to a subset of the hairy black hole of Figure~\ref{Penrose_simplified}.   $\mathcal{I}_L^+$ and $\mathcal{I}_R^+$ are the  components of null infinity $\mathcal{I}^+$.}
\label{construction}
\end{figure}
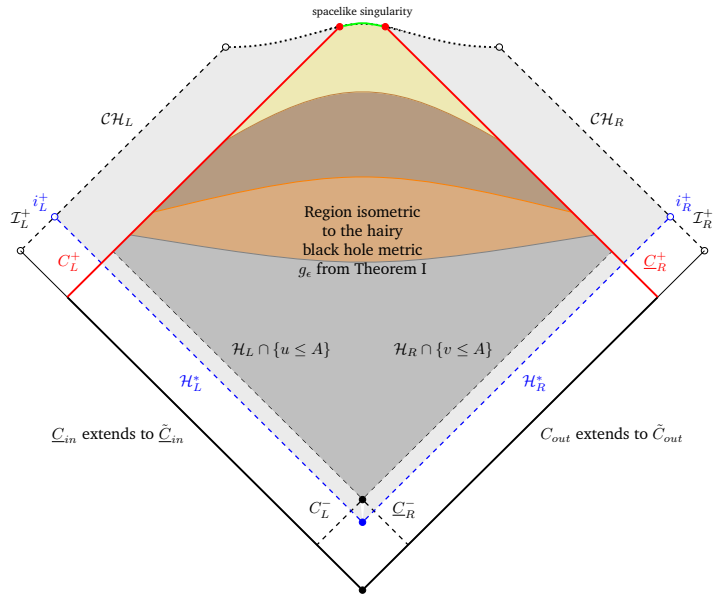

We want to point out that the only unknown step is the proof of \eqref{decay.s}  (step~\ref{G4}) which relies on establishing polynomial decay on the event horizon. We note, furthermore, that for small charge $|q_0 \mathbf{e}|\ll 1$ and $m^2=0$, step~\ref{G4} should follow  from (a slight generalization)  of \cite{Moi2}.

\blue{\begin{rmk}
Note that $A$ must be chosen to be large enough (depending on $\epsilon$) for the construction to \red{feature any part of the spacelike singularity constructed in Theorem~\ref{thm.intro}}. As a consequence,  the weighted norm $\|( 1+|v|^{s}) \phi\|_{L^{\infty}}$ (where $s>\frac{1}{2}$) is large, thus we are not in the small-data regime considered in  \cite{nospacelike} (this is consistent with the non-empty spacelike singularity our construction, recalling the earlier discussion on the global structure of two-ended black holes).
	\end{rmk}}

Running steps~\ref{G1}-\ref{G4} successfully provides a (two-ended) black hole with Penrose diagram as in the rightmost picture of Figure~\ref{Fig1}, i.e.\ a black hole with a Cauchy horizon $\mathcal{CH}^+ \neq  \emptyset$ and a spacelike singularity $\mathcal{S}\neq \emptyset$ partly given by $\mathcal{S}$ in the hairy black hole of Theorem~\ref{thm.intro}. We formalize the above strategy into the following open problem. \begin{open}\label{open1}
Construct a (one or two)-ended black hole with a Cauchy horizon $\mathcal{CH}^+ \neq  \emptyset$ and a spacelike singularity $\mathcal{S}\neq \emptyset$, which coincides with the hairy black hole singularity $\mathcal{S}$ from Theorem~\ref{thm.intro} away from   $\mathcal{CH}^+ \cap  \mathcal{S}$.
\end{open}


Our road-map towards a resolution of Open Problem~\ref{open1} shows that  the fluctuations and Kasner \blue{bounces} of Theorem~\ref{thm.intro} \textbf{\blue{can} play a role in the interior of asymptotically flat, non-hairy black holes}. We find it striking that, even when restricted to spherical symmetry, the spacelike singularity inside a black hole can obey such intricate dynamics. \blue{It is an interesting open problem to study the stability of these dynamics outside of any symmetry assumption and thus gauge their relevance to generic black hole solutions.}

To understand a larger class of black holes with both a Cauchy horizon $\mathcal{CH}^+ \neq  \emptyset$ and a spacelike singularity $\mathcal{S}\neq \emptyset$, it is of interest to perturb the hairy black hole of Theorem~\ref{thm.intro} within spherical symmetry but relaxing spatial homogeneity. Subsequently following steps~\ref{G1}-\ref{G4}, where $g_{\ep}$ is replaced by a perturbed spacetime, will yield  even more general insights than Open Problem~\ref{open1} into spherically symmetric spacelike inside black holes, which we formalize in the following open problem (note Open Problem~\ref{open2} is the charged ($q_0\neq0$) version of Open Problem v in \cite{VDM21}).

\begin{open}\label{open2}
Consider (two-ended) initial data on $\mathcal{H}^+$ such that, instead of \eqref{rough.data1}:  \begin{equation}
|\phi_{|\mathcal{H}^+}(v) -\ep | \leq |\ep|^N \cdot e^{-C_0 v}
\end{equation} for  $\ep \in E_{\eta}$, as defined in Theorem~\ref{thm.intro}, with $N>0$ and $C_0>0$ sufficiently large constants.  Prove (or disprove) that the terminal boundary is spacelike, and provide (reasonably) precise quantitative estimates.

Then, construct a (one or two)-ended black hole with a Cauchy horizon $\mathcal{CH}^+ \neq  \emptyset$ and a spacelike singularity $\mathcal{S}\neq \emptyset$, which coincides with the above perturbed hairy black hole singularity $\mathcal{S}$  away from   $\mathcal{CH}^+ \cap  \mathcal{S}$.
\end{open}

We finally want to emphasize that our quantitative methods give hope to transpose some results of Theorem~\ref{thm.intro} to towards Open Problem~\ref{open2}. We hope to return to this these very interesting questions in future work.


\subsection{Comparison with hairy black hole interiors for other matter models} \label{other.intro}

\paragraph{The charged hairy black holes with uncharged matter from \texorpdfstring{\cite{VDM21}}{[81]}}

An alternative to studying  \eqref{E1}-\eqref{E5} with a charged scalar field ($q_0\neq 0$) is to study the uncharged scalar field case $q_0=0$ where the Maxwell field $F\neq 0$ does not interact with $\phi$. This was first done numerically in \cite{hartnolletal2020} and then rigorously by the second author \cite{VDM21} and qualified as ``violent nonlinear collapse''. It is remarkable that the behavior in the $q_0=0$ case differs drastically from what we found in the $q_0\neq 0$ case in Theorem~\ref{thm.intro}, as  the following result shows.

\begin{thm}[\cite{VDM21}] \label{violent.thm}
\blue{We make} the same assumptions as Theorem~\ref{thm.intro}, except that now $q_0=0$, hence $ F= \frac{\mathbf{e}}{r^2(s)} \Omega^2 ds \wedge dt$, where $\mathbf{e}\neq 0$.	
Then, for almost \blue{all} sub-extremal parameters  $(M,\mathbf{e},\Lambda,m^2)$, there exists $\ep_0(M,\mathbf{e},\Lambda,m^2)>0$ such that, \blue{for} all $0<|\ep|<\ep_0$, the spacetime $(\mathcal{M},g)$ ends at a {spacelike singularity $\mathcal{S}=\{r=0\}$} asymptotically described by a Kasner metric with  exponents $(p_1,p_{2},p_{3}) =(1,0,0)+O(\ep^2)\in(0,1)^3$ and given by Figure~\ref{Fig7}.
Moreover, the  Kretschmann  scalar $\mathcal{K}=R^{\alpha \beta \gamma \delta} R_{\alpha \beta \gamma \delta}$  blows up at a rate $r^{-C\cdot \ep^{-2}+O(\epsilon^{-1})}$ on $\mathcal{S}=\{r=0\}$ for  $C(M,\mathbf{e},\Lambda,m^2)>0$. 
\end{thm}

We point out the following   similarities and differences between Theorem~\ref{violent.thm} and Theorem~\ref{thm.intro}. \begin{enumerate}
\item In both cases, the terminal boundary is a spacelike singularity $\mathcal{S}=\{r=0\}$ approximately described by a Kasner metric \eqref{Kasner} with positive Kasner exponents (compare Figure~\ref{Penrose_simplified} and Figure~\ref{Fig7}).

\item In both cases, the early regions are similar and governed by the almost formation of a Cauchy horizon.

\item In both cases, the Maxwell charge $Q$ is uniformly bounded away from $0$: in  Theorem~\ref{violent.thm}, this is trivial ($Q=\mathbf{e}\neq 0$ is constant), in Theorem~\ref{thm.intro}, this is item~\ref{charge.retention}, a surprising property that we call ``charge retention''.
\item Even at the linear level (i.e.~\eqref{E5} on a fixed Reissner--Nordstr\"{o}m interior),  \eqref{linear.osc} is not true if $q_0=0$: the scalar field does not oscillate, it grows instead like $\phi \approx \ep \cdot s$, where $\mathcal{CH}^+=\{s=\infty\}$ (except possibly  for an exceptional set of $(M,\mathbf{e},\Lambda,m^2)$ of $0$-Lebesgue measure that leads to the absence of growth, see \cite{Kehle2018,VDM21}, which is why  one restricts to \emph{almost} every parameter \blue{in Theorem~\ref{violent.thm}}).
\item The oscillating profile \eqref{linear.osc} \blue{turns it a Bessel-type growing oscillation as} $r\rightarrow 0$, after which $|\phi|$ becomes $O(1)$ around $r\approx \ep$ (\blue{collapsed} oscillations). There is no such mechanism in Theorem~\ref{violent.thm}.


\item The final Kasner exponent $p_1 \in (0,1)$ in Theorem~\ref{thm.intro} is related to $| \sin( \omega_0\ep^{-2}+O(\log(\ep^{-1})))|$, but we restrict $\ep$ so that $p_1$ is bounded away from $0$ and $1$, so there is no overlap (by our assumptions) with any of the Kasner exponents obtained in Theorem~\ref{violent.thm} where $|p_1-1| \lesssim \ep^2$.
\item As a consequence, the collapse in Theorem~\ref{thm.intro} is \emph{not violent} but instead rapidly fluctuating in $\ep$: one can easily see that $\mathcal{K}$ blows up at a rate $r^{-q}$, where $q$ depends on $\sin( \omega_0 \cdot \ep^{-2}+O(\log(\ep^{-1})))$.

\item There is no Kasner \blue{bounce} in the $q_0=0$ setting: in fact, in Theorem~\ref{violent.thm} one proves that the final exponent  $p_1$ \blue{lies in} $(0,1)$, so there is no mechanism triggering a Kasner \blue{bounce}. In contrast, in Theorem~\ref{thm.intro}, in the regimes where  $|\sin( \omega_0 \cdot \ep^{-2} + O(\log (\epsilon^{-1})))|$ is too small, a Kasner regime with $\blue{\acute{p}_1 }< 0$ forms, which is unstable, and ultimately disappears under Kasner \blue{bounce}, giving rise to  a second Kasner regime with $p_1 \in (0,1)$. 
\end{enumerate}

\begin{figure}

\begin{center}

\includegraphics[scale=0.075]{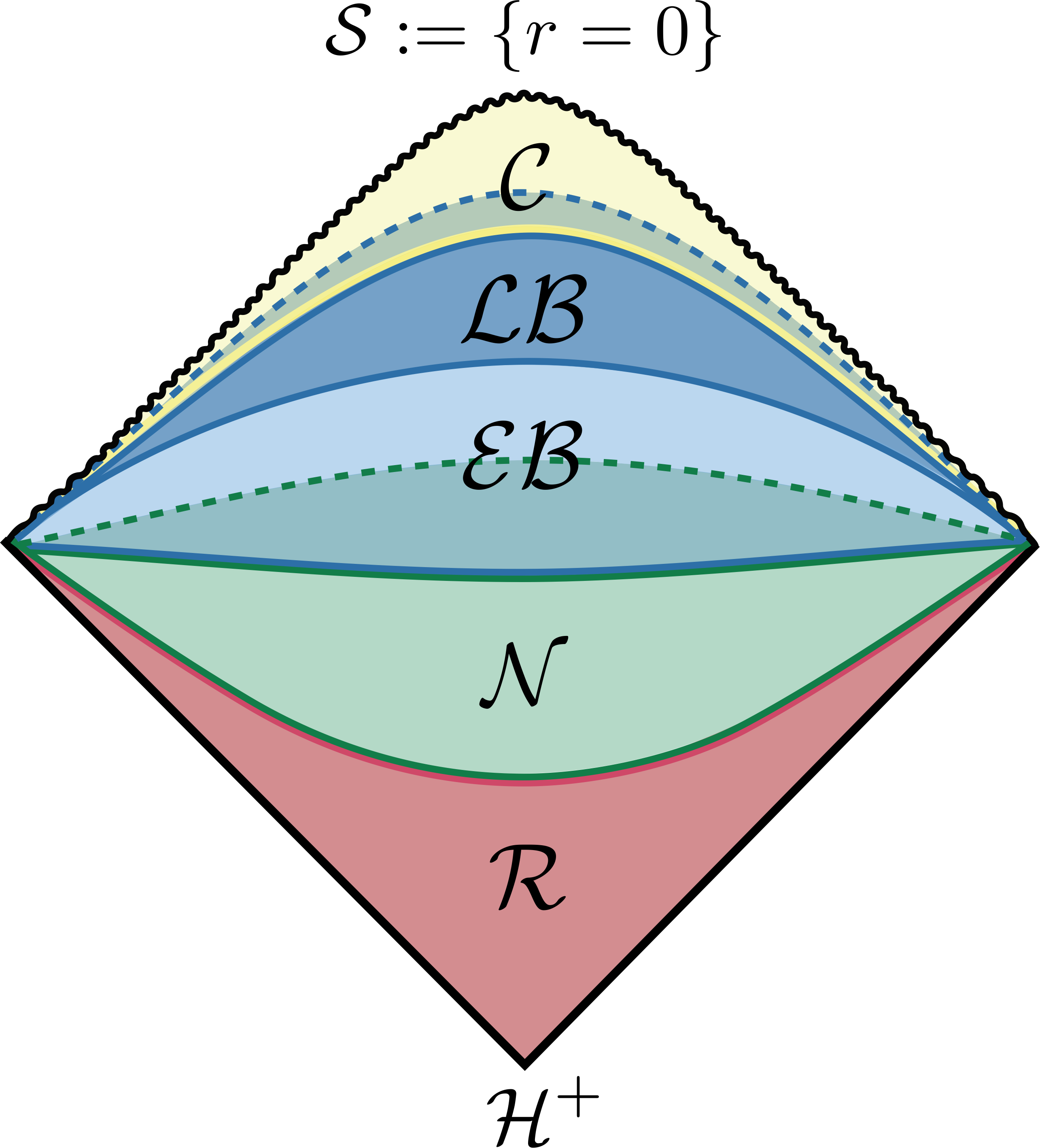}

\end{center} 	\caption{The Penrose diagram of the hairy black hole interiors constructed in Theorem~\ref{violent.thm}.}
\label{Fig7}
\end{figure}
We also remark that Theorem~\ref{thm.intro} restricts $\ep$ to a subset of $(-\ep_0,\ep_0)\setminus\{0\}$, which is the complement of a set of small measure, while Theorem~\ref{violent.thm} does not have this restriction. The restriction is to ensure that the final $p_1$ is bounded away from $\{0,1\}$ in Theorem~\ref{thm.intro}, as we explained above.

Finally, we note that Theorem~\ref{violent.thm} should lead to a resolution of Open Problem~\ref{open1} in the case of an uncharged, massive scalar field (i.e.~\eqref{E1}-\eqref{E5} with $q_0=0$, $m^2\neq 0$), upon the proof of \eqref{decay.s} for $m^2\neq 0$ (i.e.~step~\ref{G4} in the last paragraph of Section~\ref{non-hairy.intro}). 

\paragraph{Other hairy black holes}

The study of spatially homogeneous hairy black holes has been abundant both in the mathematics and physics literature: we  first mention the important examples of Einstein--Yang--Mills hairy black holes \cite{Bizon,Maison,Yau2,Galtsov1,Sarbach,Yau1}. For the Einstein--Yang--Mills black holes, \blue{the above work suggest the presence of  a spacelike singularity preceded by (potentially many) Kasner bounces, although the} qualitative behavior is different than what we obtained for charged scalar fields in Theorem~\ref{thm.intro}. Finally, we mention the existence of rotating hairy black holes with massive Klein--Gordon fields \cite{OtisYakov,numerics.KerrAdS,numerics.Kerr}; \blue{The study of the interior of these rotating hairy black holes remains an open problem.}

We  refer the reader to the introductions of \cite{OtisYakov,VDM21} for an extended discussion of various hairy black holes.

\subsection{The collapsed oscillations resulting from the charge of the scalar field} \label{osc.intro}

The \blue{collapsed} oscillations occur in a region $\mathcal{O}=\{ \ep \lessapprox r \lessapprox r_-\}$ (see Figure~\ref{Penrose_simplified}). The key point is that, schematically, $\phi$ will be shown to obey the following Bessel equation of order $0$ in $\mathcal{O}$, with respect to a new variable which is the renormalized square of the area-radius $z:=\frac{r^2}{\ep^2}$: \begin{equation}\label{Bessel.eq}
\frac{d}{dz} \left(z \frac{d\phi}{dz}\right)+  \xi_0^2 z \phi=error.
\end{equation} 
Here $\xi_0\neq 0$ is a constant  proportional\footnote{One sees, as predicted by Theorem~\ref{violent.thm} (see \cite{VDM21}), that in the $q_0=0$ case, we have $\frac{d}{dr}\left(r\frac{d\phi}{dr}\right)\approx 0$, hence $r\frac{d\phi}{dr} \approx constant=\ep^{-1}$,  which is why in the $q_0=0$ case, the behavior is violent, and not fluctuating as in the $q_0\neq0$ case, see also the discussion in Section~\ref{other.intro}.} to $q_0$. To simplify the discussion here, we normalize  $\xi_0=1$. Since $1 \lessapprox z \lessapprox \ep^{-2}$, we need to understand the behavior \blue{for large $z$}: it is given by damped oscillations of the form 
\begin{equation}\label{Bessel}
Y_0(z) \sim   \sqrt{\frac{2}{\pi z}}\cos(z-\frac{\pi}{4})\text{ or }  J_0(z) \sim   \sqrt{\frac{2}{\pi z}}\sin(z-\frac{\pi}{4}) \text{ as } z \rightarrow +\infty.
\end{equation}Note, however, that $z$ is a \underline{past}-directed timelike variable, so the damping is ``backwards-in-time''. \sloppy Thus $|Y_0|(z),\ |J_0|(z)\lesssim \ep$ on the past boundary $z\sim \ep^{-2}$ of $\mathcal{O}$, but $|Y_0|(z),\ |J_0|(z)\lesssim 1$  on the future boundary $z\approx 1$ of $\mathcal{O}$; modulo the oscillations, this means that \emph{the scalar field  amplitude has experienced growth of size $\ep^{-1}$} in $\mathcal{O}$.
\begin{rmk}
Note that, as long as $r$ is bounded away from $0$,  \eqref{Bessel.eq} is consistent with (linear) oscillations giving rise to \eqref{linear.osc}: it is only as $r$ gets close to $0$ that these oscillations provide growth,  hence the name ``collapsed oscillations''. We will show, however, that as soon as $r\ll \ep$, $\phi$ no longer oscillates, see Section~\ref{cosmo.intro}.
\end{rmk}

The algebraic relations connecting \eqref{Bessel} to the ODE initial conditions $\phi(r\approx r_-)$ will ultimately show that $\phi$ has the following schematic form at the exit of the \blue{collapsed} oscillations region $\mathcal{O}$: \begin{equation}\label{exit.osc}
\phi(r) \approx \blue{C_J} \cos( \omega_0\cdot  \ep^{-2}+O(\log(\ep^{-1})))  J_0 \left( \xi_0 \frac{r^2}{\ep^2}\right) + \blue{C_Y} \sin( \omega_0 \cdot  \ep^{-2}+O(\log(\ep^{-1})))  Y_0 \left(\xi_0 \frac{r^2}{\ep^2}\right).
\end{equation}  Contrary to appearances, \eqref{exit.osc} is not symmetric in $J_0$ and $Y_0$: when $r\ll\ep$, the function $ Y_0 \left(\xi_0 \frac{r^2}{\ep^2}\right)$ dominates $J_0 \left(\xi_0 \frac{r^2}{\ep^2} \right)$, since the Bessel functions $J_0(z)$ and $Y_0(z)$ obey the asymptotics  \begin{equation}
J_0(z) = O(1) \text{ and } Y_0(z) \sim \log(z^{-1}) \text{ as } z \rightarrow 0.
\end{equation} 
Hence, for $ e^{-\delta_0 \ep^{-2} } \ll r\ll \ep$ (the lower bound will be explained in Section~\ref{cosmo.intro}), we  schematically show \begin{equation}\label{phi.protoKasner}
\phi(r) \approx  \blue{C_Y}  \sin( \omega_0\cdot  \ep^{-2}+O(\log(\ep^{-1})))  \log (\xi_0^{-1} \frac{\ep^2}{r^2}) \approx  \blue{C_Y}  \sin( \omega_0\cdot  \ep^{-2}+O(\log(\ep^{-1})))  \log (\frac{\ep}{r}).
\end{equation}

Since on a fixed Kasner metric \eqref{Kasner}, we find $\phi = p_{\phi} \log(\tau^{-1})$, where $\tau$ is roughly a power of $r$, and $p_{\phi}$ is chosen so that $p_{1}^2 + p_{2}^2+  p_{3}^2+2p_{\phi}^2=1$ (see already \eqref{kasner2}), the expression \eqref{phi.protoKasner} explains why we obtain final Kasner exponents that depend on $ \sin( \omega_0\cdot  \ep^{-2}+O(\log(\ep^{-1}))) $.

\begin{rmk}\label{retention.rmk}
Most other quantities, such as the charge $Q$, already determine their final values at $r=0$ inside $\mathcal{O}$ (up to $O(\ep)$-errors). Therefore, the charge retention mechanism from Theorem~\ref{thm.intro} results from an explicit computation in $\mathcal{O}$, see Lemma~\ref{lem:jo_charge_retention}.
\end{rmk}

For more details on Bessel equations and functions, we refer the reader to our Appendix~\ref{sec:appendix_A}.

\subsection{Kasner bounces and connections to cosmology} \label{cosmo.intro}

We now relate the results of this paper to the heuristic observations of BKL \cite{BK1, BKL1, BKL2, KL} regarding problems in relativistic cosmology, and explain how these heuristics manifest themselves rigorously in our work. 

\paragraph{The BKL heuristics and Kasner bounces} 
In \cite{KL}, Khalatnikov and Lifschitz propose an asymptotic form of the metric for a spacetime obeying the vacuum Einstein equations in the vicinity of a spacelike singularity. Assuming $\mathcal{M} \cong I \times \Sigma = (0, T) \times \Sigma$ for some spatial $3$-manifold $\Sigma$, \blue{they suggest that $g$ locally takes the form}
\begin{equation} \label{kl_ansatz}
g \approx - d \tau^2 + \sum_{I = 1}^3 \tau^{2 p_I(x)} \omega_I(x) \otimes \omega_I(x).
\end{equation}
Here, the exponents $p_I(x)$ are smooth functions on $\Sigma$, the $\omega_I(x)$ form a \blue{(local)} basis of $1$-forms on $\Sigma$, and the metric is `synchronized' so that the singularity is located at $\tau = 0$. The exponents $p_I(x)$ are further constrained to obey the following two so-called \emph{Kasner relations}:
\begin{equation} \label{kasner1}
\sum_{I = 1}^3 p_I(x) = \sum_{I = 1}^3 p_I^2(x) = 1.
\end{equation}

However, \cite{KL} argues that generically, there is an inconsistency in the ansatz \eqref{kl_ansatz}, so long as near $\tau = 0$, one fails to obey the \emph{subcriticality condition}:
\begin{equation} \label{subcriticality}
\tau^{p_I - p_J - p_K} \ll \tau^{-1} \text{ for all } I, J , K \in \{1, 2, 3\} \text{ with } J \neq K. 
\end{equation}
Further, in $1+3$-dimensional vacuum, the relations (\ref{kasner1}) mean that the subcriticality condition \eqref{subcriticality} can never hold, outside of the exceptional case where $(p_1, p_2, p_3) = (1, 0, 0)$ or a permutation thereof. \cite{KL} thus concludes that singularities of the form \eqref{kl_ansatz} are not generic.

Subsequently, in \cite{BKL1}, the authors suggest that the metric \eqref{kl_ansatz} may be valid in some interval $(\tau_1, \tau_2) \subset I$, but as $\tau$ decreases further towards $0$, there must be a transition to a new modified Kasner-like regime. If we order $p_1 < p_2 \leq p_3$, then \cite{BKL1} calculates this transition\footnote{\red{We remark that we exclude the particular case $p_1=p_2$ corresponding to the degenerate exponents $(p_1,p_2,p_3)=(0,0,1)$.}} explicitly to be:
\begin{equation*}
g \approx - d \tau^2 + \sum_{I = 1}^3 \tau^{2 \acute{p}_I(x)} \acute{\omega}_I(x) \otimes \acute{\omega}_I(x),
\end{equation*} \begin{equation} \label{transition}
\acute{p}_1 = \frac{- p_1}{1 + 2 p_1}, \quad \acute{p}_2 = \frac{p_2 + 2 p_1}{1 + 2 p_1}, \quad \acute{p}_3 = \frac{p_3 + 2 p_1}{1 + 2 p_1}.
\end{equation}
%
Such a transition is what we call a \emph{Kasner \blue{bounce}}, and also may be described in the literature as a Kasner \blue{inversion (see \cite{hartnolletal2021})} or oscillation.

The new Kasner exponents $\acute{p}_I$ also obey the Kasner relations \eqref{kasner1}, and as such, will also fail to obey the subcriticality condition. Hence, \cite{BKL1} predicts that the generic behavior in the vicinity of a spacelike singularity \emph{in vacuum} is an infinite cascade of such transitions, which they term as the \emph{oscillatory approach} to singularity, and is expected to be highly chaotic in nature, see again Conjecture~\ref{conj.cosmo}.

To avoid this infinite cascade of transitions occuring in vacuum, the authors of \cite{BK1} then consider gravity coupled to a massless scalar field $\phi$, and modify the ansatz \eqref{kl_ansatz} and relations \eqref{kasner1} to 
\begin{gather} \label{kl_ansatz2}
g \approx - d \tau^2 + \sum_{I = 1}^3 \tau^{2 p_I(x)} \omega_I(x) \otimes \omega_I(x), \quad \phi \approx p_{\phi}(x) \log \tau,
\\[0.5em]  \label{kasner2}
\sum_{I = 1}^3 p_I(x) = \sum_{I = 1}^3 p_I^2(x) + 2 p_{\phi}^2(x) = 1.
\end{gather}

For particular choices of generalized exponents $(p_1, p_2, p_3, p_{\phi})$, it is now possible for the subcriticality condition \eqref{subcriticality} to hold near $\tau = 0$, and as such the ansatz \eqref{kl_ansatz2} is consistent, and moreover, conjecturally stable. We note that in this context, the condition \eqref{subcriticality} is identical to $\min \{p_I(x) \} > 0$, i.e.~all Kasner exponents being positive.
There still exist, of course, choices of exponents that violate \eqref{subcriticality}; the corresponding spacetimes are then subject to an instability with the same Kasner transition map \eqref{transition}, to which we append the transition of the scalar field coefficient: $p_{\phi} \mapsto \acute{p}_{\phi} = \frac{p_{\phi}}{1 + 2 p_1}$. After a finite number of such transitions \cite{BK1}, one will reach a tuple of generalized Kasner coefficients obeying \eqref{subcriticality}. Hence, a scalar field is often referred to as a stiff matter model, as in Conjecture~\ref{conj.cosmo}.

We make one final observation. The source of the instability in \cite{BK1, BKL1} is a spatial curvature term, which is actually suppressed in spherical symmetry. However, in \cite{BK2}, the authors argue that one can alternatively use an electromagnetic field to source the instability, and that the transition map \eqref{transition} between different regimes of Kasner exponents is identical.  This is consistent with the stability of the Schwarzschild interior in spherical symmetry for electromagnetism-free matter models \cite{DejanAn, Christo1}, in contrast with  Theorem~\ref{thm.intro} and Theorem~\ref{violent.thm}.

For further discussions regarding the BKL ansatz in relativistic cosmology, including generalization to higher dimension and other matter models, see also \cite{BelinskiHenneaux, Billiards, DemaretHenneauxSpindel, Henneaux}.

\paragraph{Rigorous constructions and stability results of Kasner metrics} 

Beyond the heuristics of  \cite{BK1, BK2, BKL1, BKL2, KL},  one may ask the following questions -- can one actually 
construct a large class of spacetimes containing a spacelike singularity, obeying the asymptotics \eqref{kasner1} and \eqref{kasner2}, and what does one know about their stability? 
For the first problem, \cite{AnderssonRendall} constructs a large class of \emph{real analytic} solutions to the Einstein-scalar field system  obeying the asymptotics \eqref{kl_ansatz2}. Beyond the real analytic regime, \cite{FournodavlosLuk}  recently constructed  a reasonably general class of \underline{vacuum} spacetimes 
obeying \eqref{kl_ansatz}, which are moreover allowed to be only $C^k$, for large  $k$. 


Regarding stability, the state of the art is due to Fournodavlos--Rodnianski--Speck \cite{FournodavlosRodnianskiSpeck} \blue{and a recent generalization by Oude Groeniger--Petersen--Ringstr\"om \cite{groeniger2023formation}. Loosely speaking, the former proves the stability of exact generalized Kasner spacetimes on $(0, + \infty) \times \mathbb{T}^3$ obeying the subcriticality condition \eqref{subcriticality}, while the latter both permits more general closed spatial topologies and greater inhomogeneity in the initial data, so long as the inhomogeneity is compensated by prescribing the data ``sufficiently close to the singularity''.} For other related results, we refer the reader to \cite{RodnianskiSpeck2, RodnianskiSpeck1, SpeckS3, BeyerOliynyk, FajmanUrban}. 

\blue{In particular, \cite{groeniger2023formation} provides evidence for stability, \emph{including outside of symmetry}, of the spacetimes constructed in Theorem~\ref{thm.intro} and Theorem~\ref{violent.thm}; there only remains the technical issue of extending the stability arguments to non-compact spatial topologies, for instance by localizing the analysis of \cite{FournodavlosRodnianskiSpeck, groeniger2023formation}. We also mention several works of the first author. In \cite{Li_Kasner}, it is shown that all spacelike singularities, including spatially inhomogeneous singularities, in the spherically symmetric Einstein--scalar field model exhibit Kasner-like behavior, as does a class of singularities in the spherically symmetric Einstein--Maxwell--(uncharged) scalar field model obeying certain a priori assumptions. In \cite{LiSurfaceSymPaper}, the first author constructs a large class of (inhomogeneous) data in this latter model such that the corresponding spacetime achieves these a priori assumptions, but so that the spacetime also exhibits nonlinear dynamics reminiscent of Kasner bounces \eqref{transition}. See also \cite{LiGowdyPaper} concerning similar results in Gowdy symmetry.}

Finally, we mention the work of Ringstr\"om \cite{ringstrombianchi} \blue{and B\'eguin--Dutilleul \cite{BeguinDutilleul}} on  Bianchi IX cosmologies, containing a rigorous study of  a large class of spatially homogeneous spacetimes. 
Among other things, \cite{ringstrombianchi} provides  examples of spacetimes with infinitely many Kasner bounces in vacuum and, in contrast, proves the convergence to a stable Kasner-like regime in the presence of stiff matter. 

\paragraph{The Kasner bounce mechanism for charged scalar fields} 

We will now explain the schematic mechanism behind the Kasner \blue{bounce} phenomenon as obtained in our Theorem~\ref{thm.intro}, which is the second main novelty \blue{compared} to the $q_0=0$ case of \cite{VDM21}. We will show that the Kasner \blue{bounce}, \emph{in the regimes where it occurs} (namely, for $\ep \in E_{\eta, \sigma}^{'\ \blue{bo}}$ of positive measure), is located in a region of the following form; for constants $D>0$, $N>0$: \begin{equation}\label{Kinv.intro}
\mathcal{K}_{\blue{bo}} \subset \{   e^{- D \cdot \ep^{-2}} \ep^{N} \lesssim    r \lesssim e^{- D \cdot \ep^{-2}} \ep^{-N} \}.
\end{equation} 

We define the key quantity $\Psi$, which is a dimensionless derivative of $\phi$: for $r_-(M,\mathbf{e},\Lambda)>0$ and $\delta_0(M, \mathbf{e}, \Lambda) > ~0$ to be defined later, let \begin{equation}\label{Psi.def}
\Psi:= -r \frac{d\phi}{dr}, \text{ and define } \Psi_{i}:= \Psi \big|_{{r}=e^{-\delta_0\cdot \ep^{-2}}r_-}.
\end{equation}

The condition for the presence of a \blue{bounce}  will end up being \begin{equation}\label{inv.cond} \tag{\blue{bo}}
\eta \leq |\Psi_i| \leq 1-\sigma, \text{ for some } \eta, \sigma>0 \text{ independent of } \ep.
\end{equation}
The reason for assuming  $\eta \leq |\Psi_i|$ 	is that,  based on numerics (see Section~\ref{holo.intro}), there could be multiple Kasner \blue{bounces} when$ |\Psi_i|$ is close to $0$, and the dynamics would be even more complicated. If $|\Psi_i|$ is too close to $1$, though we are still able to produce a spacelike singularity (see already Theorem~\ref{maintheorem}), we do not claim  further quantitative estimates, as some Kasner exponents degenerate towards $0$ in this case

As a consequence of \eqref{phi.protoKasner}, we find that, for some $C(M,\mathbf{e},\Lambda,m^2,q_0) \neq 0$, \begin{equation}\label{psii}
|\Psi_i| \approx  |C|  \cdot |\sin( \omega_0\cdot  \ep^{-2}+O(\log (\ep^{-1})))| .
\end{equation} Combining \eqref{inv.cond} and \eqref{psii} explains heuristically why the presence of a \blue{bounce} depends on $\ep$ and why any small neighborhood of $0$ of the form $(-\delta,\delta)$ still contains infinitely many spacetimes featuring a \blue{bounce}.

Our non-\blue{bounce} condition in Theorem~\ref{thm.intro} is not the complement of \eqref{inv.cond}, it is instead \begin{equation}\label{ninv.cond} 
|\Psi_i| \geq  1+\sigma, \text{ for some } \sigma>0 \text{ independent of } \ep.
\tag{no-\blue{bo}}, \end{equation}\blue{where, as above, we arrange that $|\Psi_i|$ is not too close to $1$ to avoid the Kasner exponents degenerating.}  

We now explain why it is that if \eqref{inv.cond} is satisfied, there is a \blue{bounce}, whereas if \eqref{ninv.cond} is satisfied then there is no \blue{bounce}.
Since the wave equation is second-order, $\Psi$ should satisfy a first-order ODE. 
Though our system is highly nonlinear, this ODE surprisingly turns out to be presentable in a simple form, written schematically\footnote{Note that the discussion of the ODE \eqref{Psi.ODE} and its solutions was already present in \cite{hartnolletal2021} at the heuristic and numerical level.} as follows: \begin{equation}\label{Psi.ODE}
\frac{d\Psi}{dR} = - \Psi (\Psi-\Psi_i) (\Psi- \Psi_i^{-1})+error, \text{ where } R:= \log(\frac{r_-}{r}),
\end{equation}
\blue{where the error terms are of size $O(e^{-\delta_0 \epsilon^{-2}})$ for some $\delta_0>0$ and\red{, thus,} have no impact on the qualitative behavior of the ODE.} The dynamics of $\Psi$ relies on the linearized stability of \eqref{Psi.ODE} near $\Psi=\Psi_i$  of the schematic form $$ \frac{d(\delta \Psi)}{dR} = -   (\Psi_i^2- 1) \cdot \delta \Psi +error.$$

If $|\Psi_i|> 1$, then $\Psi=\Psi_i$ is a stable fixed point as $R \rightarrow+\infty$ (corresponding to $r\rightarrow 0$): this is what happens if \eqref{ninv.cond} is true and then, there is no \blue{bounce} and $\Psi \approx \Psi_i$ up to $r=0$.
In contrast, if $|\Psi_i|<1$, then $\Psi=\Psi_i$ is \blue{an} \emph{unstable} fixed point, but $\Psi=\Psi_i^{-1}$ is a stable fixed point. If \eqref{inv.cond} is true, then we find that $\Psi$ gets inverted from $\Psi_i$ to $\Psi_i^{-1}$ in the region \eqref{Kinv.intro}, over which the change in $R$ is $\Delta R=O(\log(\ep^{-1}))$.

Once the behavior of $\Psi$ has been quantified, one can  immediately retrieve the Kasner exponents from Theorem~\ref{thm.intro} using similar techniques as in \cite{VDM21}%
. The most important (but algebraically trivial) feature of these relations is that $p_1=P(\blue{\alpha})>0$ if and only if $|\blue{\alpha}|>1$, see the formula \eqref{first.Kasner}, \eqref{last.Kasner} (where $\alpha \approx \Psi_i$): thus, the final Kasner  \blue{regime} always has $p_1>0$ in Theorem~\ref{thm.intro}. Hence, the final Kasner \blue{regime} indeed lies in the subcritical regime of exponents as explained earlier in Section~\ref{cosmo.intro}.

\subsection{Holographic superconductors in the AdS-CFT correspondence} \label{holo.intro}

In this section, we will discuss two separate questions, which end up being connected by a vast literature. \begin{itemize}
\item \label{holo1} In what sense is the spacetime of Theorem~\ref{thm.intro}  the interior of a hairy black hole, i.e.\ can we construct a corresponding ``hairy black hole exterior''?

\item \label{holo2}   What do numerics tell us about the hairy black hole interior of Theorem~\ref{thm.intro}?
\end{itemize}
\blue{Note that  Theorem~\ref{thm.intro} is valid for any value of the cosmological constant $\Lambda \in \RR$. However, in the $\Lambda=0$ case, there are no  non-trivial static, spherically symmetric solutions of \eqref{E1}-\eqref{E5}, see \cite{Bekenstein} (although \cite{Bekenstein} assumes $q_0=0$).}

\blue{On the other hand, in the $\Lambda < 0$ case, recent works of Zheng \cite{Weihao1,Weihao2} construct   asymptotically Anti-de-Sitter hairy black hole exteriors bifurcating off \red{of} (a large subset of) sub-extremal  Reissner--Nordstr\"{o}m--AdS spacetimes. The interior of these black holes correspond\red{s} to the spacetimes considered in Theorem~\ref{thm.intro} (in the $q_0\neq 0 $ case) and Theorem~\ref{violent.thm} (in the $q_0=0$ case), thereby resolving Open Problem iv in \cite{VDM21}.}

Anti-de-Sitter asymptotics  impose that, for a negative cosmological constant $\Lambda<0$ in \eqref{E1}: \begin{equation}\label{ads1}
	g= \left(1-\frac{2M}{r}- \frac{\Lambda r^2}{3}+o(r^{-1})\right)[-dt^2+ds^2] + r^2(s) d\sigma_{\mathbb{S}^2},
\end{equation}  which gives the following asymptotics on $\phi$ in \eqref{E5}, for $m^2<0$, there exists constants $\phi_{(0)}$, $\phi_{(1)}$ such that
\begin{equation}\label{ads2}\phi(r) =  \phi_{(0)}  \cdot u_D(r)+ \phi_{(1)}\cdot  u_N(r),  \text{ where } u_D (r)\sim r^{-\frac{3}{2}+\sqrt{\frac{9}{4}-m^2}} \text{ and } u_N (r)\sim r^{-\frac{3}{2}-\sqrt{\frac{9}{4}-m^2}} \text{  as } r \rightarrow+\infty.\end{equation} Here $\phi_{(0)}=0$ corresponds to Dirichlet-type boundary conditions,  while $\phi_{(1)}=0$ is Neumann-type. \blue{Both boundary conditions are allowed in   \cite{Weihao1,Weihao2}. We note that the Dirichlet case is the $\Lambda<0$ analogue of asymptotic flatness in the $\Lambda=0$ case and thus truly corresponds to what should be called an asymptotically-AdS ``hairy black hole''.}

\blue{Previous works in physics already argued for the existence of  such stationary asymptotically-AdS black holes  \cite{holorigin1,holorigin3,holorigin2} as part of the broader field of AdS/CMT (for Condensed Matter Theory) aiming at improving the understanding of condensed matter on flat spacetime by embedding it in a higher-dimensional AdS spacetime (see e.g.\ \cite{hartnollbook} or the related discussion in \cite{VDM21}). In the charged scalar field case $(q_0 \neq 0)$, such black holes represent holographic analogues of superconductors due to a spontaneous breaking of the $\mathbb{U}(1)$-symmetry that the charged scalar field is associated with. The presence of black hole hair (meaning a non-trivial stationary black hole) reflects the phenomenon of condensation, with	   Dirichlet boundary condition $\phi_{(0)}\neq 0$   corresponding to a so-called stimulated emission, and  $\phi_{(0)}= 0$ to spontaneous emission, in analogy with real-life superconductors (see e.g.\  \cite{introholo} for more details).

In summary,  the first question is answered in the affirmative: the regions considered in Theorem~\ref{thm.intro} and Theorem~\ref{violent.thm} indeed correspond to the interior of an asymptotically-AdS black hole (with Dirichlet or Neumann boundary conditions) recently rigorously constructed in \cite{Weihao1,Weihao2}, whose existence had already been speculated in physics works \cite{holorigin1,holorigin2} in view of their importance as holographic models of superconductors.}


We turn to the second question, \blue{regarding previous numerical works in the interior of asymptotically-AdS hairy black holes. Prior to discussing \eqref{E1}--\eqref{E5} with $q_0\neq 0$ as studied in Theorem~\ref{thm.intro}, let us mention that the black hole interior for \ \eqref{E1}--\eqref{E5} with $F\equiv 0$ (uncharged black hole) was studied in \cite{holo3}, and \eqref{E1}--\eqref{E5} with $F\neq 0$, $q_0=0$ (uncharged matter) in \cite{hartnolletal2020}, see \cite{VDM21} for an extended discussion. We also mention the interesting follow-up works \cite{holo1,holo5,holo4,holo2,holo6} for different matter models.} 
 \blue{Turning to the case  \eqref{E1}--\eqref{E5} with $q_0\neq 0$,} we  start by remarking that the \blue{collapsed} oscillations and Kasner \blue{bounces} from Theorem~\ref{thm.intro} were previously anticipated\footnote{\blue{While \cite{hartnolletal2021} is in the setting of planar symmetry which obeys slightly different equations, it happens to make no difference to the fluctuating collapse and Kasner bounce phenomena and the conclusion of Theorem~\ref{thm.intro} also hold in this setting, as we already noted in Remark~\ref{planar.remark}.}} numerically in \cite{hartnolletal2021}. This scenario \blue{was} verified numerically \blue{and} is  entirely consistent with our findings from Theorem~\ref{thm.intro}. \blue{The motivation to investigate the interior of holographic superconductors in \cite{holo3,hartnolletal2020,hartnolletal2021}  is two-fold: at first, it originates as a rudimentary model for quantum effects in the interior of black holes, following the AdS/CFT paradigm. In this context, the authors of \cite{hartnolletal2021} promote a \emph{holographic version of Strong Cosmic Censorship} roughly asserting that generic \emph{stationary}\footnote{In contrast, the usual notion of Strong Cosmic Censorship is not tied to stationary solutions and formulates genericity considerations inside a moduli space of initial data leading to a large class of dynamical spacetimes, see, for instance, the discussion in \cite{review}.} (asymptotically-AdS) black holes  (corresponding to  thermal states) possess an internal singularity. Secondly, another motivation originates in the paradigm of AdS/CMT to understand real-life superconductors using holographic models (as already formulated in \cite{holorigin1,holorigin2,holorigin3} and pursued in \cite{holo3,hartnolletal2020,hartnolletal2021}). In this context, an analogy was established between the oscillations found inside holographic superconductors and  the (AC) Josephson effect inside a \blue{``real-life''} superconductor, see \cite{Jo1,Jo2}, ultimately causing the Kasner \blue{bounce} in some regime. The physical interpretation of these Kasner bounces inside holographic superconductors is still being debated \cite{hartnolletal2021}.} 

Finally, note that, by the formula \eqref{first.Kasner}, when $\alpha \approx \Psi_i$ gets close to $0$, then the Kasner exponents $p_{2}=p_{3}$ get close to $0$: when this happens, \cite{hartnolletal2021} numerically observes  a \emph{second Kasner \blue{bounce}} and have also mentioned the possibility of arbitrarily many Kasner \blue{bounces}. These very interesting aspects are not covered by Theorem~\ref{thm.intro}, since its  assumptions \emph{specifically require to choose $\ep$} so that $\alpha(\ep)$ is bounded away from $0$ (at the cost, of course, of reducing the measure of the set of eligible $\ep$ by an arbitrarily small amount $\eta>0$). It would be of great interest to prove rigorous results confirming the numerical breakthroughs of \cite{hartnolletal2021}.  \begin{open}\label{openc}
Generalize the \red{conclusion} of Theorem~\ref{thm.intro} for a larger set of $\ep \in (-\ep_0,\ep_0)\setminus\{0\}$ such that $|\Psi_i| \lesssim \ep^N$ are allowed for a suitably chosen $N>0$, where $\Psi_i$  defined  in \eqref{Psi.def} is schematically of the form \eqref{psii}. In the $|\Psi_i| \lesssim \ep^N$ situation, control the occurrence of two (or more) Kasner \blue{bounces}.
\end{open}

We finally note that the techniques of Theorem~\ref{thm.intro} still allow to control quantitatively the spacetime up to the Proto--Kasner region (see Figure~\ref{Penrose_detailed} and Section~\ref{sec:protokasner}), even if  $|\Psi_i|$ is close to $\{0,1\}$, but not beyond.



\subsection{Outline of the paper and the different regions of \texorpdfstring{Figure~\ref{Penrose_detailed}}{Figure 6}}\label{outline.sec}

The paper (and the proof) will follow the various regions $\mathcal{R}$, $\mathcal{N}$, $\mathcal{EB}$, $\mathcal{LB}$, $\mathcal{O}$, $\mathcal{PK}$ and $\mathcal{K}$ depicted on Figure~\ref{Penrose_detailed}.

\begin{itemize}
\item  In Section~\ref{prelim.section}, we give some geometric preliminaries, and explain the gauge used in Theorem~\ref{thm.intro}.

\item In Section~\ref{sec:theorem}, we give a precise statement of the main result corresponding to Theorem~\ref{thm.intro}, together with the precise definition of the regions $\mathcal{R}$, $\mathcal{N}$, $\mathcal{EB}$, $\mathcal{LB}$, $\mathcal{O}$, $\mathcal{PK}$ and $\mathcal{K}$.
\item  In Section~\ref{sec:einsteinrosen}, we prove estimates in the  red-shift region $\mathcal{R}$, the no-shift region $\mathcal{N}$, the early blue-shift region $\mathcal{EB}$, and the late blue-shift region $\mathcal{LB}$. These estimates are similar to the ones appearing in \cite{VDM21} and feature the almost formation of a Cauchy horizon.
\item  In Section~\ref{sec:oscillations}, we prove estimates in the oscillation region $\mathcal{O}$. This section corresponds to the \blue{collapsed} oscillations discussed in Section~\ref{osc.intro}.

\item  In Section~\ref{sec:protokasner}, we prove estimates in the proto-Kasner region $\mathcal{PK}$. In this section, we demonstrate the onset of a Kasner geometry transitioning from the \blue{collapsed} oscillations in  $\mathcal{O}$ to the Kasner behavior in $\mathcal{K}$.

\item In Section~\ref{sec:sing}, we linearize the system of Einstein equations to control precisely the phase $\Theta(\ep)= \omega_0 \ep^{-2} + O(\log (\ep^{-1}))$ appearing in \eqref{psii}; this step is essential in constructing the set $E_{\eta}$ of acceptable $\ep$ in Theorem~\ref{thm.intro}.

\item In Section~\ref{sec:kasner}, we prove estimates in the Kasner region $\mathcal{K}$. In particular, we prove the Kasner \blue{bounce}  phenomenon discussed in Section~\ref{cosmo.intro}.

\item In Section~\ref{section:quantitative}, we conclude the proof of (the precise version of) Theorem~\ref{thm.intro} by providing precise geometric estimates characteristic of the Kasner behavior in a sub-region of  $\mathcal{PK}\cup\mathcal{K}$.

\end{itemize}

We will also introduce the following regions that overlap with some of the regions of Figure~\ref{Penrose_detailed}: the restricted proto-Kasner region $\blue{\mathcal{PK}'}$, the first Kasner region $\mathcal{K}_1$, the second Kasner region $\mathcal{K}_2$ and the Kasner-\blue{bounce}  region $\mathcal{K}_{\blue{bo}}$ (see Figure~\ref{FigN} and Section~\ref{sec:theorem}). Note, however, that in the absence of a Kasner \blue{bounce}, $\mathcal{K}_2=\mathcal{K}_{\blue{bo}}=\emptyset$.

\begin{figure}

\centering

\scalebox{0.7}{
\begin{tikzpicture}
\path[fill=red!70!black, opacity=0.3] (0, -6) -- (-6, 0)
.. controls (0, -6) .. (6, 0) -- (0, -6);
\path[fill=green!70!black, opacity=0.3] (6, 0)
.. controls (0, -4.5) .. (-6, 0)
.. controls (0, -6) .. (6, 0);
\path[fill=blue, opacity=0.3] (6, 0)
.. controls (0, -4.5) .. (-6, 0)
.. controls (0, -2.5) .. (6, 0);
\path[fill=blue!70!black, opacity=0.4] (6, 0)
.. controls (0, 0.2) .. (-6, 0)
.. controls (0, -2.5) .. (6, 0);
\path[fill=orange!70!black, opacity=0.6] (6, 0)
.. controls (0, 1.5) .. (-6, 0)
.. controls (0, -1) .. (6, 0);
\path[fill=brown!70!black, opacity=0.6] (6, 0)
.. controls (0, 1.5) .. (-6, 0)
.. controls (0, 4) .. (6, 0);
\path[fill=yellow!80!black, opacity=0.4] (6, 0)
.. controls (0, 4) .. (-6, 0)
.. controls (0, 6) .. (6, 0);

\node (p) at (0, -6) [circle, draw, inner sep=0.5mm, fill=black] {};
\node (r) at (6, 0) [circle, draw, inner sep=0.5mm] {};
\node (l) at (-6, 0) [circle, draw, inner sep=0.5mm] {};

\draw [thick] (p) -- (r)
node [midway, below right] {$\mathcal{H}_R$};
\draw [thick] (p) -- (l)
node [midway, below left] {$\mathcal{H}_L$};
\draw [red] (l) .. controls (0, -6) .. (r)
node [black, midway, above] {}; 
\draw [color=green!50!black] (l) .. controls (0, -4.5) .. (r)
node [black, midway, above] {}; 
\draw [blue] (l) .. controls (0, -2.5) .. (r)
node [black, midway, above] {}; 
\draw [dashed, thick, color=orange]
(l) .. controls (0, -1) .. (r)
node [midway, above, black] {}; 
\draw [dashed, thick, color=blue] (l) .. controls (0, 0.2) .. (r)
node [midway, above, black] {}; 
\draw [orange] (l) .. controls (0, 1.5) .. (r)
node [black, midway, above] {}; 
\draw [brown!70!black, dashed, thick] 
(l) .. controls (0, 3.2) .. (r)
node [black, midway, above=-0.7mm] {}; 
\draw [dashed, brown, thick] (l) .. controls (0, 4) .. (r)
node [black, midway, above=-0.7mm] {}; 
\draw [very thick, dotted] (l) .. controls (0, 6) .. (r)
node [midway, above] {\tiny $r = 0$};

\node at (0, -5.1) {$\mathcal{R}$};
\node at (0, -3.8) {$\mathcal{N}$};
\node at (0, -2.6) {$\mathcal{EB}$};
\node at (0, -0.8) {$\mathcal{LB}$};
\node at (0, +0.2) {$\mathcal{O}$};
\node at (0, +2) {$\mathcal{PK}$};
\node at (0, +2.65) {$\blue{\mathcal{PK}'}$};
\node at (0, +3.8) {$\mathcal{K}$};
\end{tikzpicture}}

\caption{A more detailed version of Figure~\ref{Penrose_simplified}, partitioning the hairy black hole interior into the different regions $\mathcal{R}, \mathcal{N}, \mathcal{EB}, \mathcal{LB}, \mathcal{O}, \mathcal{PK}, \mathcal{K}$, to be precisely  defined in Section~\ref{sub:regions}.}
\label{Penrose_detailed}
\end{figure}
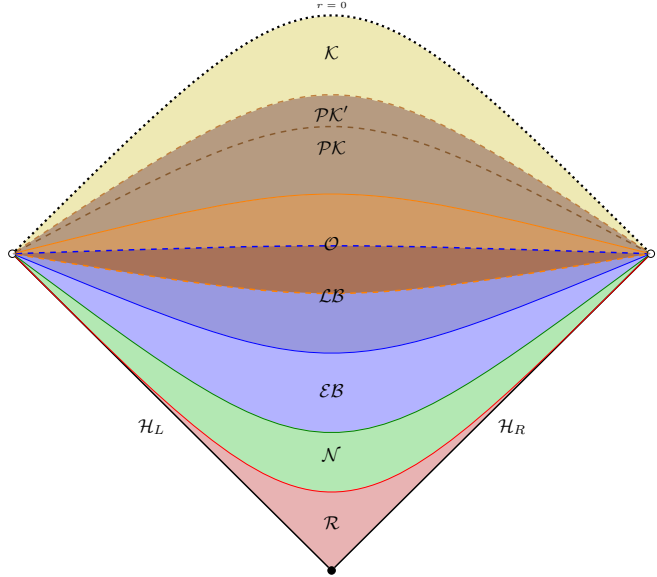 
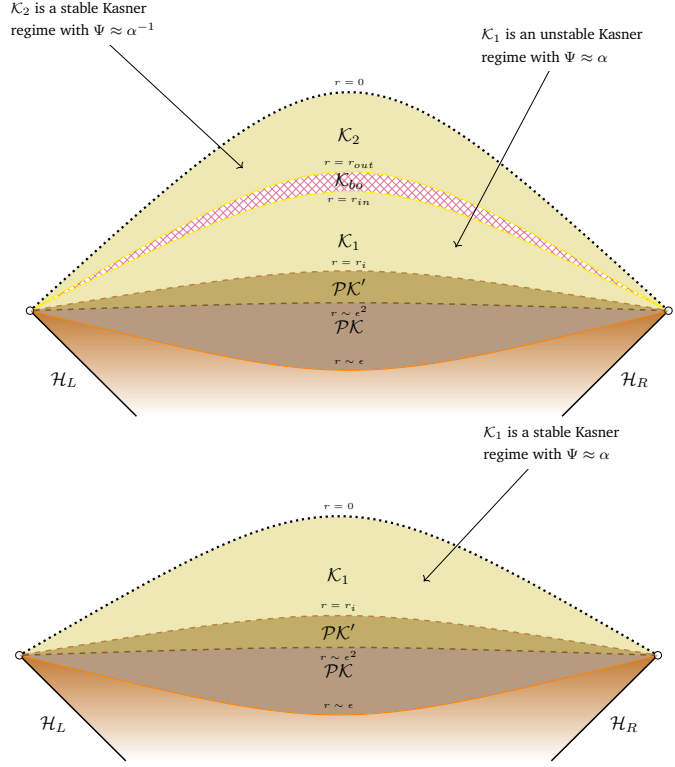
\begin{figure}

\begin{center}

\scalebox{0.7}{
\begin{tikzpicture}
\path [shade, top color=orange!70!black, opacity=0.8] (-6, 0) .. controls (0, -1.5) .. (6, 0)
-- (4, -2) -- (-4, -2) -- (-6, 0);
\path[fill=brown!70!black, opacity=0.6] (6, 0)
.. controls (0, -1.5) .. (-6, 0)
.. controls (0, +1) .. (6, 0);
\path[fill=yellow!80!black, opacity=0.4] (6, 0)
.. controls (0, 0.2) .. (-6, 0)
.. controls (0, 3) .. (6, 0);
\path[pattern=crosshatch, pattern color=purple!50!white] (6, 0)
.. controls (0, 3) .. (-6, 0)
.. controls (0, 3.5) .. (6, 0);
\path[fill=yellow!80!black, opacity=0.4] (6, 0)
.. controls (0, 5.5) .. (-6, 0)
.. controls (0, 3.5) .. (6, 0);

\node (r) at (6, 0) [circle, draw, inner sep=0.5mm] {};
\node (l) at (-6, 0) [circle, draw, inner sep=0.5mm] {};

\draw [thick] (4, -2) -- (r)
node [midway, below right] {$\mathcal{H}_R$};
\draw [thick] (-4, -2) -- (l)
node [midway, below left] {$\mathcal{H}_L$};
\draw [orange] (l) .. controls (0, -1.5) .. (r)
node [midway, above, black] {\tiny $r \sim \epsilon$};
\draw [brown, dashed, thick] (l) .. controls (0, +1) .. (r)
node [midway, above=-0.7mm, black] {\tiny $r = r_i$};
\draw [brown!70!black, dashed, thick] (l) ..controls (0, +0.2) .. (r)
node [midway, below=-0.7mm, black] {\tiny $r \sim \epsilon^2$};
\draw [yellow] (l) ..controls (0, 3) .. (r)
node [midway, below, black] {\tiny $r = r_{in}$};
\draw [yellow] (l) ..controls (0, 3.5) .. (r)
node [midway, above=-0.5mm, black] {\tiny $r = r_{out}$};
\draw [very thick, dotted] (l) .. controls (0, 5.5) .. (r)
node [midway, above] {\tiny $r=0$};

\node at (0, -0.3) {$\mathcal{PK}$};
\node at (0, +0.45) {$\blue{\mathcal{PK}'}$};
\node at (0, +1.3) {$\mathcal{K}_1$};
\node at (0, +2.45) {$\mathcal{K}_{\blue{bo}}$};
\node at (0, +3.3) {$\mathcal{K}_2$};

\draw [->] (4, 5) -- (2, 1.2);
\node [align=left, fill=white] at (4, 5) {\footnotesize $\mathcal{K}_1$ is an unstable Kasner \\ \footnotesize regime with $\Psi \approx \alpha$};
\draw [->] (-5, 5.5) -- (-2, 2.7);
\node [align=left, fill=white] at (-5, 5.5) {\footnotesize $\mathcal{K}_2$ is a stable Kasner \\ \footnotesize regime with $\Psi \approx \alpha^{-1}$};
\end{tikzpicture}}
\vspace{0.5cm}
\scalebox{0.7}{
\begin{tikzpicture}
\path [shade, top color=orange!70!black, opacity=0.8] (-6, 0) .. controls (0, -1.5) .. (6, 0)
-- (4, -2) -- (-4, -2) -- (-6, 0);
\path[fill=brown!70!black, opacity=0.6] (6, 0)
.. controls (0, -1.5) .. (-6, 0)
.. controls (0, +1) .. (6, 0);
\path[fill=yellow!80!black, opacity=0.4] (6, 0)
.. controls (0, 0.2) .. (-6, 0)
.. controls (0, 3.5) .. (6, 0);

\node (r) at (6, 0) [circle, draw, inner sep=0.5mm] {};
\node (l) at (-6, 0) [circle, draw, inner sep=0.5mm] {};

\draw [thick] (4, -2) -- (r)
node [midway, below right] {$\mathcal{H}_R$};
\draw [thick] (-4, -2) -- (l)
node [midway, below left] {$\mathcal{H}_L$};
\draw [orange] (l) .. controls (0, -1.5) .. (r)
node [midway, above, black] {\tiny $r \sim \epsilon$};
\draw [brown, dashed, thick] (l) .. controls (0, +1) .. (r)
node [midway, above=-0.5mm, black] {\tiny $r= r_i$};
\draw [brown!70!black, dashed, thick] (l) ..controls (0, +0.2) .. (r)
node [midway, below=-0.7mm, black] {\tiny $r \sim \epsilon^2$};
\draw [very thick, dotted] (l) .. controls (0, 3.5) .. (r)
node [midway, above] {\tiny $r=0$};

\node at (0, -0.3) {$\mathcal{PK}$};
\node at (0, +0.45) {$\blue{\mathcal{PK}'}$};
\node at (0, +1.5) {$\mathcal{K}_1$};

\draw [->] (4, 4) -- (1.6, 1.4);
\node [align=left, fill=white] at (4, 4) {\footnotesize $\mathcal{K}_1$ is a stable Kasner \\ \footnotesize regime with $\Psi \approx \alpha$};
\end{tikzpicture}}

\end{center}

\caption{A zoom on  $\mathcal{PK} \cup \mathcal{K}$ in  Figure~\ref{Penrose_detailed}, with the top picture representing the \blue{bounce} case ($\mathcal{K}_2, \mathcal{K}_{\blue{bo}}\neq \emptyset$) and the bottom picture  the \blue{no bounce} case ($\mathcal{K}_2, \mathcal{K}_{\blue{bo}}= \emptyset$). Note the inclusion $\blue{\mathcal{PK}'} \subset \mathcal{PK}$. }
\label{FigN}
\end{figure} 

\paragraph{Acknowledgements} We are grateful to Mihalis Dafermos for suggesting the construction of Figure~\ref{construction}  and for useful comments on the manuscripts. We also would like to thank Grigorios Fournodavlos, Jonathan Luk and Yakov Shlapentokh-Rothman for useful comments on the manuscript, and Jorge Santos for helpful discussions, \blue{as well as an anonymous referee for extremely detailed and thorough suggestions which have vastly improved the manuscript.}

\section{Geometric set-up and preliminaries}\label{prelim.section}

\subsection{Einstein-Maxwell-Klein-Gordon in double null coordinates}

We consider a spherically symmetric Lorentzian metric $(M, g)$ with a choice of double null coordinates $(u, v)$:
\begin{equation} \label{eq:metric}
g = - g_{\mathcal{Q}} + r^2 (u, v) d \sigma_{\mathbb{S}^2} = - \Omega^2(u, v) du dv + r^2 (u, v) d \sigma_{\mathbb{S}^2}.
\end{equation}
Here $(u, v)$ are coordinates on the quotient manifold $\mathcal{Q} = M / SO(3)$ and $d\sigma_{\mathbb{S}^2} = d \theta^2 + \sin^2 \theta d \varphi^2$ is the standard metric on the unit sphere. We call $r = r(u, v)$ the \textit{area-radius} function. 

Due to the presence of charged scalar matter, the Maxwell field will itself be dynamical, and is described via the following function $Q(u, v)$ on $\mathcal{Q}$:
\begin{equation} \label{eq:em_doublenull}
F = \frac{Q \Omega^2}{2r^2} du \wedge dv.
\end{equation}
To describe the coupling to the scalar field, we must choose a gauge for the Maxwell field. In spherical symmetry, we specify the gauge using a one-form $A = A_u du + A_v dv$ on $\mathcal{Q}$ which satisfies $dA = F$. 

Define the covariant derivative by $D_{\mu} = \nabla_{\mu} + i q_0 A_{\mu}$. Then the scalar field $\phi$ is a complex-valued function on $\mathcal{Q}$ satisfying the following covariant wave equation:
\begin{equation}
g^{\mu \nu} D_{\mu} D_{\nu} \phi = 0.
\end{equation}
Recall that the whole system of equations must be invariant under the gauge transformation $A \mapsto A + df$, $\phi \mapsto \phi e^{- i q_0 f}$, where $f$ is any smooth function on $\mathcal{Q}$.

We make a few more standard definitions. The \textit{Hawking mass} $\rho$ is given by
\begin{equation} \label{eq:hawking_mass}
\rho \coloneqq \frac{r}{2} ( 1 - g_{\mathcal{Q}}(\nabla r, \nabla r)) = \frac{r}{2} ( 1 - 4 \Omega^{-2} \partial_u r \partial_v r).
\end{equation}
In the presence of the Maxwell field and the cosmological constant, we further define the renormalized Vaidya mass $\varpi$ and the $r$-constant surface gravity $2K$ as
\begin{equation} \label{eq:vaidya_surface}
\varpi = \rho + \frac{Q^2}{2r} - \frac{\Lambda r^3}{6}, \hspace{0.5cm} 2K = \frac{2}{r^2} \left( \varpi - \frac{Q^2}{r} - \frac{\Lambda r^3}{3} \right).
\end{equation}

Suppose that $(M, g, F, \phi)$ are a solution to the Einstein-Maxwell-Klein-Gordon system \eqref{E1}--\eqref{E5}. In the double-null coordinates of (\ref{eq:metric}), the quantities $(r, \Omega^2, Q, A, \phi)$ then satisfy the following system of PDEs:

\begin{equation} \label{eq:raych_u_emkgss}
\partial_u ( \Omega^{-2} \partial_u r ) = - \Omega^{-2} r |D_u \phi|^2,
\end{equation}
\begin{equation} \label{eq:raych_v_emkgss}
\partial_v ( \Omega^{-2} \partial_v r ) = - \Omega^{-2} r |D_v \phi|^2,
\end{equation}
\begin{equation} \label{eq:wave_r_emkgss}
\partial_u \partial_v r = - \frac{\Omega^2}{4r} - \frac{\partial_u r \partial_v r}{r} + \frac{\Omega^2 Q^2}{4r^3} + \frac{\Omega^2r (m^2 |\phi|^2 + \Lambda)}{4} ,
\end{equation}
\begin{equation} \label{eq:wave_omega_emkgss}
\partial_u \partial_v \log(\Omega^2) = \frac{\Omega^2}{2r^2} + \frac{2 \partial_u r \partial_v r}{r^2} - \frac{\Omega^2 Q^2}{r^4} - 2 \mathfrak{Re} (D_u \phi \overline{D_v  \phi}),
\end{equation}
\begin{equation} \label{eq:q_u_emkgss}
\partial_u Q = - q_0 r^2 \mathfrak{Im} (\phi \overline{D_u \phi}),
\end{equation}
\begin{equation} \label{eq:q_v_emkgss}
\partial_v Q = + q_0 r^2 \mathfrak{Im} (\phi \overline{D_v \phi}),
\end{equation} 
\begin{equation} \label{eq:wave_psi_emkgss}
D_u D_v \phi = - \frac{\partial_u r \cdot D_v \phi}{r} - \frac{\partial_v r \cdot D_u \phi}{r} + \frac{i q_0 Q \Omega^2}{4r^2} \phi - \frac{m^2 \Omega^2}{4} \phi,
\end{equation}
\begin{equation} \label{eq:emgauge_emkgss}
\partial_u A_v - \partial_v A_u = \frac{Q \Omega^2}{2 r^2}.
\end{equation}
The equations (\ref{eq:raych_u_emkgss}) and (\ref{eq:raych_v_emkgss}) are the celebrated \textit{Raychaudhuri equations}, the equations (\ref{eq:wave_r_emkgss}) and (\ref{eq:wave_omega_emkgss}) can be viewed as wave equations for the geometric quantities $r$ and $\Omega^2$ on $\mathcal{Q}$, and the remaining equations describe the dynamics of the coupled Maxwell field and charged scalar field.

We recall also the transport equations for the Vaidya mass $\varpi$:
\begin{equation} \label{eq:modifiedhawkingmass_u_emkgss}
\partial_u \varpi = - 2r^2 (\Omega^{-2} \partial_v r)^{-1} | D_u \phi |^2 + \frac{m^2}{2} r^2 |\phi|^2 \partial_u r - q_0 Q r \mathfrak{Im} (\phi \overline{D_u \phi}),
\end{equation}
\begin{equation} \label{eq:modifiedhawkingmass_v_emkgss}
\partial_v \varpi = - 2r^2 (\Omega^{-2} \partial_u r)^{-1} | D_v \phi |^2 + \frac{m^2}{2} r^2 |\phi|^2 \partial_v r + q_0 Q r \mathfrak{Im} (\phi \overline{D_v \phi}).
\end{equation}

\begin{rmk} \label{planar.rmk}
    \blue{In the case of planar symmetry, the only changes to the equations are that the first term on the right hand sides of \eqref{eq:wave_r_emkgss} and \eqref{eq:wave_omega_emkgss} are removed. Our results remain valid in this symmetry class.}
\end{rmk}

\subsection{The Reissner-Nordstr\"om(-dS/AdS) interior metric} \label{sub:reissner_nordstrom}

We are interested in charged hairy perturbations of sub-extremal Reissner-Nordstr\"om interiors. To define sub-extremality, given some parameters $M > 0$, $\mathbf{e}, \Lambda \in \R$, consider the polynomial
\begin{equation} \label{eq:rn_polynomial}
P_{M, \mathbf{e}, \Lambda}(X) = X^2 - 2 M X + \mathbf{e}^2 - \tfrac{1}{3} \Lambda X^4.
\end{equation}

Then the set of subextremal-parameters $(M, \mathbf{e}, \Lambda)$ is $\blue{\mathcal{P}_{se}} = \blue{ \mathcal{P}_{se}^{\Lambda < 0} \cup \mathcal{P}_{se}^{\Lambda = 0} \cup  \mathcal{P}_{se}^{\Lambda > 0}}$, where $\mathcal{P}_{se}^{\Lambda < 0}$ is such that $\Lambda < 0$ and the polynomial $P_{M, \mathbf{e}, \Lambda}(X)$ has two distinct positive real roots $r_- < r_+$, \blue{$\mathcal{P}_{se}^{\Lambda = 0}$ is such that $\Lambda = 0$ and the polynomial $P_{M, \mathbf{e}, \Lambda}(X)$ has two distinct positive real roots $r_- < r_+$, while $\mathcal{P}_{se}^{\Lambda > 0}$ is such that $\Lambda > 0$ and $P_{M, \mathbf{e}, \Lambda}(X)$ has \underline{three} distinct positive real roots $r_- < r_+ < r_c$.}

The Reissner-Nordstr\"om(-dS/AdS) spacetime is a solution to (\ref{E1})--(\ref{E5}) in electrovacuum (i.e.\ $\phi \equiv 0$), and can be written in standard $(t, r)$ coordinates as
\begin{equation} \label{eq:reissner_nordstrom1}
g_{RN} = - \left ( 1 - \frac{2M}{r} + \frac{\mathbf{e}^2}{r^2} - \frac{\Lambda r^2}{3}\right ) dt^2 + \left ( 1 - \frac{2M}{r} + \frac{\mathbf{e}^2}{r^2} - \frac{\Lambda r^2}{3} \right )^{-1} dr^2 + r^2 d \sigma_{\mathbb{S}^2}.
\end{equation}  
In particular, the Reissner-Nordstr\"om(-dS/AdS) interior metric is given by (\ref{eq:reissner_nordstrom1}), restricted to the coordinate range $r_- < r < r_+$, $t \in \R$. Note that, in the interior, $t$ is a spacelike coordinate while $r$ is a timelike coordinate. 

The Maxwell field is given by \blue{setting} $Q \equiv \mathbf{e}$ in (\ref{eq:em_doublenull}), and $\Omega^2 = \Omega^2_{RN}$ will be defined shortly. One choice of gauge field $A$ which will be consistent with the remainder of this article is
\begin{equation} \label{eq:reissner_nordstrom_gauge}
A = - \left( \frac{\mathbf{e}}{r_+} - \frac{\mathbf{e}}{r} \right) dt.
\end{equation}

To recast the metric (\ref{eq:reissner_nordstrom1}) into the double null form (\ref{eq:metric}), we define
\begin{equation}
\frac{dr}{dr^*} \coloneqq \frac{\Omega^2_{RN}}{4}, \hspace{0.5cm} \Omega^2_{RN} \coloneqq - 4 \left( 1 - \frac{2M}{r} + \frac{\mathbf{e}^2}{r^2} - \frac{\Lambda r^2}{3} \right),
\end{equation}
\begin{equation}
u \coloneqq \frac{r^* - t}{2}, \hspace{0.5cm} v \coloneqq \frac{r^* + t}{2}.
\end{equation}
%
In this $(u, v)$ coordinate system, the metric can now be written as
\begin{equation} \label{eq:reissner_nordstrom}
g_{RN} = - \Omega^2_{RN} du dv + r^2 d \sigma_{\mathbb{S}^2}.
\end{equation}
In the sequel, we denote the Reissner-Nordstr\"om area-radius function by $r_{RN}$.

Recalling the definition of $2K$ in (\ref{eq:vaidya_surface}), we define the surface gravity of the event horizon $2K_+$, and the surface gravity of the Cauchy horizon $2K_-$ by
\begin{equation} \label{eq:surface_gravity}
2K_{\pm} = 2K(r = r_{\pm}) = \frac{2}{r_{\pm}^2} \left( M - \frac{\mathbf{e}^2}{r_{\pm}} - \frac{\Lambda r_{\pm}^3}{3} \right).
\end{equation}

It is then a well-known fact that the null lapse $\Omega^2_{RN}$ obeys the following asymptotics, where $\alpha_{\pm} > 0$ are fixed constants depending on the black hole parameters:
\begin{equation} \label{eq:omega_asymptotic_1}
\log( \frac{\Omega^2_{RN}}{\alpha_+ e^{2 K_+ (M, \mathbf{e}, \Lambda) \cdot r^*}} ) = O ( e^{2 K_+ (M, \mathbf{e}, \Lambda) \cdot r^*}) \; \text{ as } \; r^* = u + v \to - \infty,
\end{equation}
\begin{equation} \label{eq:omega_asymptotic_2}
\log( \frac{\Omega^2_{RN}}{\alpha_- e^{2 K_- (M, \mathbf{e}, \Lambda) \cdot r^*}} ) = O ( e^{2 K_- (M, \mathbf{e}, \Lambda) \cdot r^*}) \; \text{ as } \; r^* = u + v \to + \infty.
\end{equation}
We note that, \blue{for subextremal black hole parameters, one has} $2K_+ > 0$ while $2K_- < 0$.

Introducing the following ``regular coordinates'', 
\begin{equation} \label{eq:coord_trans}
U = e^{2K_+u}, \quad V =  e^{2K_+ v},
\end{equation}
it is well-known that in the coordinate system $(U, v)$, the metric $g_{RN}$ can be smoothly extended beyond $U = 0$, and the right event horizon $\mathcal{H}_R = \{ (U, v): U = 0, v \in [- \infty, + \infty)\}$ is realised as a smooth null hypersurface. A similar construction can be made for the coordinates $(u, V)$, with $\mathcal{H}_L = \{ (u, V): u \in [-\infty, + \infty), V = 0 \}$.

Indeed, using the coordinate system $(U, V)$, the metric $g_{RN}$ is defined for $0 \leq U, V < + \infty$, and can be smoothly beyond both $\mathcal{H}_R$ and $\mathcal{H}_L$, including the bifurcation sphere $\mathcal{H}_R \cap \mathcal{H}_L = \{ U = 0, V = 0 \}$.

\subsection{Black hole interiors with charged scalar hair} \label{sub:data}

\begin{figure}[ht]
\centering

\begin{tikzpicture}[domain=-5:5]
\coordinate (p) at  (0,0) node[below] {};
\coordinate (L) at  (-4.5,4.5);
\coordinate (R) at  (4.5,4.5);

\draw (p) --
node[midway, below, sloped] {$\mathcal{H}_R$}
(R);
\draw (p) --
node[midway, below, sloped] {$\mathcal{H}_L$}
(L);

\path[fill=lightgray, opacity=0.7] (p) -- (1.5, 1.5)
-- (0, 3)
-- (-1.5, 1.5) -- (p);

\draw plot[id=x,domain=-4.3:4.3, samples=100]
(\x, {sqrt(4 + \x * \x});

\fill [lightgray, opacity=0.5]
(p)
-- (R) -- (4, 5)
-- plot[domain = 5:-5] ({\x}, {sqrt(9 + \x * \x)})
-- (-4, 5) -- (L)
-- cycle;

\draw[->, very thick]
(-2, {sqrt(8)}) -- ++({sqrt(8)/5}, -2/5)
node [above, midway] {$T$};
\draw[->, very thick]
(+1.5, {sqrt(6.25)}) -- ++({sqrt(6.25)/5}, +1.5/5)
node [above, midway] {$T$};
\node at (0, 2) [anchor=north] {\scriptsize $s = \text{const}$};
\node at (0, 1) {$\mathcal{R}$};
\end{tikzpicture}
\caption{The local solution to the characteristic initial value problem for a spatially homogeneous black hole interior with charged scalar hair} \label{fig:char_ivp}
\end{figure}
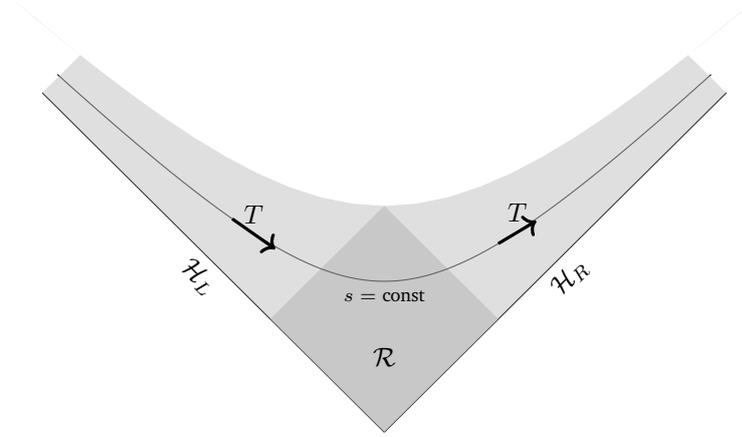

\noindent
Consider the characteristic initial value problem with initial data given on the two affine complete null hypersurfaces $\mathcal{H}_L = \{ (U, V): U \geq 0, V = 0 \} = \{ (u, v): u \in [- \infty, + \infty), v = - \infty \}$ and $\mathcal{H}_R = \{ (U, V): U = 0, V \geq 0 \} = \{ (u, v): u = -\infty, v \in [- \infty, + \infty) \}$, intersecting at the bifurcation sphere $(U, V) = (0, 0)$. 

We normalize the regular coordinates $(U, V)$, which are related to the usual interior coordinates $(u,v)$ via (\ref{eq:coord_trans}), using the following gauge choice:
\begin{equation} \label{eq:gauge_v}
\Omega_R^2(U, v) |_{\mathcal{H}_R} = \frac{1}{2K_+} e^{-2K_+(M ,\mathbf{e}, \Lambda) \cdot u} \, \Omega^2(u, v) |_{\mathcal{H}_R} = \frac{\alpha_+}{2K_+} e^{2K_+ (M, \mathbf{e}, \Lambda) \cdot v},
\end{equation}
\begin{equation} \label{eq:gauge_u}
\Omega_L^2(u, V) |_{\mathcal{H}_L} = \frac{1}{2 K_+} e^{-2K_+(M ,\mathbf{e}, \Lambda) \cdot v} \, \Omega^2(u, v) |_{\mathcal{H}_L} = \frac{\alpha_+}{2 K_+} e^{2K_+ (M, \mathbf{e}, \Lambda) \cdot u},
\end{equation}
\blue{These choices are made such that,} \red{with respect to the} \emph{generalized Kruskal-Szekeres} coordinate system $(U, V)$, \red{$\Omega^2$ is normalized as follows on $\mathcal{H} = \mathcal{H}_R \cup \mathcal{H}_L$:}
\begin{equation} \label{eq:gauge_data} \tag{$\Omega$-data}
\Omega^2_{reg}(U, V) |_{\mathcal{H}} = \frac{\alpha_+}{(2 K_+)^2}.
\end{equation}

In the context of this article, we pose the following characteristic initial data:
\begin{gather}
\tag{$r$-data} \label{eq:r_data}
r |_{\mathcal{H}} = r_+ (M ,\mathbf{e}, \Lambda), \\
\tag{$\varpi$-data} \label{eq:m_data}
\varpi|_{\mathcal{H}} =M>0, \\
\tag{$Q$-data} \label{eq:Q_data}
Q |_{\mathcal{H}} = \mathbf{e}\in \R \setminus \{0\},\\
\tag{$\phi$-data} \label{eq:phi_data}
\phi |_{\mathcal{H}} = \epsilon \in \R \setminus \{0\},
\end{gather}
as well as a gauge for the Maxwell field such that the components $A_V, A_U$ vanish on $\mathcal{H}_R, \mathcal{H}_L$ respectively. In particular $D_V \phi$ and $D_U \phi$ will vanish on their respective horizon pieces. \blue{As a result}, this data is compatible with the null constraints (\ref{eq:raych_u_emkgss}), (\ref{eq:raych_v_emkgss}). 

\blue{In the following proposition, we show that} the data (\ref{eq:gauge_data}), (\ref{eq:r_data}), (\ref{eq:Q_data}), (\ref{eq:phi_data}) uniquely specifies a solution to the Einstein-Maxwell-Klein-Gordon system, which \blue{is} moreover spatially homogeneous. Of course, in order to have uniqueness we must impose a gauge for the Maxwell field $A$. \red{In our setting, for the purpose of reducing the Einstein-Maxwell-Klein-Gordon system to a system of ODEs} (see already \red{Section~\ref{sub:evol_eqs}} and \eqref{eq:gauge_maxwell}), it will be convenient to choose
\begin{equation} \tag{$A$-gauge} \label{eq:maxwell_choice}
U A_U + V A_V = 0.
\end{equation}
We now identify the spacetime describing the hairy black hole interior spacetimes studied in this article, which we firstly describe in the regular coordinate system $(U, V)$.

\begin{proposition} \label{prop:char_ivp}
Consider characteristic initial data (\ref{eq:gauge_data}), (\ref{eq:r_data}), (\ref{eq:Q_data}), (\ref{eq:phi_data}) to the Einstein-Maxwell-Klein-Gordon system (\ref{eq:raych_u_emkgss})--(\ref{eq:emgauge_emkgss}), with $u, v$ replaced by $U, V$ respectively. Then imposing also (\ref{eq:maxwell_choice}), there exists a unique maximal future development\footnote{Maximality is meant in the sense that there is no larger causal region of the $(U, V)$-plane where one may smoothly extend the solution.} of the system $(\Omega^2, r, \phi, Q, A_U, A_V)$ \red{that is smooth} up to and including the horizon $\mathcal{H} = \mathcal{H}_L \cup \mathcal{H}_R$.

Furthermore, the domain of definition of this maximal development is given by $\{ (U, V): 0 \leq UV < D_{max} \}$ for some $D_{max}(M, \mathbf{e}, \Lambda, m^2, q_0, \epsilon) \in (0, + \infty]$, and letting $T$ be the vector field (see Figure \ref{fig:char_ivp})
\begin{equation} \label{eq:kvf}
T \coloneqq 2K_+ \left( - U \frac{\partial}{\partial U} + V \frac{\partial}{\partial V} \right),
\end{equation}
one finds that $T$ is a Killing vector field, satisfying $T r = T \Omega^2 = T Q = T \phi = T (U A_U) = 0$. In particular, the spacetime is foliated by Cauchy hypersurfaces $UV = const$, which are each spatially homogeneous with isometry group $\R \times SO(3)$.
\end{proposition}

\begin{proof}
\blue{We wish to appeal to a local existence result, for instance \cite[Proposition 4.1]{Kommemi}, with the caveat that our gauge choice (\ref{eq:maxwell_choice}) is not well-suited to such local well-posedness statements. Following \cite{Kommemi}, we firstly instead find a solution with respect to the Maxwell gauge choice:
\begin{equation} \tag{$A$-gauge'} \label{eq:maxwell_choice'}
A^{(0)}_V (U, V) = 0, \; A^{(0)}_U (U, 0) = 0.
\end{equation}
\cite{Kommemi} asserts that there exists a unique solution to (\ref{eq:raych_u_emkgss})--(\ref{eq:emgauge_emkgss}) for $(U, V) \in [0, \delta] \times [0, \delta]$ with $\delta$ sufficiently small (see the darker shaded region of Figure~\ref{fig:char_ivp}), attaining the prescribed data and using the gauge (\ref{eq:maxwell_choice'}), with $A$ replaced by $A^{(0)}$ in the equations. In particular, Maxwell gauge-independent quantities such as $\Omega^2_{reg}, r, Q$ are already uniquely determined.}

We seek a gauge transformation that relates an $A^{(0)}$ satisfying (\ref{eq:maxwell_choice'}) to an $A$ satisfying (\ref{eq:maxwell_choice}). The correct gauge transformation is $A = A^{(0)} - dh$, with $h$ given by
\begin{equation} \label{eq:h}
h(U, V) \coloneqq \int^1_0 U A^{(0)} (U\blue{x}, V\blue{x}) \, d\blue{x}.
\end{equation}
The reason is that this choice of $h$ implies that
\begin{align*}
U \frac{\partial h}{\partial U} + V \frac{\partial h}{\partial V} 
&=
\int^1_0 \left( U A^{(0)} (Ux, Vx) + U^2 x \frac{\partial A^{(0)}}{\partial U} (Ux, Vx) + UVx \frac{\partial A^{(0)}}{\partial V} (Ux, Vx) \right) \, dx, \\
&= \int^1_0 \frac{d}{dx} (U x A^{(0)}(Ux, Vx)) \, dx = U A^{(0)} (U, V),
\end{align*}
so that
\begin{equation*}
U A_U + V A_V = U A^{(0)} - U \frac{\partial h}{ \partial U} - V \frac{\partial h}{\partial V} = 0,
\end{equation*}
as required. Furthermore, the $h$ chosen in (\ref{eq:h}) is the unique such gauge transformation that is regular at $(0,0)$: If $\tilde{h}$ was another such function, then the difference $g = h - \tilde{h}$ would satisfy
\begin{equation*}
U \frac{\partial g}{\partial U} + V \frac{\partial g}{\partial V} = 0,
\end{equation*}
whose general solution is of the form $g(U, V) = G( U / V)$. So regularity at $(0,0)$ implies that $G$, and thus $g$, is constant, and $dh = d \tilde{h}$ after all.

Hence we have constructed a unique regular solution in the characteristic rectangle $[0, \delta] \times [0, \delta]$. We next show that the vector field $T$ defined in (\ref{eq:kvf}) annihilates all the relevant quantities. 
For this purpose, we argue geometrically as follows: let $a > 0$ be any positive real number, and consider the double null coordinate transformation $U \mapsto U' = a U, V \mapsto V' = a^{-1} V$. Then in the $(U', V')$ coordinate system, we note that it still holds that on $\mathcal{H} = \{ (U', V'): U' = 0 \text{ or } V' = 0\}$, we have
\begin{gather*}
\Omega'^2_{reg} (U', V') = \frac{\alpha_+}{(2 K_+)^2}, \\
r = r_+,\\ Q = \mathbf{e},\\ \phi = \epsilon, \\
U' A_{U'} + V' A_{V'} = 0.
\end{gather*}
Hence by the aforementioned existence and uniqueness result, we have a unique solution in the characteristic rectangle $(U', V') \in [0, \delta] \times [0, \delta]$. Furthermore, this agrees with the solution in the original $(U, V)$ coordinates, so that, for $f \in \{ \Omega^2, r, Q, \phi, UA_U \}$, we have
\begin{equation} \label{eq:magic}
f(U', V') = f(a U, a^{-1} V) = f(U, V).
\end{equation}
Allowing $a$ to vary across all positive reals, it is clear that we have a solution in the whole of $\{ (U, V) \in \mathbb{R}_{\geq 0}^2: 0 \leq UV \leq \delta^2 \}$, i.e.\ the lighter shaded region in Figure \ref{fig:char_ivp} that arises from sweeping out the darker shaded region for different choices of $a > 0$. Furthermore, it is immediate from (\ref{eq:magic}) that $T$ annihilates all such quantities $f$.

The extension to the region $\{ 0 \leq UV < D_{max} \}$ is then straightforward using standard extension principle arguments for nonlinear waves and again appealing to this geometric trick.
\end{proof}

\begin{rmk}
    By the generalized extension principle of \cite{Kommemi}, if the quantity $D_{max}$ of Proposition \ref{prop:char_ivp} is finite, then one must have $r \to 0$ as $UV \to D_{max}$. However, Proposition \ref{prop:char_ivp} is qualitative in nature and says little about quantitative properties of the interior, or if and how any spacelike singularity is formed. 
\end{rmk}

Following Proposition~\ref{prop:char_ivp}, we would also like to understand the transversal derivatives of $r$ and $\phi$ along $\mathcal{H}$, for which we shall need the full system of equations (\ref{eq:raych_u_emkgss})--(\ref{eq:emgauge_emkgss}). 
For $r$, we see that using (\ref{eq:wave_r_emkgss}) on $\mathcal{H}_R$,
\begin{equation} \label{eq:initial_data_r}
\partial_V \partial_U r = \frac{\alpha_+}{4 (2 K_+)^2} \left( -\frac{1}{r_+} + \frac{\mathbf{e}^2}{r_+^3} + r_+ \Lambda+ r_+ m^2 \epsilon^2 \right).
\end{equation}
Since $r_+ = r_+(M, \mathbf{e}, \Lambda)$ satisfies the equation
\begin{equation*}
P_{M, \mathbf{e}, \Lambda} (r_+) = r_+^2 - 2 M r_+ + \mathbf{e}^2 - \tfrac{1}{3} \Lambda r_+^4 = 0,
\end{equation*}
it is readily checked that the expression in the parentheses in (\ref{eq:initial_data_r}) is equal to $ - 2 K_+ + r_+ m^2 \epsilon^2$. We then integrate (\ref{eq:initial_data_r}), noting that $\partial_U r = 0$ at the bifurcation sphere $(U, V) = (0, 0)$, to find
\begin{equation*}
\partial_U r |_{\mathcal{H}_R} = \frac{\alpha_+ V}{8 K_+} \left( - 1 + \frac{r_+ m^2 \epsilon^2}{2K_+} \right).
\end{equation*}

Returning to $(u, v)$ coordinates, and performing a similar procedure on $\mathcal{H}_L$, we deduce
\begin{equation}
\lim_{u \to - \infty} \frac{- 4 \partial_u r}{\Omega^2} (u, v) = 
\lim_{v \to - \infty} \frac{- 4 \partial_v r}{\Omega^2} (u, v) =  1 - \frac{r_+ m^2 \epsilon^2}{2K_+}.
\end{equation}
We remark here that due to the presence of the Klein-Gordon mass, there is already an $O(\epsilon^2)$ deviation from the corresponding Reissner-Nordstr\"om quantity.
A similar procedure applied to (\ref{eq:wave_psi_emkgss}) will yield
\begin{gather}
\frac{2 K_+}{\alpha_+ V} \cdot D_U \phi (0, V) = \lim_{u \to - \infty} \frac{D_u \phi}{\Omega^2} (u, v) = \beta_+ \epsilon, \\
\frac{2 K_+}{\alpha_+ U} \cdot D_V \phi (U, 0) = \lim_{v \to - \infty} \frac{D_v \phi}{\Omega^2} (u, v) = \beta_+ \epsilon,
\end{gather}
where $\beta_+ = \beta_+(M, \mathbf{e}, \Lambda, m)$ is some fixed constant we do not explicitly determine.

\subsection{System of ODEs for spatially homogeneous solutions} \label{sub:evol_eqs}

Define $s = u + v$, $t = v - u$, where the null coordinates $(u, v)$ are fixed by the gauge choices (\ref{eq:gauge_v}), (\ref{eq:gauge_u}), (\ref{eq:coord_trans}). Since $\partial_t = \frac{1}{2} ( \partial_v - \partial_u ) = \frac{1}{2} T$, Proposition~\ref{prop:char_ivp} proves that the maximal future development arising from the characteristic data of Section~\ref{sub:data} obeys $\partial_t r = 0, \partial_t \Omega^2 = 0, \partial_t Q = 0, \partial_t \phi = 0$.
So we may consider these just as functions of $s$.

Of course, this is only true after imposing (\ref{eq:maxwell_choice}).
In the $(u, v)$ coordinate system, we notice that
\begin{equation} \label{eq:gauge_maxwell}
A = A_U dU + A_V dV = 2 K_+ (U A_U du + V A_V dv ) \eqqcolon \tilde{A} (du - dv) = - \tilde{A} dt,
\end{equation}
where, due to Proposition \ref{prop:char_ivp}, $\tilde{A} = UA_U(s)$ is a real-valued function of $s$, with $\lim_{s \to -\infty} \tilde{A}(s) = 0$.
We next show that this choice of gauge will in fact constrain the scalar field $\phi$ to be real.

\begin{lemma}
With the gauge choice (\ref{eq:gauge_maxwell}) and the initial data of Section \ref{sub:data}, $\phi = \phi(s)$ is everywhere real.
\end{lemma}

\begin{proof}
Consider the transport equations (\ref{eq:q_u_emkgss}), (\ref{eq:q_v_emkgss}) for the quantity $Q$. Since $\partial_t Q = \frac{1}{2} (\partial_v - \partial_u) Q = 0$,
\begin{equation*}
q_0 r^2 \mathfrak{Im} ( \phi \overline{ D_s \phi } ) 
= \tfrac{1}{2} q_0 r^2 \mathfrak{Im} ( \phi \overline{ (D_u \phi + D_v \phi) } ) = 0.
\end{equation*}
Hence by (\ref{eq:gauge_maxwell}) we must have $\mathfrak{Im} ( \phi \overline{D_s \phi} ) = \mathfrak{Im} (\phi \overline{\partial_s \phi}) = 0$.

Next, we decompose $\phi$ into its modulus and argument; $\phi = \Phi e^{i \theta}$. Then
\begin{equation*}
\mathfrak{Im} ( \phi \overline{\partial_s \phi} ) = - \Phi^2 \partial_s \theta = 0,
\end{equation*}
so that the phase $\theta$ is constant whenever $\phi$ is nonzero. But as $\phi$ is smooth in the variable $s$, it does not change phase when it reaches $0$, hence $\phi$ is real everywhere.
%
\end{proof}

\noindent
Using the identities $\partial_u = \partial_s - \partial_t, \partial_v = \partial_s + \partial_t$, we have that\blue{,} for $f \in \{ r(s), \Omega(s), \phi(s), Q(s), \tilde{A}(s) \}$, one has
\begin{equation*}
\partial_u f = \partial_v f = \frac{df}{ds} \eqqcolon \dot{f}.
\end{equation*}
We now proceed to rewrite the Einstein-Maxwell-Klein-Gordon system (\ref{eq:raych_u_emkgss})--(\ref{eq:emgauge_emkgss}) as a system of ODEs.

Firstly, the Raychaudhuri equation becomes
\begin{equation} \label{eq:raych}
\frac{d}{ds} (- \Omega^{-2} \dot{r}) = \Omega^{-2} r ( |\dot{\phi}|^2 + |\tilde{A}|^2 q_0^2 |\phi|^2 ).
\end{equation}
Defining the quantity $\kappa = - \frac{1}{4} \Omega^2 \dot{r}^{-1}$, which is exactly $1$ in Reissner-Nordstr\"om(-dS/AdS), we may rewrite this as
\begin{equation} \label{eq:raych_kappa}
\frac{d}{ds} \kappa^{-1} = 4 \Omega^{-2} r (|\dot{\phi}|^2 + |\tilde{A}|^2 q_0^2 |\phi|^2 ).
\end{equation}
We also at times appeal to \eqref{eq:raych} in the form:
\begin{equation} \label{eq:raych_transport}
\frac{d}{ds} (- \dot{r}) - \frac{d}{ds} \log (\Omega^2) \cdot (- \dot{r}) = r ( |\dot{\phi}|^2 + |\tilde{A}|^2 q_0^2 |\phi|^2).
\end{equation}

The wave equation for $r$ is now written as
\begin{equation} \label{eq:r_evol} 
\ddot{r} = - \frac{\Omega^2}{4r} - \frac{\dot{r}^2}{r} + \frac{\Omega^2}{4r^3} Q^2 + \frac{\Omega^2 r}{4} (m^2 |\phi|^2 + \Lambda), 
\end{equation}
which may be conveniently rewritten as
\begin{equation} \label{eq:r_evol_2}
\frac{d}{ds} ( - r \dot{r}) = \frac{\Omega^2}{4} - \frac{\Omega^2 Q^2}{4 r^2} - \frac{\Omega^2 r^2}{4} (m^2|\phi|^2 + \Lambda).
\end{equation}

The wave equation (\ref{eq:wave_omega_emkgss}) for the null lapse $\Omega^2$ becomes
\begin{equation} \label{eq:lapse_evol}
\frac{d^2}{ds^2} \log(\Omega^2) = \frac{\Omega^2}{2r^2} + 2 \frac{\dot{r}^2}{r^2} - \frac{\Omega^2}{r^4} Q^2 - 2 \dot{\phi}^2 + 2 |\tilde{A}|^2 q_0^2 |\phi|^2,
\end{equation}
or alternatively
\begin{equation} \label{eq:omega_evol_2}
\frac{d^2}{ds^2} \log (r \Omega^2) = \frac{\Omega^2}{4r^2} - \frac{3}{4} \frac{\Omega^2 Q^2}{r^4} - 2 |\dot{\phi}|^2 + 2 |\tilde{A}|^2 q_0^2 |\phi|^2 + \frac{\Omega^2 m^2}{4} |\phi|^2.
\end{equation}

For the Maxwell field $Q$ and the gauge field $\tilde{A}$, the equations (\ref{eq:q_u_emkgss}), (\ref{eq:q_v_emkgss}), (\ref{eq:emgauge_emkgss}) become
\begin{equation} \label{eq:Q_evol}
\dot{Q} = \tilde{A} q_0^2 r^2 |\phi|^2,
\end{equation}
\begin{equation} \label{eq:gauge_evol}
\dot{\tilde{A}} = -\frac{Q \Omega^2}{4r^2}.
\end{equation}

Finally, the wave equation (\ref{eq:wave_psi_emkgss}) for the scalar field may be written as the second order ODE
\begin{equation} \label{eq:phi_evol}
\ddot{\phi} = - \frac{2 \dot{r} \dot{\phi}}{r} - q_0^2 |\tilde{A}|^2 \phi - \frac{m^2 \Omega^2}{4} \phi,
\end{equation}
which we often use in the form
\begin{equation} \label{eq:phi_evol_2}
\frac{d}{ds} (r^2 \dot{\phi}) = - r^2 q_0^2 |\tilde{A}|^2 \phi - \frac{m^2 \Omega^2 r^2}{4} \phi.
\end{equation}

We also reformulate the initial data of Section~\ref{sub:data} so as to satisfy the ODE system (\ref{eq:raych})--(\ref{eq:phi_evol_2}). Data is posed at the limit $s \to -\infty$, and is given by
\begin{gather}
\lim_{s \to -\infty} r(s)  = r_+,\; \lim_{s\to -\infty} Q(s) = \mathbf{e},\; \lim_{s\to -\infty} \phi(s) = \epsilon, \label{eq:ode_data_rqphi} \\
\lim_{s \to -\infty} \Omega^2(s) \cdot e^{- 2 K_+ s} = \alpha_+, \label{eq:ode_data_lapse} \\
\lim_{s \to -\infty} - 4 \Omega^{-2}(s)\dot{r}(s)= 1 - \frac{r_+ m^2 \epsilon^2}{2K_+}, \label{eq:ode_data_rdot} \\
\lim_{s \to - \infty} \frac{d}{ds} \log(\Omega^2) (s) = 2 K_+, \label{eq:ode_data_lapsedot}\\
\lim_{s \to - \infty} \Omega^{-2}(s) \tilde{A}(s) = - \frac{\mathbf{e}}{8K_+ r_+^2}, \label{eq:ode_data_gauge} \\
\lim_{s \to - \infty} \Omega^{-2}(s) \dot{\phi}(s) = \beta_+ \epsilon. \label{eq:ode_data_phidot}
\end{gather}
This concludes the set-up for the analytical problem considered in this paper.

\subsection{Linear scattering in the Reissner-Nordstr\"om interior} \label{sub:scattering}

We will often need to make comparisons to various quantities in exact Reissner-Nordstr\"om(-dS/AdS). While this is straightforward for $(r, \Omega^2, Q)$, the scalar field $\phi$ vanishes in Reissner-Nordstr\"om, \blue{and we instead compare $\phi$} to solutions to the linear (\textit{charged}) covariant Klein-Gordon equation in the Reissner-Nordstr\"om interior:
\begin{equation} \label{eq:rn_kg}
g^{\mu \nu}_{RN} D_{\mu} D_{\nu} \phi = m^2 \phi.
\end{equation}
Here $g_{RN}$ and $D$ are the metric and covariant derivative in Reissner-Nordstr\"om(-dS/AdS), with $A$ as specified in Section \ref{sub:reissner_nordstrom}, and $m^2 \in \mathbb{R}$ and $q_0 \neq 0$ are fixed.

Since we are interested in the spatially homogeneous problem, we consider only solutions $\phi = \phi_{\mathcal{L}}$ to (\ref{eq:rn_kg}) that satisfy $T \phi_{\mathcal{L}} = 0$, $S \phi_{\mathcal{L}} = 0$, where $T = \partial_t$ is the Killing vector field of Reissner-Nordstr\"om(-dS/AdS) associated to stationarity and $S$ is any vector field on the sphere.

Then denoting $\psi (s) = \phi_{\mathcal{L}}(t = 0, r^* = s)$, where $r^*, t$ are as defined in Section \ref{sub:reissner_nordstrom}, it can be checked that $\psi$ satisfies the following second order ODE in $s$ (see, for instance, the equation (\ref{eq:phi_evol})):
\begin{equation} \label{eq:rn_kg_coord}
\ddot{\psi} = \frac{\Omega^2_{RN}(s)}{4} \frac{\dot{\psi}}{2r_{RN}(s)} - \frac{\Omega^2_{RN}(s)}{4} m^2 \psi - q_0^2 \left( \frac{\mathbf{e}}{r_+} - \frac{\mathbf{e}}{r_{RN}(s)} \right)^2 \psi.
\end{equation}

We define the quantities
\begin{equation} \label{eq:rn_omega}
\tilde{A}_{RN, \infty} = \frac{\mathbf{e}}{r_+} - \frac{\mathbf{e}}{r_-} \neq 0, \hspace{1cm}
\red{\omega_{RN} \coloneqq |q_0 \tilde{A}_{RN, \infty}| > 0}.
\end{equation}
Then, following \cite{MoiChristoph}, we  define four functions solving (\ref{eq:rn_kg_coord}): $\psi_{\mathcal{H}, 1}, \psi_{\mathcal{H}, 2}, \psi_{\mathcal{CH}, 1}$ and $\psi_{\mathcal{CH}, 2}$. These satisfy the following asymptotics towards the event horizon $\mathcal{H} = \{ s = - \infty \}$ and the Cauchy horizon $\mathcal{CH} = \{ s = + \infty \}$ \begin{align}
&  \psi_{\mathcal{H}, 1}(s) = 1 + o(1) \hspace{4cm} \text{ as } s \to - \infty, \\ &  \psi_{\mathcal{H}, 2}(s) = s + o(1) \hspace{4cm} \text{ as } s \to - \infty,\\ & \psi_{\mathcal{CH}, 1}(s) = e^{i \omega_{RN} s} + o(1) \hspace{3cm} \text{ as } s \to + \infty, \\ &  	\psi_{\mathcal{CH}, 2}(s) = \overline{\psi_{\mathcal{CH}, 1}} (s) = e^{-i \omega_{RN} s} + o(1) \hspace{1cm} \text{ as } s \to + \infty.
\end{align}


The results of \cite{MoiChristoph} then imply the following:

\begin{proposition} \label{prop:linscat}
Recalling the definition of $2K_-(M, \mathbf{e}, \Lambda) < 0$ from (\ref{eq:surface_gravity}), and $\alpha_-(M,  \mathbf{e}, \Lambda)>0$ from (\ref{eq:omega_asymptotic_2}), there exists some constant $C > 0$ such that
\begin{equation}
|\psi_{\mathcal{CH}, 1}(s) - e^{i \omega_{RN} s}| + \left| \frac{d \psi_{\mathcal{CH}, 1}}{ds}(s) - i \omega_{RN} e^{i \omega_{RN} s} \right| \leq C \Omega^2_{RN}(s) \leq 2 C \alpha_- e^{2 K_- s},
\end{equation}
\begin{equation}
|\psi_{\mathcal{CH}, 2}(s) - e^{- i \omega_{RN} s}| + \left| \frac{d \psi_{\mathcal{CH}, 2}}{ds}(s) + i \omega_{RN} e^{- i \omega_{RN} s} \right| \leq C \Omega^2_{RN}(s) \leq 2 C \alpha_- e^{2 K_- s}.
\end{equation}
Furthermore, there exists a scattering coefficient $B = B(M, q_0,\mathbf{e}, \Lambda) \in \CC \setminus \{0\}$ such that
\begin{equation}\label{B.def}
\psi_{\mathcal{H}, 1}(s) = B \psi_{\mathcal{CH}, 1}(s) + \overline{B} \psi_{\mathcal{CH}, 2}(s) = 2 \mathfrak{Re}( B \psi_{\mathcal{CH}, 1}(s)).
\end{equation}
\end{proposition}

\begin{corollary} \label{cor:scattering}
Let $\phi_{\mathcal{L}}$ be the solution to (\ref{eq:rn_kg}) with constant data $\phi = \epsilon$ on the event horizon $\mathcal{H} = \{ s = -  \infty \}$. Then there exists $C(M, \mathbf{e}, \Lambda, m^2, q_0) > 0$ and  $\tilde{S}(M, \mathbf{e}, \Lambda, m^2, q_0) > 0$ such that, for $s \geq \tilde{S}$, one has
\begin{equation}
\left| \phi_{\mathcal{L}}(s) - B \epsilon e^{i \omega_{RN} s} - \overline{B} \epsilon e^{- i \omega_{RN} s} \right| \leq C \epsilon e^{2K_- s}.
\end{equation}
\begin{equation}
\left| \dot{\phi}_{\mathcal{L}}(s) - i \omega_{RN} B \epsilon e^{i \omega_{RN} s} + i \omega_{RN} \overline{B} \epsilon e^{- i \omega_{RN} s} \right| \leq C \epsilon e^{2K_- s}.
\end{equation}
\end{corollary}

\section{Precise statement of the main theorems}\label{sec:theorem}

\subsection{Definition of the spacetime sub-regions} \label{sub:regions}
We now give a precise definition of the regions in \blue{Figure~\ref{Penrose_detailed}}. \blue{Using the gauge choice \eqref{eq:gauge_data}, the following regions are part of  the spacetime described in Proposition~\ref{prop:char_ivp}. The following regions, which each correspond to $s$ being inside a \red{non-empty} interval $I \subset \R$, are such that, for $\ep$ chosen sufficiently small, their union covers the entire spacetime of Theorem~\ref{maintheorem}.}
\begin{itemize}
\item The red-shift region \blue{is} $\mathcal{R} \coloneqq \{ - \infty < s \leq - \Delta_{\mathcal{R}} \}$ for some $\Delta_{\mathcal{R}} \blue{\gg 1}$: Here, we make strong use of the positive surface gravity of the event horizon (red-shift effect, see \cite{MihalisPHD, Red} and subsequent works).
\item The no-shift region \blue{is} $\mathcal{N} \coloneqq \{ - \Delta_{\mathcal{R}} \leq s \leq S \}$ for some $S \gg 1$: Here, we use a Cauchy stability argument and Gr\"onwall's inequality to show that quantities are still $\epsilon^2$-close to their Reissner-Nordstr\"om values. 
\item The early blue-shift region \blue{is} $\mathcal{EB} \coloneqq \{ S \leq s \leq s_{lin}\blue{(\epsilon)} := |2 K_-|^{-1} \log(\nu \epsilon^{-1}) \}$ for  $\nu>0$: Here, we begin to exploit the \textit{blue-shift} effect of the Cauchy horizon of Reissner-Nordstr\"om, \red{which is due to} negative surface gravity $2K_- < 0$.
\item The late blue-shift region \blue{is} $\mathcal{LB} \coloneqq \{ s_{lin} \blue{(\epsilon)} \leq s \leq \Delta_{\mathcal{B}} \epsilon^{-1} \}$ for some some $\Delta_{\mathcal{B}}>0$: Here, the spacetime geometry begins to depart from that of Reissner-Nordstr\"om, and we provide the key ingredients to help us with the analysis of subsequent regions. In particular, this region starts to see a growth of the Hawking mass (a relic of mass inflation, see \cite{Ori,Moi4} and the introduction of \cite{VDM21}). 

\item The oscillation region \blue{is} $\mathcal{O}:=\{ s_O(\ep):= 50 s_{lin}\blue{(\ep)} \leq s \leq s_{PK}(\ep)\}$ where $r ( s_{PK})= 2 |B| \mathfrak{W} r_- \ep$ for \blue{$\mathfrak{W} = \mathfrak{W}(M, \mathbf{e}, \Lambda, q_0) >0$ defined in \eqref{eq:frakw}}: Here, the Bessel-type behavior kicks in, leading to the \blue{collapsed} oscillations discussed in Section~\ref{osc.intro}.

\item The proto-Kasner region \blue{is} $\mathcal{PK}:=\{ s_{PK}\blue{(\ep)} \leq s \leq s_{i}(\ep)\}$ where $r(s_{i})=  e^{-\delta_0 \ep^{-2}}$ for \blue{$\delta_0 = \delta_0(M, \mathbf{e}, \Lambda, q_0) > 0$ defined in \eqref{eq:delta0}}: Here, the Bessel-type behavior continues but \blue{exhibits logarithmic growth rather than oscillations}.

We also define the sub-region  $\blue{\mathcal{PK'}} :=\{ s_{K_1}\blue{(\ep)} \leq s \leq s_{i}(\ep)\} \subset \mathcal{PK}$ where $r ( s_{K_1})= 2 |B| \mathfrak{W} r_- \ep^2$. \blue{In this sub-region, we will establish that Kasner-type behavior starts to take place.}

\item The Kasner region \blue{is} $\mathcal{K}:=\{ s_{i}\blue{(\ep)} \leq s < s_{\infty}(\ep)\}$, where $\lim_{s\rightarrow s_{\infty}}r (s)=  0$. In this region, we prove that the metric is \blue{described by either one Kasner regime, or two Kasner regimes connected by a bounce. In $\mathcal{K}$,} the scalar field is governed by a first-order ODE.
\end{itemize}
We also introduce the additional sub-regions of $\mathcal{PK} \cup \mathcal{K}$ depicted in Figure~\ref{FigN}, \blue{where $\Psi$ is defined in \eqref{Psi.def} and $\alpha$ in \eqref{alpha.def.intro}}.
\begin{itemize}

\item The first Kasner region $\mathcal{K}_1\subset \blue{\mathcal{PK}'}  \cup \mathcal{K}$ overlaps with  $\mathcal{PK}$ and   $\mathcal{K}$. In this region, we will show the metric is in a Kasner regime. 
Anticipating Theorem~\ref{maintheorem2}, we will find that $\mathcal{K}_1= \blue{\mathcal{PK}'}  \cup \mathcal{K}$ in the no Kasner \blue{bounce} case \eqref{Ninv.eq} (i.e.\ the first Kasner regime is the final Kasner regime, with all Kasner exponents being positive) and $\mathcal{K}_1=\{s_{K_1} \leq s \leq s_{in}\} \neq  \blue{\mathcal{PK}'}  \cup \mathcal{K}$ in the Kasner \blue{bounce} case \eqref{inv.eq} (with one negative Kasner exponent, which is thus expected to be unstable). Here $s_{in}:= \min \{s  \in \mathcal{K}: |\Psi(s)| = |\alpha|+ \ep^2\}$.

\item The Kasner \blue{bounce} region $\mathcal{K}_{\blue{bo}} \subset \mathcal{K}$.  We  have $\mathcal{K}_{\blue{bo}}=\emptyset$ in the no Kasner \blue{bounce} case \eqref{Ninv.eq}, and $\mathcal{K}_{\blue{bo}}=\{s_{in} \leq s \leq s_{out}\}$ in the Kasner \blue{bounce} case \eqref{inv.eq},  where $s_{out}:= \min \{s \in \mathcal{K}: |\Psi(s)| = |\alpha|^{-1}- \ep^2\}$. We have weaker control of the metric in $\mathcal{K}_{\blue{bo}}$, but we show that it is very short in terms of proper time.

\item The second Kasner  region $\mathcal{K}_{2} \subset \mathcal{K}$.  We will have $\mathcal{K}_{2}=\emptyset$ in the no Kasner \blue{bounce} case \eqref{Ninv.eq}, and $\mathcal{K}_{2}=\{s_{out} \leq s < s_{\infty}\}$ in the Kasner \blue{bounce} case \eqref{inv.eq}, where we exhibit a second Kasner regime (with positive Kasner exponents, in contrast to the first Kasner regime $\mathcal{K}_{1} $).
\end{itemize}
\subsection{First statement: formation of a spacelike singularity}

We first start with our main result, which covers part of Theorem~\ref{thm.intro} (namely the formation of the spacelike singularity). We reiterate that Theorem~\ref{maintheorem} contains the statements~\ref{I1}, \ref{I2}, \ref{I5} of our rough Theorem~\ref{thm.intro}; more precisely statement~\ref{I1} of Theorem~\ref{thm.intro} corresponds to statement~\ref{SS1} of Theorem~\ref{maintheorem}, while statements~\ref{I2} and \ref{I5} of Theorem~\ref{thm.intro} are covered by statement~\ref{SS2} of Theorem~\ref{maintheorem} as well as the estimate (\ref{Q.PK}). 

The statements~\ref{SS3} and \ref{SS4} of Theorem~\ref{maintheorem} will also lay the groundwork towards proving the more specific Kasner asymptotics claimed in statements~\ref{I3} and \ref{I4} of Theorem~\ref{thm.intro}; the precise nature of these asymptotics will be covered in Theorem~\ref{maintheorem2}, upon \blue{further restricting} $\ep$. 

In the following theorem and the rest of the paper, \red{for $A, B \geq 0$} we will use the notation $A \lesssim B$ \red{to denote that} there exists $C(M,\mathbf{e},\Lambda,q_0,m^2, \eta)>0$ such that $A \leq C B$.

\begin{thm}\label{maintheorem} 
Let	 $(M,\mathbf{e},\Lambda) \in \mathcal{P}_{se}$ with $\mathbf{e}\neq0$,  $q_0\neq 0$, $m^2 \in \mathbb{R}$ \blue{be subextremal Reissner-Nordstr\"om(-dS/AdS) parameters (see Section~\ref{sub:reissner_nordstrom})}. 
Then, for $\eta>0$ chosen sufficiently small, there exists $\epsilon_0(M,\mathbf{e},\Lambda,q_0,m^2,\eta)>~0$ and a subset $E_{\eta} \subset (-\ep_0,\ep_0)\setminus \{0\}$ satisfying $\frac{|(-\delta,\delta) \setminus E_{\eta}|}{2\delta}= O(\eta)$ for any $0<\delta \leq \ep_0$, such that \red{the following holds:} For all $\ep\in E_{\eta}$, the maximal future development $\mathcal{M}$ for \eqref{E1}-\eqref{E5} of the characteristic data from  Section~\ref{prelim.section} (i.e.\  \eqref{eq:gauge_data}, \eqref{eq:r_data},  \eqref{eq:Q_data},  and \eqref{eq:phi_data}) terminates at a spacelike singularity $\mathcal{S}$ on which $r$ extends continuously to $0$, and the Penrose diagram is given by Figure~\ref{Penrose_detailed}.

More precisely, there exists a foliation of $\mathcal{M}$ by spacelike hypersurfaces $\Sigma_{s}$, with $s\in (-\infty,s_{\infty}(\ep))$, where $s$ is defined in \eqref{eq:gauge_data}, and $s_{\infty}= \red{\frac{ |K_-| }{ 4 |B|^2 \omega_{RN}^2}}  \epsilon^{-2}+O(\log(\ep^{-1}))$, where  $B(M,\mathbf{e},\Lambda,m^2, q_0) \neq 0$ is defined in \eqref{B.def}, and $2K_-(M,\mathbf{e},\Lambda)<0$, \red{$\omega_{RN}(M,\mathbf{e},\Lambda,q_0) > 0$} \blue{are defined in} Section~\ref{sub:reissner_nordstrom}.
The subsequent spacetime dynamics are described as follows:

\begin{enumerate}[i.]
\item \label{SS1} (Almost formation of a Cauchy horizon). In  the late blue-shift region $\mathcal{LB} \coloneqq \{ s_{lin}(\ep) \leq s \leq \Delta_{\mathcal{B}} \epsilon^{-1} \}$, we have the following stability estimates with respect to the Reissner--Nordstr\"om(-dS/AdS) metric: 
\begin{multline}
\ep^{-1}	|\phi(s) - 2 \ep \Re( Be^{i \omega_{RN} s}) | +  \ep^{-1}\left|\frac{d}{ds}[\phi(s) - 2 \ep \Re( B e^{i \omega_{RN} s})] \right|+ 	| r(s) - r_- |  \\[0.5em] +	 | Q(s)-\mathbf{e} | +	\left| \frac{d}{ds} \log (\Omega^2)(s) - 2 K_- \right| \lesssim  \epsilon^2 s \lesssim \ep.
\end{multline}	          For $s \in \mathcal{LB}$, we find also the following estimate for $-r \dot{r}(s)$:
\begin{equation} \label{lb_rrdot}			\left | - r\dot{r}(s) - \frac{4 |B|^2 \omega_{RN}^2 \epsilon^2 r_-^{\blue{2}}} {2 |K_-|} \right | \lesssim e^{2K_- s} + \epsilon^4 s.			\end{equation}

\item \label{SS2} (\blue{Collapsed} oscillations and loss of charge). In the \blue{oscillation} region $\mathcal{O} \coloneqq \{ s_{O}(\ep) \leq s \leq  s_{PK}(\ep) \}$, we have the following Bessel-type oscillations for the scalar field: for some $\xi_0 = \red{\frac{|K_-|}{4 |B|^2 \omega_{RN}}}+O(\ep^2 \log(\ep^{-1}))$,
\begin{equation}\label{Bessel.statement}
\left | \phi(s) - \left(  C_J (\epsilon) J_0 \left( \frac{\xi_0 r^2(s)}{r_-^2 \epsilon ^2} \right) + C_Y(\epsilon) Y_0 \left( \frac{ \xi_0 r^2(s)} {r_-^2 \epsilon^2} \right)\right) \right| \lesssim\epsilon^2 \log(\epsilon^{-1}),\end{equation} 
\begin{equation} \label{Bessel.statement.2} \left |\frac{d}{ds}\left( \phi(s) - \left(  C_J (\epsilon) J_0 \left( \frac{\xi_0 r^2(s)}{r_-^2 \epsilon ^2} \right) + C_Y(\epsilon) Y_0 \left( \frac{ \xi_0 r^2(s)} {r_-^2 \epsilon^2} \right)\right) \right)\right| \lesssim\epsilon^2 \log(\epsilon^{-1}),\end{equation}			 where the constants $ C_J (\epsilon)$, $ C_Y (\epsilon)$ are highly oscillatory in $\ep$; namely for $\mathfrak{W}(M,\mathbf{e},q_0,\Lambda) >0$ given by
\begin{equation} \mathfrak{W}(M,\mathbf{e},\Lambda,q_0)=\sqrt{\frac{\omega_{RN}}{2|K_-|}}= \sqrt{ \frac{|q_0 \mathbf{e}|}{|\frac{\mathbf{e}^2}{r_-^2}- \frac{\Lambda}{3} r_+ (r_+ + 2r_-)|}}>0,\label{M.formula}
\end{equation}
one finds that
\begin{equation}
\left| C_J(\epsilon) - \frac{\sqrt{\pi}}{2} \mathfrak{W}^{-1} \cos (\Theta(\epsilon)) \right|+ \left| C_Y(\epsilon) - \frac{\sqrt{\pi}}{2} \mathfrak{W}^{-1} \sin (\Theta(\epsilon)) \right| \lesssim \epsilon^2 \log (\epsilon^{-1}),\end{equation} 
\begin{equation} \left| \Theta(\ep)-\frac{\ep^{-2}}{8 |B|^2 \mathfrak{W}  } \right| \lesssim \log(\ep^{-1}). \end{equation} 
For $s \in \mathcal{O}$, one has (note this improves upon \eqref{lb_rrdot} in $\mathcal{LB} \cap \mathcal{O}$):
\begin{equation}  \label{o_rrdot}			\left | - r\dot{r}(s) - \frac{4 |B|^2 \omega_{RN}^2 \epsilon^2 r_-^{\blue{2}}} {2 |K_-|} \right | \lesssim \epsilon^4 \log(\epsilon^{-1}).			\end{equation}
Moreover, the charge $Q$ transitions from $\mathbf{e}$ to   $Q_{\infty}(M, \mathbf{e},\Lambda):=\frac{3}{4}\mathbf{e} + \Lambda\frac{r_-^2 r_+(2r_-+r_+)}{12 \mathbf{e}}$ (up to $O(\ep^{2-})$ errors), and \begin{equation}
\left|	Q(s)- \mathbf{e} +(\mathbf{e}- Q_{\infty})\left(1-\frac{r^2(s)}{r_-^2}\right)\right|\lesssim \ep^2 \log(\ep^{-1} ) \text{ for all } s\in \mathcal{O},\end{equation}
\begin{equation} \frac{Q_{\infty}(M, \mathbf{e}, 0)}{\mathbf{e}}=\frac{3}{4},\end{equation} 
\begin{equation} \left\{ \frac{Q_{\infty}(M,\mathbf{e},\Lambda)}{\mathbf{e}}:\ (M,\mathbf{e},\Lambda)\in \mathcal{P}_{se},\ \Lambda<0 \right\}= \left( \frac{1}{2},\frac{3}{4} \right),\end{equation}
\begin{equation} \left\{\frac{Q_{\infty}(M,\mathbf{e},\Lambda)}{\mathbf{e}}:\ (M,\mathbf{e},\Lambda)\in \mathcal{P}_{se},\  \Lambda>0 \right \}= \left(\frac{3}{4},1\right). \end{equation}

\item (Logarithmic Bessel-type divergence). \label{SS3} \sloppy In the proto-Kasner region $\mathcal{PK}:=\{ s_{PK}(\ep) \leq s \leq s_i(\ep)\}$, the Bessel-type behavior persists with modified coefficients, namely there exist
\[
C_{YK}(\ep) = C_Y(\ep) + O(\ep^2 \log(\ep^{-1})), \quad C_{JK}(\ep) = C_J(\ep)+ O(\ep^2 \log(\ep^{-1})), \xi_K = \xi_0 + O(\ep^2 \log(\ep^{-1})),
\]
such that \eqref{Bessel.statement}, \eqref{Bessel.statement.2} remain true for all $s\in \mathcal{PK}$, after replacing $ (C_Y(\ep),C_J(\ep),\xi_0) $ by  $(C_{YK}(\ep), C_{JK}(\ep),\xi_K)$. Consequently, we have the following logarithmic \blue{growth}: for all $s\in  \mathcal{PK}$,	\begin{equation}
\left | \phi(s) + \frac{2}{\pi} C_{YK}(\ep)\log(\frac{r_-^2 \ep^2}{ \xi_K r^2(s)}) - C_{YK}(\ep) -  2\pi^{-1} (\gamma-\log2) C_{J K }(\ep)\right| \lesssim\epsilon^2 \log(\epsilon^{-1}),\end{equation} \begin{equation} 	\left |\frac{d}{ds}\left(  \phi(s) + \frac{2}{\pi} C_{YK}\log(\frac{r_-^2 \ep^2}{ \xi_K r^2(s)}) \right)\right| \lesssim\epsilon^2 \log(\epsilon^{-1}).
\end{equation}		

Moreover, \eqref{o_rrdot} still holds for $s \in \mathcal{PK}$, while the charge $Q$ remains close to its value at $s_{PK}$; for all $s\in  \mathcal{PK}$: \begin{equation}\label{Q.PK}
| Q(s) - Q_{\infty}|\lesssim \ep^2 \log(\ep^{-1}).
\end{equation}

Finally, in the sub-region $\blue{\mathcal{PK}'}=\{ s_{K_1}(\ep) \leq s \leq s_i(\ep)\} \subset \mathcal{PK}$, the quantity \blue{$\Psi(s)$ defined by}
\[
\Psi(s):= -r(s) \frac{d\phi}{dr}(s)= \frac{r(s)}{-\frac{dr}{ds}(s)}\frac{d\phi}{ds}(s)
\]
obeys the estimates:
\begin{equation}
\left|\Psi(s) + \frac{2}{\sqrt{\pi}} \mathfrak{W}^{-1} \sin(\Theta(\ep))\right| \lesssim  \ep^2 \log(\ep^{-1}).
\end{equation} 
%
%
\begin{equation}
|\Psi(s) - \Psi(s_i) | \lesssim r^2(s)\log(\ep^{-1}) \lesssim \ep^4\log(\ep^{-1}) .
\end{equation}

\item \label{SS4}	 In the Kasner region $\mathcal{K}:=\{ s_{i} \leq s < s_{\infty}(\ep)\}$, $-r \dot{r}(s)$ obeys the lower bound
\begin{equation} \label{k_rrdot}
- r \dot{r} (s) \geq  \frac{4 |B|^2 \omega_{RN}^2 \epsilon^2 r_-^{\blue{2}}} {2 |K_-|} \cdot \frac{1}{2} \eta^2.
\end{equation}
The quantity $\Psi$ obeys the following ODE: introducing $R:= \log(\frac{r_-}{r})$, \begin{equation}
\frac{d\Psi}{dR} = - \Psi (\Psi-\alpha) (\Psi-\alpha^{-1}) + \mathcal{F},
\end{equation}\begin{equation} \label{alpha.def}
|\alpha- \Psi(s_i)| \lesssim e^{-\delta_0 \ep^{-2}},\ \mathcal{F}(R) \lesssim  e^{-\delta_0 \ep^{-2}} r(R).
\end{equation}

Finally, the charge retention remains in the sense that \eqref{Q.PK} is also valid for all $s \in \mathcal{K}$.
\end{enumerate}

\end{thm}
\subsection{Second statement: Kasner asymptotics in the \texorpdfstring{$\mathcal{PK}$}{PK} and \texorpdfstring{$\mathcal{K}$}{K} regions}

We now enter into the details of the regions $\blue{\mathcal{PK}'} \cup \mathcal{K}$, in which the Kasner-like behavior manifests. The following theorem requires a (slightly) stronger assumption on $\ep$ than Theorem~\ref{maintheorem}.
\begin{thm}\label{maintheorem2}
Let	 $(M,\mathbf{e},\Lambda) \in \mathcal{P}_{se}$ with $\mathbf{e}\neq0$,  $q_0\neq 0$, $m^2 \in \mathbb{R}$. Then, for any sufficiently small $\eta, \sigma > 0$, there exists $\epsilon_0(M,\mathbf{e},\
\Lambda,q_0,m^2,\eta,\sigma)>0$ \blue{such that there exists a subset $E_{\eta,\sigma}' \subset E_{\eta}$ satisfying $\frac{|(-\delta,\delta) \setminus E'_{\eta,\sigma}|}{2\delta}= O(\eta + \sigma)$ for any $0< \delta \leq \ep_0$, and such that $E'_{\eta, \sigma}$ may be expressed as a disjoint union $E'_{\eta, \sigma} = E_{\eta,\sigma}^{'\ bo} \cup E_{\eta,\sigma}^{'\ Nbo}$, where}
\begin{align}\label{inv.eq}
& \ep \in E_{\eta,\sigma}^{'\ \blue{bo}} \text{ if }	\eta < |\Psi(s_i)| = \frac{2}{\sqrt{\pi}} \mathfrak{W}^{-1}|\sin(\Theta(\ep))| + O(\epsilon^2 \log (\epsilon^{-1})) \leq  1- \sigma <1.\\
\label{Ninv.eq} & 
\ep \in E_{\eta,\sigma}^{'\ \blue{Nbo}} \text{ if  }	|\Psi(s_i)| = \frac{2}{\sqrt{\pi}} \mathfrak{W}^{-1}|\sin(\Theta(\ep))| + O(\epsilon^2 \log(\epsilon^{-1})) \geq 1+ \sigma >1.
\end{align} 
We call these respectively the Kasner \blue{bounce} case and the no \blue{bounce} inversion case.

Moreover, we have the following two possibilities. \begin{itemize}
\item If $\mathfrak{W}(M,\mathbf{e},\Lambda,q_0) \geq \frac{2}{\sqrt{\pi}}$, then $E_{\eta,\sigma}^{'\ \blue{Nbo}}=\emptyset$, and therefore $\frac{|(-\ep,\ep)\setminus E_{\eta,\sigma}^{'\ \blue{bo}}|}{2\ep}=O(\eta + \sigma)$, for $\eta,\sigma$ small. 

\item	If $\mathfrak{W}(M,\mathbf{e},\Lambda,q_0) < \frac{2}{\sqrt{\pi}}$, then $|E_{\eta,\sigma}^{'\ \blue{Nbo}}|, |E_{\eta,\sigma}^{'\ \blue{bo}}|>0$, and for $\eta,\sigma$ small: 
\begin{align*}
& \frac{|E_{\eta,\sigma}^{'\ \blue{Nbo}}\cap(-\ep, \ep)|}{2\ep}= \frac{2}{\pi} \arcsin(1-\frac{\sqrt{\pi}}{2} \mathfrak{W})+O(\sigma), \\ & \frac{|E_{\eta,\sigma}^{'\ \blue{bo}} \cap (-\ep, \ep)|}{2\ep}= 1-\frac{2}{\pi} \arcsin(1-\frac{\sqrt{\pi}}{2} \mathfrak{W})+O(\eta + \sigma).
\end{align*}
\end{itemize}
Then, for all $\ep \in E_{\eta,\sigma}^{'}$, we have the following Kasner-like behavior, \blue{where, in what follows, $b_-=\frac{|B|\omega_{RN}}{|K_-|}$.}
\begin{enumerate}

\item In the first Kasner region $\mathcal{K}_1$, we have \blue{the following}, recalling $\alpha$ from \eqref{alpha.def}.

\begin{enumerate}
\item In the no Kasner \blue{bounce} case \eqref{Ninv.eq}, we have $|\alpha|>1$, $\mathcal{K}_1=\blue{\mathcal{PK}'} \cup \mathcal{K}$, and for all $s\in \mathcal{K}_1$,
\begin{equation} \label{eq:no.bounce}
|\Psi(s)-\alpha| \lesssim \blue{\ep^2 \cdot \left( \frac{r(s)}{r(s_{PK})} \right)^{\beta}},
\end{equation} where we define $\beta:= \min\{\frac{1}{2}, \alpha^2-1\}>0$. Moreover, the metric takes the following Kasner-like form 
\begin{equation}
g = - d \tau^2 + \mathcal{X}_1 \cdot ( 1 + \mathfrak{E}_{X, 1}(\tau)) \,\tau^{\frac{2 (\alpha^2 -1)}{\alpha^2 + 3}} \, dt^2 + \mathcal{R}_1 \cdot ( 1 + \mathfrak{E}_{R, 1}(\tau)) \, r_-^2 \tau^{ \frac{4}{\alpha^2 + 3} } \, d \sigma_{\mathbb{S}^2},
\end{equation}	\begin{equation}\label{K}
\left| \log \mathcal{X}_1 + \frac{\alpha^2 + 1}{\alpha^2 + 3} \frac{4|K_-|^2}{|B|^2}\ \epsilon^{-2} \right| + \left| \log \mathcal{R}_1 - \frac{1}{\alpha^2 + 3} \frac{4|K_-|^2}{|B|^2}\ \epsilon^{-2} \right| \lesssim \log (\epsilon^{-1}),
\end{equation}
\begin{equation}\label{error.bound}
|\mathfrak{E}_{X, 1}(\tau)| + |\mathfrak{E}_{R, 1}(\tau)| \lesssim \epsilon^2 \cdot \left(\frac{\tau}{\tau(s_{K_1})}\right)^{\frac{2\beta}{\alpha^2 + 3}},
\end{equation} where $\tau$ is the proper time\footnote{More precisely, $\tau$ a is past-directed timelike variable, orthogonal to the hypersurfaces $\Sigma_s$, and normalized such that $g(d \tau, d \tau) = \blue{-1}$ and $\tau = 0$ at the spacelike singularity $\{ r = 0 \}$.}, and  we call $(p_1,p_{2},p_{3})= (\frac{\alpha^2-1}{\alpha^2+3}, \frac{2}{\alpha^2+3}, \frac{2}{\alpha^2+3}) \in (0,1)^3$ the Kasner exponents.
\item In the  Kasner \blue{bounce} case \eqref{inv.eq}, we have $|\alpha|<1$, $ \blue{ \mathcal{K}_1 =\{ s_{K_1} \leq s \leq s_{in}\}\subset \mathcal{PK}' \cup \mathcal{K}}$, and for all $s\in \mathcal{K}_1$,
\begin{equation}
|\Psi(s)-\alpha| \lesssim \ep^2,
\end{equation}  where $s_{in} =  \frac{b_-^{-2}}{4|K_-|}+ O(\log(\ep^{-1}))$ is such that \begin{equation}\label{rsin}
r(s_{in}) = r_- \cdot \exp(-\frac{\ep^{-2}}{2b_-^2\cdot (1-\alpha^2)} +O(\log(\ep^{-1}))).
\end{equation}  Moreover, the metric takes the following Kasner-like form, where $\tau_0>0$ is a constant.
\begin{equation}
g = - d \tau^2 + \mathcal{X}_1 \cdot ( 1 + \mathfrak{E}_{X, 1}(\tau)) \,(\tau-\tau_0)^{\frac{2 (\alpha^2 -1)}{\alpha^2 + 3}} \, dt^2 + \mathcal{R}_1 \cdot ( 1 + \mathfrak{E}_{R, 1}(\tau)) \, r_-^2 (\tau-\tau_0)^{ \frac{4}{\alpha^2 + 3} } \, d \sigma_{\mathbb{S}^2}.
\end{equation}	
We call $(p_1,p_{2},p_{3})= (\frac{\alpha^2-1}{\alpha^2+3}, \frac{2}{\alpha^2+3}, \frac{2}{\alpha^2+3}) \in (-\frac{1}{3},0)\times (\frac{1}{2},\frac{2}{3})^2$ the first Kasner exponents. 

Moreover, $\mathcal{X}_1$  and  $\mathcal{R}_1$  obey \eqref{K}, and  for all  $\tau(s_{K_1}) \leq \tau \leq \tau(s_{in})$: 
\begin{equation}
|\mathfrak{E}_{X, 1}(\tau)| + |\mathfrak{E}_{R, 1}(\tau)| \lesssim \epsilon^2 .
\end{equation}

\end{enumerate}

\item In the Kasner \blue{bounce} region \blue{$\mathcal{K}_{bo}$} (only in  the  Kasner \blue{bounce} case \eqref{inv.eq}), we have for all $s \in \blue{\mathcal{K}_{bo}}$,
\begin{equation}
r(s) = r_- \cdot \exp(-\frac{\ep^{-2}}{2b_-^2\cdot (1-\alpha^2)} +O(\log(\ep^{-1}))).
\end{equation} Moreover, in terms of proper time, we have \begin{equation}
0<\tau(s_{in})-\tau(s_{out}) \lesssim \exp(-\frac{\ep^{-2}  }{b_-^2\cdot(1-\alpha^2)} +O(\log(\ep^{-1}))).
\end{equation}

\item In the second Kasner region $\mathcal{K}_{2}$ (only in  the  Kasner \blue{bounce} case \eqref{inv.eq}),  we define   $\beta:= \min\{\frac{1}{2}, \alpha^{-2}-1\}>0$ (since $|\alpha|<1$). We have
for all $s\in \mathcal{K}_1$: \begin{equation}
|\Psi(s)-\alpha^{-1}| \lesssim \ep^2  \cdot	\left( \frac{r(s)}{r(s_{out})}\right)^{\beta}.
\end{equation}  

Moreover, the metric takes the following Kasner-like form, for all $0<\tau \leq \tau(s_{out})$ 	\begin{equation}
g = - d \tau^2 + \mathcal{X}_2 \cdot ( 1 + \mathfrak{E}_{X, 2}(\tau)) \, \tau^{\frac{2 (1 - \alpha^2)}{1 + 3 \alpha^2}} \, dt^2 + \mathcal{R}_2 \cdot ( 1 + \mathfrak{E}_{R, 2}(\tau)) \, r_-^2 \tau^{ \frac{4 \alpha^2}{1 + 3 \alpha^2} } \, d \sigma_{\mathbb{S}^2} .
\end{equation}							\begin{equation}
\left| \log \mathcal{X}_2 + \frac{1 + \alpha^{-2}}{1 + 3 \alpha^2} \frac{4|K_-|^2}{|B|^2}\ \epsilon^{-2} \right| + \left| \log \mathcal{R}_2 - \frac{1}{1 + 3\alpha^2}  \frac{4|K_-|^2}{|B|^2}\ \epsilon^{-2} \right| \lesssim \log (\epsilon^{-1}),
\end{equation}
\begin{equation}
|\mathfrak{E}_{X, 2}(\tau)| + |\mathfrak{E}_{R, 2}(\tau)| \lesssim \epsilon^2 \cdot \left(\frac{\tau}{\tau(s_{out})}\right)^{\frac{2\beta}{\alpha^{-2} + 3}}.
\end{equation} 	We call $(p_1,p_{2},p_{3})= (\frac{1-\alpha^2}{1+3\alpha^2}, \frac{2\alpha^2}{1+3\alpha^2}, \frac{2\alpha^2}{1+3\alpha^2}) \in (0,1)^3$ the second Kasner exponents.

Finally, in terms of proper time, we have (recalling that $\tau$ is normalized so that $\underset{s\rightarrow s_{\infty}}{\lim}\tau(s)=0$):
\begin{equation}
\exp(-\frac{1+\alpha^{-2}}{2b_-^2\cdot(1-\alpha^2)}\ \ep^{-2}+O(\log(\ep^{-1})))\lesssim \tau(s_{out}) \lesssim  \exp(-\frac{1+\alpha^{-2}}{2b_-^2\cdot(1-\alpha^2)}\ \ep^{-2}+O(\log(\ep^{-1}))).
\end{equation}
\end{enumerate}
\end{thm}
Without loss of generality, we will choose $\ep>0$ in all the subsequent sections.

\section{Almost formation of a Cauchy horizon} \label{sec:einsteinrosen}

In this section, we provide estimates up to a region $s \sim \epsilon^{-1}$. The analysis will be perturbative in nature, and we always bear in mind the comparison to the linear charged scalar field problem in Reissner-Nordstr\"om, see Section \ref{sub:scattering}. The analysis will largely follow \cite{VDM21}, with minor modifications due to the now dynamical nature of the charge $Q$ and the charge term for the scalar field.

As in \cite{VDM21}, the estimates up to $s \sim \epsilon^{-1}$ will be divided into the four regions $\mathcal{R}$, $\mathcal{N}$, $\mathcal{EB}$, and $\mathcal{LB}$ (see Figure~\ref{Penrose_detailed} and Figure~\ref{Fig1}). 	Where differences from \cite{VDM21} are minor, we aim to be relatively brief, and focus on the new techniques required to deal with the charge and the scalar field.

\subsection{Estimates up to the no-shift region}

We begin with some notation. As the arguments of this section are perturbative, it will help to introduce differences between the quantities $(r, \Omega^2, \phi, Q, \tilde{A})$ and their Reissner-Nordstr\"om(-dS/AdS) values. Therefore define
\begin{equation*}
\delta r (s) = r (s) - r_{RN}(s), \hspace{0.5cm} \delta \Omega^2(s) = \Omega^2(s) - \Omega^2_{RN}(s),
\end{equation*}
\begin{equation*}
\delta \phi (s) = \phi (s) - \phi_{\mathcal{L}}(s),
\end{equation*}
\begin{equation*}
\delta Q (s) = Q(s) - \mathbf{e}, \hspace{0.5cm} \delta \tilde{A}(s) = \tilde{A}(s) + \left( \frac{\mathbf{e}}{r_+} - \frac{\mathbf{e}}{r_{RN}(s)} \right),
\end{equation*}
where we used (\ref{eq:reissner_nordstrom_gauge}) to provide the Maxwell gauge field $\tilde{A}_{RN}$ to which we compare, and $\phi_{\mathcal{L}}$ is the solution to the linear charged scalar field scattering problem in Reissner-Nordstr\"om, see Corollary \ref{cor:scattering}.

We also use the quantity $\delta \log \Omega^2 (s) = \log \Omega^2 (s) - \log \Omega^2_{RN} (s) $. \blue{We commonly use that for $\delta \log \Omega^2 (s) \lesssim 1$,}
\blue{
\begin{equation} \label{eq:logomegadiff}
\delta \Omega^2(s) \lesssim \Omega^2(s) \, \delta \log \Omega^2(s) \lesssim \delta \Omega^2(s).
\end{equation}}
We now proceed to the estimates in the red-shift region $\mathcal{R}$.

\begin{proposition} \label{prop:redshift}
There exist $D_R(M, \mathbf{e}, \Lambda, m^2, q_0) > 0$ and $\Delta_{\mathcal{R}}(M, \mathbf{e}, \Lambda) \blue{ \gg 1}$ such that, in the red-shift region $\mathcal{R} = \{ - \infty < s < - \Delta_R \}$, the following estimates hold:
\begin{equation} \label{eq:rs_phi}
|\phi| + |\Omega^{-2} \dot{\phi}| \leq D_R \epsilon,
\end{equation}
\begin{equation} \label{eq:rs_kappa}
| \kappa^{-1} - 1 | \leq D_R \epsilon^2,
\end{equation}
\begin{equation} \label{eq:rs_mass}
| \varpi - M | \leq D_R \epsilon^2,
\end{equation}
\begin{equation} \label{eq:rs_omega}
\left | \frac{d \log (\Omega^2)}{ds} - 2 K(s) \right | + | \log(\alpha_+^{-1} \Omega^2) - 2 K_+ s | \leq D_R \Omega^2 \ll 1,
\end{equation}
\begin{equation} \label{eq:rs_diffs}
| \delta \log (\Omega^2) | + |\delta r| + |\delta \dot{r}| + |\delta \tilde{A}| \Omega^{-2} + | \delta Q | \leq D_R \epsilon^2,
\end{equation}
\begin{equation} \label{eq:rs_phidiff}
| \delta \phi | + | \delta \dot{\phi} | \leq D_R \epsilon^3.
\end{equation}
\end{proposition}

\begin{proof}
We provide a short sketch. For more details, see Proposition 4.5 in \cite{Moi} or Lemma 4.1 in \cite{VDM21}. We make the following bootstrap assumptions:
\begin{equation}
\left | \frac{d}{ds} \log(\Omega^2) - 2 K_+ \right | \leq K_+,
\tag{RS1} \label{eq:rs_bootstrap_lapse}
\end{equation}
\begin{equation}
| \phi | \leq 4 \epsilon,
\tag{RS2} \label{eq:rs_bootstrap_phi}
\end{equation}
\begin{equation}
|Q - \mathbf{e}| \leq \frac{\mathbf{e}}{2},
\tag{RS3} \label{eq:rs_bootstrap_Q}
\end{equation}
\begin{equation}
|r - r_+| \leq \frac{r_+}{2}.
\tag{RS4} \label{eq:rs_bootstrap_r}
\end{equation}

These will hold in neighborhood of $s = - \infty$ due to the asymptotic data (\ref{eq:ode_data_rqphi}) and (\ref{eq:ode_data_lapsedot}).
The bootstrap assumption (\ref{eq:rs_bootstrap_lapse}) will give the important estimate
\begin{equation} \label{eq:rs_omega_estimate}
\int^s_{-\infty} \Omega^2(s') \, ds' \lesssim \int^s_{-\infty} \frac{d}{ds} \log (\Omega^2) \cdot \Omega^2(s') \, ds' = \int^s_{-\infty} \frac{d}{ds} \Omega^2(s') = \Omega^2(s).
\end{equation}

The evolution equation (\ref{eq:gauge_evol}) for $\tilde{A}$ will immediately yield that $|\tilde{A}| \lesssim \Omega^2$. Therefore, turning to the equation (\ref{eq:phi_evol_2}), we use (\ref{eq:rs_omega_estimate}) to see that
\begin{equation*}
|r^2 \dot{\phi}| \lesssim \int_{- \infty}^s \Omega^2(s') \epsilon \, ds \lesssim \Omega^2(s) \epsilon,
\end{equation*}
which, after another round of integration, gives
$|\phi| \leq \epsilon + C \Omega^2(s) \epsilon,$
for some positive constant $C$. 

Since the right hand side of the equation \blue{(\ref{eq:raych_kappa})} is now bounded by $C \Omega^2 \epsilon^2$, the estimate (\ref{eq:rs_kappa}) is straightforward, and from this point (\ref{eq:rs_mass}) and (\ref{eq:rs_omega}) are also immediate. In particular, the remaining bootstraps are all improved, so long as $\Delta_{\mathcal{R}}$ is chosen large enough that $\Omega^2(s) < C^{-1}$ and $\epsilon$ is sufficiently small.

In order to get the difference estimates (\ref{eq:rs_diffs}), (\ref{eq:rs_phidiff}), we make yet another bootstrap assumption:
\begin{equation} \label{eq:rs_bootstrap_diffs} \tag{RS5}
|\delta r| + | \delta \dot{r} | + | \delta \log(\Omega^2) | + | \delta \tilde{A} | + \epsilon^{-1} | \delta \phi | + \epsilon^{-1} | \delta \dot{\phi} | \leq \epsilon^2.
\end{equation}

The claim is that we can then use the ODEs for differences (found by subtracting from the relevant ODE in Section \ref{sub:evol_eqs} the analogous ODE in Reissner-Nordstr\"om) to then improve the RHS of (\ref{eq:rs_bootstrap_diffs}) by $C \Omega^2 \epsilon^2$, hence improving the bootstrap for $\Delta_{\mathcal{R}} > 0$ sufficiently large.

We demonstrate this for the $\delta \phi$ estimate as an illustration. Taking the differences of (\ref{eq:phi_evol}), we have
\begin{equation*}
|\delta \ddot{\phi}| \lesssim |\delta \dot{r}| |\dot{\phi}| + |\dot{r}| |\delta \dot{\phi}| + |\dot{r} \dot{\phi}| |\delta r| + |\delta \tilde{A}| |\tilde{A}| |\phi| + (\tilde{A}^2 + \Omega^2) |\delta \phi| + |\delta \Omega^2| |\phi| .
\end{equation*}

Using (\ref{eq:rs_phi})--(\ref{eq:rs_omega}) as well as (\ref{eq:rs_bootstrap_diffs}) to estimate the RHS by appropriate powers of $\epsilon$ and $\Omega^2$ ((\ref{eq:logomegadiff}) is also useful here), we see that
\begin{equation*}
|\delta \ddot{\phi}| \lesssim \Omega^2 \epsilon^3,
\end{equation*}
so that, integrating this from $s = -\infty$ once and then twice, using \eqref{eq:rs_omega_estimate} we indeed get
\begin{equation*}
|\delta \dot{\phi}| + |\delta \phi| \leq C \Omega^2 \epsilon^3
\end{equation*}
as claimed. A similar procedure can be executed for the remaining equations of Section \ref{sub:evol_eqs}. This improves (\ref{eq:rs_bootstrap_diffs}) and completes the proof of this proposition.
\end{proof}

\blue{In the sequel, we fix $\Delta_{\mathcal{R}}$ such that Proposition~\ref{prop:redshift} applies.} Next, we use a Gr\"onwall argument to provide estimates in the no-shift region $\mathcal{N}$.

\begin{proposition} \label{prop:noshift}
Take $S > 0$ to be any fixed real number. Then there exists some $C(M ,\mathbf{e}, \Lambda, m^2, q_0) > 1$ and $D_N (M ,\mathbf{e}, \Lambda, m^2, q_0) > 0$ such that, in the region $\mathcal{N} = \{ - \Delta_{\mathcal{R}} \leq s \leq S \}$, the following estimates hold for $\epsilon$ sufficiently small:
\begin{equation}
|\phi| + |\dot{\phi}| \leq D_N C^s \epsilon, \label{eq:ns_phi}
\end{equation}
\begin{equation}
| \kappa^{-1} - 1 | \leq D_N C^s \epsilon^2, \label{eq:ns_kappa}
\end{equation}
\begin{equation}
| Q - \mathbf{e} | \leq  D_N C^s \epsilon^2, \label{eq:ns_Q}
\end{equation}
\begin{equation}
| \varpi - M | \leq D_N C^s \epsilon^2, \label{eq:ns_mass}
\end{equation}
\begin{equation}
| \delta r |  +  | \delta \dot{r} | + | \delta \log (\Omega^2) | + \left| \frac{d}{ds} \delta \log (\Omega^2) \right| + |\delta \tilde{A}| \leq D_N C^s \epsilon^2, \label{eq:ns_diff}
\end{equation}
\begin{equation}
| \delta \phi | + | \delta \dot{\phi} | \leq D_N C^s \epsilon^3. \label{eq:ns_phidiff}
\end{equation}
In the sequel, we only apply these estimates at $s = S$, so that allowing $D_N$ to depend also on $S$, the terms involving $C^s$ can be absorbed into $D_N$ -- indeed we define $D_N' = C^S D_N$ for convenience.
\end{proposition}

\begin{proof}
Here, we begin with some bootstrap assumptions on the geometry and the Maxwell field which clearly hold in a neighborhood of $ s = - \Delta_{\mathcal{R}}$:
\begin{equation}
|Q| \leq 2|\mathbf{e}|,
\tag{NS1} \label{eq:ns_bootstrap_Q}
\end{equation}
\begin{equation}
|\tilde{A}| \leq \frac{2|\mathbf{e}|}{r_+},
\tag{NS2} \label{eq:ns_bootstrap_gauge}
\end{equation}
\begin{equation}
|\kappa^{-1} - 1| \leq \tfrac{1}{2},
\tag{NS3} \label{eq:ns_bootstrap_raych}
\end{equation}
\begin{equation}
r \geq \tfrac{r_-}{2},
\tag{NS4} \label{eq:ns_bootstrap_r}
\end{equation}
\begin{equation}
\Omega^2 \leq 2 \Omega^2_{\max} (M, \mathbf{e}, \Lambda) \coloneqq 2 \sup_{s \in \R} \Omega^2_{RN}(s).
\tag{NS5} \label{eq:ns_bootstrap_lapse}
\end{equation}

Using these bootstraps and (\ref{eq:phi_evol}), it is straightforward to see that 
\begin{equation*}
\left | \frac{d}{ds} (\phi, \dot{\phi}) \right | \leq C_1(M, \mathbf{e}, \Lambda, q_0, m^2) | (\phi, \dot{\phi} )|,
\end{equation*}
which immediately yields (\ref{eq:ns_phi}). Using the evolution equation (\ref{eq:Q_evol}) for $Q$ then gives (\ref{eq:ns_Q}), improving (\ref{eq:ns_bootstrap_Q}) -- in fact we have $|Q - \mathbf{e}| \lesssim C_1^{2s} \epsilon^2$.

Next, we consider the following tuple of geometric and gauge quantities:
\begin{equation*}
\mathbf{X} = (\delta r, \delta \dot{r}, \delta \log \Omega^2, \tfrac{d}{ds}  \delta \log \Omega^2, \delta \tilde{A}).
\end{equation*}
Taking the differences from Reissner-Nordstr\"om for the equations (\ref{eq:r_evol}), (\ref{eq:lapse_evol}), (\ref{eq:gauge_evol}), along with the bootstrap assumptions, for some $C_2(M, \mathbf{e}, \Lambda, m^2, q_0) > 0$, we get
\begin{equation}
\left| \frac{d \mathbf{X}}{ds} \right| \leq C_2 (|\mathbf{X}| + C_1^{2s} \epsilon^2).
\end{equation}
Thus, using Gr\"onwall, we deduce (\ref{eq:ns_diff}) for some $C$ with appropriate dependence on $C_1, C_2$. For $\epsilon$ small, this improves (\ref{eq:ns_bootstrap_gauge}), (\ref{eq:ns_bootstrap_r}), (\ref{eq:ns_bootstrap_lapse}).

Next, for $s \geq - \Delta_{\mathcal R}$, we know $\Omega^{-2} \leq C_3 e^{- 2K_- s}$ for some $C_3 (M, \mathbf{e}, \Lambda, m^2, q_0) > 0$. Hence, \blue{\eqref{eq:raych_kappa}}, 
\begin{equation}
| \kappa^{-1} (s) - \kappa^{-1} (- \Delta_{\mathcal{R}}) | \lesssim C_3 e^{-2 K_- s} C_1^{2s} \epsilon^2.
\end{equation}
So, for $C$ chosen sufficiently large, we get (\ref{eq:ns_kappa}) and improve (\ref{eq:ns_bootstrap_raych}). The final estimate (\ref{eq:ns_phidiff}) is straightforward.
\end{proof}

\begin{corollary} \label{cor:noshift}
Let $S' > S$ be another arbitrarily chosen real number. Then for $S \leq s \leq S'$ and for $\epsilon$ sufficiently small, there exists some $D_-(M, \mathbf{e}, \Lambda, m^2, q_0, S, S') > 0$ such that:
\begin{equation}
|r(s) - r_-| \leq D_-( \epsilon^2 + e^{2K_- s}),
\end{equation}
\begin{equation}
\left| \tilde{A} + \frac{\mathbf{e}}{r_+} - \frac{\mathbf{e}}{r_-} \right| \leq D_- ( \epsilon^2 + e^{2 K_- s} ),
\end{equation}
\begin{equation}
D_-^{-1} e^{2 K_- s} \leq \Omega^2 \leq D_- e^{2K_- s},
\end{equation}
\begin{equation} \label{eq:ns_cor_phi}
|\phi - B \epsilon e^{i \omega_{RN} s} - \overline{B} \epsilon e^{-i \omega_{RN} s}| + | \dot{\phi} - i \omega_{RN} B \epsilon e^{i \omega_{RN} s} + i \omega_{RN} \overline{B} \epsilon e^{-i \omega_{RN} s}| \leq D_- (\epsilon^3 + \epsilon e^{2K_- s}),
\end{equation}
\begin{equation} \label{eq:ns_cor_lapse}
\left | \frac{d \log(\Omega^2)}{ds} - 2 K_- \right | \leq D_- (\epsilon^2 + e^{2K_- s}) \leq \frac{|K_-|}{100}.
\end{equation}
\end{corollary}

\begin{proof}
Note the previous proposition holds in the same way for $S$ replaced by $S'$. Therefore, the corollary follows from the difference estimates of that proposition, Corollary \ref{cor:scattering}, and some further basic computations on Reissner-Nordstr\"om such as
\begin{enumerate}
\item
$ r_{RN}(s) - r_- \lesssim e^{2 K_- s} $ for $s \geq S \gg 1$.
\item
$ \Omega^2_{RN} \sim e^{2 K_- s} $ for $s \geq S \gg 1$.
\end{enumerate}
The final part of (\ref{eq:ns_cor_lapse}) clearly follows by taking $S$ large and $\epsilon$ small. In the sequel, we will take advantage of the fact that $S$ and $S'$ can be taken to be as large as required.
\end{proof}

\subsection{Estimates on the early blue-shift region} \label{sub:earlyblueshift}

The early blue-shift $\mathcal{EB}$ is defined by $\mathcal{EB} = \{ S \leq s \leq s_{lin} = (|2 K_-|)^{-1} \log (\nu \epsilon^{\blue{-2}})\}$, where $\nu(M, \mathbf{e}, \Lambda, m^2, q_0) > 0$ is \blue{chosen such that \eqref{eq:eb_kappaestimate} holds}. In the blue-shift regions, we begin to exploit the fact that $\Omega^2$ is exponentially decaying in $s$.

\begin{proposition} \label{prop:earlyblueshift}
Take the quantity $S$ in Proposition \ref{prop:noshift} to be sufficiently large, and $\nu(M, \mathbf{e}, \Lambda, m^2, q_0) > 0$ to be sufficiently small. Then there exists some $D_E(M, \mathbf{e}, \Lambda, m^2, q_0) > 0$ such that the following hold:
\begin{equation}
| \delta \log (\Omega^2) | + s \left| \frac{d}{ds} \delta \log (\Omega^2) \right| \leq D_E \epsilon^2 s^2, \label{eq:eb_lapse}
\end{equation}
\begin{equation}
| \phi | + | \dot{\phi} | \leq D_E |B| \epsilon, \label{eq:eb_phi}
\end{equation}
\begin{equation}
| \delta r | s^{-1} + |\delta \dot{r} | + |\delta \tilde{A}| \leq D_E \epsilon^2, \label{eq:eb_diff}
\end{equation}
\begin{equation}
\left | \kappa^{-1} - 1 \right | \leq \tfrac{1}{100}, \label{eq:eb_kappa}
\end{equation}
\begin{equation}
| \varpi -  M | + | Q - \mathbf{e} | \leq D_E \epsilon^2 s, \label{eq:eb_masscharge}
\end{equation}
\begin{equation}
|\delta \phi| + |\delta \dot{\phi}| \leq D_E \epsilon^3 s. \label{eq:eb_phidiff}
\end{equation}
\end{proposition}

\begin{proof}
Recall the quantities $D_N'$ from Proposition \ref{prop:noshift} and $B \in \CC$ from \blue{Proposition~\ref{prop:linscat}}. We bootstrap the following estimates:
\begin{equation}
| \delta \log \Omega^2 | \leq 4 D_N' \epsilon^2 s^2,
\tag{EB1} \label{eq:eb_bootstrap_lapse}
\end{equation}
\begin{equation}
|\phi| + \omega_{RN}^{-1}|\dot{\phi}| \leq 20 |B| \epsilon,
\tag{EB2} \label{eq:eb_bootstrap_phi}
\end{equation}
\begin{equation}
r \geq \tfrac{r_-}{2},
\tag{EB3} \label{eq:eb_bootstrap_r}
\end{equation}
\begin{equation}
|\tilde{A}| \leq \tfrac{2\mathbf{e}}{r_+}.
\tag{EB4} \label{eq:eb_bootstrap_gauge}
\end{equation}

To start with, note that, by (\ref{eq:eb_bootstrap_lapse}) and the Reissner-Nordstr\"om asymptotics, we have that, for some $C > 0$, 
\begin{equation*}
C^{-1} e^{2 K_- s} \leq \Omega^2 \leq C e^{2 K_- s}.
\end{equation*}
Note also that $|\delta Q| \lesssim \epsilon^2 s$ is immediate from (\ref{eq:ns_Q}) and the bootstraps.

Next, we use (\ref{eq:eb_bootstrap_phi}) along with the Raychaudhuri equation (\ref{eq:raych_kappa}) to get the estimate on $\kappa$:
\begin{equation} \label{eq:eb_kappaestimate}
\left | \kappa^{-1} - 1 \right | \leq D_N' \epsilon^2 + C |B|^2 \epsilon^2 e^{2|K_-|s} \leq \frac{1}{100},
\end{equation}
where $\nu$ is chosen such that the second inequality holds. One can rewrite the left hand side of this expression as
\begin{equation*}
\left| \kappa^{-1} - 1 \right| = 4 \Omega^{-2}_{RN} | \delta \dot{r} + \dot{r} ( e^{- \delta \log \Omega^2} - 1 ) |.
\end{equation*}

Hence, we can use (\ref{eq:eb_kappaestimate}) and (\ref{eq:eb_bootstrap_lapse}) to produce a preliminary estimate on $\delta \dot{r}$:
\begin{equation}
|\delta \dot{r}| \lesssim \Omega^2 | \delta \log \Omega^2 | + e^{2K_-s} \epsilon^2 + |B|^2 \epsilon^2 \lesssim e^{2K_- s} ( D_N' + D_N' s^2 ) \epsilon^2 + |B|^2 \epsilon^2.
\end{equation}
Integrating this up\footnote{
We use here a straightforward calculus lemma: given $\alpha > 0$ and $N \in \N$, then for $s_{0} \in \R$  large, $\int^{\infty}_{s_0} s^N e^{- \alpha s} \, ds \lesssim_{\alpha, N} s_0^N e^{- \alpha s_0}.$
}
from $s = S$, we get $|\delta r| \lesssim \epsilon^2 s$, or to be more specific we have the following, where $C$ depends on the black hole parameters, $m^2$ and $q_0$ but not the choice of $S$ (unlike the quantity $D_N' = C^S D_N$ which does have exponential dependence on $S$):
\begin{equation}
|\delta r| \leq C \epsilon^2 (|B|^2 s + D_N' + D_N' S^2 e^{2 K_-S}) \leq C \epsilon^2 (|B|^2 s + D_N').
\end{equation}
Here, we use that $S^2 e^{2K_- S}$ is uniformly bounded for $S > 0$.

We wish to improve the bootstrap (\ref{eq:eb_bootstrap_lapse}). So we use the equation (\ref{eq:lapse_evol}) and take differences from Reissner-Nordstr\"om, leading to following inequality
\begin{align*}
\left | \frac{d^2}{ds^2} \delta \log \Omega^2 \right |
&\lesssim \Omega^2 ( |\delta \log \Omega^2| + |\delta r| + |\delta \dot{r}| + |\delta Q| ) + |\phi|^2 + |\dot{\phi}|^2, \\
&\lesssim e^{2K_-s} (D_N' s^2 + |B|^2 s + D_N') \epsilon^2 + |B|^2 \epsilon^2.
\end{align*}
Integrating this expression up from $s = S$ twice, we get
\begin{align*}
\left| \frac{d}{ds} \delta \log \Omega^2 \right|
&\lesssim D_N' \epsilon^2 + |B|^2 \epsilon^2 s, \\
|\delta \log \Omega^2 |
&\lesssim D_N' \epsilon^2 s + |B|^2 \epsilon^2 s^2.
\end{align*}
So, for $S$ sufficiently large, we improve the bootstrap (\ref{eq:eb_bootstrap_lapse}) -- note that $D_N' = C^S D_N$ grows exponentially in $S$ so that $D_N' \gg |B|^2$ for $S$ large.

Then taking differences between our ODEs and the Reissner-Nordstr\"om ODEs will lead to the estimates (\ref{eq:eb_diff}) and (\ref{eq:eb_masscharge}), thus improving upon the bootstraps (\ref{eq:eb_bootstrap_r}) and (\ref{eq:eb_bootstrap_gauge}).

We now turn to the estimates on the charged scalar field. We shall use the equation (\ref{eq:phi_evol_2}) and consider the quantity
\begin{equation}
H = r^4 \dot{\phi}^2 + r^4 q_0^2 |\tilde{A}|^2 \phi^2.
\end{equation}
As we already control $r$ and $|\tilde{A}|$ from below in this region, producing an upper bound for $H$ would give us desired estimates on $|\phi|, |\dot{\phi}|$. Using (\ref{eq:phi_evol_2}) as well as (\ref{eq:gauge_evol}), we compute
\begin{equation*}
\frac{dH}{ds} = - \frac{m^2 \Omega^2 r^4}{2} \phi \dot{\phi} - \frac{r^2 q_0^2 \tilde{A} Q \Omega^2}{2} \phi^2 + 4 r^3 \dot{r} q_0^2 |\tilde{A}|^2 \phi^2.
\end{equation*}

The last term is negative since $\dot{r} < 0$, so we can apply Cauchy-Schwarz (again recalling the lower bounds on $r$ and $|\tilde{A}|$) to see that 
\begin{equation*}
\frac{dH}{ds} \leq C_H \Omega^2 H,
\end{equation*}
for some $C_H = C_H(M, \mathbf{e}, \Lambda, m^2, q_0) > 0$. Therefore, by Gr\"onwall, so long as we pick $S$ large enough such that
\begin{equation*}
\int_S^{s_{lin}} C_H \Omega^2 \, ds \leq C C_H \int_S^{\infty} e^{2 K_- s} \, ds < \log 2,
\end{equation*}
we deduce that $H(s) \leq 2 H(S) \leq 16 r_-^4 \omega_{RN}^2 |B|^2 \epsilon^2$, where the latter estimate for $H(S)$ is computed using (\ref{eq:ns_cor_phi}). This improves (\ref{eq:eb_bootstrap_phi}).

Finally, for the difference estimate on the scalar field we use the modified quantity $\tilde{H}$, given by
\begin{equation}
\tilde{H} = r^4 | \delta \dot{\phi} |^2 + r^4 q_0^2 |\tilde{A}|^2 |\delta \phi|^2.
\end{equation}
To find the analogous differential inequality here, we first need to find the difference version of (\ref{eq:phi_evol}):
\begin{equation} \label{eq:eb_htilde}
\delta \ddot{\phi} = - \frac{2\dot{r}}{r} \delta \dot{\phi} - q_0^2 |\tilde{A}|^2 \delta \phi - \frac{m^2 \Omega^2}{4} \delta \phi + \mathcal{E}_{EB},
\end{equation}
where the error\footnote{In this expression for the error term, our estimates in the region $\mathcal{EB}$ show that both $- s \dot{r}$ and $s^2 \Omega^2_{RN}$ are $O(1)$. Nonetheless, we choose to write it in this form both for clarity and because we shall use the same expression again in the late blue-shift region, where both \blue{of} these are still $O(1)$ but for different reasons.} term $\mathcal{E}_{EB}$ is estimated by $|\mathcal{E}_{EB}| \lesssim \epsilon^3 ( 1 - s \dot{r} + s^2 \Omega^2_{RN})$.

Using a similar procedure as before, one can use this equation to get the differential inequality:
\begin{equation}
\frac{d\tilde{H}}{ds} \lesssim \Omega^2 \tilde{H} + \epsilon^3 (1 - s \dot{r}+ s^2 \Omega_{RN}^2) \sqrt{\tilde{H}}.
\end{equation}
Rewriting this expression as a differential inequality in terms of $\sqrt{\tilde{H}}$ and using the usual Gr\"onwall inequality with a \blue{(nondecreasing)} inhomogeneous term, we get
\begin{equation*}
\sqrt{\tilde{H}} (s) \lesssim \sqrt{\tilde{H}}(S) + \int_S^s \epsilon^3 (1 - s \dot{r} + s^2 \Omega^2_{RN}) \, ds \lesssim \epsilon^3 s
\end{equation*}
as required. This concludes the proof of (\ref{eq:eb_phidiff}) and the proposition.
\end{proof}

\subsection{Estimates on the late blue-shift region} \label{sub:lateblueshift}

To close this section, we give the estimates on the late blue-shift region. Here, we leave the regime where Cauchy stability holds, in particular showing that $\dot{r}$ remains at size $O(\epsilon^2)$  rather than continuing to decay exponentially.

\begin{proposition} \label{prop:lateblueshift}
Choose $\Delta_{\mathcal{B}} > 0$ sufficiently small, and define $b_- = \frac{2|B|\omega_{RN}}{2|K_-|}$ . Then, in the region $ \mathcal{LB} = \{ s_{lin} \leq s \leq \Delta_{\mathcal{B}} \epsilon^{-1} \}$, there exists some $D_L (M, \mathbf{e}, \Lambda, m^2, q_0) > 0$ such that the following estimates hold, \blue{recalling the quantity $\alpha_-$ from \eqref{eq:omega_asymptotic_2}}:
\begin{equation}
| r(s) - r_- | \leq D_L \epsilon^2 s, \label{eq:lb_rdiff}
\end{equation}
\begin{equation} 
\left | - \dot{r} - \frac{4 |B|^2 \omega_{RN}^2 \epsilon^2 r(s)} {2 |K_-|} \right | \leq D_L( \Omega^2 + \epsilon^4 s), \label{eq:lb_rdot}
\end{equation}
\begin{equation}
\left | \frac{d}{ds} \log (\Omega^2) - 2 K_- \right | \leq D_L \epsilon^2 s, \label{eq:lb_lapse}
\end{equation}
\begin{equation}
\left | \frac{d}{ds} \log ( r(s)^{- b_-^{-2} \epsilon^{-2}} \cdot \Omega^2 (s) ) \right | \leq D_L (\epsilon^{-2} \Omega^2 + \epsilon^2 s), \label{eq:lb_rlapse}
\end{equation}
\begin{equation}
\left | \log \left ( \frac{\Omega^2}{\blue{\alpha_-} e^{2K_-s}} \right) \right| \leq D_L \epsilon^2 s^2, \label{eq:lb_lapse_2}
\end{equation}
\begin{equation}
|\phi| + |\dot{\phi}| \leq D_L \epsilon, \label{eq:lb_phi}
\end{equation}
\begin{equation}
|Q - \mathbf{e}| s^{-1} + |\delta \tilde{A}| \leq D_L \epsilon^2, \label{eq:lb_Q}
\end{equation}
\begin{equation}
|\phi - B \epsilon e^{i \omega_{RN} s} - \overline{B} \epsilon e^{-i \omega_{RN} s}| + | \dot{\phi} - i \omega_{RN} B \epsilon e^{i \omega_{RN} s} + i \omega_{RN} \overline{B} \epsilon e^{-i \omega_{RN} s}| \leq D_L \epsilon^3 s. \label{eq:lb_phidiff}
\end{equation}
\end{proposition}

\begin{proof}
We use the following bootstrap assumptions, where the constant $B_1(M, \mathbf{e}, \Lambda, m^2, q_0) > 0$ will be described later in the proof:
\begin{equation}
\left| \frac{d}{ds} \log \Omega^2 - 2K_- \right| \leq |K_-|,
\tag{LB1} \label{eq:lb_bootstrap_lapse}
\end{equation}
\begin{equation}
|\tilde{A}| \leq \tfrac{2|\mathbf{e}|}{r_+},
\tag{LB2} \label{eq:lb_bootstrap_gauge}
\end{equation}
\begin{equation}
- \dot{r} \leq \frac{8 |B|^2 \omega_{RN}^2 r_-}{|2 K_-|} \epsilon^2,
\tag{LB3} \label{eq:lb_bootstrap_r}
\end{equation}
\begin{equation}
|Q| \leq 2|\mathbf{e}|,
\tag{LB4} \label{eq:lb_bootstrap_Q}
\end{equation}
\begin{equation}
| \delta \log \Omega^2 | \leq B_1 \epsilon^2 s^2.
\tag{LB5} \label{eq:lb_bootstrap_lapsediff}
\end{equation}

From (\ref{eq:lb_bootstrap_r}), we have the trivial estimate $|\delta \dot{r}| \lesssim \epsilon^2$, and thus $| \delta r | \lesssim \epsilon^2 s$. 

Due to (\ref{eq:lb_bootstrap_lapse}), we have that in $\mathcal{LB}$,
\begin{equation*}
\int_{s_{lin}}^s \Omega^2 (s') \, ds' \leq |K_-|^{-1} \Omega^2 (s_{lin}) \lesssim \epsilon^2,
\end{equation*}
So given all the bootstrap assumptions, we may proceed exactly as in the proof of Proposition \ref{prop:earlyblueshift} to recover the estimates (\ref{eq:lb_phi}) and (\ref{eq:lb_phidiff}). Note that the error $\mathcal{E}_{LB}$ in (\ref{eq:eb_htilde}) is replaced by $\mathcal{E}_{LB}$ which still satisfies the same estimate
\begin{equation*}
|\mathcal{E}_{LB}| \lesssim \epsilon^3 ( 1 - s \dot{r} + s^2 \Omega^2_{RN} ) \lesssim \epsilon^3,
\end{equation*}
so the proof follows in an identical fashion.

We also get the estimate (\ref{eq:lb_Q}) just as in Proposition \ref{prop:earlyblueshift}, improving the two bootstraps (\ref{eq:lb_bootstrap_gauge}) and (\ref{eq:lb_bootstrap_Q}).

Using the difference version of (\ref{eq:omega_evol_2}), we get the inequality
\begin{equation}
\left| \frac{d^2}{ds^2} \delta \log ( r \Omega^2 ) \right| \lesssim B_1 \Omega^2 \epsilon^2 s^2 + \epsilon^2.
\end{equation}
Integrating once, we deduce that
\begin{equation} \label{eq:lb_omegaestimate}
\left| \frac{d}{ds} \delta \log (r \Omega^2) \right| \lesssim B_1 \epsilon^2 s_{lin}^2(\epsilon) \cdot \epsilon^2 + \epsilon^2 s \lesssim \epsilon^2 s.
\end{equation}

Therefore, in light of the previous estimates on $\delta r, \delta \dot{r}$, for $\epsilon$ sufficiently small we have
\begin{equation}
\left| \frac{d}{ds} \log (r \Omega^2)(s) - 2 K_- \right|
+
\left| \frac{d}{ds} \log (\Omega^2)(s) - 2 K_- \right|
\leq C \epsilon^2 s \leq \frac{|K_-|}{10},
\end{equation}
where the final inequality follows for $\Delta_{\mathcal{B}}$ taken sufficiently small.

This improves the bootstrap (\ref{eq:lb_bootstrap_lapse}), and moreover, integrating (\ref{eq:lb_omegaestimate}) once again will yield 
\begin{equation} \label{eq:lb_deltalogromega}
|\delta \log (r \Omega^2)| \leq C (\epsilon^2 s_{lin}^2(\epsilon) + B_1 \epsilon^4 s_{lin}^2(\epsilon) + \epsilon^2 s^2 ),
\end{equation}
which improves (\ref{eq:lb_bootstrap_lapsediff}) after once again accounting for $\delta r$, so long as $B_1$ is chosen sufficiently large (i.e.~larger than $4C$ for the $C$ appearing in \blue{\eqref{eq:lb_deltalogromega}}).

It remains to improve upon the bootstrap (\ref{eq:lb_bootstrap_r}). For this, as in the uncharged scalar field model of \cite{VDM21}, we use the Raychaudhuri equation in the convenient form (\ref{eq:raych_transport}). We begin by estimating the expression involving the scalar field; using (\ref{eq:lb_phidiff}) and (\ref{eq:lb_Q}), we have
%
\begin{equation*}
|\dot{\phi}|^2 + |\tilde{A}|^2 q_0^2 |\phi|^2 
=
|B \omega_{RN} \epsilon e^{i \omega_{RN}s} - \overline{B} \omega_{RN} \epsilon e^{-i \omega_{RN}s}|^2 + \omega_{RN}^2 | B \epsilon e^{i \omega_{RN} s} + \overline{B} \epsilon e^{-i\omega_{RN} s}|^2 + \mathcal{E}_B.
\end{equation*}
Expanding out the trigonometric expressions on the right, we can see this can be rewritten as
\begin{equation}
|\dot{\phi}|^2 + |\tilde{A}|^2 q_0^2 |\phi|^2
= 4 |B|^2 \omega_{RN}^2 \epsilon^2 + \mathcal{E}_B,
\end{equation}
where the error term is bounded by $|\mathcal{E}_B| \lesssim \epsilon^4 s$.

Therefore, using also (\ref{eq:lb_rdiff}) and (\ref{eq:lb_omegaestimate})
, one finds that (\ref{eq:raych_transport}) may be written in the form
\begin{equation} \label{eq:lb_raych_1}
\frac{d}{ds} ( - \dot{r} ) - 2 K_- (- \dot{r} ) = 4 |B|^2 \omega_{RN}^2 r_- \epsilon^2 + \mathcal{E}_B(s),
\end{equation}
with the error $\mathcal{E}_B(s)$ once again satisfying $|\mathcal{E}_B(s)| \lesssim \epsilon^4 s$.

We now use a classical integrating factor of $e^{- 2K_-s}$ to integrate (\ref{eq:lb_raych_1}) between $s_{lin}$ and $s$, yielding
\begin{equation} \label{eq:lb_raych_2}
- \dot{r}(s) = - e^{2 K_- (s - s_{lin})} \dot{r}(s_{lin}) + \int_{s_{lin}}^s 4 |B|^2 \omega^2_{RN} r_- \epsilon^2 e^{2 K_-(s - s')} + \mathcal{E}_B(s') e^{2 K_-(s - s')} \, ds'.
\end{equation}
Using $|\mathcal{E}_B(s')| \lesssim \epsilon^4 s'$ and computing these integrals, we find that
\begin{equation} \label{eq:lb_raych_3}
\left| - \dot{r}(s) - \frac{4 |B|^2 \omega_{RN}^2 r_- \epsilon^2}{|2K_-|} \right| \lesssim e^{2K_-(s - s_{lin})} \left| - \dot{r}(s_{lin}) - \frac{4 |B|^2 \omega_{RN}^2 r_- \epsilon^2}{|2K_-|} \right| + \epsilon^4 s.
\end{equation}
From this and prior estimates, we yield (\ref{eq:lb_rdot}), thus improving the remaining bootstrap (\ref{eq:lb_bootstrap_r}).

The estimate (\ref{eq:lb_rlapse}) simply comes from combining (\ref{eq:lb_rdot}) and (\ref{eq:lb_lapse}). \blue{Finally, \eqref{eq:lb_lapse_2} follows from combining \eqref{eq:omega_asymptotic_2} with \eqref{eq:lb_bootstrap_lapsediff}, plus using that $e^{2K_- s} \leq e^{2K_- s_{lin}} \lesssim \ep^2$ for $s \in \mathcal{LB}$.}
\end{proof}

We conclude this section with a straightforward corollary concerning quantities evaluated at a specific point $s = s_{O}(\epsilon) \coloneqq 50 s_{lin}(\epsilon)$ that will be useful in the next region. For ease of notation, we first define the dimensionless parameter $\mathfrak{W}(M, \mathbf{e}, \Lambda, q_0) > 0$ as:
\begin{equation} \label{eq:frakw}
\mathfrak{W} \coloneqq \sqrt{ \frac{\omega_{RN}(M, \mathbf{e}, \Lambda, q_0)}{|2 K_- (M, \mathbf{e}, \Lambda)|}}.
\end{equation}

\begin{corollary} \label{cor:lateblueshift}
\blue{Let $s_O(\ep) =50 s_{lin}(\epsilon) \in \mathcal{LB}$.} Then, upon defining the quantities
\begin{gather} \label{a0.def}
r_0 = r( s_O), \quad
\omega_0 = |q_0 \tilde{A}( s_O)|, \quad
\xi_0 = \omega_0 \left( - \frac{d}{ds} \left( \frac{r^2}{r_-^2 \epsilon^2} \right) \right)^{-1}( s_O), \quad Q_0 = Q(s_O),
\end{gather}
we have the following estimates:
\begin{equation}\label{a0.est}
|r_0 - r_-| + \log (\epsilon^{-1}) |\omega_0 - \omega_{RN}| + \left| \xi_0 - \frac{1}{8 |B|^2 \mathfrak{W}^2} \right| + |Q_0 - \mathbf{e}| \leq D_L \epsilon^2 \log (\epsilon^{-1}).
\end{equation}
\blue{In particular, the following upper and lower bounds hold for $\xi_0$.
\begin{equation} \label{eq:xi0_upperlower}
    \xi_0 \leq \frac{1}{4 |B|^2 \mathfrak{W}^2}, \quad \xi_0^{-1} \leq 16 |B|^2 \mathfrak{W}^2.
\end{equation}
Finally, one has
\begin{equation} 
    |\phi(s_O)| + |\dot{\phi}(s_O)| \leq D_L \epsilon.
\end{equation}}
\end{corollary}

\section{The collapsed oscillations} \label{sec:oscillations}

From this point, the results begin to \blue{differ} from that of \cite{VDM21}. We next study the region of `\blue{collapsed} oscillations' (see Section~\ref{Jo.intro}), a region where $r$ eventually becomes $O(\ep)$ small.
The dynamics in this section are largely driven by the charged scalar field, which remains oscillatory but also exhibits a slowly growing behavior, in that the amplitude grows from $O(\epsilon)$ at  start of the regime to $O(1)$ at the end of the regime. Indeed, after a change of variables, the behavior of $\phi$ will be well approximated by a \textit{Bessel function} of order $0$. Schematically, in the  region  $\epsilon \lesssim r \ll r_-$, the scalar field behaves as 
\begin{equation}\label{schm}
\phi(r) \approx \frac{C \ep}{r} \cdot \cos
\left( \frac{r^2- r_-^{2}}{\epsilon^2} +O(\log(\ep^{-1})) \right).
\end{equation} 

We will, in fact, need to track the growth-oscillatory behavior more precisely \blue{using the explicit} Bessel functions $J_0(x)$ and $Y_0(x)$, defined as the two  linearly independent solutions of Bessel's equation of order $0$: 
\begin{equation} 
\frac{d}{dz} \left( z \frac{df}{dz} \right) + z f = 0.
\end{equation}
\blue{They exhibit} the following asymptotic behavior (see also Facts \ref{fact:bessel1_taylor}, \ref{fact:bessel2_taylor} and \ref{fact:bessel_bigasymp} in Appendix \ref{sec:appendix_A} for more detailed asymptotics):
\begin{equation} \label{eq:jo_besselj}
J_0(z) = 
\begin{cases}
\sqrt{ \frac{2}{\pi z} } \cos \left( z - \frac{\pi}{4} \right) + O( z^{-3/2} )
& \text{ as } z \to +\infty, \\
1 + O(z)
& \text{ as } z \to 0,
\end{cases}
\end{equation}
\begin{equation} \label{eq:jo_bessely}
Y_0(z) = 
\begin{cases}
\sqrt{ \frac{2}{\pi z} } \sin \left( z - \frac{\pi}{4} \right) + O( z^{-3/2} )
& \text{ as } z \to +\infty, \\
\frac{2}{\pi} \left( \ln \left( \frac{z}{2} \right) + \gamma \right) + O(z)
& \text{ as } z \to 0,
\end{cases}
\end{equation}
\blue{where $\gamma \approx 0.577$} is the Euler-Mascheroni constant.  We see, therefore, that the Bessel functions \blue{$J_0(z)$ and $Y_0(z)$} both oscillate and decay\footnote{
As $z$ is decreasing in $s$ in our setting, and we are interested in behavior as $z$ decreases, the decay of Bessel functions exhibits itself as inverse polynomial growth of the amplitude of the scalar field $\phi$.
}
at a slow inverse polynomial rate as $z \to \infty$, but exhibit very different behavior in the $z \to 0$ regime. We later show the schematic estimate \eqref{schm} takes the following more precise form involving  the renormalized variable $z=\frac{r^2}{\ep^2}$ \blue{and the frequency $\xi_0$ defined in Corollary~\ref{cor:lateblueshift}}:
\begin{equation*}
\phi(r) \approx  C_J(\ep) \cdot J_0 \left( \frac{\xi_0}{r_-^2} \cdot \frac{r^2}{  \ep^2} \right) +  C_Y(\ep) \cdot Y_0 \left(
\frac{\xi_0}{r_-^2} \cdot \frac{r^2}{  \ep^2} \right), 
\end{equation*}
\begin{equation*}
C_J(\ep) \approx \sqrt{2\pi \xi_0 } \cdot |B| \cdot  \cos( \xi_0 \cdot \ep^{-2}),\ C_Y(\ep) \approx  \sqrt{2\pi \xi_0 } \cdot |B| \cdot  \sin( \xi_0 \cdot \ep^{-2}).
\end{equation*}
\blue{In light of \eqref{eq:jo_besselj}--\eqref{eq:jo_bessely}, tracking the size of the coefficients $C_J(\epsilon)$ and $C_Y(\epsilon)$ precisely will be important in later sections.}

The main objective of Section~\ref{sec:oscillations} is to prove these schematic estimates. In view of the logarithmic \blue{growth} of $Y_0(z)$ for $z\ll 1$, the precise value of the coefficient $C_Y(\ep)$ will be crucial in the geometry of the later regions where $r\ll \ep$ (see  Section \ref{sec:protokasner}). 


Once the scalar field asymptotics are understood, we also account for the scalar field \emph{backreaction} onto both the Maxwell field and the geometry. In this region, the backreaction on the spacetime geometry is  minimal; however the backreaction on the Maxwell field has a rather curious effect. \blue{That is, within the oscillatory region,} the charge $Q(s)$ will decrease in absolute value from (approximately) $\mathbf{e}$ to (approximately) $Q_{\infty}$, where $Q_{\infty}(M, \mathbf{e}, \Lambda)$ depends on the black hole parameters and lies strictly between $\frac{\mathbf{e}}{2}$ and $\mathbf{e}$, see (\ref{eq:qinfty}) and Lemma \ref{lem:jo_charge_retention}.  $Q_{\infty}$ being bounded away from $0$ is causing the \emph{charge retention} in Theorem~\ref{thm.intro}, see Section~\ref{osc.intro}.

We now state the main result of this section, \blue{where $\mathcal{O} = \{ s \geq s_O(\epsilon): r(s) \geq 2 |B| \mathfrak{W} \epsilon r_- \}$ and we recall that $ \mathfrak{W}$ was defined in \eqref{eq:frakw}.}

\begin{proposition} \label{prop:oscillation}
\blue{For $s \in \mathcal{O}$}, the lapse $\Omega^2$ \blue{obeys the following upper bound, where we recall $\alpha_-$ from \eqref{eq:omega_asymptotic_2}.}
\begin{equation} \label{eq:osc_lapse}
\Omega^2(s) \leq \alpha_- \exp( K_- s ).
\end{equation}
There exists some $D_O(M, \mathbf{e}, \Lambda, m^2, q_0) > 0$ such that, for $s \in \mathcal{O}$, we have the following estimates:
\begin{equation} \label{eq:osc_r}
| - r \dot{r} (s) - 4 |B|^2 \mathfrak{W}^2 r_-^2 \epsilon^2 \omega_{RN} | \leq D_O \epsilon^4 \log \epsilon^{-1},
\end{equation}
\begin{equation} \label{eq:osc_lapse_2}
\left| \frac{d}{ds} \log \Omega^2 - 2 K_- \right| \leq D_O \epsilon^2 ( \log (\epsilon^{-1}) + r^{-2}(s) ).
\end{equation}
\begin{equation} \label{eq:osc_gauge}
| q_0 \tilde{A}(s) - q_0 \tilde{A}_{RN, \infty} | \leq D_O \epsilon^2,
\end{equation}
\begin{equation} \label{eq:osc_Q}
\left| Q(s) - \mathbf{e} + \frac{1}{4} \frac{2 K_-}{\tilde{A}_{RN, \infty}} (r_-^2 - r^2(s)) \right| \leq D_O \epsilon^2 \log (\epsilon^{-1}).
\end{equation}

Finally, there exist coefficients $C_J(\epsilon)$ and $C_Y(\epsilon)$ determined via the formula (\ref{eq:jo_linalg}) and satisfying (\ref{eq:bessel_j_coeff}) and (\ref{eq:bessel_y_coeff})
such that we have the following two estimates on the scalar field: for all $s \in \mathcal{O}$,
\begin{equation} \label{eq:osc_phi_1}
\left | \phi(s) - \left(  C_J (\epsilon) J_0 \left( \frac{\xi_0 r^2(s)}{r_-^2 \epsilon ^2} \right) + C_Y(\epsilon) Y_0 \left( \frac{ \xi_0 r^2(s)} {r_-^2 \epsilon^2} \right)\right) \right| \leq D_O \epsilon^2 \log(\epsilon^{-1}),
\end{equation}
\begin{equation} \label{eq:osc_phi_2}
\left | \dot{\phi}(s) - \omega_0 \left( C_J (\epsilon) J_1 \left( \frac{\xi_0 r^2(s)}{r_-^2 \epsilon^2} \right) + C_Y(\epsilon) Y_1 \left( \frac{\xi_0 r^2(s)}{r_-^2 \epsilon^2} \right) \right) \right |  \leq  D_O \epsilon^2 \log(\epsilon^{-1}).
\end{equation}
Here, $\xi_0$ and $\omega_0$ are as defined in Corollary \ref{cor:lateblueshift}.
\end{proposition}

\blue{Later, it will be convenient to refer to \eqref{eq:osc_Q}, \eqref{eq:osc_phi_1} and \eqref{eq:osc_phi_2} evaluated at the right endpoint of $\mathcal{O}$.
\begin{corollary} \label{cor:oscillation}
    Let $s_{PK}(\ep) > 0$ be such that $r(s_{PK}) = 2 |B| \mathfrak{W} \ep r_-$. Then
    \begin{equation} \label{eq:s0}
\left| s_{PK} - (8|B|^2 \mathfrak{W}^2 \omega_{RN})^{-1} \epsilon^{-2} \right| \leq D_O \log( \epsilon^{-1} ).
\end{equation}
    Further, if one defines $Q_{\infty}$ as
    \begin{equation} \label{eq:qinfty}
        Q_{\infty} (M, \mathbf{e}, \Lambda) \coloneqq \mathbf{e} - \frac{2 K_- (M, \mathbf{e}, \Lambda) \cdot r_-^2(M, \mathbf{e}, \Lambda)}{4 \tilde{A}_{RN, \infty} (M, \mathbf{e}, \Lambda)},
    \end{equation}
    one may infer the following estimates at $s = s_{PK}$.
    \begin{gather} \label{eq:q_precise_end}
        |Q(s_{PK}) - Q_{\infty}| \leq D_O \epsilon^2 \log (\epsilon^{-1}),
        \\[0.5em] \label{eq:osc_phi_1_precise}
        \left | \phi(s_{PK}) - \left(  C_J (\epsilon) J_0 (4 \xi_0 |B|^2 \mathfrak{W}^2 ) + C_Y(\epsilon) Y_0 (4 \xi_0 |B|^2 \mathfrak{W}^2) \right) \right| \leq D_O \epsilon^2 \log(\epsilon^{-1}),
        \\[0.5em] \label{eq:osc_phi_2_precise}
        \left | \dot{\phi}(s_{PK}) - \omega_0 \big (  C_J (\epsilon) J_1 ( 4 \xi_0 |B|^2 \mathfrak{W}^2 ) + C_Y(\epsilon) Y_1 ( 4 \xi_0 |B|^2 \mathfrak{W}^2 ) \big ) \right| \leq D_O \epsilon^2 \log(\epsilon^{-1}).
    \end{gather}
\end{corollary}}

This section will be organized into three parts. In Section \ref{sub:osc_prelim}, we introduce the main bootstrap assumptions for the region $\mathcal{O}$, and make several preliminary observations. In Section \ref{sub:bessel}, we derive the main scalar field estimates, establishing the aforementioned Bessel type behavior. 
In Section \ref{sub:osc_qomega}, we use the results on the scalar field to improve the bootstrap assumptions, proving in particular the estimates (\ref{eq:osc_lapse}), (\ref{eq:osc_lapse_2}) and (\ref{eq:osc_Q}) regarding $\Omega^2(s)$ and $Q(s)$, 
then complete the proof of Proposition \ref{prop:oscillation}.

\subsection{Bootstraps and preliminary estimates} \label{sub:osc_prelim}

\blue{Throughout the proof of} Proposition \ref{prop:oscillation}, we make reference to the following \blue{three} bootstrap assumptions:

\begin{equation}
\Omega^2 \leq \epsilon^{40},
\tag{O1} \label{eq:jo_bootstrap_lapse}
\end{equation}
\begin{equation}
|\phi| + |\dot{\phi}| \leq \epsilon^{-1},
\tag{O2} \label{eq:jo_bootstrap_phi}
\end{equation}
\begin{equation}
|Q| \leq \epsilon^{-2}.
\tag{O3} \label{eq:jo_bootstrap_Q}
\end{equation}

By Proposition \ref{prop:lateblueshift} and Corollary~\ref{cor:lateblueshift}, these hold within a neighborhood of $s = s_O(\ep)$. It will be convenient to define the bootstrap region $\mathcal{O}_{boot}$ as \red{the connected component of $\{ s \in \mathcal{O}: s \leq \epsilon^{-4}, \text{\eqref{eq:jo_bootstrap_lapse}--\eqref{eq:jo_bootstrap_Q} apply}\}$ such that $s_O \in \mathcal{O}_{boot}$.}
We now make some preliminary estimates using these bootstraps and Proposition \ref{prop:lateblueshift}.
Morally, the following lemma will allow us to treat $- r \dot{r}$ and $\tilde{A}$ as constant inside $\mathcal{O}$, at least up to an extremely small error. 

\begin{lemma} \label{lemma:jo_geometry}
\blue{Any $s \in \mathcal{O}_{boot}$ satisfies $s \lesssim \ep^{-2}$. Furthermore, there  exists some constant $D_O(M, \mathbf{e}, \Lambda, m^2, q_0) > 0$ such that, for $s \in \mathcal{O}_{boot}$, }
\begin{equation} \label{eq:jo_rd}
\left |\frac{d}{ds} (-r \dot{r}) \right |+ 	\left |r\dot{r}(s)- r\dot{r}(s_{O})\right | \leq D_O \epsilon^{30},
\end{equation}
\begin{equation} \label{eq:jo_r}
\left| -r\dot{r}(s) - 4 |B|^2 \mathfrak{W}^2 \omega_{RN} r_-^2 \epsilon^2 \right| \leq D_O \epsilon^4 \log(\epsilon^{-1}),
\end{equation}
\begin{equation} \label{eq:jo_gauge}
|q_0\tilde{A} (s) - q_0 \tilde{A} ( s_O )| \leq D_O \epsilon^{30},
\end{equation}
\begin{equation} \label{eq:jo_gauge_2}
||q_0 \tilde{A}| (s) - \omega_{RN}| \leq D_O \epsilon^2.
\end{equation}
\end{lemma}

\begin{proof}
These estimates are immediate from (\ref{eq:r_evol_2}) and (\ref{eq:gauge_evol}). Indeed, letting $f$ be either $-r \dot{r}$ or $\tilde{A}$, we see that
\begin{equation*}
\left| \frac{df}{ds} \right| \leq \Omega^2 \cdot P_f ( |\phi|, |Q|, |\dot{\phi}|, r^{-2})
\end{equation*}
for some polynomial $P_f$ of degree less than $2$.
Then, using that (\ref{eq:jo_bootstrap_lapse}) provides a large power of $\epsilon$, the remaining bootstraps plus the fact that $r(s) \gtrsim \epsilon$ give the estimates (\ref{eq:jo_rd}) and (\ref{eq:jo_gauge}). The other two estimates then follow straightforwardly from Proposition \ref{prop:lateblueshift}, (\ref{eq:jo_rd}) and (\ref{eq:jo_gauge}).

\blue{As $- r \dot{r} \gtrsim \epsilon^2$ and $r^2(s_O) = r_-^2 + O(\epsilon^2 \log(\epsilon^{-1}))$, it is immediate that $s \lesssim \epsilon^{-2}$ for $s \in \mathcal{O}_{boot}$.} 
\end{proof}

We use Lemma \ref{lemma:jo_geometry} to rewrite several of the equations of Section \ref{sub:evol_eqs} as simplified equations with constant coefficients, plus an error term. \blue{Recalling $\xi_0$ from \eqref{a0.def}, to simplify notation, we denote}
\begin{equation} \label{eq:xlambda}
x = x(s) \coloneqq \frac{r^2(s)}{r_-^2 \epsilon^2}, \; \text{ so that } \; \xi_0 = \left | \left(- \frac{dx}{ds} \right)^{-1} \cdot q_0 \tilde{A} \right| (s =  s_O).
\end{equation}
We then proceed to rewrite \blue{(\ref{eq:phi_evol_2})} and (\ref{eq:Q_evol}) in terms of the new variable $x$. Performing this change of variables, one gets the two equations:
\begin{equation} \label{eq:jo_phi_evol}
\frac{d}{dx} \left( x \frac{d\phi}{dx} \right) + \xi_0^2 x \phi = \mathcal{E}_{\phi},
\end{equation}
\begin{equation} \label{eq:jo_q_evol}
\frac{dQ}{dx} = \sgn(\mathbf{e}) q_0 r_-^2 \xi_0 \epsilon^2 x |\phi|^2 + \mathcal{E}_{Q}.
\end{equation}
We should like to estimate the error terms $\mathcal{E}_{\phi}$ and $\mathcal{E}_Q$. To do so, we apply Lemma \ref{lemma:jo_geometry} extensively to prove the following corollary.

\begin{corollary} \label{cor:jo_errors}
\blue{Recall $x$ and $\xi_0$ from} (\ref{eq:xlambda}). Then, \blue{for $s \in \mathcal{O}_{boot}$}, one finds
\begin{equation} \label{eq:jo_lambda}
\left| \xi_0 - \left(- \frac{dx}{ds}\right)^{-1} \cdot |q_0 \tilde{A}| (s) \right| \leq D_O \blue{\epsilon^{25}}.
\end{equation}
Furthermore, in the equations (\ref{eq:jo_phi_evol}) and (\ref{eq:jo_q_evol}), one has the following control on the error terms:
\begin{equation*}
|\mathcal{E}_{\phi}| + |\mathcal{E}_{Q}| \leq D_O \epsilon^{20}.
\end{equation*}
\end{corollary}

\blue{
\begin{proof}
    For \eqref{eq:jo_lambda}, we first use that, from \eqref{eq:r_evol_2},
    \begin{equation} \label{eq:dx2}
        \left| \frac{d^2 x}{ds^2} \right| = \frac{2}{r_-^2 \ep^2} \left| \frac{d}{ds} (- r \dot{r}) \right| = \frac{\Omega^2}{2 r_-^2 \ep^2} \left | 1 - \frac{Q^2}{r^2} - r^2 |\phi|^2 - r^2 \Lambda \right | \lesssim \ep^{25},
    \end{equation}
    where the final estimate comes from the bootstrap assumptions \eqref{eq:jo_bootstrap_lapse}--\eqref{eq:jo_bootstrap_Q}. Integrating this once between $s_O$ and $s \in \mathcal{O}_{boot}$, and combining with \eqref{eq:xlambda}, \eqref{eq:jo_lambda} follows.

    For the error terms, making the substitution $s \mapsto x$ in \eqref{eq:phi_evol_2} and \eqref{eq:Q_evol} yields:
    \[
        |\mathcal{E}_{\phi}| + |\mathcal{E}_Q| \lesssim \left( \Omega^2 + \left| \frac{d^2 x}{ds^2} \right| + \left| \xi_0 - \left(- \frac{dx}{ds}\right)^{-1} |q_0 \tilde{A}| \right| \right) 
        \cdot P (\ep^{-1}, r^{-1}, |\phi|, |\dot{\phi}|),
    \]
    for some polynomial $P$ of degree less than five. Upon inserting \eqref{eq:jo_bootstrap_lapse}, \eqref{eq:dx2} and \eqref{eq:jo_lambda}, the corollary follows.
\end{proof}
}

\subsection{Scalar field oscillations} \label{sub:bessel}

The focus of this subsection will be understanding the \blue{behavior} of the \textit{charged} scalar field $\phi(s)$ via the equation (\ref{eq:jo_phi_evol}). This equation is a Bessel type ODE that we would like to understand as $x$ decreases from $\epsilon^{-2} + O(\log (\epsilon^{-1}))$ to $r^2(s_{PK}) ( r_-^2  \epsilon^2)^{-1} = 4 |B|^2 \mathfrak{W}^2$. 

\begin{proposition} \label{prop:oscillation_phi}
\blue{For $s \in \mathcal{O}_{boot}$,} there exists a constant $D_O>0$, and coefficients $C_J(\epsilon)$ and $C_Y(\epsilon)$ such that\footnote{
For a definition of the Bessel functions $J_{\blue{\nu}}(z)$ and $Y_{\blue{\nu}}(z)$, see Appendix \ref{sec:appendix_A} and Fact \ref{fact:besselrelation}.
}
\begin{equation} \label{eq:jo_phi}
| \phi - C_J(\epsilon) J_0( \xi_0 x ) - C_Y(\epsilon) Y_0 (\xi_0 x) | \leq D_O \epsilon^{10},
\end{equation}
\begin{equation} \label{eq:jo_phi_x}
\left | \frac{d\phi}{dx} + \xi_0 C_J(\epsilon) J_1( \xi_0 x ) + \xi_0 C_Y(\epsilon) Y_1 (\xi_0 x) \right | \leq D_O \epsilon^{10}.
\end{equation}

Defining $x_{\mathcal{B}} = x(s_O)$, the coefficients $C_J(\epsilon)$ and $C_Y(\epsilon)$ are determined by
\begin{equation} \label{eq:jo_linalg}
\begin{bmatrix}
C_J(\epsilon) \\ C_Y(\epsilon)
\end{bmatrix}
=
\frac{\pi x_{\mathcal{B}}}{2}
\begin{bmatrix}
- \xi_0 Y_1(\xi_0 x_{\mathcal{B}}) & - Y_0(\xi_0 x_{\mathcal{B}}) \\
\xi_0 J_1(\xi_0 x_{\mathcal{B}}) & J_0(\xi_0 x_{\mathcal{B}})
\end{bmatrix}
\begin{bmatrix}
\phi (x_{\mathcal{B}}) \\ \tfrac{d \phi}{d x}(x_{\mathcal{B}})
\end{bmatrix}
.
\end{equation}
The coefficients obey the estimates
\begin{equation} \label{eq:bessel_j_coeff}
\left| C_J(\epsilon) - \frac{\sqrt{\pi}}{2} \mathfrak{W}^{-1} \cos (\Theta(\epsilon)) \right| \leq D_O \epsilon^2 \log (\epsilon^{-1}),
\end{equation}
\begin{equation} \label{eq:bessel_y_coeff}
\left| C_Y(\epsilon) - \frac{\sqrt{\pi}}{2} \mathfrak{W}^{-1} \sin (\Theta(\epsilon)) \right| \leq D_O \epsilon^2 \log (\epsilon^{-1}),
\end{equation}
where the phase function $\Theta(\epsilon)$ is given by
\begin{equation} \label{eq:bessel_coeff_phase}
\Theta(\epsilon) = \left . |q_0 \tilde{A}|(s) \cdot r^2(s) \cdot \left( - \frac{d}{ds} r^2(s) \right)^{-1} + \omega_{RN} s + \arg B - \frac{\pi}{4} \right|_{s =  s_O}.
\end{equation}


Finally, one has the following upper bounds for $\phi(s)$ and $\dot{\phi}(s)$:
\begin{equation} \label{eq:jo_phi_upperbound}
\max\{|\phi|, \omega^{-1}_{RN} | \dot{\phi} | \} \leq \frac{100 |B| \epsilon r_-}{r}.
\end{equation}
\end{proposition}

\begin{proof}
Ignoring the error term $\mathcal{E}_{\phi}$, the equation (\ref{eq:jo_phi_evol}) is exactly Bessel's equation of order $0$ after a rescaling $z = \xi_0 x$. We should like to separate the analysis of the linear behavior arising from Bessel's equation and analysis of the error term.

For this purpose, we recast the equation (\ref{eq:jo_phi_evol}) as a first-order linear evolution equation in the variables $(\phi, \frac{d \phi}{dx})$, treating the error term as an inhomogeneity. We get the following (see also the discussion preceding Lemma \ref{lem:bessel_solution}):
\begin{equation} \label{eq:firstorder}
\frac{d}{dx} \begin{bmatrix} \phi(x) \\ \frac{d \phi}{dx} (x) \end{bmatrix}
=
\begin{bmatrix} 0 & 1 \\ -\xi_0^2 & - \frac{1}{x} \end{bmatrix} \begin{bmatrix} \phi(x) \\ \frac{d \phi}{dx}(x) \end{bmatrix} + \begin{bmatrix} 0 \\ \frac{1}{x} \mathcal{E}_{\phi} \end{bmatrix}.
\end{equation}

The analysis of this first-order system then follows from Lemma \ref{lem:bessel_solution}, or more precisely Corollary \ref{cor:bessel_solution_rescale}. 
Following these, we aim to write the solution using Duhamel's principle; let $\mathbf{S}(z_1; z_0)$ be the linear solution operator for the usual Bessel's equation defined in Lemma \ref{lem:bessel_solution}, and $\mathbf{S}_{\xi_0}(x_1; x_0) \coloneqq \mathbf{Q}_{\xi_0} \circ \mathbf{S}(\xi_0 x_1; \xi_0 x_0) \circ \mathbf{Q}_{\xi_0}^{-1}$ as in Corollary \ref{cor:bessel_solution_rescale}, where $\mathbf{Q}_{\chi}: \R^2 \to \R^2$ is the linear stretching operator $\mathbf{Q}_{\chi}: (A, B) \mapsto (A, \chi B)$.

Then, following Corollary \ref{cor:bessel_solution_rescale}, $\mathbf{S}_{\xi_0}(x_1; x_0)$ can be interpreted as the linear solution (semigroup) operator for the first-order system (\ref{eq:firstorder}) without the inhomogeneous term $\mathcal{E}_{\phi}$.
Defining $x_{\mathcal{B}}$ as in the statement of the proposition, we have from Corollary \ref{cor:lateblueshift} that $|x_{\mathcal{B}} - \epsilon^{-2}| \lesssim \log(\epsilon^{-1})$.
Duhamel's principle (i.e.\ Corollary \ref{cor:bessel_solution_rescale}) applied to first-order systems\footnote{
Note that in this section, we are always integrating backwards in $x$. Nonetheless, the standard theory of first order systems will still apply, with the usual convention that $\int_{x_0}^{x_1} = - \int_{x_1}^{x_0}$.
} gives that
\begin{equation} \label{eq:jo_bessel_inhomog}
\begin{bmatrix}
\phi \\ \frac{d \phi}{dx}
\end{bmatrix}
(x) = 
\mathbf{S}_{\xi_0}(x; x_{\mathcal{B}})
\begin{bmatrix}
\phi \\ \frac{d \phi}{dx}
\end{bmatrix}
(x_{\mathcal{B}})
+
\bigintsss_{x_{\mathcal{B}}}^x \mathbf{S}_{\xi_0}(x; \tilde{x})
\begin{bmatrix}
0 \\ \frac{1}{\tilde{x}} \mathcal{E}_{\phi}(\tilde{x})
\end{bmatrix}
\, d \tilde{x} .
\end{equation}

We define:
\begin{equation} \label{eq:separation}
\begin{bmatrix} \phi_B \\ \frac{d \phi_B}{dx} \end{bmatrix} (x) \coloneqq
\mathbf{S}_{\xi_0}(x; x_{\mathcal{B}}) \begin{bmatrix} \phi \\ \frac{d\phi}{dx} \end{bmatrix} (x_{\mathcal{B}}), \hspace{0.5cm}
\begin{bmatrix} \phi_e \\ \frac{d \phi_e}{dx} \end{bmatrix} (x) \coloneqq
\bigintsss_{x_{\mathcal{B}}}^x \mathbf{S}_{\xi_0}(x; \tilde{x})
\begin{bmatrix}
0 \\ \frac{1}{\tilde{x}} \mathcal{E}_{\phi}(\tilde{x})
\end{bmatrix}
\, d \tilde{x} .
\end{equation}
Thus $\phi_B(x)$ represents exactly the solution of the homogeneous linear Bessel ODE with data given at $x = x_{\mathcal{B}}$, while $\phi_e(x)$ is to be treated as an error term due to the small inhomogeneity. To treat $\phi_B(x)$ and $\phi_e(x)$, we shall use parts (\ref{bessel_uno'}) and (\ref{bessel_dos'}) of Corollary \ref{cor:bessel_solution_rescale} respectively.

As $\phi_B(x)$ must be a solution to (a rescaled) Bessel's equation, the solution must be exactly (using Fact \ref{fact:besselrelation} to justify the appearance of $J_1(\xi_0 x)$ and $Y_1(\xi_0 x)$ in the derivatives):
\begin{equation} \label{eq:jo_bessel_exact}
\phi_B(x) = C_J J_0(\xi_0 x) + C_Y Y_0(\xi_0 x),
\end{equation}
\begin{equation} \label{eq:jo_bessel_exact2}
\frac{d \phi_B}{d x}(x) = - \xi_0 C_J J_1(\xi_0 x) - \xi_0 C_Y Y_1(\xi_0 x).
\end{equation} 
In these equations, the coefficients $C_J = C_J(\epsilon)$ and $C_Y = C_Y(\epsilon)$ \blue{are chosen such that} $\phi(x_{\mathcal{B}})=\phi_B(x_{\mathcal{B}})$ and $\frac{d\phi}{dx}(x_{\mathcal{B}})=\frac{d\phi_B}{dx}(x_{\mathcal{B}})$, \blue{i.e.~so that $\phi_B(x)$ agrees with $\phi(x)$ to first order at $x = x_{\mathcal{B}}$}. (It is immediate from \blue{\eqref{eq:separation}} that $\phi_e(x)$ vanishes to first order at $x = x_{\mathcal{B}}$.) Part (\ref{bessel_uno'}) of Corollary \ref{cor:bessel_solution_rescale} then implies (\ref{eq:jo_linalg}).


In order to deal with $\phi_e(x)$, we recall part \ref{bessel_dos'} of Corollary \ref{cor:bessel_solution_rescale}. \blue{Using Corollary~\ref{cor:oscillation}, in particular the upper and lower bounds for $\xi_0$ in \eqref{eq:xi0_upperlower}, the operator norm of the rescaled solution operator satisfies}
$$\| \mathbf{S}_{\xi_0}(x_1; x_0) \|_{l^2 \to l^2} \blue{ \leq \max \left \{ \frac{1}{4 |B|^2 \mathfrak{W}^2} , 16 |B|^2 \mathfrak{W}^2 \right \} \cdot } \max \left \{ 1, \frac{x_0}{x_1} \right \}.$$
\blue{Integrating} backwards in the variable $x$, i.e.\ $x < \bar{x}$ in (\ref{eq:jo_bessel_inhomog}), \blue{we can see that upon}
defining the $l^{\infty}$ norm on $\R^2$ as $\|(x,y)\|_{l^{\infty}}= \max\{|x|,|y|\}$,
\begin{equation*}
\left \| \mathbf{S}_{\xi_0} (x, \tilde{x}) \begin{bmatrix} 0 \\ \frac{1}{\tilde{x}} \mathcal{E}_{\phi}(\tilde{x}) \end{bmatrix} \right \|_{l^{\infty}} 
\lesssim \max \blue{ \left \{ \frac{1}{|B|^2 \mathfrak{W}^2}, |B|^2 \mathfrak{W}^2 \right \} } \cdot \frac{1}{x} \sup |\mathcal{E}_{\phi}| \lesssim \epsilon^{20},
\end{equation*}
where we apply Corollary \ref{cor:jo_errors} and $x \geq 4 |B|^2 \mathfrak{W}^2$ in the last step. Furthermore, the length of the interval of integration in (\ref{eq:jo_bessel_inhomog}) is bounded by $x_{\mathcal{B}} - x \leq x_{\mathcal{B}} \lesssim 2 \epsilon^{-2}$, so we have
\begin{equation} \label{eq:jo_bessel_error}
\left \| \begin{bmatrix} \phi_e(x) \\ \frac{d \phi_e}{dx} (x) \end{bmatrix} \right \|_{l^{\infty}} = \left \| 
\bigintsss_{x_{\mathcal{B}}}^x \mathbf{S}_{\xi_0}(x; \tilde{x})
\begin{bmatrix}
0 \\ \frac{1}{\tilde{x}} \mathcal{E}_{\phi}(\tilde{x})
\end{bmatrix}
\, dx' \right \|_{l^{\infty}} \lesssim \epsilon^{10}
\end{equation}
as required. Combining (\ref{eq:jo_bessel_exact}), (\ref{eq:jo_bessel_exact2}) and (\ref{eq:jo_bessel_error}) then gives both (\ref{eq:jo_phi}) and (\ref{eq:jo_phi_x}).

\blue{We now recover} the precise form for the coefficients $C_J(\epsilon), C_Y(\epsilon)$. We carefully compute $C_J(\epsilon)$; the other coefficient $C_Y(\epsilon)$ shall follows in an analogous manner. By (\ref{eq:jo_linalg}), we have the formula
\begin{equation} \label{eq:jo_cj}
C_J(\epsilon) = \frac{\pi x_{\mathcal{B}}}{2}
\left[ - \xi_0 Y_1 (\xi_0 x_{\mathcal{B}}) \phi (x_{\mathcal{B}}) - Y_0 (\xi_0 x_{\mathcal{B}}) \frac{d \phi}{d x} (x_{\mathcal{B}}) \right].
\end{equation}

We therefore obtain a precise expression for $C_J(\epsilon)$ using the large-$z$ asymptotics for the Bessel function $Y_{\nu}(z)$, and the scalar field estimates at the comparison point $s =  s_O$ given by (\ref{eq:lb_phidiff}) in Proposition \ref{prop:lateblueshift}. By Fact \ref{fact:bessel_bigasymp} and $x_{\mathcal{B}} = \epsilon^{-2} + O(\log(\epsilon^{-1}))$,
\begin{equation} \label{eq:jo_coeff_1}
Y_0 (\xi_0 x_{\mathcal{B}}) = \sqrt{ \frac{2}{\pi \xi_0 x_{\mathcal{B}}}} \sin \left( \xi_0 x_{\mathcal{B}} -\frac{\pi}{4} \right) + O(\epsilon^3),
\end{equation}
\begin{equation} \label{eq:jo_coeff_2}
Y_1 (\xi_0 x_{\mathcal{B}}) = - \sqrt{ \frac{2}{\pi \xi_0 x_{\mathcal{B}}}} \cos \left( \xi_0 x_{\mathcal{B}} -\frac{\pi}{4} \right) + O(\epsilon^3).
\end{equation}

On the other hand, writing $B$ as $B = |B| (\cos \arg B + i \sin \arg B)$ in (\ref{eq:lb_phidiff}) gives  for $\phi(x_{\mathcal{B}}) = \phi(s)|_{s = \blue{s_O}}$:
\begin{equation} \label{eq:jo_coeff_3}
\phi(x_{\mathcal{B}}) = 2 |B| \epsilon \cos ( \omega_{RN} \blue{s_O} + \arg B)| + O(\epsilon^3 \log (\epsilon^{-1})).
\end{equation}
For $\frac{d \phi}{dx}$, we proceed in several steps, using in particular Lemma \ref{lemma:jo_geometry}, \eqref{eq:jo_lambda}, \eqref{eq:xlambda} and \eqref{a0.est}: \begin{align}
    \frac{d \phi}{d x} (x_{\mathcal{B}})		
    &= \left. \left( \frac{dx}{ds} \right)^{-1} \dot{\phi} \, \right|_{s = s_O} \nonumber \\[0.5em]		
    &= - 2 |B| \epsilon \, \omega_{RN} \left. \left( \frac{dx}{ds} \right)^{-1} \sin(\omega_{RN} s + \arg B) \right|_{s = s_O} + O(\epsilon^3 \log(\epsilon^{-1})) \nonumber \\[0.5em]		
    &= + 2 |B| \epsilon \, \omega_{RN} (8 |B|^2 \mathfrak{W}^2 \omega_{RN})^{-1} \sin (\omega_{RN} \blue{s_O} + \arg B)|  + O(\epsilon^3 \log(\epsilon^{-1})) \nonumber \\[0.5em]		
    &= 2 |B| \epsilon \, \xi_0 \sin (\omega_{RN} \blue{s_O} + \arg B )| + O(\epsilon^3 \log(\epsilon^{-1})). \label{eq:jo_coeff_4}	
\end{align}

Substituting all of (\ref{eq:jo_coeff_1}), (\ref{eq:jo_coeff_2}), (\ref{eq:jo_coeff_3}), (\ref{eq:jo_coeff_4}) into (\ref{eq:jo_cj}), one deduces that
\begin{equation}
C_J(\epsilon) = \frac{\pi x_{\mathcal{B}}}{2} \sqrt{ \frac{2 \xi_0}{\pi x_{\mathcal{B}}}} \cdot 2 |B| \epsilon \cos \left( \xi_0 x_{\mathcal{B}} - \frac{\pi}{4} + \omega_{RN} s_O + \arg B \right) + O(\epsilon^2 \log(\epsilon^{-1})).
\end{equation}
To conclude, one uses Corollary \ref{cor:lateblueshift} and $x_{\mathcal{B}} = \epsilon^{-2} + O(\log (\epsilon^{-1}))$ to yield, since $\sqrt{8 \xi_0} |B| = \mathfrak{W} + O(\ep \log(\ep^{-1})\blue{)}$:
\begin{equation}
C_J(\epsilon) = 
\frac{\sqrt{\pi}}{2} \mathfrak{W}^{-1}
\cos (\Theta(\epsilon)) + O(\epsilon^2 \log(\epsilon^{-1})),
\end{equation}
where by the definition of $\xi_0$ in Corollary \ref{cor:lateblueshift}, $\Theta(\epsilon)$ is as in (\ref{eq:bessel_coeff_phase}).
%

Finally, to obtain the upper bound (\ref{eq:jo_phi_upperbound}), note first of all that $\xi_0 x \geq \xi_0 4 |B|^2 \mathfrak{W}^2 \geq 1/4$ by Corollary \ref{cor:lateblueshift}. Since one can check that, for $z \geq \frac{1}{4}$,
\begin{equation*}
\max \{ |J_0(z)|, |J_1(z)|, |Y_0(z)|, |Y_1(z)| \} \leq 10 z^{-1/2},
\end{equation*}
one can use (\ref{eq:jo_phi}) and (\ref{eq:jo_phi_x}) along with (\ref{eq:bessel_j_coeff}) and (\ref{eq:bessel_y_coeff}) to deduce that
\begin{equation}
\max \left \{ |\phi(x)|, \xi_0^{-1} \left| \frac{d \phi}{dx} \right| \right\} \leq \sqrt{\pi} \,  \mathfrak{W}^{-1} \cdot 10 (\xi_0 x)^{-1/2} + O(\epsilon^2 \log(\epsilon^{-1})) \leq \frac{60 |B| \epsilon r_-}{r}.
\end{equation}
Then, one can use Lemma \ref{lemma:jo_geometry} to change the $x$-derivative to an $s$-derivative, with negligible error, and thus find (\ref{eq:jo_phi_upperbound}). 
\end{proof}

\subsection{Precise estimates for \texorpdfstring{$Q$}{Q} and \texorpdfstring{$\Omega^2$}{Ω{2}}} \label{sub:osc_qomega}

Note that while \eqref{eq:jo_bootstrap_phi} has already been improved by \eqref{eq:jo_phi}, it remains to estimate $Q$ and $\Omega^2$ in the region \red{$\mathcal{O}_{boot}$}, and close the remaining bootstraps (\ref{eq:jo_bootstrap_Q}) and (\ref{eq:jo_bootstrap_lapse}). We shall show that $Q$ changes to (a value close to) the quantity $Q_{\infty}$, while $\Omega^2$ remains exponentially decaying in $s$.

\begin{proposition} \label{prop:jo_qomega}
\blue{For $s \in \mathcal{O}_{boot}$}, we have the following estimates for $Q(s)$ and $\Omega^2(s)$, \blue{where $D_O > 0$ may be larger than that of Lemma~\ref{lemma:jo_geometry} and Proposition~\ref{prop:oscillation_phi} if necessary.}
\begin{equation} \label{eq:jo_q_precise}
\left| Q(s) - \mathbf{e} + \frac{1}{4} \frac{2 K_-}{\tilde{A}_{RN, \infty}}(r_-^2 - r^2(s)) \right| \leq D_O \epsilon^2 \log(\epsilon^{-1}) ,
\end{equation}
\begin{equation} \label{eq:jo_lapse}
\left| \frac{d}{ds} \log \Omega^2 - 2K_- \right| \leq D_O \epsilon^2 \left( \log (\epsilon^{-1}) + r^{-2}(s) \right).
\end{equation}
\end{proposition}

\begin{proof}
We shall prove (\ref{eq:jo_q_precise}) using  (\ref{eq:jo_q_evol}) and Proposition \ref{prop:oscillation_phi}. \blue{By Corollary~\ref{cor:lateblueshift}, particularly the estimate $|Q_0 - \mathbf{e}| \lesssim \ep^2 \log(\ep^{-1})$ from \eqref{a0.est},} it remains to estimate the integral:
\begin{equation} \label{eq:jo_integral_estimate}
\sgn(\mathbf{e}) \int_{x_{\mathcal{B}}}^x q_0 r_-^2 \xi_0 \epsilon^2 \tilde{x} |\phi|^2  (\tilde{x}) \, d\tilde{x}, 
\end{equation}
\blue{where we recall for convenience that $x = (\frac{r}{r_-})^2 \epsilon^{-2}$.}

Firstly, using (\ref{eq:jo_phi}), (\ref{eq:bessel_j_coeff}), (\ref{eq:bessel_y_coeff}), and the Bessel function asymptotics in (\ref{eq:jo_besselj}), (\ref{eq:jo_bessely}), one may find that
\begin{equation}
\left| \phi - \sqrt{\frac{1}{2 \xi_0 x}} \mathfrak{W}^{-1} \cos (\xi_0 x - \pi/4 - \Theta(\epsilon)) \right| \lesssim \epsilon^2 \log(\epsilon^{-1}) x^{-1/2} + x^{-3/2},
\end{equation}
which we may use to deduce the estimate 
\begin{equation}
\left| \xi_0 x \phi^2 - \frac{\mathfrak{W}^{-2}}{2} \cos^2 (\xi_0 x - \pi/4 - \Theta(\epsilon)) \right| \lesssim \epsilon^2 \log(\epsilon^{-1}) + x^{-1}.
\end{equation}

Therefore, recalling that $\mathfrak{W}^{-2} = \frac{|2 K_-|}{\omega_{RN}} = \frac{|2 K_-|}{q_0 |\tilde{A}_{RN,\infty}|}$ by definition, we can integrate this to estimate (\ref{eq:jo_integral_estimate}) by
\begin{multline}
\left| \int_{x_{\mathcal{B}}}^x q_0 r_-^2 \xi_0 \epsilon^2 \tilde{x} |\phi|^2 (\tilde{x}) \, d\tilde{x} - \frac{2 |K_-| r_-^2}{2 |\tilde{A}_{RN, \infty}|} \epsilon^2 
\int_{x_\mathcal{B}}^x \cos^2 (\xi_0 \tilde{x} - \pi/4 - \Theta(\epsilon)) \, d\tilde{x} \right| \\ \lesssim 
\epsilon^2 ( \epsilon^2 \log(\epsilon^{-1}) x_{\mathcal{B}} + \log(x_{\blue{\mathcal{B}}} / x) ) \lesssim \epsilon^2 \log(\epsilon^{-1}).
\end{multline}
Then integrating $\cos^2(\xi_0 x - \pi/4 - \Theta(\epsilon)) = \frac{1}{2}(1 + \cos(2\xi_0 x - \pi/2 - 2 \Theta(\epsilon))$ in the usual manner,
\begin{equation}
\left| \int_{x_{\mathcal{B}}}^x q_0 r_-^2 \xi_0 \epsilon^2 \tilde{x} |\phi|^2 (\tilde{x}) \, d\tilde{x}
- \frac{1}{4} \frac{2 |K_-|}{|\tilde{A}_{RN, \infty}|} (r^2(x) - r^2(x_{\mathcal{B}}) ) \right| \lesssim \epsilon^2 \log(\epsilon^{-1}).
\end{equation}
By replacing $r^2(x_{\mathcal{B}})$ by $r_-^2$ (which carries an error of $O(\epsilon^2 \log (\epsilon^{-1}))$), we get the estimate (\ref{eq:jo_q_precise}) -- to ensure we have the right sign in (\ref{eq:jo_q_precise}), recall that $\mathbf{e}$ and $\tilde{A}_{RN, \infty}$ have opposite signs, while $2K_-$ is negative.

We now move onto the estimates for $\Omega^2$. \blue{For this, we} consider (\ref{eq:omega_evol_2}), and integrate by parts using \blue{(\ref{eq:phi_evol})} as follows, where we use (\ref{eq:jo_bootstrap_lapse}) to control error terms involving $\Omega^2$:
\begin{align}
\left| \frac{d}{ds} \log( r \Omega^2) - 2K_- \right| (s)
&\lesssim \left| \frac{d}{ds} \delta \log (r\Omega^2) \right| ( s_O) + \left | \int_{ s_O}^s (- |\dot{\phi}|^2 + q_0^2 |\tilde{A}|^2 |\phi|^2 )\, ds \right | + \epsilon^{30} \nonumber \\[0.5em]
&\lesssim \epsilon^2 \log(\epsilon^{-1}) + |\phi \dot{\phi} (s)| + |\phi \dot{\phi} ( s_O)| + \left| \int_{ s_O}^s \phi (\ddot{\phi} + q_0^2 |\tilde{A}|^2 \phi) \, ds\right| \nonumber \\[0.5em]
&\lesssim \epsilon^2 \log(\epsilon^{-1}) + |\phi \dot{\phi} (s)| + \left| \int_{ s_O}^s \frac{- \dot{r} \phi \dot{\phi}}{r} \, ds\right|. \label{eq:jo_lapse_final}
\end{align}
\blue{In the third line, we used Corollary~\ref{cor:lateblueshift} to absorb $|\phi \dot{\phi} (s_{\mathcal{O}})|$ into the $\ep^2 \log(\ep^{-1})$ error.}

By (\ref{eq:jo_phi_upperbound}) in Proposition \ref{prop:oscillation_phi}, \blue{for $s \in \mathcal{O}_{boot}$, one has}
\begin{equation*}
| \phi \dot{\phi} (s) | \leq \frac{10^4 |B|^2 \omega_{RN} \epsilon^2 r_-^2}{r^2}.
\end{equation*}
Inserting this estimate into all instances of $\phi \dot{\phi}$ in (\ref{eq:jo_lapse_final}) and then evaluating the integral, one eventually arrives at the estimate (\ref{eq:jo_lapse}), \blue{after} applying (\ref{eq:jo_r}) again. Indeed, one has
\begin{equation*}
\left| \frac{d}{ds} \log \Omega^2 - 2 K_- \right| \leq \left| \frac{d}{ds} \log(r \Omega^2) - 2 K_- \right| + \frac{ - \dot{r}}{r} \lesssim \epsilon^2 \left[ \log (\epsilon^{-1}) + r^{-2}(s) \right].
\end{equation*}
This completes the proof of the proposition.
\end{proof}

\begin{proof}[Proof of Proposition~\ref{prop:oscillation} and Corollary~\ref{cor:oscillation}]
Suppose the bootstrap assumptions (\ref{eq:jo_bootstrap_lapse}), (\ref{eq:jo_bootstrap_Q}), (\ref{eq:jo_bootstrap_phi}) hold in $\blue{\mathcal{O}_{boot}} \subset \mathcal{O}$. Then the conclusions of Lemma \ref{lemma:jo_geometry} and Propositions \ref{prop:oscillation_phi} and \ref{prop:jo_qomega} hold. In particular, (\ref{eq:jo_phi_upperbound}) and (\ref{eq:jo_q_precise}) show that (\ref{eq:jo_bootstrap_phi}) and (\ref{eq:jo_bootstrap_Q}) are indeed improved in the bootstrap region. \red{Since, by Lemma~\ref{lemma:jo_geometry}, we know that $s \lesssim \ep^{-2}$ for all $s \in \mathcal{O}_{boot}$, in order to show $\mathcal{O} = \mathcal{O}_{boot}$} it remains only to improve (\ref{eq:jo_bootstrap_lapse}).

For this purpose, we simply need to integrate (\ref{eq:jo_lapse}). \blue{ By (\ref{eq:jo_r}), we know that $- \dot{r} r \gtrsim \ep^2$}, hence due to (\ref{eq:lb_lapse_2}) \blue{and $s \geq s_{O} \gtrsim \log(\ep^{-1})$, one has}
\begin{align} 
| \log \Omega^2 (s) - 2 K_- s - \log \alpha_- | 
&\lesssim \epsilon^2 \log(\epsilon^{-1})^2 + \int^s_{ s_O} \left( \epsilon^2 \log (\epsilon^{-1}) - \frac{\dot{r}}{r} \right) \nonumber \\[0.5em] 
&\lesssim \epsilon^2 \log(\epsilon^{-1}) s + \log \frac{r( s_O)}{r(s)}.
\label{eq:jo_lapse_final2}
\end{align}
In fact, the final term on the right hand side is also bounded by $\epsilon^2 \log(\epsilon^{-1}) s$. To see this, note that if $s$ is such that $r(s) \geq \frac{r_-}{2}$, then (\ref{eq:jo_r}) implies that
\begin{equation*}
\log \frac{r( s_O)}{r(s)} = \int_{ s_O}^s \frac{-\dot{r}(\tilde{s})}{r(\tilde{s})} \, d \tilde{s} \lesssim \int_{ s_O}^s \frac{\epsilon^2}{r(\tilde{s})^2} \, d\tilde{s} \lesssim \epsilon^2 s.
\end{equation*}
On the other hand, if $r(s) \blue{ \leq } \frac{r_-}{2}$ then (\ref{eq:jo_r}) implies that $s \gtrsim \epsilon^{-2}$, \blue{and thus} $\log (\frac{r( s_O)}{r(s)}) \lesssim \log (\epsilon^{-1}) \lesssim \epsilon^2 \log(\epsilon^{-1}) s$. 

Hence, \blue{(\ref{eq:jo_lapse_final2})} implies that
\begin{equation}
\Omega^2 \leq \alpha_- e^{(2 K_- \blue{+} D_O \epsilon^2 \log(\epsilon^{-1}))s},
\end{equation}
which clearly implies (\ref{eq:osc_lapse}) for $\epsilon$ small. Moreover, since $e^{K_- \cdot  s_O(\epsilon)} \lesssim \epsilon^{50}$, we have improved the final bootstrap (\ref{eq:jo_bootstrap_lapse}). \blue{Therefore, $\mathcal{O}_{boot} = \mathcal{O}$ and  its right endpoint is $s =s_{PK}$, where} $r(s_{PK}) = 2 |B| \mathfrak{W} \epsilon r_-$. 

\blue{Integrating (\ref{eq:jo_r}) in the region $\mathcal{O}$, one gets:
	\begin{equation*}
	\left| \frac{1}{2} r^2 (s_O) - \frac{1}{2} r^2(s_{PK}) - 4 |B|^2 \mathfrak{W}^2 \omega_{RN} r_-^2 \epsilon^2 (s_{PK} - s_O) \right| \lesssim \epsilon^2 \log(\epsilon^{-1}),
	\end{equation*}
	which simplifies after applying the estimates of Proposition \ref{prop:lateblueshift} to 
	\begin{equation*}
	\left| r_-^2 - 8 |B|^2 \mathfrak{W}^2 \omega_{RN} r_-^2 \epsilon^2  s_{PK}\right| \lesssim \epsilon^2 \log( \epsilon^{-1} ).
	\end{equation*}
 This yields the estimate \eqref{eq:s0} for $s_{PK}$. }The remaining parts of \blue{both Proposition~\ref{prop:oscillation} and Corollary~\ref{cor:oscillation}} are straightforward.
\end{proof}



Remarkably, (\ref{eq:q_precise_end}) shows that the spacetime exhibits a nonzero, yet controlled \textit{discharge} within the oscillatory region $\mathcal{O}$. We conclude this section with a (purely algebraic) lemma revealing that the final charge $Q_{\infty}(M, \mathbf{e}, \Lambda)$ lies strictly between $\mathbf{e}$ and $\mathbf{e}/2$.

\begin{lemma} \label{lem:jo_charge_retention}
Define $Q_{\infty}$ as in \eqref{eq:qinfty}. 
Then one has the following alternative form for $Q_{\infty}$, \blue{where we recall $\mathcal{P}_{se} = \mathcal{P}_{se}^{\Lambda < 0} \cup \mathcal{P}_{se}^{\Lambda = 0} \cup \mathcal{P}_{se}^{\Lambda > 0}$ from Section~\ref{sub:reissner_nordstrom}.}
\begin{equation} \label{eq:qinfty2}
Q_{\infty} = \frac{3}{4} \mathbf{e} + \frac{\Lambda r_-^2 r_+ (2 r_- + r_+)}{ 12 \mathbf{e}}.
\end{equation}
From this, we make the following observations regarding $Q_{\infty}$:
\begin{enumerate}[(i)]
\item \label{case:lambda0}
If $\Lambda = 0$, then $Q_{\infty} = \frac{3}{4} \mathbf{e}$.
\item \label{case:lambda<}
If $\Lambda<0$, then $Q_{\infty}$ lies strictly between $\frac{1}{2} \mathbf{e}$ and $\frac{3}{4} \mathbf{e}$. Furthermore, as $(M, \mathbf{e}, \Lambda)$ varies across the sub-extremal parameter space $\mathcal{P}_{se}^{\Lambda < 0}$, $Q_{\infty} / \mathbf{e}$ achieves all values in $(\frac{1}{2}, \frac{3}{4})$.
\item \label{case:lambda>}
If $\Lambda > 0$, then $Q_{\infty}$ lies strictly between $\frac{3}{4} \mathbf{e}$ and $\mathbf{e}$. Furthermore, as $(M, \mathbf{e}, \Lambda)$ varies across the sub-extremal parameter space $\mathcal{P}_{se}^{\Lambda > 0}$, $Q_{\infty} / \mathbf{e}$ achieves all values in $(\frac{3}{4}, 1)$.
\end{enumerate}
In particular, in all cases, one has $|Q_{\infty}| > \frac{1}{2}|\mathbf{e}| > 0$.
\end{lemma}

\begin{proof}

To get (\ref{eq:qinfty2}), we need to find a clean expression for $2 K_-$ in terms of the black hole parameters. For this purpose, recall the polynomial (\ref{eq:rn_polynomial}), and define the function $f(X)$ as
\begin{equation*}
f(X) = X^{-2} P_{M, \mathbf{e}, \Lambda}(X) = X^{-2} (r_+ - X) (r_- - X) \left( \tfrac{\mathbf{e}^2}{r_+ r_-} - \tfrac{\Lambda}{3} (r_+ + r_-) X - \tfrac{\Lambda}{3} X^2 \right),
\end{equation*}
from which we may alternatively define the surface gravity $2K_-$ as:
\begin{equation*}
2 K_-(M, \mathbf{e}, \Lambda) = f'(X) |_{X = r_-} =r_-^{-2} (r_- - r_+) \left( \tfrac{\mathbf{e}^2}{r_+r_-} - \tfrac{\Lambda}{3} r_-(r_+ + r_-) - \tfrac{\Lambda}{3} r_-^2 \right).
\end{equation*}
Recalling once more that $\tilde{A}_{RN, \infty} (M, \mathbf{e}, \Lambda) = \frac{\mathbf{e}}{r_+} - \frac{\mathbf{e}}{r_-} = \frac{\mathbf{e}(r_- - r_+)}{r_+ r_-}$, we therefore get (\ref{eq:qinfty2}).

Once we have (\ref{eq:qinfty2}), the case (\ref{case:lambda0}) is immediate. \blue{For the} first statement of (\ref{case:lambda>}), \blue{one may assume} without loss of generality that $\mathbf{e} > 0$. \blue{Then, for $\Lambda > 0$,} (\ref{eq:qinfty2}) gives $Q_{\infty} > \frac{3}{4} \mathbf{e}$, while (\ref{eq:qinfty}) gives $Q_{\infty} < \mathbf{e}$, \blue{since} $\tilde{A}_{RN,\infty} < 0$ \blue{and $2K_- < 0$}. For the $\Lambda < 0$ case, we require one final observation, namely
\begin{equation*}
0= r_+^{-1} P_{M, \mathbf{e}, \Lambda} (r_+) - r_-^{-1} P_{M, \mathbf{e}, \Lambda} (r_-) = \mathbf{e}^2 \left( \frac{1}{r_+} -\frac{1}{ r_-}\right) + (r_+ - r_-) - \frac{\Lambda}{3} (r_+^3 - r_-^3).
\end{equation*}
Dividing by $r_+ - r_-$ and multiplying by $r_+ r_-$, this gives
\begin{equation*}
\tfrac{\Lambda}{3}r_+ r_- (r_+^2 + r_+ r_- + r_-^2) = - \mathbf{e}^2 + r_+ r_-.
\end{equation*}

In particular, in the case $\Lambda < 0$, we then have
\begin{equation*}
0 < - \tfrac{\Lambda}{3} r_-^2 r_+ (2 r_- + r_+) < - \tfrac{\Lambda}{3} r_+ r_- (r_+^2 + r_+ r_- + r_-^2) < \mathbf{e}^2.
\end{equation*}
Substituting this into (\ref{eq:qinfty2}), we get the first statement of (\ref{case:lambda<}).

To prove the second statements of (\ref{case:lambda<}), (\ref{case:lambda>}), we first state without proof several facts about the sub-extremal parameter spaces. Without loss of generality, we fix $\mathbf{e}$ to be positive. Then letting $\mathcal{P}^+_{se} = \blue{\mathcal{P}_{se}} \cap \{ \mathbf{e} > 0 \}$ be the space of subextremal parameters with positive $\mathbf{e}$,
\begin{itemize}
\item
$\mathcal{P}^+_{se}$ is an open, connected set in $\R^3$.
\item
Define the set $\mathcal{P}_{ex}^{\Lambda < 0}$ to be the set of $(M, \mathbf{e}, \Lambda) \in \R^+ \times \R \times (- \infty, 0)$ such that the polynomial $P_{M, \mathbf{e}, \Lambda}(X)$ has a single repeated positive root $X = R$, and $\mathcal{P}_{ex}^{\Lambda > 0}$ to be the set of $(M, \mathbf{e}, \Lambda) \in \R^+ \times \R \times (0, + \infty)$ such that the polynomial $P_{M, \mathbf{e}, \Lambda}(X)$ has two positive roots $X = R$ and $X = R_C$, with $X = R$ having multiplicity $2$ and $R < R_C$. Then, $\mathcal{P}^+_{ex} \coloneqq \blue{(\mathcal{P}_{ex}^{\Lambda < 0} \cup \mathcal{P}_{ex}^{\Lambda > 0})} \cap \{ \mathbf{e} > 0 \}$ is a subset of the boundary $\partial \mathcal{P}^+_{se}$.
\item
For $(M, \mathbf{e}, \Lambda) \in \mathcal{P}_{se}^+$, one must have $ - \infty < \Lambda M^2 < \frac{2}{9}$, i.e.\ we are constrained to have $\Lambda M^2 < \frac{2}{9}$ in order for there to exist a choice of $\mathbf{e} > 0$ such that the polynomial $P_{M, \mathbf{e}, \Lambda}(X)$ has the correct number of distinct positive roots. Moreover, all these values are achieved when restricting to the extremal case i.e.\ $\{ \Lambda M^2: (M, \mathbf{e}, \Lambda) \in \mathcal{P}_{ex}^+ \} = (- \infty, \frac{2}{9})$.
\item
The functions $r_+(M, \mathbf{e}, \Lambda)$ and $r_-(M, \mathbf{e}, \Lambda)$ are continuous in $\mathcal{P}_{se}^+$. Furthermore, they \blue{may be continuously extended onto} $\mathcal{P}_{ex}^+ \subset \partial \mathcal{P}_{se}^+$, \blue{upon defining}  $r_- = r_+ = R$ on \blue{$\mathcal{P}_{ex}^+$}.
\end{itemize}

In light of these facts, $Q_{\infty}$ is a continuous function of $(M, \mathbf{e}, \Lambda)$ in $\mathcal{P}_{se}^+$, and \blue{may be continuously extended onto} the space of extremal parameters $\mathcal{P}_{ex}^+$, where we have
\begin{equation} \label{eq:qinfty3}
Q_{\infty} = \frac{3}{4} \mathbf{e} + \frac{\Lambda R^4}{4 \mathbf{e}}.
\end{equation}
We would like to estimate $\Lambda R^4$. For extremal parameters, we must have $P_{M, \mathbf{e}, \Lambda} (R) = \frac{d}{dX}P_{M, \mathbf{e}, \Lambda} (R) = 0$, from which we can get the two identities\footnote{These will arise from $\frac{d}{dX}P_{M,\mathbf{e},\Lambda}(R) = 0$ and $\frac{d}{dX}(X^{-4} P_{M, \mathbf{e},\Lambda})(R) = 0$ respectively.}:
\begin{equation} \label{eq:rel1}
- 2 \Lambda R^3 + 3 R - 3 M = 0,
\end{equation}
\begin{equation} \label{eq:rel2}
R^2 - 3 M R + 2 \mathbf{e}^2 = 0.
\end{equation}

Using (\ref{eq:rel1}) and then (\ref{eq:rel2}) in (\ref{eq:qinfty3}), we then get
\begin{equation} \label{eq:qinfty4}
Q_{\infty} = \frac{3}{4} \mathbf{e} \left( 1 + \frac{R (R - M)}{2 \mathbf{e}^2} \right) = \frac{3}{4} \mathbf{e} \left( 1 + \frac{R - M}{3 M - R} \right)= \frac{3}{2} \mathbf{e}  \cdot \frac{1}{3 - R/M}.
\end{equation}
It remains to work out the range of values that $R/M$ can take. We rearrange (\ref{eq:rel1}) to get the equation:
\begin{equation}
1 - \frac{R}{M} + \frac{2}{3} \Lambda M^2 \left ( \frac{R}{M} \right )^3 = 0.
\end{equation}
As $\Lambda M^2$ varies between $- \infty$ and $\frac{2}{9}$, we can see that the unique positive root of this cubic expression in $R/M$ varies between $0$ and $\frac{3}{2}$, not including the endpoints. (Note that if $\Lambda M^2 > \frac{2}{9}$, then in fact there are no positive roots! This is why $\Lambda M^2$ must be upper bounded.)

Therefore, as $\Lambda M^2$ and thus $R/M$ vary within their allowed ranges, from (\ref{eq:qinfty4}) we have that $Q_{\infty}$ is allowed to take every value between $\mathbf{e}/2$ and $\mathbf{e}$, not including the endpoints.

Of course, this computation was done for extremal parameters in $\mathcal{P}_{ex}^+$ rather than the subextremal parameters in $\mathcal{P}_{se}^+$. But $\mathcal{P}_{se}^+$ is connected, so it is straightforward to show that indeed $\{ Q_{\infty}(M , \mathbf{e}, \Lambda): (M, \mathbf{e}, \Lambda) \in \mathcal{P}_{se}^+ \} = ( \mathbf{e}/2, \mathbf{e} )$ also. This concludes the proof of Lemma \ref{lem:jo_charge_retention}.
\end{proof}

\section{The Proto-Kasner Region} \label{sec:protokasner}

\subsection{Estimates beyond the oscillatory region -- statement of \texorpdfstring{Proposition~\ref{prop:oscillation+}}{Proposition 6.1}} \label{sub:osc_kasner}

In the previous \blue{section,} we considered the region $\mathcal{O} = \{ s \geq  s_O : r(s) \geq 2 |B| \mathfrak{W} \epsilon r_- \}=\{ s_O \leq s \leq s_{PK} \}$, where  $s_O(\ep)=50s_{lin}$ and we  $s_{PK}=O(\ep^{-2})$ was defined such that $r(s_{PK})=2 |B| \mathfrak{W} \epsilon r_- $. The reason it was necessary to end this region at $r \sim \epsilon r_-$ was twofold:
\begin{enumerate}
\item \label{block1}
Considering the Bessel function form of $\phi$,
\begin{equation*}
\phi = C_J (\epsilon) J_0 \left( \frac{\xi_0 r^2}{r_-^2 \epsilon^2} \right) + C_Y (\epsilon) Y_0 \left( \frac{\xi_0 r^2}{r_-^2 \epsilon^2} \right) + \text{error},
\end{equation*}
then the Bessel functions $J_0(z), Y_0(z)$ will change behavior from oscillatory (at large $z$) to convergent or logarithmically growing  (at small $z$) once $r^2(s) \leq \xi_0^{-1} r_-^2 \epsilon^2 \approx 8 |B|^2 \mathfrak{W}^2 \epsilon^2 r_-^2$. (See (\ref{eq:jo_besselj}), (\ref{eq:jo_bessely}))
\item \label{block2}
We \blue{tackled} many error terms using the \blue{smallness of} $\Omega^2(s)$ to dominate polynomial powers of $r^{-1}$ and $|\phi|$. At the start of the region $\mathcal{O}$, i.e.\ $s = s_O(\ep)= 50 s_{lin}(\ep)$, we only had that $\Omega^2 \lesssim \epsilon^{100}$, so in order for the errors to be controlled we required $r^{-1}$ to be at worst an inverse power of $\epsilon$.
\end{enumerate}

Nevertheless, we will show that we may extend many of the important estimates beyond $s = s_{PK}$ to a region $\mathcal{PK} = \{ s \geq s_{PK}: r(s) \geq e^{- \delta_0 \epsilon^{-2}} r_-\}$, which for reasons that will later become clear we denote \blue{by} the \textit{proto-Kasner} region. Here, the dimensionless constant $\delta_0(M, \mathbf{e}, \Lambda, q_0)$ is selected to be (recall that $b_- = \frac{2 |B| \omega_{RN}}{|2 K_-|}= 2|B| \mathfrak{W}^2$)
\begin{equation} \label{eq:delta0}
\delta_0 \coloneqq \frac{1}{\blue{2000}} |B(M, \mathbf{e}, \Lambda)|^{-2} \, \mathfrak{W}^{-4}(M, \mathbf{e}, \Lambda, q_0) = \frac{1}{\blue{500}} b_-^{-2}.
\end{equation}

In the region $\mathcal{PK}$, we will overcome difficulty \ref{block1} by instead using new bootstrap assumptions that reflect the now monotonic behavior of $\phi$ and $\dot{\phi}$. The more fundamental difficulty \ref{block2} is dealt with by now using the fact that $\Omega^2$ is now such that \red{$- \log \Omega^2 \gtrsim - K_- s \gtrsim \ep^{-2}$} \blue{in regimes such as $\mathcal{PK}$ where} $s$ is of order $\ep^{-2}$. 


\begin{proposition} \label{prop:oscillation+}
Choose $\delta_0(M, \mathbf{e}, \Lambda, q_0) > 0$ as in (\ref{eq:delta0}). Then, in the region $\mathcal{PK} = \{ 2 |B| \mathfrak{W} \epsilon r_- \geq r(s) \geq e^{- \delta_0 \epsilon^{-2}} r_- \}$, there exists some $D_{PK}(M, \mathbf{e}, \Lambda, m^2, q_0) > 0$ such that
that one has the following exponential behavior for the lapse $\Omega^2(s)$:
\begin{equation} \label{eq:pk_lapse}
\Omega^2(s) \leq
D_{PK} \exp( - 50 \delta_0 \epsilon^{-2}).
\end{equation}

Recalling the quantity $Q_{\infty}$ \blue{from} (\ref{eq:qinfty}), we have the following for all $s\in \mathcal{PK}$: 
\begin{equation} \label{eq:pk_r}
\left| - r \dot{r}(s) - 4 |B|^2 \mathfrak{W}^2 \omega_{RN} r_-^2  \ep^2 \right| \leq D_{PK} \epsilon^4 \log(\epsilon^{-1}),
\end{equation}
\begin{equation} \label{eq:pk_gauge}
\left| q_0 \tilde{A}(s) - q_0 \tilde{A}_{RN, \infty} \right| \leq D_{PK} \epsilon^2,
\end{equation}
\begin{equation} \label{eq:pk_Q}
| Q(s) - Q_{\infty} | \leq D_{PK} \epsilon^2 \log(\epsilon^{-1}).
\end{equation}
Using these, \blue{and recalling that $\omega_{RN} = |q_0 \tilde{A}_{RN, \infty}|$, }we define the quantities $\omega_K$ and $\xi_K$ by 
\begin{gather} 
\omega_K \coloneqq | q_0 \tilde{A} | (s_{PK}) = \omega_{RN} + O(\epsilon^2), \label{eq:omegak} \\[0.5em] \xi_K \coloneqq \omega_K \left( - \frac{d}{ds} \frac{r^2}{r_-^2 \epsilon^2} \right)^{-1} (s_{PK}) = \frac{1}{8 |B|^2 \mathfrak{W}^2} + O(\epsilon^2 \log(\epsilon^{-1})). \label{eq:xik}
\end{gather}

For the scalar field $\phi$, there will exist coefficients $C_{JK}(\epsilon)$ and $C_{YK}(\epsilon)$ obeying, for $C_J(\epsilon)$ and $C_Y(\epsilon)$ as in Proposition \ref{prop:oscillation_phi}, the estimates $|C_{JK}(\epsilon) - C_J(\epsilon)| + |C_{YK}(\epsilon) - C_Y(\epsilon)| \leq D_{PK}\epsilon^2 \log(\epsilon^{-1})$, such that for all $s\in \mathcal{PK}$:
\begin{equation} \label{eq:pk_phi_1}
\left | \phi(s) - \left(  C_{JK} (\epsilon) J_0 \left( \frac{\xi_K r^2(s)}{r_-^2 \epsilon ^2} \right) + C_{YK}(\epsilon) Y_0 \left( \frac{ \xi_K r^2(s)} {r_-^2\epsilon^2} \right)\right) \right| \leq D_{PK} \epsilon^2 \log(\epsilon^{-1}),
\end{equation}
\begin{equation} \label{eq:pk_phi_2}
\left | \dot{\phi}(s) - \omega_K \left( C_{JK} (\epsilon) J_1 \left( \frac{\xi_K r^2(s)}{r_-^2\epsilon^2} \right) + C_{YK}(\epsilon) Y_1 \left( \frac{\xi_K r^2(s)}{r_-^2 \epsilon^2} \right) \right) \right |  \leq  D_{PK} \epsilon^2 \log(\epsilon^{-1}).
\end{equation}

Moreover, the following upper bounds for $\phi$ and $\dot{\phi}$ are satisfied for all $s \in \mathcal{PK}$:
\begin{equation} \label{eq:obar_phi}
|\phi(s)| \leq 100 \, \mathfrak{W}^{-1} \left( 1 + \log \left( \frac{8 |B|^2 \mathfrak{W}^2 r_-^2 \epsilon^2}{r^2 }\right)\right),
\end{equation}
\begin{equation} \label{eq:obar_phid}
|\dot{\phi}| \leq 100 \, \omega_{RN} \cdot \frac{8 |B|^2 \mathfrak{W} r_-^2 \epsilon^2}{r^2}.
\end{equation}
\end{proposition}

\begin{rmk}
Both the statement and the proof of Proposition \ref{prop:oscillation+} will show large similarities to Proposition \ref{prop:oscillation} in Section \ref{sec:oscillations}. One notes, however, that, in order to obtain \blue{the necessary error estimates, such as those of Lemma~\ref{lem:pk_geometry},} the quantities $\xi_0$ and $\omega_0$, which are defined as the values of certain quantities evaluated at $s = s_O$, to the quantities $\xi_K$ and $\omega_K$, which are the values of the same quantities but evaluated at $s = s_{PK}$ instead.
\end{rmk}

\subsection{Proof of \texorpdfstring{Proposition~\ref{prop:oscillation+}}{Proposition 6.1}} \label{sub:proof*}

\blue{We often} make reference to the following \blue{four} bootstrap assumptions in the region $\mathcal{PK}$.  
\begin{equation} \label{eq:pk_bootstrap_lapse} \tag{PK1}
\Omega^2 \leq \alpha_- \exp( - 50 \delta_0 \epsilon^{-2}),
\end{equation}
\begin{equation} \label{eq:pk_bootstrap_phi} \tag{PK2}
|\phi| \leq 100 \, \mathfrak{W}^{-1} \left( 1 + \log \left( \frac{8 |B|^2 \mathfrak{W}^2 r_-^2 \epsilon^2}{r^2 }\right)\right),
\end{equation}
\begin{equation} \label{eq:pk_bootstrap_phi_s} \tag{PK3}
|\dot{\phi}| \leq 100 \, \omega_{RN} \cdot \frac{8 |B|^2 \mathfrak{W} r_-^2 \epsilon^2}{r^2}.
\end{equation}
\begin{equation} \label{eq:pk_bootstrap_Q} \tag{PK4}
|Q| \leq 2 |\mathbf{e} |.
\end{equation}

\blue{From Corollary~\ref{cor:oscillation} and Lemma~\ref{lem:jo_charge_retention}, it follows that} that the bootstraps (\ref{eq:pk_bootstrap_phi}), (\ref{eq:pk_bootstrap_phi_s}), (\ref{eq:pk_bootstrap_Q}) hold in a neighborhood of $s=s_{PK}$.
\blue{On the other hand, \eqref{eq:pk_bootstrap_lapse} holds in a neighborhood of $s = s_{PK}$ due to \eqref{eq:osc_lapse} and \eqref{eq:s0}}: note that, for $\epsilon$ chosen sufficiently small, \blue{using \eqref{eq:delta0} and \eqref{eq:s0}, we have}
\begin{equation} \label{eq:s0lower}
 K_- s_{PK} \leq - \frac{2 |K_-|}{\omega_{RN}} \frac{1}{\blue{32} |B|^2 \mathfrak{W}^2} \blue{\ep^{-2}} \leq -  \frac{1}{\blue{32} |B|^2 \mathfrak{W}^4} \blue{\ep^{-2}} < - 51 \delta_0 \epsilon^{-2}.
\end{equation}

We now proceed in largely the same manner as in the region $\mathcal{O}$. Always \blue{assuming} the bootstraps, we begin as in Section~\ref{sub:osc_prelim} with some preliminary estimates on $- r \dot{r}$ and $\tilde{A}$, then as in \blue{Section}~\ref{sub:bessel} we use these together with the equation (\ref{eq:phi_evol_2}) written in \blue{Bessel-like} form to find precise asymptotics for the scalar field $\phi$. Finally one concludes using these Bessel asymptotics to close estimates for $Q$ and $\Omega^2$ as in \blue{Section}~\ref{sub:osc_qomega}. We shall aim to be rather terse, as modifications from the analysis of the region $\mathcal{O}$ are generally minor.

Before proceeding, we briefly address how to deal with the difficulty \ref{block2} mentioned in Section \ref{sub:osc_kasner}. The idea is that the \blue{new} bootstrap (\ref{eq:pk_bootstrap_lapse}) now means \blue{the smallness of $\Omega^2$ dominates by powers of $r^{-1}$}, since $r^{-1} \lesssim \red{e^{- \delta_0 \ep^{-2}}}$ in the region under consideration. For instance, consider the following expression arising in (\ref{eq:r_evol_2})\blue{; in light of the bootstrap assumptions (\ref{eq:pk_bootstrap_lapse})--(\ref{eq:pk_bootstrap_Q}), one estimates}
\begin{equation*}
\frac{\Omega^2}{4} \left( \frac{Q^2}{r^2} + m^2 r^2|\phi|^2 \right) \lesssim e^{- 48 \delta_0 \epsilon^{-2}}.
\end{equation*}
We now proceed with the proof in a series of lemmas similar to those of Section \ref{sec:oscillations}. As in that section, it is convenient to define  $\mathcal{PK}_{boot}$ as \red{the connected component of the set $\{ s \in \mathcal{PK}: s \leq \epsilon^{-4}, \text{\eqref{eq:pk_bootstrap_lapse}--\eqref{eq:pk_bootstrap_Q} apply}\}$ such that $s_{PK} \in \mathcal{PK}_{boot}$.}
Then we have the following.
\begin{lemma} \label{lem:pk_geometry}
\blue{For $s \in \mathcal{PK}_{boot}$,} one finds
\begin{equation} \label{eq:pk_rd}
\left |\frac{d}{ds} (-r \dot{r})(s) \right | + \left| \red{ - r \dot{r}(s) + r \dot{r}(s_{PK})} \right| \leq D_{PK} e^{-40 \delta_0 \epsilon^{-2}},
\end{equation} 
\begin{equation} \label{eq:pk_gauge_2}
|q_0\tilde{A} (s) - q_0 \tilde{A} ( s_{PK} )| \leq D_{PK} e^{- 40 \delta_0 \epsilon^{-2}}.
\end{equation}
Further, (\ref{eq:pk_r}) and (\ref{eq:pk_gauge}) hold. Moreover, letting $x(s) = \frac{r^2(s)}{r_-^2 \epsilon^2}$, one \blue{finds that}
\begin{equation} \label{eq:pk_phi_evol}
\frac{d}{dx} \left( x \frac{d\phi}{dx} \right) + \xi_K^2 x f = \mathcal{E}_{\phi},
\end{equation}
where the error term $\mathcal{E}_{\phi}$ is bounded by
\begin{equation*}
|\mathcal{E}_{\phi}(s)| \leq D_{PK} e^{- 30 \delta_0 \epsilon^{-2}}.
\end{equation*}
\end{lemma}

\begin{proof}
The proof is almost identical to that of Lemma \ref{lemma:jo_geometry} and \ref{cor:jo_errors}, in light of the comments above. 
\end{proof}

\begin{lemma} \label{lem:pk_phi}
\blue{Recalling $x(s) = \frac{r^2(s)}{r_-^2 \ep^2}$}, let $x_{PK} = x(s_{PK})$, and define the coefficients $C_{JK}(\epsilon)$ and $C_{YK}(\epsilon)$ via 
\begin{equation}
\begin{bmatrix} C_{JK}(\epsilon) \\ C_{YK}(\epsilon) \end{bmatrix} =
\begin{bmatrix} J_0(\xi_K x_{PK}) & Y_0(\xi_K x_{PK}) \\ - \xi_K J_1(\xi_K x_{PK}) & - \xi_K Y_1(\xi_K x_{PK}) \end{bmatrix}^{-1}
\begin{bmatrix} \phi(x_{PK}) \\ \frac{d\phi}{dx}(x_{PK}) \end{bmatrix}. \label{eq:pk_coeffs}
\end{equation}
Then one has that $|C_{JK}(\epsilon) - C_J(\epsilon)| + |C_{YK}(\epsilon) - C_Y(\epsilon)| \leq D_{PK} \epsilon^2 \log(\epsilon^{-1})$, and that\blue{, for $x = x(s)$ corresponding to $s \in \mathcal{PK}_{boot}$,}
\begin{equation} \label{eq:pkl_phi}
| \phi(x) - C_{JK}(\epsilon) J_0( \xi_K x ) - C_{JK}(\epsilon) Y_0 (\xi_K x) | \leq D_{PK} e^{- 20 \delta_0 \epsilon^{-2}},
\end{equation}
\begin{equation} \label{eq:pkl_phi_x}
\left | \frac{d\phi}{dx}(s) + \xi_K C_{JK}(\epsilon) J_1( \xi_K x ) + \xi_K C_{YK}(\epsilon) Y_1 (\xi_K x) \right | \leq D_{PK} e^{- 20 \delta_0 \epsilon^{-2}}.
\end{equation}
\end{lemma}

\begin{proof}
The idea is that \blue{in light of} (\ref{eq:pk_phi_evol}), we can use the equation (\ref{eq:jo_bessel_inhomog}) describing the solution for the scalar field using the solution operator to the linear Bessel equation, just as in Proposition \ref{prop:oscillation_phi}. 
To get \blue{improved} error estimates, however, we evolve the system from $x = x_{PK}$ as opposed to $x = x_{\mathcal{B}}$, hence the \blue{appearance of the updated Bessel coefficients $C_{JK}(\ep)$ and $C_{YK}(\ep)$} in (\ref{eq:pk_coeffs}). 

We first show these are close to the original coefficients as claimed; from Proposition \ref{prop:oscillation_phi}, we have
\begin{equation*}
\begin{bmatrix} \phi(x_{PK}) \\ \frac{d \phi}{dx} (x_{PK}) \end{bmatrix} 
= 
\begin{bmatrix} J_0(\xi_0 x_{PK}) & Y_0(\xi_0 x_{PK}) \\ - \xi_0 J_1(\xi_0 x_{PK}) & - \xi_0 Y_1(\xi_0 x_{PK}) \end{bmatrix} 
\begin{bmatrix} C_J(\epsilon) \\ C_Y(\epsilon) \end{bmatrix}
+ O(\epsilon^{10}).
\end{equation*}
Moreover, using $x_{PK} = 4 |B|^2 \mathfrak{W}^2$, \blue{from \eqref{a0.est} and Corollary~\ref{cor:jo_errors}, we know that $\xi_0 x_{PK}$ and $ \xi_K x_{PK}$ lie within the interval $(\frac{1}{4}, 1)$ and that $|\xi_0 x_{PK} - \xi_K x_{PK}| \lesssim \ep^{25} \lesssim \epsilon^2 \log(\epsilon^{-1})$}. \blue{Thus,} since the derivatives of $J_{\nu}(z)$ and $Y_{\nu}(z)$ are bounded for $1/4 \leq z \leq 1$, and \blue{$|C_J(\ep)|, |C_Y(\ep)| \lesssim 1$}, we can modify the above formula to
\begin{equation} \label{eq:pk_coeff_linalg_0}
\begin{bmatrix} \phi(x_{PK}) \\ \frac{d \phi}{dx} (x_{PK}) \end{bmatrix} 
= 
\begin{bmatrix} J_0(\xi_K x_{PK}) & Y_0(\xi_K x_{PK}) \\ - \xi_K J_1(\xi_K x_{PK}) & - \xi_K Y_1(\xi_K x_{PK}) \end{bmatrix} 
\begin{bmatrix} C_J(\epsilon) \\ C_Y(\epsilon) \end{bmatrix}
+ O(\epsilon^2 \log(\epsilon^{-1})).
\end{equation}

Finally, since by Lemma \ref{lem:bessel_wronskian} the inverse matrix
\begin{equation*}
\begin{bmatrix} J_0(\xi_K x_{PK}) & Y_0(\xi_K x_{PK}) \\ - \xi_K J_1(\xi_K x_{PK}) & - \xi_K Y_1(\xi_K x_{PK}) \end{bmatrix}^{-1}
=
\frac{2}{\pi x_{PK}}
\begin{bmatrix} - \xi_K Y_1(\xi_K x_{PK}) & - Y_0(\xi_K x_{PK}) \\ \xi_K J_1(\xi_K x_{PK}) & J_0(\xi_K x_{PK}) \end{bmatrix}
\end{equation*}
has uniformly bounded entries, combining (\ref{eq:pk_coeff_linalg_0}) with (\ref{eq:pk_coeffs}) shows, as claimed, the relation
\begin{equation*}
\begin{bmatrix} C_{JK}(\epsilon) \\ C_{YK}(\epsilon) \end{bmatrix}
=
\begin{bmatrix} C_J(\epsilon) \\ C_Y(\epsilon) \end{bmatrix}
+ O(\epsilon^2 \log(\epsilon^{-1})).
\end{equation*}

The remainder of the proof is as in Proposition \ref{prop:oscillation_phi}. Adopting the same notation from the proof of this proposition, we record the analogue of (\ref{eq:jo_bessel_inhomog}) again here:
\begin{equation} \label{eq:pk_bessel_inhomog}
\begin{bmatrix}
\phi \\ \frac{d \phi}{dx}
\end{bmatrix}
(x) =
\mathbf{S}_{\xi_K}(x; x_{PK})
\begin{bmatrix}
\phi \\ \frac{d \phi}{dx}
\end{bmatrix}
(x_{PK})
+
\bigintsss_{x_{PK}}^x \mathbf{S}_{\xi_K}(x; \tilde{x})
\begin{bmatrix}
0 \\ \frac{1}{\tilde{x}} \mathcal{E}_{\phi}(\tilde{x})
\end{bmatrix}
\, d\tilde{x} .
\end{equation}
The first term on the right hand side will correspond to a solution of the linear Bessel equation after a rescaling $z = \xi_K x$, just as before. Using once again part (\ref{bessel_uno'}) of Corollary \ref{cor:bessel_solution_rescale} we see this corresponds to the objects on the left hand sides of (\ref{eq:pkl_phi}) and (\ref{eq:pkl_phi_x}).

For the second term on the right hand side of (\ref{eq:pk_bessel_inhomog}), we once again appeal to part (\ref{bessel_dos'}) of Corollary \ref{cor:bessel_solution_rescale}. From this corollary, the \blue{appearance} of the operator $\mathbf{S}_{\xi_K}(x; x_{PK})$ will contribute at worst an additional factor of $x^{-1} \leq \epsilon^2 e^{2 \delta_0 \epsilon^{-2} }$ in our $l^{\infty}$ estimates. Since the length of the integration interval $|x - x_{PK}| \leq x_{PK}$ is uniformly bounded, (\ref{eq:pkl_phi}) and (\ref{eq:pkl_phi_x}) follow straightforwardly.
\end{proof}

\begin{proof}[Completing the proof of Proposition \ref{prop:oscillation+}]
We are now in a position to close all the bootstraps and complete the proof. We first use Lemma \ref{lem:pk_phi} to improve (\ref{eq:pk_bootstrap_phi}) and (\ref{eq:pk_bootstrap_phi_s}). We shall only show the latter; the former follows \blue{in a similar fashion}. 

By Lemma \ref{lem:pk_phi} and the bounds on the coefficients in Proposition \ref{prop:oscillation_phi}, \blue{within the region $\mathcal{PK}_{boot}$,} one has
\begin{equation*}
\left| \frac{d \phi}{dx} \right| \leq 2 \sqrt{\pi} \, \mathfrak{W}^{-1} \, \xi_K\cdot  ( |J_1(\xi_K x)| + |Y_1(\xi_K x)|).
\end{equation*}
In light of Facts \ref{fact:bessel1_taylor} and \ref{fact:bessel2_taylor}, \blue{one may check that} $\left|\frac{d \phi}{dx}\right| \leq 4 \sqrt{\pi} \, \mathfrak{W}^{-1} x^{-1}$. Hence, using $- \frac{dx}{ds} = 8 |B|^2 \mathfrak{W}^2 \omega_{RN} + O(\epsilon^2 \log(\epsilon^{-1}))$ from Lemma \ref{lem:pk_geometry}, we improve (\ref{eq:pk_bootstrap_phi_s}). \blue{As mentioned, improving \eqref{eq:pk_bootstrap_phi} is similar.}

We now move onto $Q$ and $\Omega^2$. For $Q(s)$, it remains to understand only how $Q$ changes in the region $\mathcal{PK}$, in particular that it changes only by $O(\epsilon^2 \log(\epsilon^{-1}))$. Looking at (\ref{eq:Q_evol}), we need to study
\begin{equation*}
\left| \int_{s_{PK}}^s \tilde{A} q_0^2 r^2 |\phi|^2 (s') \, ds' \right|.
\end{equation*}

Now, we use (\ref{eq:jo_gauge_2}) and (\ref{eq:obar_phi}) to estimate
\begin{equation}
\left| \int_{s_{PK}}^s \tilde{A} q_0^2 r^2 |\phi|^2 (s') \, ds' \right| \lesssim \int_{s_{PK}}^s \epsilon^2 \cdot \frac{r^2}{r_-^2 \epsilon^2} \left( 1 + \left| \log (\frac{\epsilon^2}{r^2})\right| \right)^2 (s') \, ds'.
\end{equation}
Substituting $x = \frac{r^2}{r_-^2 \epsilon^2}$ as usual, and noting that $-\frac{dx}{ds} \sim 1$, we therefore see that
\begin{equation}
\left| \int_{s_{PK}}^s \tilde{A} q_0^2 r^2 |\phi|^2 (s') \, ds' \right| \lesssim \int_0^{4|B|^2 \mathfrak{W}^2} \epsilon^2 x ( 1 + |\log (x)| )^2 \, dx \lesssim \epsilon^2.
\end{equation}
Combined with (\ref{eq:q_precise_end}), this will yield (\ref{eq:pk_Q}), and hence improve the bootstrap (\ref{eq:pk_bootstrap_Q}). 

Last of all, for the quantity $\Omega^2$, we proceed using that the Raychaudhuri equation (\ref{eq:raych}) implies $- \Omega^{-2} \dot{r}(s)$ \blue{is nonincreasing}. \blue{Therefore, for $s \in \mathcal{PK}_{boot}$},
\blue{\begin{equation*}
    \Omega^2(s) \leq \frac{\Omega^2(s_{PK})}{- \dot{r}(s_{PK})} \cdot (-\dot{r}(s)) = \Omega^2 (s_{PK}) \cdot \frac{- r \dot{r}(s)}{- r \dot{r} (s_{PK})} \cdot \frac{r(s_{PK})}{r(s)}.
\end{equation*}}
Therefore, one applies (\ref{eq:pk_r}) and (\ref{eq:osc_lapse}), \blue{alongside $r(s_{PK}) = 2 |B| \mathfrak{W} r_- \ep$ and $r(s) \geq e^{- \delta_0 \ep^{-2}}$ for $s \in \mathcal{PK}$,} to find
\begin{equation}
{\Omega^2(s)} \leq (1 + 2 D_O \epsilon^2 \log(\epsilon^{-1})) \cdot \alpha_- e^{\blue{K_-} s_{PK} }\cdot \blue{2 |B| \mathfrak{W} r_- \ep} \cdot e^{\delta_0 \epsilon^{-2}}.
\end{equation}
By (\ref{eq:s0lower}), \blue{this means that $\Omega^2(s) \leq \alpha_- e^{- 50 \delta_0 \ep^{-2}} \cdot (1 + 2 D_O \epsilon^2 \log(\epsilon^{-1}))2 |B| \mathfrak{W} r_- \ep$. Thus, for $\ep$ chosen sufficiently small}, it is straightforward to bound the right hand side such that we improve the final bootstrap (\ref{eq:pk_bootstrap_lapse}). \blue{As a result, $\mathcal{PK}_{boot} = \mathcal{PK}$, and the remaining assertions of Proposition~\ref{prop:oscillation+} are straightforward.}
\end{proof}


\subsection{The onset of the Kasner-like geometry}

Following the proof of Proposition \ref{prop:oscillation+}, we mention a corollary of Proposition \ref{prop:oscillation+} that will be important in later sections, and in interpreting the subset $\blue{\mathcal{PK}'} = \mathcal{K}_1 \cap \mathcal{PK} = \{2 |B| \mathfrak{W} \epsilon^2 r_- \geq r(s) \geq e^{- \delta_0 \epsilon^{-2}} r_- \}$ as a genuine Kasner-like region. (Note that, for this step, we restrict to $r(s) \lesssim \epsilon^2 r_-$ rather than $r(s) \lesssim \epsilon r_-$ \blue{in order} to ignore the $O(1)$ terms in the Bessel \blue{function} asymptotics.)

\begin{corollary} \label{cor:protokasner}
Consider the region $s \in \blue{\mathcal{PK}'} = \{ 2 |B| \mathfrak{W} \epsilon^2 r_- \geq r(s) \geq e^{- \delta_0 \epsilon^{-2}}r_- \}$. In this region, we have the following  forms for $(\phi,\dot{\phi})$, where $c_1=1$ and $c_2 = 2 \pi^{-1} (\gamma-\log2)$ \blue{and $\gamma \approx 0.577$ is the Euler-Mascheroni constant.}
\begin{equation} \label{eq:pk_phi_asymp}
\left| \phi(s) + \frac{2}{\pi} C_{YK}(\epsilon) \log \left( \frac{r_-^2 \epsilon^2 }{r^2(s) \xi_K} \right) - c_1 C_{JK}(\epsilon) - c_2 C_{YK}(\epsilon) \right| \leq D_{PK} \epsilon^2 \log (\epsilon^{-1}),
\end{equation}
\begin{equation} \label{eq:pk_phi_s_asymp}
\left| \dot{\phi}(s) + \frac{2}{\pi} C_{YK}(\epsilon) \frac{r_-^2 \omega_{K} \epsilon^2 }{r^2(s) \xi_K} \right| \leq D_{PK} \epsilon^2 \log (\epsilon^{-1}).
\end{equation}

Next, recall the function $\Theta(\epsilon)$ arising in (\ref{eq:bessel_coeff_phase}). 
Then, defining $\Psi(s):= \frac{r^2 \dot{\phi}}{- r \dot{r}}(s)$  [see already \eqref{eq:Psi} in the next section], one finds that
\begin{equation} \label{eq:pk1_Psi}
\left|\Psi(s) + \frac{2}{\sqrt{\pi}} \, \mathfrak{W}^{-1} \sin(\Theta(\epsilon)) \right| \leq D_{PK} \epsilon^2 \log(\epsilon^{-1}).
\end{equation}
Furthermore, if we let $\Psi_i = \Psi(s_i)$, \blue{where $r(s_i) = e^{- \delta_0 \ep^{-2}}$, then for $s \in \mathcal{PK}'$,}
\begin{equation} \label{eq:pk1_Psi_2}
|\Psi(s) - \Psi_i| \leq D_{PK} r^2(s) \log (\epsilon^{-1}) \leq D_{PK}^2 \ep^4 \log (\epsilon^{-1}) .
\end{equation}

Finally, one obtains the following \blue{estimates} for the lapse $\Omega^2$:
\begin{equation} \label{eq:pk1_lapse_der}
\left| \frac{ \frac{d}{ds} \log \Omega^2(s) }{ \frac{d}{ds} \log r(s) } - \Psi_i^2 + 1 \right| \leq D_{PK} \epsilon^4 \log(\epsilon^{-1}),
\end{equation}
\begin{equation} \label{eq:pk1_lapse}
\left| \log \left[ \Omega^2(s) \left( \frac{r(s)}{r_-} \right)^{1 - \Psi_i^2} \right] + \frac{1}{2} b_-^{-2} \epsilon^{-2} \right| \leq D_{PK} \log (\epsilon^{-1}).
\end{equation}
\end{corollary}

\begin{proof}
The equations (\ref{eq:pk_phi_asymp}) and (\ref{eq:pk_phi_s_asymp}) follow immediately from (\ref{eq:pk_phi_1}), (\ref{eq:pk_phi_2}) and the Bessel function asymptotics in Facts \ref{fact:bessel1_taylor} and \ref{fact:bessel2_taylor}. Indeed, by restricting $r(s) \leq 2 |B| \mathfrak{W} \epsilon^2 r_-$, we guarantee that
\begin{equation*}
\frac{\xi_K r^2(s)}{r_-^2 \epsilon^2} \leq \epsilon^2,
\end{equation*}
hence Facts \ref{fact:bessel1_taylor} and \ref{fact:bessel2_taylor} guarantee that
\begin{equation*}
\left| J_0 \left( \frac{\xi_K r^2(s)}{r_-^2 \epsilon^2} \right) - 1 \right| \lesssim \epsilon^4, \quad \left| J_1\left( \frac{\xi_K r^2(s)}{r_-^2 \epsilon^2} \right) \right| \lesssim \epsilon^2,
\end{equation*}
\begin{equation*}
\left| Y_0 \left( \frac{\xi_K r^2(s)}{r_-^2 \epsilon^2} \right) + \frac{2}{\pi} \log \frac{r_-^2 \epsilon^2}{\xi _K r^2(s)} - \frac{2}{\pi} ( \gamma - \log 2 ) \right| \lesssim \epsilon^4, \quad \left| Y_1\left( \frac{\xi_K r^2(s)}{r_-^2 \epsilon^2} \right) + \frac{2}{\pi} \frac{r_-^2 \epsilon^2}{\xi_K r^2(s)} \right| \lesssim \epsilon^2 \log(\epsilon^{-1}).
\end{equation*}
Substituting these into (\ref{eq:pk_phi_1}) and (\ref{eq:pk_phi_2}), we get (\ref{eq:pk_phi_asymp}) and (\ref{eq:pk_phi_s_asymp}) with $c_1 = 1$ and $c_2 = 2 \pi^{-1} (\gamma - \log 2)$.

Next, combining (\ref{eq:pk_phi_s_asymp}) with the estimates (\ref{eq:omegak}) and (\ref{eq:xik}) it is straightforward to get
\begin{equation} \label{eq:pk_rphi_s_asymp}
\left| r^2 \dot{\phi}(s) + 16 \pi^{-1} C_{YK}(\epsilon) |B|^2 \mathfrak{W}^2 \omega_{RN} r_-^2 \epsilon^2 \right | \lesssim \epsilon^4 \log (\epsilon^{-1}).
\end{equation}
Hence from (\ref{eq:pk_r}) we find that
\begin{equation}
\left| \Psi (s) + 4 \pi^{-1} C_{YK}(\epsilon) \right| = \left| \frac{r^2 \dot{\phi}(s)}{- r \dot{r}(s)} + 4 \pi^{-1} C_{YK}(\epsilon) \right| \lesssim \epsilon^2 \log(\epsilon^{-1}),
\end{equation}
thus (\ref{eq:pk1_Psi}) follows from $|C_{YK}(\epsilon) - C_Y(\epsilon)| \lesssim \epsilon^2 \log(\epsilon^{-1})$ and (\ref{eq:bessel_y_coeff}).

To get (\ref{eq:pk1_Psi_2}), note that (\ref{eq:pk_phi_s_asymp}) will also yield $|r^2 \dot{\phi} (s) - r^2 \dot{\phi} (s_i) | \lesssim \epsilon^2 \log(\epsilon^{-1}) r^2(s)$, while we also know from (\ref{eq:pk_rd}) that $|- r \dot{r}(s) \blue{+} r \dot{r}(s_i)| \lesssim e^{-40 \delta_0 \epsilon^{-2}} \lesssim \epsilon^2 \log(\epsilon^{-1}) r^2(s)$ changes little in the region $\mathcal{PK}$. \blue{Since $- r \dot{r}(s) \sim \epsilon^2$ by (\ref{eq:pk_r})}, \blue{the computation
\[
    \left| \frac{r^2 \dot{\phi}}{- r \dot{r}} (s)  - \frac{r^2 \dot{\phi}}{- r \dot{r}} (s_i) \right | \lesssim \frac{1}{- r \dot{r}(s)} \cdot \left[ |r^2 \dot{\phi} (s) - r^2 \dot{\phi} (s_i) | + | {}- r \dot{r}(s) + r \dot{r}(s_i)| \right],
\]}
\blue{together with $r(s) \lesssim \ep^2$ for $s \in \mathcal{PK}'$,} yields (\ref{eq:pk1_Psi_2}).

We now move on to the estimate (\ref{eq:pk1_lapse_der}) for the derivative of $\Omega^2(s)$. We proceed here using the Raychaudhuri equation in the form (\ref{eq:raych_transport}), which after multiplying by $r(s) (\dot{r}(s))^{-2}$ gives
\begin{equation*}
\frac{r}{\dot{r}} \frac{d}{ds} \log \Omega^2 - \frac{r \ddot{r}}{\dot{r}^2} - \frac{r^2 \dot{\phi}^2}{\dot{r}^2} = \frac{r^2 |\tilde{A}|^2 q_0^2 |\phi|^2}{\dot{r}^2}.
\end{equation*}
Using (\ref{eq:r_evol}) to substitute for $\ddot{r}$, and noticing that the rightmost term on the left hand side is exactly $\Psi^2$, one finds
\begin{equation*}
\frac{r}{\dot{r}} \frac{d}{ds} \log \Omega^2 + 1 - \Psi^2 = ( - r \dot{r})^{-2} \left[ r^4 \tilde{A}^2 q_0^2 |\phi|^2 + \frac{\Omega^2}{4} \left( Q^2 - r^2 + r^4 \Lambda + r^4 m^2 |\phi|^2 \right) \right].
\end{equation*}
From Proposition \ref{prop:oscillation+} and \eqref{eq:pk_phi_asymp}, the right hand side of this is bounded by a multiple of $ \ep^{-4} (r^2 \log(\frac{1}{r^2}))^2\lesssim \epsilon^4 \log (\epsilon^{-1})$, so that if we \blue{also} use (\ref{eq:pk1_Psi_2}), then we get (\ref{eq:pk1_lapse_der}) as required.

Finally, for (\ref{eq:pk1_lapse}), one would like to integrate (\ref{eq:pk1_lapse_der}). However, one needs to first estimate $\Omega^2(s_{K_1})$, where $s = s_{K_1}$ is such that $r(s_{K_1}) = 2 |B| \mathfrak{W} \epsilon^2 r_-$ (corresponding to the past boundary of $\blue{\mathcal{PK}'}$). For this purpose, one finds exactly as in the proof of Proposition \ref{prop:oscillation}\footnote{We need to extend from $r(s) \gtrsim \epsilon$ to $r(s) \gtrsim \epsilon^2$, but the proof still applies as we only needed $|\log (r / r_-)| \lesssim \log (\epsilon^{-1})$.} that, for $s \in \red{\mathcal{O} \cup (\mathcal{PK} \setminus \mathcal{PK}')} = \{  s_O \leq s \leq s_{K_1} \}$,
\begin{equation}
\left | \frac{d}{ds} \log \Omega^2(s) - 2 K_-  \right| \lesssim \epsilon^2 [\log (\epsilon^{-1}) + r^{-2}(s)].
\end{equation}
Integrating this for $s \in \red{\mathcal{O} \cup (\mathcal{PK} \setminus \mathcal{PK}')}$, using that the length of the integration interval is of $O(\epsilon^{-2})$, one gets
\begin{equation*}
\left | \log \Omega^2(s_{K_1})  - 2 K_- s_{K_1} \right | \lesssim \log (\epsilon^{-1})+ |\log \Omega^2 ( s_O )| \lesssim \log (\epsilon^{-1}).
\end{equation*}

We also need to estimate the expression $2 K_- s_{K_1}$. This follows from an identical computation to that of $s_{PK}$ as in (\ref{eq:s0}), and one finds that
\begin{equation*}
2 K_- s_{K_1} = - \frac{|2 K_-|}{8 |B|^2 \mathfrak{W}^2 \omega_{RN} \epsilon^2} + O(\log (\epsilon^{-1})) = - \frac{1}{2} b_-^2 \epsilon^{-2} + O(\log (\epsilon^{-1})).
\end{equation*}
Since $\frac{r(s_{K_1})}{r_-} \sim \epsilon^2$, so that additional factors of $\log(\frac{r(s_{K_1})}{r_-})$ may be added, we arrive at
\begin{equation} \label{eq:pk1_lapse_init}
\left | \log \left[ \Omega^2 (s_{K_1}) \left( \frac{ r(s_{K_1}) }{r_-} \right)^{1- \Psi_i^2} \right] + \frac{1}{2} b_-^2 \epsilon^{-2} \right | \lesssim \log(\epsilon^{-1}).
\end{equation}

Finally, we rewrite (\ref{eq:pk1_lapse_der}) as
\begin{equation*}
\left| \frac{d}{ds} \log \left[ \Omega^2(s) \left( \frac{r(s)}{r_-} \right)^{1 - \Psi_i^2} \right] \right| \lesssim \epsilon^2 \log (\epsilon^{-1}) \frac{- \dot{r}(s)}{r(s)}.
\end{equation*}
Integrating this \blue{for $s \in \mathcal{PK}'$}, and using (\ref{eq:pk1_lapse_init}) to estimate the boundary term at $s = s_{K_1}$, we find (\ref{eq:pk1_lapse}).
\end{proof}

\section{Construction of the sets \texorpdfstring{$E_{\eta}$}{E{η}} and \texorpdfstring{$E'_{\eta,\sigma}$}{E{η,σ}}  for further quantitative estimates} \label{sec:sing}	
So far, we have a description of the hairy black hole interior up to the region $\mathcal{PK}$, where $s \approx (8 |B|^2 \mathfrak{W}^2 \omega_{RN} \epsilon^2)^{-1}$ and $r(s) \geq \exp(- \delta_0 \epsilon^{-2}) r_-$. These estimates hold for $0 < \epsilon \leq \epsilon_0$, where $\epsilon_0(M, \mathbf{e}, \Lambda, m^2, q_0) > 0$ is taken sufficiently small, and $\delta_0$ is a fixed quantity determined by (\ref{eq:delta0}). 
In particular, up to this point, we have placed no restriction on the value of $\epsilon$, other than its smallness. 

However, within the present article, to proceed further we will need to restrict attention to a smaller subset of values of $\epsilon$ verifying a certain condition. We define the following important quantity:
\begin{equation} \label{eq:Psi}
\Psi \coloneqq \frac{ r^2 \dot{\phi}}{-r \dot{r}} = - r \frac{d \phi}{dr}. 
\end{equation}

Then the precise condition on $\epsilon$ we shall use is that 
\begin{equation} \tag{$\dagger$} \label{condition_eta}
|\Psi_i|:=	|\Psi(s_i)| \geq \eta > 0,
\end{equation}
where $s_i$ marks the end of the region $\mathcal{PK}$, i.e.\ $r(s_i) = e^{- \delta_0 \epsilon^{-2}} r_-$, and $\eta > 0$ is an arbitrary small constant. 
Before using (\ref{condition_eta}) to analyse the spacetime beyond $\mathcal{PK}$ in Section~\ref{sec:kasner}, we first quantitatively characterize the set of $\epsilon$ for which a condition such as (\ref{condition_eta}) holds. In light of Corollary \ref{cor:protokasner}, we are required to study the quantity $\Theta(\epsilon)$ defined in (\ref{eq:bessel_coeff_phase}). 

\subsection{Improved estimates on \texorpdfstring{$\Theta(\epsilon)$}{ϴ(ε)}}

From Proposition \ref{prop:lateblueshift}, $\Theta(\epsilon)$ is identified to have the following dependence on $\epsilon$:
\begin{equation} \label{eq:vartheta}
\Theta(\epsilon) = \frac{1}{8 |B|^2 \mathfrak{W}^2 \epsilon^2} + O(\log (\epsilon^{-1})).
\end{equation}
However, the $O(\log (\epsilon^{-1}))$ error prevents us from having any quantitative control on quantities such as $\sin(\Theta(\epsilon))$, \blue{i.e.~}given only the above expression, for any fixed $\eta > 0$, if $|\cdot|$ denotes Lebesgue measure, the \textit{limiting density}
\begin{equation*}
\limsup_{\epsilon_0 \to 0} \epsilon_0^{-1} | \{ \epsilon \in (0, \epsilon_0]: |\sin (\Theta(\epsilon))| > \eta \} |
\end{equation*} 
could be arbitrarily small, or even vanish asymptotically.

To overcome this issue, we instead consider the quantity:
\begin{equation} \label{eq:bessel_coeff_phase_der}
\frac{d}{d \epsilon} \Theta(\epsilon) = \left ( \frac{\partial}{\partial \epsilon} + \frac{d \, s_O(\epsilon)}{d \epsilon} \frac{\partial}{\partial s} \right) \left (|q_0 \tilde{A}|(s) \cdot r^2(s) \cdot \left( - \frac{d}{ds} r^2(s) \right)^{-1} + \omega_{RN} s \right ),
\end{equation}
where to make sense of the right hand side, we now interpret $f(s) \in \{ r(s), \log \Omega^2(s), Q(s), \tilde{A}(s), \phi(s)\blue{, \dot{\phi}(s)} \}$ as (smooth) functions of $\epsilon$ as well as $s$. 

Denoting the $\epsilon$-derivatives by using a subscript, i.e.\ $f_{\epsilon} = \frac{\partial}{\partial \epsilon} f$, while still using $\dot{f} = \frac{\partial}{\partial s} f$ to denote $s$-derivatives, we shall take an $\epsilon$-derivative of the system (\ref{eq:raych})--(\ref{eq:phi_evol_2}) to find a system of \textit{linear} evolution equations for the quantities $f_{\epsilon}(s)$. For instance, the $\epsilon$-derivative of (\ref{eq:gauge_evol}) is
\begin{equation*}
\dot{\tilde{A}}_{\epsilon} = - \frac{\Omega^2}{4 r^2} \left( Q_{\epsilon} + Q (\log \Omega^2)_{\epsilon} - \frac{2 Q r_{\epsilon}}{r} \right),
\end{equation*}
while the corresponding evolution equations for the other linearized quantities are more complicated and will not be written explicitly. \blue{We can also perform this linearization for the difference quantities
\[
    \delta \phi_{\ep} = \frac{\partial}{\partial \ep} \left( \phi(s) - \phi_{\mathcal{L}}(s) \right), \quad \delta \dot{\phi}_{\ep} = \frac{\partial}{\partial s} \delta \phi_{\ep}.
\]} \indent
Along with the evolution equations, \blue{one must} pose data for the quantities $f_{\epsilon}(s)$ in the $s \to - \infty$ limit. For this purpose, one should return to the $1+1$-dimensional formulation of the problem in the regular $(U, V)$ coordinates as in Section \ref{sub:data}, and take the appropriate $\epsilon$-derivatives there. The correct asymptotic data is
\begin{gather}
\lim_{s \to -\infty} r_{\epsilon}(s) = 0, \; \lim_{s \to - \infty} Q_{\epsilon} (s) = 0, \; \lim_{s \to - \infty} \phi_{\epsilon} = 1, \label{eq:ode_datae_rqphi} \\
\lim_{s \to - \infty} (\log \Omega^2)_{\epsilon} = \lim_{s \to - \infty} \frac{d}{ds} (\log \Omega^2)_{\epsilon} = 0, \label{eq:ode_datae_omega} \\
\lim_{s \to - \infty} \Omega^{-2} \tilde{A}_{\epsilon} (s) = 0, \label{eq:ode_datae_gauge} \\
\lim_{s \to - \infty} 4 \Omega^{-2} \dot{r}_{\epsilon}(s) = \frac{2 r_+ m^2 \epsilon}{2 K_+}, \label{eq:ode_datae_rdot} \\
\lim_{s \to - \infty} \Omega^{-2} \dot{\phi}_{\epsilon} = \beta_+. \label{eq:ode_datae_phidot}
\end{gather}

The plan is now to use a similar procedure to Section \ref{sec:einsteinrosen} to find sufficiently strong estimates for the quantities $f_{\epsilon} \in \{ r_{\epsilon}, (\log \Omega^2)_{\epsilon}, Q_{\epsilon}, \tilde{A}_{\epsilon}, \phi_{\epsilon} \}$ up to the late blue shift region $\mathcal{LB}$, where we have nontrivial overlap with the oscillatory region $\mathcal{O}$, and compute $\frac{d}{d \epsilon}\Theta(\epsilon)$. The most crucial estimate will be to determine that $\dot{r}_{\epsilon}$ is comparable to $\epsilon$.

\begin{proposition} \label{prop:epsilon_der}
For $s \in \mathcal{EB} \cup \mathcal{LB} = \{ S \leq s \leq \Delta_{\mathcal{B}} \epsilon^{-1} \}$, there exists some constant $D_{LE}(M, \mathbf{e}, \Lambda, m^2, q_0)$ such that:
\begin{gather}
|\dot{r}_{\epsilon}| + s^{-1} |r_{\epsilon}| \leq D_{LE} \epsilon, \label{eq:lb_e_r} \\
\left| \frac{d}{ds} (\log \Omega^2)_{\epsilon} \right| + s^{-1} |(\log \Omega^2)_{\epsilon}| \leq D_{LE} \epsilon s, \label{eq:lb_e_lapse} \\
|\tilde{A}_{\epsilon}| + s^{-1} |Q_{\epsilon}| \leq D_{LE} \epsilon, \label{eq:lb_e_maxwell} \\
|\dot{\phi}_{\epsilon}| + |\phi_{\epsilon}| \leq D_{LE}, \label{eq:lb_e_phi} \\
|\delta \dot{\phi}_{\epsilon}| + | \delta \phi_{\epsilon}| \leq D_{LE} \epsilon^2 s. \label{eq:lb_e_phidiff}
\end{gather}
Furthermore, we have the more precise estimate for $ s_O=50 s_{lin} \leq s \leq \Delta_{\mathcal{B}} \epsilon^{-1}$:
\begin{equation} \label{eq:lb_e_r_precise}
\left| - \dot{r}_{\epsilon}(s) - 8 |B|^2 \mathfrak{W}^2 \omega_{RN} r_- \epsilon \right| \leq D_{LE} \epsilon^3 s.
\end{equation}
\end{proposition}

\begin{corollary} \label{cor:phase_der}
Consider the expression (\ref{eq:bessel_coeff_phase_der}). Then
\begin{equation} \label{eq:bessel_coeff_phase_epsilon}
\left| \frac{d}{d \epsilon} \Theta(\epsilon) + \frac{1}{4 |B|^2 \mathfrak{W}^2 \epsilon^3} \right| \leq D_{LE} \epsilon^{-1} \log (\epsilon^{-1}).
\end{equation}
\end{corollary}

\begin{proof}[Proof of Corollary \ref{cor:phase_der} given Proposition \ref{prop:epsilon_der}]
First consider the $\frac{\partial}{\partial s}$ derivative in (\ref{eq:bessel_coeff_phase_der}). One checks from (\ref{eq:r_evol_2}), (\ref{eq:gauge_evol}) and Proposition \ref{prop:lateblueshift}, that one has
\begin{equation*}
\left |\frac{\partial}{\partial s} \left (|q_0 \tilde{A}|(s) \cdot r^2(s) \cdot \left( - \frac{d}{ds} r^2(s) \right)^{-1} + \omega_{RN} s \right ) \right|
\lesssim \epsilon^{-4} \, \Omega^2 \cdot \blue{(1 + |\phi|^2)} + \red{| \omega_{RN} - |q_0 \tilde{A}||}  \lesssim \epsilon^2
\end{equation*}
where the final step follows from $\Omega^2 \lesssim \epsilon^{100}$ at $s =  s_O$ and Proposition \ref{prop:lateblueshift}. Hence even with the $\frac{d  s_O}{d \epsilon} \sim \epsilon^{-1}$ factor in front, this term contributes at worst $O(\epsilon)$ and can be ignored.

The main term $(4 |B|^2 \mathfrak{W}^2 \epsilon^3)^{-1}$ on the left hand side of (\ref{eq:bessel_coeff_phase_epsilon}) comes from taking the $\epsilon$-derivative of $\left( - \frac{d}{ds} r^2(s) \right)^{-1}$. Indeed, the expression that arises from this is
\begin{equation*}
I = |q_0 \tilde{A}|(s) \cdot r^2(s) \cdot \left( - \frac{d}{ds} r^2(s) \right)^{-2} \cdot 2 ( r_{\epsilon}(s) \dot{r}(s) + r(s) \dot{r}_{\epsilon}(s)).
\end{equation*}
Using Propositions \ref{prop:lateblueshift} and \ref{prop:epsilon_der}, particularly \eqref{eq:lb_e_r_precise}, we can evaluate (note $\dot{r}(s) r_{\epsilon}(s) = O(\epsilon^3 s)$ is treated as part of the error):
\begin{align*}
I &= \frac{\omega_{RN} r_-^2}{(8 |B|^2 \mathfrak{W}^2 r_-^2 \omega_{RN} \epsilon^2)^2} \cdot  2(- 8 |B|^2 \mathfrak{W}^2 \omega_{RN} r_-^2 \epsilon )+ O(\epsilon^{-1} \log (\epsilon^{-1})) \\
&= - \frac{1}{4 |B|^2 \mathfrak{W}^2 \epsilon^3} +  O(\epsilon^{-1} \log(\epsilon^{-1})).
\end{align*}

Therefore, the $\epsilon$-derivative on the left hand side of (\ref{eq:bessel_coeff_phase_epsilon}) can be evaluated using Proposition~\ref{prop:epsilon_der} as:
\begin{equation*}
I + (|q_0 \tilde{A}|_{\epsilon}(s)  r^2(s) + 2 |q_0 \tilde{A}|(s) r(s) r_{\epsilon}(s) ) \cdot \left( - \frac{d}{ds} r^2(s) \right)^{-1} = - \frac{1}{4 |B|^2 \mathfrak{W}^2 \epsilon^3} + O(\epsilon^{-1} \log(\epsilon^{-1})),
\end{equation*}
completing the proof of the corollary.
\end{proof}

\begin{proof}[Proof of Proposition \ref{prop:epsilon_der}]
The first step will be to find the upper bounds. For this purpose, we will have to write out the full linear system of ODEs, however, as mentioned previously, this would be \blue{cumbersome}, and we therefore only include upper bounds for $|\dot{f}_{\epsilon}|$, \blue{often using that $r \sim 1, - \dot{r} \lesssim 1, |\phi| \lesssim \ep, |\dot{\phi}|\lesssim \ep, Q \sim 1$ and $\tilde{A} \lesssim 1$ within this region}. Differentiating (\ref{eq:r_evol})--(\ref{eq:phi_evol}) in $\epsilon$, the appropriate inequalities and equations are:
\begin{equation} \label{eq:r_e_evol}
|\ddot{r_{\epsilon}}| \lesssim \Omega^2 ( |r_{\epsilon}| + |(\log \Omega^2)_{\epsilon}| + |Q_{\epsilon}| ) + ( - \dot{r} ) (|r_{\epsilon}| + |\dot{r}_{\epsilon}|) + \Omega^2 | \phi \phi_{\epsilon} |,
\end{equation}
\begin{equation} \label{eq:lapse_e_evol}
\left| \frac{d^2}{ds^2} (\log \Omega^2)_{\epsilon} \right| \lesssim \Omega^2 ( |r_{\epsilon}| + |(\log \Omega^2)_{\epsilon}| + |Q_{\epsilon}| ) + ( - \dot{r} ) (|r_{\epsilon}| + |\dot{r}_{\epsilon}|) + |\tilde{A}| |\phi \phi_{\epsilon}| + |\dot{\phi} \dot{\phi}_{\epsilon}| + |\tilde{A}_{\epsilon}| |\phi|^2,
\end{equation}
\begin{equation} \label{eq:Q_e_evol}
| \dot{Q}_{\epsilon} | \lesssim (|\tilde{A}_{\epsilon}| + |\tilde{A} r_{\epsilon}|) |\phi|^2 + |\tilde{A} \phi \phi_{\epsilon}|,
\end{equation}
\begin{equation} \label{eq:gauge_e_evol}
|\dot{\tilde{A}}_{\epsilon}| \lesssim \Omega^2 ( |r_{\epsilon}| + |(\log \Omega^2)_{\epsilon}| + |Q_{\epsilon}|),
\end{equation}
\begin{equation} \label{eq:phi_e_evol}
\ddot{\phi}_{\epsilon} = - \frac{2 \dot{r} \dot{\phi}_{\epsilon}}{r} - q_0^2 \tilde{A}^2 \phi_{\epsilon} - \frac{m^2 \Omega^2}{4} \phi_{\epsilon} + J_{\phi}, \hspace{1cm} |J_{\phi}| \lesssim |\dot{r}_{\epsilon} - \dot{r} r_{\epsilon}| |\dot{\phi}| + |\tilde{A}_{\epsilon}| |\phi| + \Omega^2 (\log \Omega^2)_{\epsilon} |\phi|.
\end{equation}

Note that, for the final equation (\ref{eq:phi_e_evol}), it is necessary to exactly keep the terms corresponding to $\phi$ being a solution of the linear charged scalar wave equation. We now proceed through the regions $\mathcal{R}$, $\mathcal{N}$, $\mathcal{EB}$ and $\mathcal{LB}$ exactly as in Section \ref{sec:einsteinrosen}.

\medskip \noindent
\underline{Step 1: The redshift region $\mathcal{R} = \{ - \infty < s \leq - \Delta_{\mathcal{R}} \}$}

\medskip \noindent
In this region, we consider the following bootstrap assumptions, which hold in a neighborhood of $s = -\infty$ by the asymptotic data above:
\begin{equation}
|r_{\epsilon}| + |\dot{r}_{\epsilon}| + |(\log \Omega^2)_{\epsilon}| + |Q_{\epsilon}| + |\Omega^{-2} \tilde{A}_{\epsilon} | \leq \epsilon,
\end{equation}
\begin{equation}
|\phi_{\epsilon}| + |\dot{\phi}_{\epsilon}| \leq 2.
\end{equation}

Note that by Proposition \ref{prop:redshift}, we know already that $- \dot{r}, |\tilde{A}| \lesssim \Omega^2$ \blue{and $|\phi|, |\dot{\phi}| \lesssim \ep$}, so given these bootstrap assumptions and the inequality (\ref{eq:rs_omega_estimate}) we may simply integrate up the equations (\ref{eq:r_e_evol})--(\ref{eq:phi_e_evol}) to get
\begin{equation}
|r_{\epsilon}| + |\dot{r}_{\epsilon}| + |(\log \Omega^2)_{\epsilon}| + \left| \frac{d}{ds} (\log \Omega^2)_{\epsilon} \right| + |Q_{\epsilon}| + |\Omega^{-2} \tilde{A}_{\epsilon} | \lesssim \epsilon \Omega^2,
\end{equation}
\begin{equation}
|\phi_{\epsilon} - 1| + |\dot{\phi}_{\epsilon}| \lesssim \Omega^2.
\end{equation}
(Recall that the asymptotic data for $r_{\epsilon}$, $\dot{r}_{\epsilon}$, $\log \Omega^2_{\epsilon}$ and $Q_{\epsilon}$ is $0$.)	Hence with a choice of $\Delta_{\mathcal{R}}$ large enough, the bootstraps are easily improved.

\medskip \noindent
\underline{Step 2: The no-shift region $\mathcal{N} = \{ - \Delta_{\mathcal{R}} \leq s \leq S \}$}

\medskip \noindent
Since we are integrating only in a finite $s$-region, the no-shift region is easily dealt with using Gr\"onwall. To do this, let $\mathbf{X}_{\epsilon}$ and $\mathbf{\Phi}_{\epsilon}$ denote the tuples:
\begin{equation*}
\mathbf{X}_{\epsilon} = \left( r_{\epsilon}, \dot{r}_{\epsilon}, (\log \Omega^2)_{\epsilon}, \frac{d}{ds} (\log \Omega^2)_{\epsilon}, Q_{\epsilon}, \tilde{A}_{\epsilon} \right), \; \mathbf{\Phi}_{\epsilon} = (\phi_{\epsilon}, \dot{\phi}_{\epsilon}),
\end{equation*}
then in light of Proposition \ref{prop:noshift}, the system (\ref{eq:r_e_evol})--(\ref{eq:phi_e_evol}) can be translated into
\begin{equation*}
|\dot{\mathbf{X}}_{\epsilon}| \lesssim |\mathbf{X}_{\epsilon}| + \epsilon |\mathbf{\Phi}_{\epsilon}|, \;
|\dot{\mathbf{\Phi}}_{\epsilon}| \lesssim |\mathbf{\Phi}_{\epsilon}| + \epsilon |\mathbf{X}_{\epsilon}|.
\end{equation*}

Hence a straightforward use of Gr\"onwall in the bounded $s$-region $s \in[- \Delta_{\mathcal{R}}, S]$ yields
\begin{equation*}
\sup_{s \in \mathcal{N}} |\mathbf{X}_{\epsilon}|(s) \lesssim |\mathbf{X}_{\epsilon}|(-\Delta_{\mathcal{R}}) + \epsilon \sup_{s \in \mathcal{N}} |\mathbf{\Phi}_{\epsilon}|(s)
\lesssim |\mathbf{X}_{\epsilon}|(-\Delta_{\mathcal{R}}) + \epsilon |\mathbf{\Phi}_{\epsilon}|(-\Delta_{\mathcal{R}})  + \epsilon^2 \sup_{s \in \mathcal{N}} |\mathbf{X}_{\epsilon}|(s)
\end{equation*}
\blue{So, for $\epsilon$ sufficiently small, we absorb the rightmost term into the left hand side, and repeat the argument for $\mathbf{\Phi}_{\ep}(s)$, to derive that 
\begin{gather*}
    \sup_{s \in \mathcal{N}} |\mathbf{X}_{\epsilon}|(s) \lesssim |\mathbf{X}_{\epsilon}|(-\Delta_{\mathcal{R}}) + \epsilon |\mathbf{\Phi}_{\epsilon}|(-\Delta_{\mathcal{R}}), \\[0.5em]
    \sup_{s \in \mathcal{N}} |\mathbf{\Phi}_{\epsilon}|(s) \lesssim |\mathbf{\Phi}_{\epsilon}|(-\Delta_{\mathcal{R}}) + \epsilon |\mathbf{X}_{\epsilon}|(-\Delta_{\mathcal{R}}).
\end{gather*}

Inserting our bounds from the red-shift region $\mathcal{R}$, there exists $D_{NE}(M, \mathbf{e}, \Lambda, m^2, q_0)$ such that, for all $s \in \mathcal{N}$,
\begin{equation}
|\mathbf{X}_{\epsilon}|(\blue{s}) \leq D_{NE} \epsilon, \; |\mathbf{\Phi}_{\epsilon}|(\blue{s}) \leq D_{NE}.
\end{equation}}
\medskip \noindent
\underline{Step 3: Upper bounds in the blue shift regions $\mathcal{EB} \cup \mathcal{LB} = \{ S \leq s \leq \Delta_{\mathcal{B}} \blue{\epsilon^{-1}} \}$.}

\medskip \noindent
This step will be much simpler than the corresponding nonlinear estimates in Sections \ref{sub:earlyblueshift} and \ref{sub:lateblueshift}. We use the bootstrap assumptions:
\begin{gather}
|\dot{r}_{\epsilon}| \leq 10 D_{NE} \epsilon, \label{eq:b_e_bootstrap_r} \\
|(\log \Omega^2)_{\epsilon}| \leq 10 D_{NE} \epsilon s^3, \label{eq:b_e_bootstrap_lapse} \\
|Q_{\epsilon}| \leq 10 D_{NE} \epsilon s^2.
\end{gather}
Note that the first of these trivially implies $|r_{\epsilon}| \leq 10 D_{NE} \epsilon s$. Then integration of (\ref{eq:gauge_e_evol}) gives that $|\tilde{A}_{\epsilon} |\lesssim D_{NE} \epsilon$.

We next use these bootstraps to estimate $\phi_{\epsilon}$ and $\dot{\phi}_{\epsilon}$. Note that the expression $J_{\phi}$ in \eqref{eq:phi_e_evol} now obeys the estimate $|J_{\phi}| \lesssim D_{NE} \epsilon^2$. We now follow the proof of Proposition \ref{prop:earlyblueshift} and consider the quantity
\begin{equation*}
H^{(\epsilon)} = r^4 \dot{\phi}^2_{\epsilon} + r^4 q_0^2 |\tilde{A}|^2 \phi_{\epsilon}^2.
\end{equation*}
Completely analogously to before, one finds that
\begin{equation*}
\dot{H}^{(\epsilon)} \lesssim \Omega^2 H^{(\epsilon)} + |J_{\phi}| |\dot{\phi}| \lesssim \Omega^2 H^{(\epsilon)} + D_{NE} \epsilon^2 \sqrt{ H^{(\epsilon)} }.
\end{equation*}
So, as \blue{$\int_S^{\Delta_{\mathcal{B}}\ep^{-1}} \Omega^2(\tilde{s}) \, d \tilde{s} \lesssim 1$}, one can apply Gr\"onwall to the quantity $\sqrt{ H^{(\epsilon)} }$ to yield
\begin{equation*}
\sqrt{ H^{(\epsilon)} } (s) \lesssim \sqrt{ H^{(\epsilon)} }(S) + D_{NE} \epsilon^2 s \lesssim |\phi_{\epsilon}|(S) + |\dot{\phi}_{\epsilon}|(S) + D_{NE} \epsilon^2 s \lesssim D_{NE}.
\end{equation*}
Since $H^{(\epsilon)} \sim \phi_{\epsilon}^2 + \dot{\phi}_{\epsilon}^2$, one gets the estimate (\ref{eq:lb_e_phi}). 

The next step is to integrate (\ref{eq:r_e_evol}). Note that, by Propositions \ref{prop:earlyblueshift} and \ref{prop:lateblueshift}, we have $- \dot{r} \lesssim \max \{ \Omega^2, \epsilon^2 \}$, \blue{$|\phi| \lesssim \ep$} and $\Omega^2 \lesssim e^{2 K_- s}$, so
\begin{equation*}
|\dot{r}_{\epsilon}| \leq D_{NE} \epsilon + C D_{NE} \int_S^{\Delta_{\mathcal{B}} \epsilon^{-1}} (e^{2 K_- s'} s'^3 \epsilon + \epsilon^3 s') \, ds' \leq 2 D_{NE} \blue{\ep},
\end{equation*}
where $C$ is a constant independent of $\epsilon$ and $D_{NE}$. \blue{For the second inequality above, note that evaluating the integral gives a quantity bounded by a multiple of $\epsilon(e^{2K_- S} S^3 + \Delta_{\mathcal{B}}^2)$, hence the inequality follows} for $S$ chosen sufficiently large and $\Delta_{\mathcal{B}}$ chosen sufficiently small. This improves (\ref{eq:b_e_bootstrap_r}) and in fact yields (\ref{eq:lb_e_r}) after a further integration.

The estimates (\ref{eq:lb_e_lapse}) and (\ref{eq:lb_e_maxwell}) then follow from integration of (\ref{eq:lapse_e_evol}), (\ref{eq:Q_e_evol}) and (\ref{eq:gauge_e_evol}), and this also improves the remaining two bootstrap assumptions.

\medskip \noindent
\underline{Step 4: Precise bounds for the scalar field}

\medskip \noindent
We now move onto the estimate (\ref{eq:lb_e_phidiff}). Note that $\delta \phi_{\epsilon} = \phi_{\epsilon} - (\phi_{\mathcal{L}})_{\epsilon}$, but recall that since $\phi_{\mathcal{L}}$ is the exactly the solution to the linear charged wave equation in a Reissner-N\"ordstrom background having initial data $\lim_{s \to - \infty} \phi_{\mathcal{L}} = \epsilon$, it is clear that $(\phi_{\mathcal{L}})_{\epsilon}$ is exactly the solution to the linear charged wave equation in Reissner-N\"ordstrom with $\lim_{s \to - \infty} (\phi_{\mathcal{L}})_{\epsilon} (s) = 1$. Namely, we have
\begin{equation*}
(\ddot{\phi}_{\mathcal{L}})_{\epsilon} = - \frac{2 \dot{r}_{RN} (\dot{\phi}_{\mathcal{L}})_{\epsilon}}{r_{RN}} - q_0^2 \tilde{A}_{RN}^2 (\phi_{\mathcal{L}})_{\epsilon} - \frac{m^2 \Omega_{RN}^2}{4} (\phi_{\mathcal{L}})_{\epsilon}.
\end{equation*}

Hence subtracting this equation from (\ref{eq:phi_e_evol}), and using both the estimates of Section \ref{sec:einsteinrosen} and earlier within the proof of this proposition, one finds
\begin{equation} \label{eq:lb_e_equation}
\delta \dot{\phi}_{\epsilon} = - \frac{2 \dot{r} \delta \dot{\phi}_{\epsilon}}{r} - q_0^2 \tilde{A}^2 \delta \phi - \frac{m^2 \Omega^2}{4} \delta \phi_{\epsilon} + J_{\phi} + \tilde{J}_{\phi},
\end{equation}
where $J_{\phi}$ is as in (\ref{eq:phi_e_evol}), $\tilde{J}_{\phi}$ arises from taking the differences of $r, \dot{r}, \tilde{A}, \Omega^2$ from their Reissner-Nordstr\"om quantities, and one gets the estimate
\begin{equation*}
|J_{\phi}(s)| + |\tilde{J}_{\phi}(s)| \lesssim
\begin{cases}
\epsilon^2 \Omega^2 & \text{ for } s \in \mathcal{R}, \\
\epsilon^2 & \text{ for } s \in \mathcal{N} \cup \mathcal{EB} \cup \mathcal{LB}.
\end{cases}
\end{equation*}

We now proceed in the usual way using the quantity
\begin{equation*}
\tilde{H}^{(\epsilon)} = r^4 | \delta \dot{\phi}_{\epsilon} |^2 + r^4 q_0^2 |\tilde{A}|^2 |\delta \phi_{\epsilon}|^2.
\end{equation*}
Then by using the equation (\ref{eq:lb_e_equation}), one finds that for all $s \in \mathcal{R} \cup \mathcal{N} \cup \mathcal{EB} \cup \mathcal{LB}$, we get
\begin{equation*}
\frac{d}{ds} \tilde{H}^{(\epsilon)} \lesssim \Omega^2 \tilde{H}^{(\epsilon)} + (|J_{\phi}| + |\tilde{J}_{\phi}|) \sqrt{ \tilde{H}^{(\epsilon)} }.
\end{equation*}
Using now the fact that $\lim_{s \to - \infty} \tilde{H}^{(\epsilon)}(s) = 0$ and $\int_{- \infty}^{\Delta_{\mathcal{B}} \epsilon^{-1}} \Omega^2$ is uniformly bounded in $\epsilon$, Gr\"onwall applied to this differential inequality gives:
\begin{equation}
\sqrt{\tilde{H}^{(\epsilon)}} \lesssim \int_{- \infty}^{\Delta_{\mathcal{B}} \epsilon^{-1}} (|J_{\phi}(\tilde{s})| + |\tilde{J}_{\phi}(\tilde{s})| )\, d\tilde{s} \lesssim \epsilon^2 s.
\end{equation}

This of course will yield the estimate (\ref{eq:lb_e_phidiff}). Note that by Corollary \ref{cor:scattering}, one then has for $s \in \mathcal{LB}$, 
\begin{gather}
\left| \phi_{\epsilon} - B e^{i \omega_{RN} s} - \overline{B} e^{- i \omega_{RN} s} \right| \lesssim \epsilon^2 s, \label{eq:lb_e_phi_est1} \\
\left| \dot{\phi}_{\epsilon} - i \omega_{RN} B e^{i \omega_{RN} s} + i \omega_{RN} \overline{B} e^{- i \omega_{RN} s} \right| \lesssim \epsilon^2 s. \label{eq:lb_e_phi_est2}
\end{gather}

\medskip \noindent
\underline{Step 5: The precise estimate for $\dot{r}_{\epsilon}$}

\medskip \noindent
We finally move to the estimate (\ref{eq:lb_e_r_precise}). We now use the differentiated version of the Raychaudhuri equation in the convenient form (\ref{eq:raych_transport}). We see that using (\ref{eq:lb_e_r}), (\ref{eq:lb_e_lapse}) and (\ref{eq:lb_e_maxwell}) we can find
\begin{equation*}
\frac{d}{ds}(- \dot{r}_{\epsilon}) - \frac{d}{ds} \log (\Omega^2) \cdot ( - \dot{r}_{\epsilon}) = 2 r ( \dot{\phi} \dot{\phi}_{\epsilon} + |\tilde{A}|^2 q_0^2 \phi \phi_{\epsilon}) + J_{\Omega},
\end{equation*}
where the error $J_{\Omega}$ satisfies $|J_{\Omega}| \lesssim \epsilon^3 s$.

Now we use (\ref{eq:lb_e_phi_est1}), (\ref{eq:lb_e_phi_est2}), alongside the estimates (\ref{eq:lb_rdiff}), (\ref{eq:lb_lapse}), (\ref{eq:lb_phidiff}), to find further that
\begin{equation*}
\frac{d}{ds}(- \dot{r}_{\epsilon}) - 2 K_- ( - \dot{r}_{\epsilon}) = 8 |B|^2 \omega_{RN}^2 r_- \epsilon + J_{\Omega} + \tilde{J}_{\Omega},
\end{equation*}
where $|\tilde{J}_{\Omega}| \lesssim \epsilon^3 s$. We now use a standard integrating factor to integrate between $s_{lin}$ and $s \in \mathcal{LB}$, yielding
\begin{equation*}
- \dot{r}_{\epsilon}(s) = e^{2 K_-(s - s_{lin})} (- \dot{r}_{\epsilon})(s_{lin}) + \int_{s_{lin}}^s e^{2 K_-(s - s')} (8 |B|^2 \omega_{RN}^2 r_- \epsilon + J_{\Omega}(s') + \tilde{J}_{\Omega}(s')) \, ds' .
\end{equation*}

One may then simply compute the relevant integrals to find that
\begin{equation*}
\left| - \dot{r}_{\epsilon} - \frac{8 |B|^2 \omega_{RN}^2 r_- \epsilon}{2 K_-} \right| \lesssim \epsilon^{-1} \Omega^2 + \epsilon^3 s.
\end{equation*}
But for $s \geq s_O$, we know $\Omega^2 \lesssim \epsilon^{100}$, and (\ref{eq:lb_e_r_precise}) follows immediately.
\end{proof}

\subsection{The measure of the set \texorpdfstring{$E_{\eta}$}{E{η}}}

We now use Proposition~\ref{prop:epsilon_der}, or more precisely Corollary~\ref{cor:phase_der}, to control the measure of the set of values of $\epsilon$ such that condition (\ref{condition_eta}) holds. Let $\epsilon_0(M, \mathbf{e}, \Lambda, m^2, q_0) > 0$ be such that the results of Section~\ref{sec:protokasner} hold for $0 < |\epsilon| < \epsilon_0$. Then we have the following corollary:

\begin{corollary} \label{cor:condition_eta}
Let $\eta > 0$ be a sufficiently small constant. We define $E_{\eta}$ to be the set of $\epsilon$ such that the hairy black hole interior corresponding to $\phi = \epsilon$ on $\mathcal{H}$ obeys the condition (\ref{condition_eta}) at $s = s_i$:
\begin{equation} \label{eq:eeta}
E_{\eta} =  \{ \epsilon \in (0, \epsilon_0) : |\Psi_i| \geq \eta \}.
\end{equation}
Then the set $E_{\eta}$ is non-empty, has $0$ as a limit point, and we have the following upper bound for the limiting density of values of $\epsilon$ violating (\ref{condition_eta}): there exists some constant $K$ such that
\begin{equation} \label{eq:density}
\limsup_{\tilde{\epsilon} \downarrow 0} \tilde{\epsilon}^{-1} | (0, \tilde{\epsilon}) \setminus E_{\eta} | \leq K \mathfrak{W} \eta.
\end{equation}
\end{corollary}

\begin{proof}
We will estimate the measure of the set $F_{\eta, \tilde{\epsilon}} = (0, \tilde{\epsilon}) \setminus E_{\eta} = \{ \epsilon \in (0, \tilde{\epsilon}): |\Psi_i| < \eta \}$. We first use Corollary~\ref{cor:protokasner} to change variable from $\epsilon$ to $\blue{\tilde{\Theta} = } \Theta(\epsilon)$. For $\epsilon_0$ and $\eta$ sufficiently small, 
\begin{align*}
|F_{\eta, \tilde{\epsilon}}|
&= \int_0^{\tilde{\epsilon}} \mathbbm{1}_{\{|\Psi_i| < \eta\}} \, d \epsilon\\[0.5em]
&\leq \int_0^{\tilde{\epsilon}} \mathbbm{1}_{\{|\sin(\Theta(\epsilon))| < \sqrt{\pi} \mathfrak{W} \eta\} } \, d \epsilon \\[0.5em]
&= \int_{\Theta(\tilde{\epsilon})}^{+\infty} \mathbbm{1}_{\{|\sin \blue{\tilde{\Theta}}| < \sqrt{\pi} \mathfrak{W} \eta\}} \, \left| \frac{d}{d\epsilon} \Theta(\epsilon) \blue{\Big|_{\ep = \Theta^{-1}(\tilde{\Theta})}}\right|^{-1} \, \blue{d \tilde{\Theta}}.
\end{align*}

We now apply Corollary~\ref{cor:phase_der} along with the previous estimate (\ref{eq:vartheta}) for $\Theta(\epsilon)$. This will yield:
\begin{equation}
\frac{d}{d\epsilon} \Theta(\epsilon) = - 4 \sqrt{2} |B| \mathfrak{W} \, \Theta(\epsilon)^{3/2} + O(\Theta^{1/2} \log \Theta).
\end{equation}
Combining with the above, we therefore see that
\begin{equation*}
| F_{\eta, \tilde{\epsilon}} | \leq \int^{+\infty}_{\Theta(\tilde{\epsilon})} \mathbbm{1}_{\{|\sin \blue{\tilde{\Theta}}| < \sqrt{\pi} \mathfrak{W} \eta\}} \, \frac{1}{4 |B|\mathfrak{W} }\, \blue{\tilde{\Theta}^{-3/2} \, d \tilde{\Theta}.}
\end{equation*}

To evaluate this integral, note that if $\eta$ is taken sufficiently small, then the set $\{ |\sin\Theta| < \sqrt{\pi} \mathfrak{W} \eta \}$ is simply a union of intervals of width $2 \sqrt{\pi} \mathfrak{W} \eta + O(\eta^2)$ and centered on the integer lattice $\pi \Z$. So the integral is akin to taking a discrete sum of the form $\sum^{+\infty}_{n = \Theta(\tilde{\epsilon})} \frac{\eta}{n^{3/2}}$, multiplied by appropriate weights. Keeping only the weights $|B|$ and $\mathfrak{W}$ which depend on the background parameters, we have
\begin{equation}
|F_{\eta, \tilde{\epsilon}}| \lesssim |B|^{-1} \, \Theta(\tilde{\epsilon})^{-1/2} \eta \lesssim \mathfrak{W}\eta \tilde{\epsilon},
\end{equation}
using (\ref{eq:vartheta}) again in the last step. This yields (\ref{eq:density}), and the remainder of the corollary follows easily.
\end{proof}

\begin{rmk}
Recalling Theorem~\ref{maintheorem2}, as well as studying the values of $\ep$ obeying (\ref{condition_eta}), we also need to investigate the measure of the set of values of $\ep$ obeying the `\blue{bounce}' condition $\eta \leq |\Psi_i| \leq 1 - \sigma$ or the `non-\blue{bounce}' condition $|\Psi_i| \geq 1 + \sigma$. This can be done in the same way as the proof of Corollary~\ref{cor:condition_eta}. 
\end{rmk}

\section{Kasner regimes and bounces} \label{sec:kasner}
Now, assuming the condition \eqref{condition_eta}, i.e\ that $\epsilon$ lies in $E_{\eta}$ as defined in Theorem~\ref{maintheorem}, we will complete the proof of Theorem~\ref{maintheorem}. In particular, we claim there exists some $\epsilon_0(\eta) > 0$ depending on $\eta$ as well as the usual parameters $M ,\mathbf{e}, \Lambda, m^2, q_0$, such that if $0 < \epsilon \leq \epsilon_0 = \epsilon_0(\eta)$ \textbf{and} \eqref{condition_eta} holds, then our corresponding hairy black hole interior contains a crushing spacelike singularity, with more quantitative Kasner-like asymptotics to follow.

Firstly, we briefly describe the expected dynamics between the end of the region $\mathcal{PK}$ and the eventual spacelike singularity at $s = s_{\infty}$. It turns out the intermediate dynamics will be highly sensitive to the value of $\Psi(s_i)$, often denoted as $\Psi_i$ in the sequel, where $\Psi(s)$ is given by \eqref{eq:Psi}. 

Following the discussion of the introduction, particularly Section \ref{cosmo.intro}, if $|\Psi_i| > 1$, then the region $\blue{\mathcal{PK}'} \subset \mathcal{PK}$ (see Corollary~\ref{cor:protokasner}) should already lie in a regime associated with positive Kasner exponents. For $\epsilon$ sufficiently small, we will show that many of the estimates of Proposition \ref{prop:oscillation+} will persist all the way to the spacelike singularity $\{ s = s_{\infty}\}$, meaning a single (stable) Kasner-like regime.

On the other hand, if $\eta \leq |\Psi_i| < 1$, then initially, the spacetime lies in a regime associated with one negative and two positive Kasner exponents -- known to be unstable in the cosmological setting. In this case, we shall observe the aforementioned \textit{Kasner \blue{bounce}} phenomenon. In simplified terms, between $s=s_i$ and $s = s_{\infty}$, the quantity $\Psi$ will \textit{invert} from its initial value $\Psi_i = \Psi(s_i)$ to a final value $\Psi_f \approx \Psi_i^{-1}$ satisfying $|\Psi_f| > 1$. The spacetime in turn evolves into a regime associated with positive Kasner exponents, which in turn persists up to the singularity.

The cause and nature of such Kasner \blue{bounces}, as well as why this allows for spacelike singularity formation, will be the focus of this section. Denote by $\mathcal{K}$ the region $\{s \geq s_i\} = \{s: r(s) \leq e^{- \delta_0 \epsilon^{-2}} r_- \}$. The section will be organized as follows:
\begin{itemize}
\item 
In Section \ref{sub:kasner_why}, we give further background into why we distinguish between $|\Psi| < 1$ and $|\Psi| > 1$, and hence the need to study Kasner \blue{bounces}. We will also introduce new renormalized quantities and the equations that they obey. 

\item
In Section \ref{sub:kasner_prelim}, we state the main Proposition \ref{prop:k_ode} regarding the quantity $\Psi$ in the region $\mathcal{K}$. We then state the main bootstrap assumptions used in this proof, and prove some preliminary results used in the proof of Proposition \ref{prop:k_ode}.
\item
In Section \ref{sub:kasner_main}, we provide the proof of Proposition \ref{prop:k_ode}. This will entail a detailed estimate for all the error terms involved in deriving an ODE of the form (\ref{Psi.ODE}). Once completed, this proof shows that the spacetime will exist up to a spacelike singularity at some $s = s_{\infty}$ with $r(s_{\infty}) = 0$.
\item
In Section \ref{sub:kasner_geom}, we apply our results regarding $\Psi$ in Proposition \ref{prop:k_ode} to find quantitative estimates on geometric quantities such as $\Omega^2$. This will be crucial in showing quantitative closeness to Kasner spacetimes in Section \ref{section:quantitative}.
\end{itemize}

\subsection{Background on Kasner bounces} \label{sub:kasner_why}

In this section, we briefly explore the important role played by the quantity $\Psi$, and the differing evolutionary dynamics that occur when $|\Psi_i| = |\Psi(s_i)|$ is greater than or less than $1$.
The \blue{two} key equations are the evolution equation (\ref{eq:r_evol_2}) for $-r \dot{r}$ and the Raychaudhuri equation (\ref{eq:raych}). 
It will be useful to rewrite \eqref{eq:raych} as an evolution equation for $\log (\Omega^2 / (- \dot{r}))$, and \blue{express the first term on} the right hand side in terms of the quantity $\Psi$:
\begin{equation} \label{eq:raych_2}
\frac{d}{ds} \log \left( \frac{\Omega^2}{-\dot{r}} \right) = \frac{\dot{r}}{r} \Psi^2 + \frac{r^2}{r \dot{r}} |\tilde{A}|^2 q_0^2 |\phi|^2.
\end{equation}

Supposing for now that the rightmost term is integrable and small, one sees that the value of $\Psi^2$ determines the leading order behavior of $\Omega^2$ in $r$. Letting $\alpha^2 = \inf \Psi^2$ and $\beta^2 = \sup \Psi^2$ in our region of interest, integrating (\ref{eq:raych_2}) gives:
\begin{equation} \label{eq:kasner_upperlower}
C_{lower} r^{\beta^2} \leq \frac{\Omega^2}{- \dot{r}} \leq C_{upper} r^{\alpha^2}.
\end{equation}
For this heuristic discussion, it is only important to note that $C_{lower}$ and $C_{upper}$ are independent or $r$, indeed the important feature will be the powers of $r$.

In light of (\ref{eq:kasner_upperlower}), we turn to the evolution equation for $-r \dot{r}(s)$ written above. Inside a region such as $\mathcal{K}$ where $r$ is small, one expects that the dominant term on the right hand side is the term $\frac{Q^2 \Omega^2}{4 r^2}$, and (\ref{eq:kasner_upperlower}) will yield upper and lower bounds of the form:
\begin{equation}
- C_{lower} r^{\beta^2 - 2} \dot{r} \lesssim - \frac{d}{ds} (- r \dot{r}) \lesssim - C_{upper} r^{\alpha^2 - 2} \dot{r}.
\end{equation}

Integrating this expression, and ignoring the degenerate case where $\alpha$ or $\beta$ are equal to $1$, \blue{we have}
\begin{equation}
| r^{\beta^2-1}(s) - r^{\beta^2 - 1}(s_i) | \lesssim r \dot{r}(s) - r \dot{r}(s_i) \lesssim | r^{\alpha^2 - 1}(s) - r^{\alpha^2 - 1}(s_i) |.
\end{equation}
We can therefore make \blue{two} observations:
\begin{itemize}
\item
If $\beta < 1$, then $|r(s)^{\beta^2 - 1} - r(s_i)^{\beta^2 - 1}|$ is unbounded as $r(s) \to 0$, which suggests that $- r \dot{r}$ is unbounded. In fact, since this suggests $r \dot{r}(s)$ will at some point become positive, we have somehow exited the trapped region.	So something must have gone wrong, either in our a priori assumptions or in assuming that $\beta < 1$. In our context, we shall show the latter issue arises; there must be a `\blue{bounce}' which forces $\beta \geq 1$.
\item
If $\alpha > 1$, then  $|r(s)^{\alpha^2 - 1} - r(s_i)^{\alpha^2 - 1}|$ is bounded by a multiple of $r(s_i)^{\alpha^2 - 1}$, which in our case has order of magnitude $O (e^{- (\alpha - 1) \delta_0 \epsilon^{-2}}) \ll - r \dot{r}(s_i)$. So $- r \dot{r}$ changes little from its value at $s = s_i$.	The insight is that if $\Psi(s_i) > 1$ initially, then there is a hope of closing estimates in a bootstrap argument, such that $- r \dot{r}$ only changes by this extremely small quantity, and thus allows for formation of a spacelike singularity.
\end{itemize}

In what follows, we \blue{formalize these observations in a quantitative manner, in order to} obtain more \emph{quantitative information} on the behavior of spacetime, both \blue{in cases} when \blue{a bounce} occurs, and when it does not. 

\subsection{The bootstraps, preliminary estimates and statement of \texorpdfstring{Proposition~\ref{prop:k_ode}}{Proposition 8.1}} \label{sub:kasner_prelim}

In this section, we state and initiate the proof of Proposition \ref{prop:k_ode}. As well as asserting the eventual formation of an $r=0$ spacelike singularity, the key content of this proposition is that we can control the quantity $\Psi(s)$ via a certain nonlinear ODE -- which, as we show in Section \ref{sub:kasner_geom}, in turn allows us to control the other geometric quantities including $\Omega^2(s)$.


\begin{proposition} \label{prop:k_ode}
Fix some $\eta \in \R$ with $0 <\eta < \min \{ \frac{1}{2} \mathfrak{W}, \frac{1}{4} \}$. Then there exists some $\epsilon_0(\eta) > 0$ depending on $\eta$ as well as the usual parameters $M, \mathbf{e}, \Lambda, m^2, q_0$, such that if both
\begin{equation} \label{eq:condition} \tag{$*$}
0 < \epsilon \leq \epsilon_0(\eta) \hspace{0.5cm} \textbf{and} \hspace{0.5cm} |\Psi_i| \coloneqq |\Psi(s_i)| \geq \eta,
\end{equation}
then the solution of the system (\ref{eq:raych})--(\ref{eq:phi_evol_2}) exists in the interval $s \in (- \infty, s_{\infty})$, with $s_i < s_{\infty} < + \infty$ and $r(s) \to 0$ as $s \to s_{\infty}$. In fact, \blue{for $s \in \mathcal{K} = \{ s \geq s_i: r(s) > 0 \}$, one obtains the following bounds on $- r \dot r(s)$.}
\begin{equation} \label{eq:k_r_lower}
- r \dot{r} (s) \geq 2 |B|^2 \mathfrak{W}^2 r_-^2 \omega_{RN} \eta^2 \epsilon^2, \blue{\quad - r \dot{r} (s) \leq 8 |B|^2 \mathfrak{W}^2 r_-^2 \omega_{RN} \ep^2.}
\end{equation}

Regarding the quantity $\Psi$ defined in (\ref{eq:Psi}), there exists some constant $D_K(\eta, M, \mathbf{e}, \Lambda, m^2, q_0) > ~0$, a real number $\alpha = \alpha(\epsilon)$ and a function $\mathcal{F}(s)$ satisfying:
\begin{equation*}
| \alpha - \Psi_i | \leq D_K \exp( -\delta_0 \epsilon^{-2} ), \hspace{0.5cm} |\mathcal{F}(s)| \leq D_K \exp( -\delta_0 \epsilon^{-2} ) r(s) = D_K r_- e^{- \delta_0 \epsilon^{-2} - R},
\end{equation*}
where we define $R = \log (r_- / r(s))$, so that $\Psi = \Psi(R)$ satisfies the following ODE in the region $\mathcal{K}$:
\begin{equation} \label{eq:k_ode_main}
\frac{d \Psi}{d R} = - \Psi \left( \Psi - \alpha \right) \left( \Psi - \frac{1}{\alpha} \right) + \mathcal{F}.
\end{equation}
Finally, one has the following upper and lower bounds for $|\Psi(s)|$ for $s \in \mathcal{K}$:
\begin{equation} \label{eq:k_Psi_upperlower}
\min\{ |\Psi_i|, |\Psi_i^{-1}| \} - D_K e^{-\delta_0 \epsilon^{-2}} \leq |\Psi(s)|
\leq \max\{ |\Psi_i|, |\Psi_i^{-1}| \} + D_K e^{-\delta_0 \epsilon^{-2}}.
\end{equation}
\end{proposition}

\begin{remark}
We use the same notation $D_K = D_K(\eta, M, \mathbf{e}, \Lambda, m^2, q_0) > 0$ throughout the various lemmas and propositions of Section \ref{sec:kasner}. Note that it is possible to track the dependence of $D_K$, \blue{and any constants implied by the notation $\lesssim$,} on $\eta$, and \blue{thereby} strengthen the result from having a fixed $\eta$ in (\ref{eq:condition}) to allowing $\eta = \eta(\epsilon)$ to decay as $\epsilon \to 0$, for instance $\eta = \epsilon^{0.01}$. However, for the purpose of simpler exposition, we do not pursue that here.
\end{remark}

The proof of Proposition \ref{prop:k_ode} will be broken into several lemmas in Sections \ref{sub:kasner_prelim} and \ref{sub:kasner_main}. We first list the three main bootstrap assumptions of the region $\mathcal{K}$:
\begin{equation} \label{eq:k_bootstrap_Psi} \tag{K1}
|\Psi| \leq 4 \eta^{-1},
\end{equation}
\begin{equation} \label{eq:k_bootstrap_r} \tag{K2}
- r \dot{r}(s) \geq 
|B|^2 \mathfrak{W}^2 r_-^2 \omega_{RN} \eta^2 \epsilon^2,
\end{equation}
\begin{equation} \label{eq:k_bootstrap_Q} \tag{K3}
\frac{|Q_{\infty}|}{2} \leq |Q(s)| \leq 2 |Q_{\infty}|.
\end{equation}
Here, $Q_{\infty} = Q_{\infty}(M, \mathbf{e}, \Lambda) \neq 0$ is defined in (\ref{eq:qinfty}); by Lemma \ref{lem:jo_charge_retention}, $Q_{\infty}$ lies strictly between $\mathbf{e}/2$ and $\mathbf{e}$. In light of Proposition \ref{prop:oscillation+} and Corollary \ref{cor:protokasner}, these bootstrap assumptions hold in a neighborhood of $s = s_i$. 

In the remainder of Section \ref{sub:kasner_prelim}, we state and prove two preliminary lemmas, the first of which provides estimates for the Maxwell quantities $\tilde{A}(s)$ and $Q(s)$, as well as the scalar field $\phi(s)$. The second lemma then produces a crucial lower bound for $|\Psi|$ as well as a useful preliminary upper bound on $\Omega^2(s)$.

\begin{lemma} \label{lem:k_prelim_1}
Assuming the bootstraps (\ref{eq:k_bootstrap_Psi}), (\ref{eq:k_bootstrap_r}), (\ref{eq:k_bootstrap_Q}), we have the following preliminary estimates on $\phi$ as well as the Maxwell quantities $Q$ and $\tilde{A}$:
\begin{equation} \label{eq:k_phi}
|\phi(s)| \leq 2 \eta^{-1} \log \left( \frac{8 |B|^2 \mathfrak{W}^2 r_-^2 \epsilon^2}{r^2(s)} \right),
\end{equation}
\begin{equation} \label{eq:k_Q}
|Q(s) - Q(s_i)| \leq D_K \exp(- \delta_0 \epsilon^{-2}),
\end{equation}
\begin{equation} \label{eq:k_gauge}
| |q_0 \tilde{A}|(s) - \omega_{RN} | \leq D_K \epsilon^2.
\end{equation}
In particular, using (\ref{eq:k_Q}), we immediately improve the bootstrap assumption (\ref{eq:k_bootstrap_Q}).
\end{lemma}

\begin{lemma} \label{lem:k_prelim_2}
Assuming the bootstraps (\ref{eq:k_bootstrap_Psi}), (\ref{eq:k_bootstrap_r}), (\ref{eq:k_bootstrap_Q}), we have the following upper bound on $- r \dot{r}$ and estimate on $r^2\dot{\phi}$ which we use to get a corresponding lower bound on $|\Psi|$:
\begin{equation} \label{eq:k_r_upper}
- r \dot{r} (s) \leq 4 |B|^2 \mathfrak{W}^2 \omega_{RN} r_-^2 \epsilon^2 + D_K \epsilon^4 \log (\epsilon^{-1}),
\end{equation}
\begin{equation} \label{eq:k_phid}
\left| \frac{r^2 \dot{\phi} (s)}{ r^2 \dot{\phi} (s_i) } - 1 \right| \leq D_K \exp( - \delta_0 \epsilon^{-2} ),
\end{equation}
\begin{equation} \label{eq:k_Psi_lower}
|\Psi| \geq \eta - D_K \exp( - \delta_0 \epsilon^{-2}).
\end{equation}
In particular, $\Psi$ will never vanish, and thus never change sign, in the region $\mathcal{K}$. Moreover, we have the following initial upper bound for $\Omega^2(s)$:
\begin{equation} \label{eq:k_lapse_upper}
\frac{\Omega^2}{-\dot{r}} (s) \leq \frac{\Omega^2}{- \dot{r}} (s_i) \leq \exp(- 50 \delta_0 \epsilon^{-2}).
\end{equation}
\end{lemma}

\begin{proof}[Proof of Lemma \ref{lem:k_prelim_1}]
To prove the preliminary bound on $\phi$, \blue{recall} that we may rewrite $\Psi$ as $- r \frac{d \phi}{dr}$. So we use (\ref{eq:k_bootstrap_Psi}) as follows:
\begin{equation} \label{eq:psiPsi1}
|\phi(s) - \phi(s_i)| \leq \int^{r(s_i)}_{r(s)} \frac{|\Psi(r)|}{r} \, dr \leq 4 \eta^{-1} \log (\frac{r(s_i)}{r(s)}) \leq 2 \eta^{-1} \log (\frac{r^2(s_i)}{r^2(s)}).
\end{equation}
But, from \eqref{eq:pk_phi_asymp} in Corollary \ref{cor:protokasner}, 
\begin{equation} \label{eq:psiPsi2}
|\phi(s_i)| \leq \blue{\frac{2}{\pi} C_{YK}(\ep)} \mathfrak{W}^{-1} \log \left( \frac{r_-^2 \epsilon^2}{\xi_K r^2(s_i)} \right) \blue{ + |c_1 C_{JK}(\ep)| + |c_2 C_{YK}(\ep)|} \leq 2 \eta^{-1} \log \left( \frac{8 |B|^2 \mathfrak{W}^2 r_-^2 \epsilon^2}{r^2(s_i)} \right).
\end{equation} 
\blue{In the second step, we used  $|C_{YK}(\ep)|, |C_{JK}(\ep)| \leq \sqrt{\pi} \mathfrak{W}$ from Propositions~\ref{prop:oscillation_phi} and \ref{prop:oscillation+}, as well as $\eta < \frac{1}{2} \mathfrak{W}$ and the definition \eqref{eq:xik}. In particular, $r(s_i) = \red{e^{- \delta_0 \ep^{-2}}}$ is small enough that the contribution of $|c_1 C_{JK}(\ep)| + |c_2 C_{YK}(\ep)|$ is negligible.}
Combining the two inequalities \eqref{eq:psiPsi1} and \eqref{eq:psiPsi2} will clearly yield (\ref{eq:k_phi}).

For the gauge field $\tilde{A}(s)$, we use the following trick: 
\blue{For $s \in \mathcal{K}$, we have $r(s) \leq \exp(-\delta_0 \epsilon^{-2}) r_- \leq \min \left \{ \frac{|Q_{\infty}|}{8}, \frac{1}{2 \sqrt{|\Lambda|}} \right \}$, and $|m^2 r^2 |\phi|^2| < 1$ via \eqref{eq:k_phi}, so
\[
    \frac{\Omega^2}{4} - \frac{\Omega^2 Q^2}{4 r^2} - \frac{\Omega^2 r^2}{4} (m^2|\phi|^2 + \Lambda) \leq - \frac{\Omega^2 Q^2}{8 r^2}.
\]}
Thus, (\ref{eq:r_evol_2}) tells us that $- r \dot{r}(s)$ is decreasing in $\mathcal{K}$, \blue{and one further has a bound on the integral}
\begin{equation} \label{eq:k_rr_int}
\int_{s_i}^s \frac{Q^2 \Omega^2}{4 r^2} (\tilde{s}) \, d\tilde{s} \leq - 2 r \dot{r} (s_i).
\end{equation}

Hence, using the lower bound of (\ref{eq:k_bootstrap_Q}) once more, the equation (\ref{eq:gauge_evol}) yields
\begin{equation}
| q_0 \tilde{A}(s) - q_0 \tilde{A}(s_i) | \leq \int_{s_i}^s \frac{|Q| \Omega^2}{4 r^2} (\tilde{s}) \, d \tilde{s} \leq \frac{2}{|Q_{\infty}|} \int_{s_i}^s \frac{Q^2 \Omega^2}{4 r^2} (\tilde{s}) \, d \tilde{s} \leq \frac{- 4 r \dot{r}(s_i)}{|Q_{\infty}|}.
\end{equation}
\blue{Thus, \eqref{eq:pk_r} and \eqref{eq:pk_gauge} of Proposition \ref{prop:oscillation+} allows one to deduce (\ref{eq:k_gauge}).}

To close the estimate on $Q$, we simply integrate (\ref{eq:Q_evol}):
\begin{equation}
|Q(s) - Q(s_i)| \leq \int^s_{s_i} q_0^2 |\tilde{A}| r^2 |\phi|^2 (\tilde{s}) \, d\tilde{s}  \lesssim \int^s_{s_i} r^2(\tilde{s}) \log^2 \left( \frac{8 |B|^2 \mathfrak{W}^2 r_-^2 \epsilon^2}{r^2(\tilde{s})} \right) \, d\tilde{s}.
\end{equation}
But in the region $\mathcal{K}$, where $r(s) \leq e^{- \delta_0 \epsilon^{-2}} r_-$ and the interval of integration has length $|s- s_i| = O (\epsilon^{-2} e^{- 2 \delta_0 \epsilon^{-2}})$ by (\ref{eq:k_bootstrap_r}), it is straightforward to determine (\ref{eq:k_Q}) and improve the bootstrap (\ref{eq:k_bootstrap_Q}).
\end{proof}

\begin{proof}[Proof of Lemma \ref{lem:k_prelim_2}]
The first estimate (\ref{eq:k_r_upper}) is immediate by monotonicity. Indeed, as \blue{argued in the proof of} Lemma~\ref{lem:k_prelim_1}, the quantity $- r \dot{r}(s)$ is decreasing in $s$. So the estimate (\ref{eq:pk_r}) evaluated at $s= s_i$ yields (\ref{eq:k_r_upper}).

For the lower bound on $|\Psi|$, we first produce the estimate (\ref{eq:k_phid}) regarding $r^2 \dot{\phi}$. Using the monotonicity of $- \Omega^{-2} \dot{r}(s)$ from the Raychaudhuri equation (\ref{eq:raych}), and Proposition \ref{prop:oscillation+}, we first find the preliminary estimate 
\begin{equation}
\frac{\Omega^2}{- \dot{r}} (s) \leq \frac{\Omega^2}{- \dot{r}} (s_i) \leq \frac{\Omega^2 r(s_i)}{-r  \dot{r}(s_i)} \blue{\lesssim}  \exp(- 50 \delta_0 \epsilon^{-2}) \cdot \ep^{-2}  \exp(-  \delta_0 \epsilon^{-2}) \blue{\lesssim}  \exp(- 50 \delta_0 \epsilon^{-2}).
\end{equation}
This is (\ref{eq:k_lapse_upper}), and we use this together with (\ref{eq:k_r_upper}) and the conclusions of Lemma \ref{lem:k_prelim_1} to yield from \eqref{eq:phi_evol_2}
\begin{equation}
\left| \frac{d}{ds} (r^2 \dot{\phi})(s) \right| \lesssim \exp( - \delta_0 \epsilon^{-2}) \cdot  r(s)  \log \left( \frac{8 |B|^2 \mathfrak{W}^2 r_-^2 \epsilon^2}{r^2(s)} \right) \lesssim \epsilon^2 \exp( - \delta_0 \epsilon^{-2}).
\end{equation}

Integrating this up, and using once again that the interval of integration has length bounded by $O(\epsilon^{-2} e^{ - 2\delta_0 \epsilon^{-2}})$, we therefore find that
\begin{equation} \label{eq:k_phiddiff}
|r^2 \dot{\phi} (s) - r^2 \dot{\phi} (s_i)| \lesssim \exp(- 3 \delta_0 \epsilon^{-2}).
\end{equation}
Finally, since $|\Psi(s_i)| \geq \eta$ and $- r \dot{r}(s_i)$ is bounded below using (\ref{eq:jo_r}), we have $|r^2 \dot{\phi} (s_i)| = - r \dot{r}(s_i) \cdot |\Psi(s_i)| \gtrsim \epsilon^2$. Combining this with (\ref{eq:k_phiddiff}) we obtain (\ref{eq:k_phid}).

Therefore, for $\epsilon$ sufficiently small, we see that
\begin{equation}
|\Psi(s)| = \frac{ - r \dot{r} (s_i) }{- r \dot{r} (s)} \cdot \frac{ r^2 \dot{\phi} (s)}{ r^2 \dot{\phi} (s_i)} \cdot |\Psi(s_i)| \geq \eta - D_K \exp(- \delta_0 \epsilon^{-2}).
\end{equation}
Here, we used (\ref{eq:k_phid}) and that $-r \dot{r}(s)$ is decreasing for $s \in \mathcal{K}$, as well as the original assumption (\ref{eq:condition}).
\end{proof}

\subsection{The dynamical system for \texorpdfstring{$\Psi$}{Ψ} and the proof of \texorpdfstring{Proposition~\ref{prop:k_ode}}{Proposition 8.1}} \label{sub:kasner_main}

In this section, we complete the proof of Proposition \ref{prop:k_ode}. The main step will be to find the ODE (\ref{eq:k_ode_main}). For this purpose, we start with a lemma concerning the first and second derivatives of $\Psi$ with respect to the timelike variable $r$.

\begin{lemma} \label{lem:k_first_second}
Assume that the bootstraps (\ref{eq:k_bootstrap_Psi}), (\ref{eq:k_bootstrap_r}), (\ref{eq:k_bootstrap_Q}) hold. \blue{If} $\mathcal{E}$ is defined such that
\begin{equation} \label{eq:k_first_derivative}
\frac{d\Psi}{dr} =
\Psi \frac{1}{-r\dot{r}} \frac{\Omega^2}{-4\dot{r}}
\left( 1 - \frac{Q^2}{r^2} - m^2 r^2 |\phi|^2 -r^2 \Lambda \right)
+
\mathcal{E},
\end{equation}
then we have the following estimates in the region $s \in \mathcal{K}$:
\begin{equation} \label{eq:k_error}
| \mathcal{E} | \leq D_K \exp( - 2 \delta_0 \epsilon^{-2}) \cdot r , \hspace{0.5cm}
\left| \frac{d \mathcal{E}}{d r} \right| \leq D_K \exp( - 2 \delta_0 \epsilon^{-2} ) \cdot r^{-1}.
\end{equation}

For the second derivative, if the error term $\mathcal{F}_1$ is defined such that
\begin{equation} \label{eq:k_ode_1}
\frac{d^2 \Psi}{d r^2} - 2 \Psi^{-1} \left( \frac{ d \Psi } { d r} \right)^2 - \frac{\Psi^2 - 2}{r} \frac{d\Psi}{dr} = \mathcal{F}_1,
\end{equation}
then for $s \in \mathcal{K}$, the expression $\mathcal{F}_1$ satisfies the following bound:
\begin{equation} \label{eq:k_error_1}
\mathcal{F}_1 \leq D_K \exp(- 2\delta_0 \epsilon^{-2}) r^{-1}.
\end{equation}
\end{lemma}

\begin{proof}[Proof of Lemma \ref{lem:k_first_second}]
Assuming the bootstraps, the conclusions of Lemmas \ref{lem:k_prelim_1} and \ref{lem:k_prelim_2} will hold. In light of (\ref{eq:r_evol_2}) and (\ref{eq:phi_evol_2}), differentiating $\Psi$ in the variable $r$ one yields:
\begin{equation} \label{eq:k_first_derivative*}
\frac{d\Psi}{dr} = \frac{1}{\dot{r}} \frac{d \Psi}{ds} = 
\Psi \frac{1}{-r\dot{r}} \frac{\Omega^2}{-4\dot{r}}
\left( 1 - \frac{Q^2}{r^2} - m^2 r^2 |\phi|^2-r^2 \Lambda \right)
+
\mathcal{E},
\end{equation}
where the error $\mathcal{E}$ is given by
\begin{equation} \label{eq:k_error*}
\mathcal{E} \coloneqq \frac{1}{-r\dot{r}} \frac{d}{dr} ( r^2 \dot{\phi} )
=
\frac{1}{-r\dot{r}} \left( \frac{q_0^2 |\tilde{A}|^2 r^3 \phi}{-r \dot{r}} + \frac{m^2 \Omega^2 r^2 \phi }{-4 \dot{r}} \right).
\end{equation}

We seek the estimates (\ref{eq:k_error}). Using Lemmas \ref{lem:k_prelim_1} and \ref{lem:k_prelim_2} \blue{together with \eqref{eq:k_bootstrap_r},} it is straightforward to get
\begin{equation} \label{eq:k_error_est1}
|\mathcal{E}| \lesssim \epsilon^{-4} r^3 \log \left( \frac{\epsilon^2}{\xi^* r^2} \right) + \exp(- 2 \delta_0 \epsilon^{-2}) \epsilon^{-2} r^2 \log \left( \frac{\epsilon^2}{\xi^* r^2} \right) \lesssim \exp(- 2 \delta_0 \epsilon^{-2}) r.
\end{equation}
For the derivative estimate in (\ref{eq:k_error}), we simply differentiate the expression (\ref{eq:k_error*}) term by term. This is rather cumbersome, and we simplify the exposition by merely considering the new multiplicative factors that arise when differentiating the various terms, noting that it is most crucial to keep track of the additional powers of $r^{-1}$ that arise:

\begin{enumerate}[(i)]
\item \label{k_i}
Differentiating $(- r \dot{r})^{-1}$ in $r$ will yield an extra multiplicative factor of
\begin{equation*}
\frac{\Omega^2}{- 4\dot{r}} \frac{1}{-r\dot{r}} \left( \frac{Q^2}{r^2} + m^2 r^2 |\phi|^2+ r^2 \Lambda - 1 \right).
\end{equation*}
\blue{Since $Q \neq 0$,} this will contribute a multiplicative factor of $r^{-2}$. Note that though one could fear that additional powers of $\epsilon^{-1}$ will appear (e.g.\ in the $-r \dot{r}$ appearing in the denominator), these will be negated by the smallness of $\frac{\Omega^2}{- \dot{r}}$, see (\ref{eq:k_lapse_upper}).
\item
Differentiating $\tilde{A}$ in $r$ will introduce a similar factor of
\begin{equation*}
\frac{Q \Omega^2}{- 4 r^2 \dot{r}}
\end{equation*}
which we treat in exactly the same way as (\ref{k_i}).
\item
Differentiating $\phi$ in $r$ simply yields
\begin{equation*}
\frac{d \phi}{dr} = - \Psi r^{-1},
\end{equation*}
so that in light of bootstrap (\ref{eq:k_bootstrap_Psi}) providing an upper bound on $|\Psi|$, differentiating this term contributes only one power of $r^{-1}$.
\item
Differentiating the $\frac{\Omega^2}{4\dot{r}}$ in the second term of (\ref{eq:k_error*}) introduces, by the Raychaudhuri equation (\ref{eq:raych_2}), an extra multiplicative factor of
\begin{equation*}
\frac{r}{\dot{r}^2} ( |\dot{\phi}|^2 + q_0^2 |\tilde{A}|^2 \phi^2 ) = \frac{\Psi^2}{r} + \frac{r^3 q_0^2 |\tilde{A}|^2 \phi^2}{ (- r \dot{r})^2} \leq \frac{\Psi^2}{r} + 1,
\end{equation*}
where we used the bootstrap (\ref{eq:k_bootstrap_r}) and Lemma \ref{lem:k_prelim_1}. So differentiating this term contributes at worst one power of $r^{-1}$.
\item \label{k_v}
Finally, differentiating any powers of $r$ that arise in (\ref{eq:k_error*}) will trivially only lose one power of $r$.
\end{enumerate}
From (\ref{k_i})--(\ref{k_v}) we get the estimate (\ref{eq:k_error}).

\blue{Due to \eqref{eq:k_error}, \eqref{eq:k_phi}, (\ref{eq:k_lapse_upper}) and the bootstrap assumptions, \eqref{eq:k_first_derivative} yields the upper bound}
\begin{equation} \label{eq:k_psi_r_weak}
\left| \frac{d \Psi}{d r} \right| \leq D_K \exp(-50 \delta_0 \epsilon^{-2}) r^{-2} + D_K \exp(- 2 \delta_0 \epsilon^{-2}) r \leq D_K \exp( - 5 \delta_0 \epsilon^{-2}) r^{-2}.
\end{equation}

We now treat the second derivative. We differentiate the expression (\ref{eq:k_first_derivative}) again in $r$, and find
\begin{align} \label{eq:k_second_derivative}
\frac{d^2 \Psi}{dr^2} =\ & 
\Psi^{-1} \left( 2 \frac{d\Psi}{dr} - \mathcal{E} \right) \left( \frac{d \Psi}{dr} - \mathcal{E}  \right)
+ \left( \frac{d\Psi}{dr} - \mathcal{E} \right) \left( \frac{\Psi^2}{r} + \frac{r}{\dot{r}^2} q_0^2 |\tilde{A}|^2 \phi^2 \right) \\[0.5em] \nonumber &
+ \left( \frac{d\Psi}{dr} - \mathcal{E} \right) \frac{d}{dr} \log \left | 1 - \frac{Q^2}{r^2} - m^2 r^2 \phi ^2  - r^2 \Lambda\right | + \frac{d \mathcal{E}}{dr}.
\end{align}
To explain the derivation of the equation (\ref{eq:k_second_derivative}), we take (\ref{eq:k_first_derivative}), subtract $\mathcal{E}$ from both sides, take a logarithm, and \blue{then finally differentiate}, to write 
\begin{equation*}
\frac{d}{dr} \log \left| \frac{d\Psi}{dr} - \mathcal{E} \right| = \frac{d}{dr} \left( \log |\Psi| - \log (- r \dot{r}) + \log \left( \frac{\Omega^2}{- 4 \dot{r}} \right) + \log \left| 1 - \frac{Q^2}{r^2} - m^2 r^2 \phi^2 - r^2 \Lambda \right| \right ),
\end{equation*}
from which (\ref{eq:k_second_derivative}) follows after using the usual evolution equations (\ref{eq:r_evol_2}) and (\ref{eq:raych_2}).

We now wish to recover the ODE (\ref{eq:k_ode_1}) along with (\ref{eq:k_error_1}) i.e.\ we just need to show that all error terms present within (\ref{eq:k_second_derivative}) are bounded by $\exp(- \delta_0 \epsilon^{-2}) r^{-1}$. This is mostly straightforward by (\ref{eq:k_error}) and (\ref{eq:k_psi_r_weak}), with the most complicated term being the one involving the $\log$.

Firstly, using (\ref{eq:k_psi_r_weak}) and (\ref{eq:k_error}) as well as the lower bound (\ref{eq:k_Psi_lower}) for $|\Psi|$, we see that
\begin{equation} \label{eq:k_ode_part1}
\left|
\Psi^{-1} \left( 2 \frac{d \Psi}{dr} - \mathcal{E} \right) \left( \frac{d \Psi}{dr} - \mathcal{E} \right) - 2 \Psi^{-1} \left( \frac{d \Psi}{dr} \right)^2 
\right|
\lesssim \exp( - 2\delta_0 \epsilon^{-2} ) r^{-1}.
\end{equation}
Next, note that by (\ref{eq:k_r_upper}), (\ref{eq:k_phi}) and (\ref{eq:k_gauge}), the expression involving $\phi^2$ can be bounded by
\begin{equation*}
\frac{r}{\dot{r}^2} q_0^2 |\tilde{A}|^2 \phi^2 \lesssim \ep^{-4} r^3 \log(r^{-1})\leq  \exp(- \delta_0 \epsilon^{-2}) r^2 \leq  \exp(-2 \delta_0 \epsilon^{-2}) r,
\end{equation*}
from which we may also deduce
\begin{equation} \label{eq:k_ode_part2}
\left|
\left( \frac{d\Psi}{dr} - \mathcal{E} \right) \left( \frac{\Psi^2}{r} + \frac{r}{\dot{r}^2} q_0^2 |\tilde{A}|^2 \phi^2 \right) 
- \frac{\Psi^2}{r} \frac{d \Psi}{dr} 
\right|
\lesssim \exp(- 2\delta_0 \epsilon^{-2}) r^{-1}.
\end{equation}

We now move onto the more tedious term involving $\frac{d}{dr} \log(\cdots)$. Computing the derivative using the system of equations (\ref{eq:raych})--(\ref{eq:phi_evol_2}) one eventually finds
\begin{align*}
\frac{d}{dr} \log \left| 1 - \frac{Q^2}{r^2} - m^2 r^2 \phi^2 -r^2 \Lambda\right| 
&= 
- \frac{2}{r} \frac{1 - f}{1 - g},
\end{align*}
where the expressions $f$ and $g$ are given by
\begin{equation*}
f = \frac{r^3}{Q^2} \left ( - \frac{Q \tilde{A} q_0^2 \phi^2}{ - \dot{r}} + m^2 r \phi^2 - m^2 r \phi \Psi + r \Lambda \right ), 
\end{equation*}
\begin{equation*}
g = \frac{r^2}{Q^2} \cdot ( 1 - m^2 r^2 \phi^2 -r^2 \Lambda).
\end{equation*}

Using the bootstraps along with Lemmas \ref{lem:k_prelim_1} and \ref{lem:k_prelim_2}, it is straightforward to deduce that $|f| + |g| \lesssim r^2$.
The conclusion of this computation is therefore that
\begin{equation} \label{eq:k_ode_aux}
\left|
\frac{d}{dr} \log \left | 1 - \frac{Q^2}{r^2} - m^2 r^2 \phi^2 \right | + \frac{2}{r} 
\right|
\lesssim r,
\end{equation}
which along with (\ref{eq:k_psi_r_weak}) and (\ref{eq:k_error}) finally yields the required estimate
\begin{equation} \label{eq:k_ode_part3}
\left|
\left( \frac{d \Psi}{dr} - \mathcal{E} \right) \frac{d}{dr} \log \left | 1 - \frac{Q^2}{r^2} - m^2 r^2 \phi^2 \right |
+ \frac{2}{r} \frac{d \Psi}{dr} 
\right|
\lesssim \exp(-2 \delta_0 \epsilon^{-2}) r^{-1}.
\end{equation}

The desired equation (\ref{eq:k_ode_1}) is then found by combining the identity (\ref{eq:k_second_derivative}) with the estimates (\ref{eq:k_ode_part1}), (\ref{eq:k_ode_part2}), (\ref{eq:k_ode_part3}) and (\ref{eq:k_error}) to get the required error term $\mathcal{F}_1$ satisfying (\ref{eq:k_error_1}). 
\end{proof}

We have now finished all the preparation for the proof of Proposition \ref{prop:k_ode}. Before turning to the actual proof below, we first give a brief sketch for the \blue{benefit} of the reader. After changing variables to $R = \log (r_- / r(s))$, we may integrate up the second-order ODE (\ref{eq:k_ode_1}) to derive a first-order ODE of a similar form to (\ref{eq:k_ode_main}). More precisely, one finds an equation of the form
\begin{equation} \label{eq:k_ode_main_sketch}
- r \frac{ d \Psi}{dr} = \frac{d \Psi}{d R} = - \Psi (\Psi^2 - K \Psi + 1) + \text{error},
\end{equation}
where the expression $K$ appearing here is a constant of integration. By evaluating (\ref{eq:k_ode_main_sketch}) at $s = s_i$ we find that $K \approx \Psi_i + \Psi_i^{-1}$.

We treat (\ref{eq:k_ode_main_sketch}) as a one-dimensional dynamical system for the unknown $\Psi = \Psi(R)$, and consider what happens as $r \to 0$, i.e.\ $R \to + \infty$. Assuming the error to be negligible, the dynamical system (\ref{eq:k_ode_main_sketch}) will prohibit $\Psi$ from growing too large -- in particular allowing us to improve the bootstrap (\ref{eq:k_bootstrap_Psi}). This, in turn, will allow us to improve the bootstrap (\ref{eq:k_bootstrap_r}) due to the definition of $\Psi$ and the upper bound on $r^2 \dot{\phi}$ from (\ref{eq:k_phid}). Recall also that bootstrap \eqref{eq:k_bootstrap_Q} has already been improved in \blue{Lemma~\ref{lem:k_prelim_1}}.

Having improved all the bootstraps, the lower bound (\ref{eq:k_bootstrap_r}) allows us to continue the solution all the way to some $s = s_{\infty}$ with $r(s_{\infty}) = 0$. To obtain (\ref{eq:k_ode_main}), we again integrate up the second-order ODE (\ref{eq:k_ode_1}), but by determining the constant of integration $K$ teleologically at $R = + \infty$ (i.e.\ $r\rightarrow0$), we get the precise bound for the error term $\mathcal{F}$ as in Proposition \ref{prop:k_ode}.

Finally, some soft arguments using known upper and lower bounds for $|\Psi|$ will allow us to determine that $K \geq 2$, so that, writing $K = \alpha + \alpha^{-1}$, we get the required ODE (\ref{eq:k_ode_main}). The remaining assertions are then straightforward. We now make this argument precise.

\begin{proof}[Proof of Proposition \ref{prop:k_ode}]
Performing the change of variables $R = \log ( r_- / r(s)) = \log r_- - \log r(s)$ on (\ref{eq:k_ode_1}), one finds the ODE
\begin{equation} \label{eq:k_ode_3}
\frac{d^2 \Psi}{d R^2} - 2 \Psi^{-1} \left( \frac{d \Psi}{d R} \right)^2 + (\Psi^2 - 1) \frac{d \Psi}{d R} 
= r^2 \mathcal{F}_1.
\end{equation}
Multiplying by the integrating factor $\Psi^{-2}$, we write the left hand side as a total derivative:
\begin{equation} \label{eq:k_ode_4}
\frac{d}{dR} \left( \Psi^{-2} \frac{d \Psi}{dR} + \Psi + \frac{1}{\Psi} \right)
= \Psi^{-2} r^2 \mathcal{F}_1.
\end{equation}

Due to (\ref{eq:k_Psi_lower}) and \blue{(\ref{eq:k_error_1})}, the right hand side of (\ref{eq:k_ode_4}) can be bounded by $ 4 \, D_K \eta^{-2} \exp( - 2 \delta_0 \epsilon^{-2}) r \leq \exp( - \delta_0 \epsilon^{-2}) e^{-R}$. So this error is integrable as $R \to+ \infty$, and we proceed by integrating up (\ref{eq:k_ode_4}). 


For now, we can only integrate (\ref{eq:k_ode_4}) in a finite bootstrap region; for $R_0 > R_i := \log (r_- / r(s_i)) = \delta_0 \epsilon^{-2}$ lying in our bootstrap region, there exists a constant of integration $K_{R_0}$ and an error term $\mathcal{F}_{R_0}(R)$, such that, for $R \in [R_i, R_0]$, the following holds:
\begin{equation} \label{eq:k_ode_5}
\Psi^{-2} \frac{d \Psi}{d R} + \Psi - K_{R_0} + \frac{1}{\Psi} = \mathcal{F}_{R_0}(R).
\end{equation}
The choice of $K_{R_0}$ is made such that $\mathcal{F}_{R_0}(R_0) = 0$ and hence from the aforementioned bound on $\mathcal{F}_1$, one has $\mathcal{F}_{R_0}(R) \leq \exp( -\delta_0 \epsilon^{-2}) e^{-R}$. We also rewrite the above in the form
\begin{equation} \label{eq:k_ode_6}
\frac{d \Psi}{d R} = - \Psi ( \Psi^2 - K_{R_0} \Psi + 1 ) + \Psi^2 \mathcal{F}_{R_0}.
\end{equation}

In order to proceed, we must estimate the constant of integration $K_{R_0}$. For this purpose, we evaluate (\ref{eq:k_ode_5}) at $R = R_i$, and then apply Proposition \ref{prop:oscillation+} (\eqref{eq:pk_lapse} specifically). This proposition, together with (\ref{eq:k_first_derivative}), will give
\begin{equation*}
\left| \Psi^{-2} \frac{d \Psi}{d R} (R_i) \right | \leq \exp ( -2 \delta_0 \epsilon^{-2}),
\end{equation*}
so that, when evaluating (\ref{eq:k_ode_5}) at $R=R_i$, one finds the estimate
\begin{equation} \label{eq:k_ode_const}
\left| K_{R_0} - \Psi_i - \frac{1}{\Psi_i} \right| \leq 2 \exp( -2 \delta_0 \epsilon^{-2} ).
\end{equation}

In particular, for $\epsilon$ chosen sufficiently small one can get the key upper bound $|K_{R_0}| \leq \frac{5}{3} \eta^{-1}$. Here we used the estimate (\ref{eq:pk1_Psi}) at $s=s_i$ in Corollary \ref{cor:protokasner} and that \blue{we chose} $\eta \leq \min \{\frac{1}{2} \mathfrak{W}, \frac{1}{4} \}$, as well as $\pi^{-\frac{1}{2}}< \frac{2}{3}$.
We can use this to improve the bootstraps (\ref{eq:k_bootstrap_Psi}) and (\ref{eq:k_bootstrap_r}). We first improve the upper bound (\ref{eq:k_bootstrap_Psi}) on $|\Psi|$; we prove $|\Psi| \leq \frac{5}{3} \eta^{-1}$. Without loss of generality, suppose that $\Psi$, and hence $K_{R_0}$, are positive for $s \in \mathcal{K}$, and suppose for contradiction that $\sup_{s \in \mathcal{K}} \Psi > \frac{5}{3} \eta^{-1}$. 

So we may choose $R_1$ to be $R_1 = \inf \{ R > R_i: \Psi(R) = \frac{5}{3} \eta^{-1} \}$. This trivially implies that $\frac{d \Psi}{d R} (R_1) \geq 0$. However, looking at (\ref{eq:k_ode_6}) for any $R_0 \geq R_1$ we get
\begin{equation*}
\frac{d \Psi}{d R}(R_1) = - \left( \tfrac{5}{3} \eta^{-1} \right)^2 \left( \tfrac{5}{3} \eta^{-1} - K_{R_0} \right) - \tfrac{5}{3} \eta^{-1} + \Psi^2 \mathcal{F}_{R_0} (R_1) < 0,
\end{equation*}
where the final step follows from $K_{R_0} \leq \frac{5}{3} \eta^{-1}$ and $ \Psi^2 |\mathcal{F}_{R_0}(R_1)| \leq \Psi^2 \exp( - 2 \delta_0 \epsilon^{-2} ) \leq \eta^{-1}$ for $\epsilon$ small. This is a contradiction, and thus ensures that $\sup_{s \in \mathcal{K}} |\Psi(s)| \leq \frac{5}{3} \eta^{-1}$. This improves the bootstrap (\ref{eq:k_bootstrap_Psi}).

For the remaining bootstrap (\ref{eq:k_bootstrap_r}), we combine the above with the estimate (\ref{eq:k_phid}). To be precise, we have
\begin{equation}
- r \dot{r} (s) = - r \dot{r} (s_i) \cdot \frac{ \Psi(s_i) }{ \Psi(s)} \cdot \frac{ r^2 \dot{\phi}(s) }{ r^2 \dot{\phi}(s_i) } \geq \frac{3}{5} \eta^2 \cdot ( -r \dot{r} (s_i)) \cdot \frac{ r^2 \dot{\phi}(s)}{r^2 \dot{\phi}(s_i)}.
\end{equation}
Thus, upon using Proposition \ref{prop:oscillation+} and (\ref{eq:k_phid}), one finds \blue{(with $D_K$ made larger if necessary)} the lower bound 
\begin{equation} \label{eq:k_r_eventuallb}
- r \dot{r} (s) \geq \frac{3}{5} \eta^2 \cdot [ 4 |B|^2 \mathfrak{W}^2 \omega_{RN} r_-^2 \epsilon^2 - D_K \epsilon^4 \log(\epsilon^{-1}) ],
\end{equation}
which clearly improves (\ref{eq:k_bootstrap_r}) for $\epsilon$ sufficiently small. So the bootstrap argument is complete, and in light of (\ref{eq:k_r_eventuallb}), we conclude that the spacetime extends all the way to $r = 0$, i.e.\ $R = + \infty$.

To find the final one-dimensional dynamical system (\ref{eq:k_ode_main}) we require a argument involving taking limits. Consider the identity (\ref{eq:k_ode_5}); another consequence is that for $R_0 > R_1 \geq R_i$, 
\begin{equation*}
|K_{R_0} - K_{R_1}| = | \mathcal{F}_{R_0}(R_1) | \leq e^{- \delta_0 \epsilon^{-2}} e^{-R_1}.
\end{equation*}
In particular, if $(R_n)_{n \in \N}$ is a sequence with $R_n \to + \infty$, then the sequence $(K_{R_n})_{n \in \N}$ is Cauchy, and the limit is independent of the sequence taken. So there exists some $K \in \R$ such that $K_R \to K$ as $R \to + \infty$, which by (\ref{eq:k_ode_const}) also satisfies
\begin{equation} \label{eq:k_constant}
|K - \Psi_i - \Psi_i^{-1}| \leq 2 \exp(- 2 \delta_0 \epsilon^{-2}).
\end{equation}

Hence we may take the limit in equation (\ref{eq:k_ode_6}) where we fix $R$ and take $R_0 \to + \infty$. We find that there will exist some function $\mathcal{F}$ with $\Psi^2 \mathcal{F}_{R_0}(R) \to \mathcal{F}(R)$ as $R_0 \to + \infty$, also satisfying $|\mathcal{F}| \leq D_K e^{- \delta_0 \epsilon^{-2}} e^{-R}$, such that one has
\begin{equation} \label{eq:k_ode_7}
\frac{d \Psi}{d R} = - \Psi ( \Psi^2 - K \Psi + 1) + \mathcal{F}.
\end{equation}

By construction, $K$ has the same sign as $\Psi(s)$ in $\mathcal{K}$. We argue now that $|K| \geq 2$. Suppose otherwise that $|K| < 2$, which implies that $\Psi^2 - K \Psi + 1$ is bounded below by a positive constant $\beta$. Then (\ref{eq:k_ode_7}) implies that
\begin{equation*}
\frac{d |\Psi|}{d R} \leq - \beta |\Psi| + |\mathcal{F}|.
\end{equation*}
We then apply Gr\"onwall's inequality, finding that
\begin{equation*}
|\Psi(R)| \leq e^{- \beta(R - R^*)} |\Psi(R^*)| + \int^{R}_{R^*} e^{- \beta(R - \tilde{R})} | \mathcal{F} (\tilde{R}) | \, d \tilde{R} \xrightarrow{R \to + \infty} 0.
\end{equation*}
But this is a contradiction to the lower bound (\ref{eq:k_Psi_lower})! Hence $|K| \geq 2$, and  we may write $K = \alpha + \alpha^{-1}$ for some $\alpha \in \R$. By $|K| \leq \frac{5}{3} \eta^{-1}$, it is clear that $|\alpha| \leq \frac{5}{3} \eta^{-1}$ also. From (\ref{eq:k_constant}), we can deduce that
\begin{equation*}
| ( \alpha - \Psi_i ) (\alpha - \Psi_i^{-1}) | \lesssim \exp( - 2 \delta_0 \epsilon^{-2}) \eta^{-1} \lesssim \exp( - 2 \delta_0 \epsilon^{-2}),
\end{equation*}
so, \blue{interchanging $\alpha \leftrightarrow \alpha^{-1}$ if necessary,} we may choose $\alpha$ accordingly such that $|\alpha - \Psi_i| \lesssim \exp(- \delta_0 \epsilon^{-2})$ as required. Therefore, we have arrived at the ODE (\ref{eq:k_ode_main}).

For the final statement in Proposition \ref{prop:k_ode}, note that the upper bound follows from rather straightforward analysis of this ODE, while the lower bound arises from a straightforward modification from the proof of (\ref{eq:k_Psi_lower}) in Lemma \ref{lem:k_prelim_2}.
\end{proof}

\subsection{Geometric features of the region \texorpdfstring{$\mathcal{K}$}{K} in the bounce case \texorpdfstring{$|\Psi_i|<1$}{|Ψ|<1}} \label{sub:kasner_geom}

In this section, we make use of Proposition \ref{prop:k_ode} to derive more quantitative information regarding the quantities $r(s)$ and $\Omega^2(s)$ in the region $s \in \mathcal{K}$, focusing for now on the interesting case where $|\Psi_i| < 1$ and there is a Kasner \blue{bounce}. In particular, we will estimate the value of $r(s)$ where the \blue{bounce} occurs, and bound $\Omega^2(s)$ in such a way that we can infer quantitative closeness to Kasner-like spacetimes before and after the \blue{bounce}.

To make precise statements about the \blue{closeness} of these regions to Kasner-like geometries, we will have to assume further that $\Psi_i$, and therefore the quantities $\alpha$ and $\alpha^{-1}$ of Proposition \ref{prop:k_ode}, are bounded strictly away from $1$ in absolute value, where several important quantities will begin to degenerate\footnote{Assuming the Kasner correspondence of Section \ref{cosmo.intro}, having $|\alpha| = 1$ would imply a spacetime of Kasner exponents $0$, $1/2$ and $1/2$, which already begins to display degenerate features in the BKL picture, see Section~\ref{cosmo.intro}.}.

So, for this section, we strengthen the condition (\ref{eq:condition}) to the following assumption on $\Psi_i$, any $0< \sigma< 1/4$:
\begin{equation} \tag{$**$} \label{inversion_assumption}
\eta \leq |\Psi_i| \leq 1 - \sigma.
\end{equation}
In light of Corollary \ref{cor:condition_eta}, or more precisely the remark following it, the assumption (\ref{inversion_assumption}) is not vacuous; for $\eta$ sufficiently small and any choice of $\sigma \in \left( 0, \frac{1}{4}\right)$, there are certainly arbitrarily small $\epsilon$ such that (\ref{inversion_assumption}) is satisfied, and in fact the measure of this set of $\epsilon$ is controlled, as claimed in Theorem~\ref{maintheorem2}.

Assuming (\ref{inversion_assumption}), we make precise the region of spacetime where the \blue{bounce} occurs in the following lemma.

\begin{lemma} \label{lem:inv_interval}
Given $n \in \mathbb{N}$ with $n \geq 2$, there exists $\epsilon_0(n, \eta, \sigma) > 0$ such that if $0 < |\epsilon| < \epsilon_0(n, \eta, \sigma)$ \textbf{and} the assumption (\ref{inversion_assumption}) holds, then for any $z > 0$ such that $z \in [|\alpha| + \epsilon^n, |\alpha|^{-1} - \epsilon^n]$, there exists a unique $s_z \in \mathcal{K}$ such that $|\Psi(s_z)| = z$. For this domain of $z$, the function $z \mapsto s_z$ is increasing, smooth and invertible, we may define the inversion interval $\mathcal{K}_{\blue{bo}}^n$ to be  $ \{ s_{in}(\epsilon) \leq s \leq s_{out}(\epsilon) \}$ with $|\Psi(s_{in})| = (|\alpha| + \epsilon^n)$ and $|\Psi(s_{out})| = (|\alpha|^{-1} - \epsilon^n)$. 

Moreover there exists a constant $D_I(M ,\mathbf{e}, \Lambda, m^2, q_0, \eta, \sigma, n) > 0$ depending on $\eta$ and $\sigma$ as well as the usual parameters $M, \mathbf{e}, \Lambda, m^2, q_0$ such that we have the following for all $s \in \mathcal{K}_{\blue{bo}}^n$:
\begin{equation} \label{eq:inv_interval}
\left| \epsilon^2 \cdot \log \frac{r_-}{r(s)} - \frac{b_-^{-2}}{2 (1 - \alpha^2)} \right| \leq D_I \epsilon^2 \log(\epsilon^{-1}).
\end{equation}
\end{lemma}

\begin{remark}
The idea is that Lemma \ref{lem:inv_interval} identifies precisely the region where the quantity $\Psi$ transitions from having absolute value smaller than $1$ to having absolute value greater than $1$, and that this transition occurs entirely within a region where $\log (\frac{r_-}{ r(s)})$ is of size $\sim \epsilon^{-2}$, but the  $\log (\frac{r_-}{ r(s)})$-difference within the region is only $O(\log(\ep^{-1}))$.
\end{remark}

In what follows, we will, for the most part, take $n=2$ and define $\mathcal{K}_{\blue{bo}}:= \mathcal{K}_{\blue{bo}}^{n=2}$. 

\begin{proof}
We shall prove Lemma \ref{lem:inv_interval} in \blue{three} steps:
\begin{itemize}
\item
First, we show using equations (\ref{eq:raych_2}) and (\ref{eq:r_evol_2}) that when $\log(\frac{r_- } {r(s)}) \leq \frac{1}{2} b_-^{-2} \epsilon^{-2} (1 - \alpha^{2})^{-1} - O(\log (\epsilon^{-1}))$, we still have $- r \dot{r} (s) \approx - r \dot{r} (s_i)$ and $\Psi(s) \approx \Psi_i \approx \alpha$.
\item
Next we show using the same equations, that conversely, once we proceed to $\log(\frac{r_- } {r(s)}) \geq \frac{1}{2} b_-^{-2} \epsilon^{-2} (1 - \alpha^2)^{-1} + O(\log (\epsilon^{-1}))$, we must have that $\Psi(s) \geq |\alpha| + \epsilon^n$, hence identifying an $s = s_{in}(\epsilon)$ which which satisfies (\ref{eq:inv_interval}).
\item
Now applying Proposition \ref{prop:k_ode}, particularly the ODE (\ref{eq:k_ode_main}), we \blue{deduce that $\frac{d \Psi}{d R}$ \red{is positive} and bounded strictly away from $0$ for $|\Psi| \in [ |\alpha| + \blue{\epsilon^n}, |\alpha^{-1}| - \blue{\epsilon^n} ]$. \red{Using this, we then show that one proceeds} from $|\Psi| =|\alpha| + \epsilon^n$ to $|\Psi| = |\alpha|^{-1} - \epsilon^n$ in an $R$-interval of length $O(\log (\epsilon^{-1}))$, thus proving (\ref{eq:inv_interval}) for $s \in \mathcal{K}_{\blue{bo}}$, as well as the remaining assertions of the lemma.}
\end{itemize}

For ease of notation we suppose, without loss of generality, that $\Psi, \alpha > 0$ in this proof.

By (\ref{eq:k_Psi_upperlower}), we know that, $\Psi \geq \alpha - \epsilon^2$ in $\mathcal{K}$. Therefore, one has from (\ref{eq:raych_2}) that
\begin{equation*}
\frac{d}{ds} \log \left[ \frac{\Omega^2(s)}{- \dot{r}} \left( \frac{r(s)}{r_-} \right)^{-(\alpha - \epsilon^2)^2} \right] \leq \frac{d}{ds} \log \left( \frac{\Omega^2(s)}{- \dot{r}(s)} \right) - \frac{\dot{r}}{r} \Psi^2 \leq 0.
\end{equation*}
So  $\frac{\Omega^2}{- \dot{r}}(s) \left( \frac{r(s)}{r_-} \right)^{- (\alpha - \epsilon^2)^2}$ is decreasing. But by Corollary \ref{cor:protokasner}, specifically \eqref{eq:pk1_lapse}, and $|\alpha - \Psi_i| \lesssim \epsilon^2 \log (\epsilon^{-1})$, \blue{as well as the fact that $\frac{1}{|\dot{r}(s_i)|} \lesssim r(s_i) \cdot \epsilon^{-2} \lesssim \exp \left( 2 \log(\epsilon^{-1})\right)$}, we have
\begin{equation}\label{monot}
\frac{\Omega^2}{- \dot{r}}(s) \left( \frac{r(s)}{r_-} \right)^{- (\alpha - \epsilon^2)^2} \leq \frac{\Omega^2}{- \dot{r}}(s_i) \left( \frac{r(s_i)}{r_-} \right)^{- (\alpha - \epsilon^2)^2} \leq \exp \left( - \frac{1}{2} b_-^{-2} \epsilon^{-2} \right) \cdot \exp \left( D_K \log(\epsilon^{-1}) \right).
\end{equation}  Furthermore we used that, because $\Psi_i^2 - (\alpha-\ep^2)^2 = O(\ep^2)$, one has $(\frac{r(s_i)}{r_-})^{\Psi_i^2 - (\alpha-\ep^2)^2} \lesssim e^{O(\epsilon^{-2}) \cdot O(\epsilon^2)} \lesssim 1$.

We now use this upper bound when integrating the equation (\ref{eq:r_evol_2}). Due to Lemma \ref{lem:k_prelim_1}, we can use the following estimate for the integral of the right hand side of (\ref{eq:r_evol_2}):
\begin{align*}
|- r \dot{r}(s) + r \dot{r}(s_i)|
&\leq \int^s_{s_i} \frac{Q_{\infty}^2 \Omega^2(\tilde{s})}{2 r^2(\tilde{s})} \, d \tilde{s}+ \int^s_{s_i} \frac{ \Omega^2(\tilde{s}) r^2(\tilde{s})}{4}  (|\Lambda| + m^2 |\phi|^2 ) \, d \tilde{s} \blue{ + \int^s_{s_i} \frac{\Omega^2(\tilde{s})}{4} \, d \tilde{s} } \\[0.5em]
&\lesssim \int^{r(s_i)}_{r(s)} \frac{\Omega^2}{- \dot{r}} \frac{1}{r^2} \, dr \\[0.5em]
&\lesssim \exp \left( - \frac{1}{2} b_-^{-2} \epsilon^{-2} \right) \cdot \exp( D_K \log(\epsilon^{-1}) ) \int^{r(s_i)}_{r(s)} \cdot \left( \frac{r}{r_-} \right)^{(\alpha - \epsilon^2)^2 - 2} \, dr \\[0.5em]
&\lesssim \left( \frac{r_-}{r(s)} \right)^{1 - (\alpha - \epsilon^2)^2} \cdot \exp \left( - \frac{1}{2} b_-^{-2} \epsilon^{-2} \right) \cdot \exp ( D_K \log(\epsilon^{-1})).
\end{align*}

Here, the last step follows as $1 - (\alpha - \epsilon^2)^2 \geq \sigma > 0$ for sufficiently small $\epsilon$ depending on $\sigma$. If $\frac{r(s)}{r_-} \geq~ \exp \left( (1 - \alpha^2)^{-1}  \left(- \frac{1}{2} b_-^{-2} \epsilon^{-2} + (D_K + 6n) \log (\epsilon^{-1})\right) \right)$, then one finds 
\begin{align*}
\left( \frac{r_-}{r(s)} \right)^{1 - (\alpha - \epsilon^2)^2} 
&\leq \exp \left( - \frac{1 - (\alpha - \epsilon^2)^2}{ 1 - \alpha^2 } \left( - \frac{1}{2}b_-^{-2} \epsilon^{-2} + (D_K + 6n) \log(\epsilon^{-1}) \right) \right) \\[0.5em]
&\lesssim \epsilon^{6n} \cdot \exp \left( \frac{1}{2}b_-^{-2} \epsilon^{-2} \right) \cdot \exp( - D_K \log (\epsilon^{-1}) ).
\end{align*}
Putting this in the above one sees that $|- r \dot{r}(s) + r \dot{r}(s_i)| \lesssim \epsilon^{6n}$, or 
\begin{equation*}
\left| \frac{- r \dot{r} (s)}{-r \dot{r}(s_i)} - 1 \right | \lesssim \epsilon^{6n - 2}.
\end{equation*}
But $- r \dot{r}(s)$ changing little from its initial value means that $\Psi(s)$ also changes little from its initial value; to see this use also (\ref{eq:k_phid}), which combined with the above implies $|\Psi(s) - \Psi_i| \lesssim \epsilon^{6n - 2}$. As $|\alpha - \Psi_i| \lesssim e^{- \delta_0 \epsilon^{-2}}$, we thus know that for $\log r(s) \geq - \frac{1}{2} (1 - \alpha^2)^{-1} b_-^{-2} \epsilon^{-2} + O(\log (\epsilon^{-1}))$ as specified, we have not yet entered the regime $\mathcal{K}_{\blue{bo}}$.

On the other hand, we show that for $s_{in} = \sup \{ s \in \mathcal{K}: \Psi(s) \leq \alpha + \epsilon^n \}$, we must have $\log r (s_{in}) \geq - \frac{1}{2} ( 1 - \alpha^2)^{-1} b_-^{-2} \epsilon^{-2} - O(\log (\epsilon^{-1}))$. For this purpose, we use the Raychaudhuri equation (\ref{eq:raych_2}) to see that, for $s_i \leq s \leq s_{in}$, one has
\begin{equation*}
\frac{d}{ds} \log \left[ \frac{\Omega^2(s)}{- \dot{r}} \left( \frac{r(s)}{r_-} \right)^{-(\alpha + \epsilon^n)^2} \right] \geq \frac{d}{ds} \log \left( \frac{\Omega^2(s)}{- \dot{r}(s)} \right) - \frac{\dot{r}}{r} \Psi^2 \geq - \frac{r^2}{- r \dot{r}} |\tilde{A}|^2 q_0^2 |\phi|^2.
\end{equation*}
Since $r(s) \leq e^{- \delta_0 \epsilon^{-2}} r_-$ for $s \in \mathcal{K}$, one sees using Proposition \ref{prop:k_ode} and Lemma \ref{lem:k_prelim_1} that the integral of the right hand side is bounded below by $- \log 2$, say. Therefore a similar application of Corollary \ref{cor:protokasner} will yield that
\begin{equation*}
\frac{\Omega^2}{- \dot{r}} \left( \frac{r(s)}{r_-} \right)^{- (\alpha + \epsilon^n)^2} \geq \frac{1}{2} \exp \left( - \frac{1}{2} b_-^{-2} \epsilon^{-2} \right) \cdot \exp( - D_K \log(\epsilon^{-1})).
\end{equation*}

One now again integrates the equation (\ref{eq:r_evol_2}), \blue{or more precisely \eqref{eq:k_rr_int}}, getting now the upper bound
\begin{align*}
|- r \dot{r}(s) + r \dot{r}(s_i)|
&\geq \int^s_{s_i} \frac{\blue{Q^2} \Omega^2(\tilde{s})}{8 r^2(\tilde{s})} \, d \tilde{s} 
\gtrsim \int^{r(s_i)}_{r(s)} \frac{\Omega^2}{- \dot{r}} \frac{1}{r^2} \, dr \\[0.5em]
&\gtrsim \exp \left( - \frac{1}{2} b_-^{-2} \epsilon^{-2} \right) \cdot \exp( - D_K \log(\epsilon^{-1}) ) \int^{r(s_i)}_{r(s)} \left( \frac{r}{r_-} \right)^{(\alpha + \epsilon^n)^2 - 2} \, dr \\[0.5em]
&\gtrsim \left( \frac{r_-}{r(s)} \right)^{1 - (\alpha + \epsilon^n)^2} \cdot \exp \left( - \frac{1}{2} b_-^{-2} \epsilon^{-2} \right) \cdot \exp ( - D_K \log(\epsilon^{-1})).
\end{align*}
But by Proposition \ref{prop:k_ode}, we always have $\ep^2 \lesssim - r \dot{r} \lesssim \epsilon^2$, so that
\begin{equation*}
\left( \frac{r(s)}{r_-} \right)^{1 - (\alpha + \epsilon^n)^2} \gtrsim \exp \left( - \frac{1}{2} b_-^{-2} \epsilon^{-2} \right) \cdot \exp ( (2- D_K)\log(\epsilon^{-1})).
\end{equation*}
Hence we do indeed find that for $s_i \leq s \leq s_{in}$, we must have $\log (r(s)/r_-) \geq - \frac{1}{2} (1 - \alpha^2)^{-1} b_-^{-2} \epsilon^{-2} - O(\log(\epsilon^{-1}))$ as claimed. This identifies $s = s_{in}$ obeying (\ref{eq:inv_interval}).

The remainder of this proof then proceeds entirely using the ODE (\ref{eq:k_ode_main}), which for $R = \log(\frac{r_-}{r(s)})$ we record again here as
\begin{equation} \label{eq:k_ode_main_copy}
\frac{d \Psi}{dR} = - \Psi ( \Psi - \alpha ) ( \Psi - \alpha^{-1} ) + \mathcal{F}, \quad |\mathcal{F}(R)| \leq D_K e^{- \delta_0 \epsilon^{-2} - R}.
\end{equation}
We have identified $R_{in} = R(s_{in}) = \frac{1}{2} (1 - \alpha^2)^{-1} b_-^{-2} \epsilon^{-2} + O(\log(\epsilon^{-1}))$ such that $\Psi(R_{in}) = \alpha + \epsilon^n$. We want to use that $\alpha$ is an unstable fixed point and $\alpha^{-1}$ a stable fixed point of this one-dimensional dynamical system.

Note that for $\Psi(s) \in [\alpha + \epsilon^n, \alpha^{-1} - \epsilon^n]$, one knows that (i) $\Psi \geq \eta$, and (ii) $\alpha^{-1} - 1 \geq 1 - \alpha \geq \sigma - O(e^{-\delta_0 \epsilon^{-2}})$, so we absorb the error term $\mathcal{F}$ and quantify the stability and instability of the fixed points to find
\begin{align}
\frac{d}{dR} (\Psi - \alpha) \geq \frac{\eta \sigma}{2} (\Psi - \alpha) \quad &\text{ if } \Psi \in [\alpha + \epsilon^n, 1], \label{eq:k_unstable} \\[0.5em]
\frac{d}{dR} (\alpha^{-1} - \Psi) \leq - \frac{\sigma}{2} (\alpha^{-1} - \Psi) \quad &\text{ if } \Psi \in [1, \alpha^{-1} - \epsilon^n]. \label{eq:k_stable}
\end{align} In particular, $\frac{d\Psi}{dR}>0$ as long as $\Psi\in [\alpha + \epsilon^n, \alpha^{-1} - \epsilon^n]$. From (\ref{eq:k_unstable}), one finds that $\Psi - \alpha$ proceeds from $\epsilon^n$ to $1 - \alpha$ in $O(\log (\epsilon^{-1}))$ time in $R$, and from (\ref{eq:k_stable}) that $\alpha^{-1} - \Psi$ proceeds from $\alpha^{-1} - 1$ to $\epsilon^n$ also in $O(\log (\epsilon^{-1}))$ time in $R$. Therefore defining $R_{out}$ to be the minimal $R$ such that $\Psi(R_{out}) = \alpha^{-1} - \epsilon^n$, one finds $R_{out} = R_{in} + O(\log (\epsilon^{-1})) = \frac{1}{2} (1- \alpha^2)^{-1} b_-^{-2} \epsilon^{-2} + O (\log(\epsilon^{-1}))$ also.

Finally, since $\frac{d \Psi}{d R} \geq \frac{1}{2} \epsilon^n \eta \sigma > 0$ when $\Psi \in [\alpha + \epsilon^n, \alpha^{-1} - \epsilon^n]$, we have that, upon entering this region of $\Psi$, it is impossible to return, and the remaining claims of the lemma are immediate.
\end{proof}

\section{Quantitative Kasner-like asymptotics}\label{section:quantitative}

Let us recapitulate the various regions so far:
\begin{enumerate}
\item  In Section~\ref{sec:protokasner}, we  worked in the Proto-Kasner region  $\mathcal{PK}=\{ s_{PK} \leq s \leq s_i\} =\{  2|B| \mathfrak{W} \ep \geq \frac{r}{r_-} \geq e^{-\delta_0 \ep^{-2}}\}$.\\ We now also define the restricted region  $\blue{\mathcal{PK}'}=\{ s_{K_1} \leq s \leq s_i\} =\{  2|B| \mathfrak{W} \ep^2 \geq \frac{r}{r_-} \geq e^{-\delta_0 \ep^{-2}}\}$.

\item In Section~\ref{sec:kasner}, we have worked in the Kasner region $\mathcal{K}=\{ s_{i} \leq s < s_{\infty}\} =\{  0< \frac{r}{r_-} \leq e^{-\delta_0 \ep^{-2}}\}$ and showed that $\lim_{s\rightarrow s_{\infty}}r(s)=0$.
\end{enumerate}
In what follows, we prove quantitative estimates on the ``Kasner-like behavior'' of the metric; for this, we will first have to restrict $\mathcal{PK}$ to its aforementioned subset $\blue{\mathcal{PK}'}$ on which $r\gtrsim \ep^2$ (recall that $r(s_{K_1})=2|B| \mathfrak{W} r_- \ep^2$). While in the previous sections, the analysis was so far oblivious to the absence/presence of a Kasner \blue{bounce}, we will now also be obliged to distinguish both cases and treat them differently. 

As we will soon show, the ``No Kasner \blue{bounce}'' condition \blue{ will be that, for some $0 < \sigma < 1$,}
\begin{equation} \tag{$*\!*\!*$} \label{***}
|\Psi(s_i)|  \geq 1 + \sigma,
\end{equation}
whereas the ``Kasner \blue{bounce}'' condition will be \blue{that, for some $0 < \sigma < 1$ and $\eta > 0$,}
\begin{equation} \tag{$**$} \label{**}
\eta \leq |\Psi(s_i)| \leq 1 - \sigma.
\end{equation}
We note that even combining \eqref{***} and \eqref{**} does not cover \blue{the range of $\Psi(s_i)$ included in \eqref{eq:condition}. In particular, we do not say anything further in the case where $\Psi(s_i) = 1$.}
In each \blue{of these two} cases, we will further sub-divide $\mathcal{PK}\cup \mathcal{K}$ as follows: 
\begin{enumerate}[i.] 
\item  If the ``No Kasner \blue{bounce}'' condition  \eqref{***} is satisfied, we define the first Kasner region \blue{by} \begin{equation}\label{K1.def.no.inv}\mathcal{K}_1= \blue{\mathcal{PK}'}\cup \mathcal{K.}\end{equation} 
Because there is no \blue{bounce} in that case, we show indeed that the metric is close to a single Kasner spacetime on the whole of $\mathcal{K}_1= \blue{\mathcal{PK}'}\cup \mathcal{K}$ in Theorem~\ref{thm.***}.
\item If the ``Kasner \blue{bounce}'' condition  \eqref{**} is satisfied, we define the first Kasner region \blue{by}  \begin{equation}\label{K1.def.inv}
\mathcal{K}_1= \{s_{K_1} \leq s \leq s_{in}\}\supset \blue{\mathcal{PK}'},
\end{equation} \blue{where $s_{in} \in \mathcal{K}$ will be defined shortly. We also define}  \begin{equation}\label{Kinv.def}
\mathcal{K}_{\blue{bo}} = \{ s_{in} \leq s \leq s_{out}\}
\end{equation} to be the Kasner \blue{bounce} region, and   \begin{equation}\label{K2.def.inv}
\mathcal{K}_2=\{   s_{out} \leq s < s_{\infty}\}
\end{equation}  to be the second Kasner region, where $s_{in}$, $s_{out} \in \mathcal{K}$ given by Lemma~\ref{lem:inv_interval} (applied to $n=2$) are defined \blue{such that} \eqref{eq:inv_interval} is satisfied on $\mathcal{K}_{\blue{bo}}$. Because of the Kasner \blue{bounce}, we show that the metric is close to a first Kasner spacetime in $\mathcal{K}_1$, and close to a second, different Kasner spacetime in $\mathcal{K}_2$ in Theorem~\ref{thm.**}. By \eqref{eq:inv_interval}, the transition region $\mathcal{K}_{\blue{bo}}$ \blue{will be shown to be small in an appropriate sense}.
\end{enumerate}

We now state the two main theorems of this section, Theorem~\ref{thm.***} and \ref{thm.**}.

\begin{theorem}\label{thm.***}
Let $(r, \Omega^2, \phi, Q, \tilde{A})$ be a solution to the system (\ref{eq:raych})--(\ref{eq:phi_evol_2}), which, by Proposition \ref{prop:oscillation+}, exists at least up to the value $s = s_i = \frac{b_-^{-2} \epsilon^{-2}}{4 |K_-| } + O(\log (\epsilon^{-1}))$ at which $r(s) = e^{- \delta_0 \epsilon^{-2}} r_-$ \blue{(see Proposition~\ref{prop:lateblueshift} for a definition of $b_-$)}. Suppose that for some given \blue{$0 < \sigma < 1$}, \eqref{***} is satisfied.

Then there exists some $\epsilon_0 (M, \mathbf{e}, \Lambda, m^2, q_0, \sigma) > 0$, such that, if $0 < |\epsilon| < \epsilon_0$, \blue{then} there exists some $s_{\infty} > s_i$ such that the solution of (\ref{eq:raych})--(\ref{eq:phi_evol_2}) exists for $s \in (- \infty, s_{\infty})$ with $\lim_{s \to s_{\infty}} r(s) = 0$.

Furthermore, we have the following Kasner-like asymptotics: denote $s_{K_1}$ such that $r(s_{K_1}) = 2 |B| \mathfrak{W} \epsilon^2 r_-$, then in the region $\mathcal{K}_1 = \{ s_{K_1} \leq s < s_{\infty} \}$, one may write the metric in the following form, where $\alpha$ is as determined in Proposition \ref{prop:k_ode}:
\begin{equation}
g = - d \tau^2 + \mathcal{X}_1 \cdot ( 1 + \mathfrak{E}_{X, 1}(\tau)) \, \tau^{\frac{2 (\alpha^2 -1)}{\alpha^2 + 3}} \, dt^2 + \mathcal{R}_1 \cdot ( 1 + \mathfrak{E}_{R, 1}(\tau)) \, r_-^2 \tau^{ \frac{4}{\alpha^2 + 3} } \, d \sigma_{\mathbb{S}^2} .
\end{equation}
Here, $\mathcal{X}_1$ and $\mathcal{R}_1$ are constants, and $\mathfrak{E}_{X, 1}(\tau)$ and $\mathfrak{E}_{R, 1}(\tau)$ are small functions of $\tau$ satisfying the following bounds for  $\beta =\min \{\frac{1}{2},1-\alpha^2\}$
\begin{equation}
\left| \log \mathcal{X}_1 + \frac{\alpha^2 + 1}{\alpha^2 + 3} b_-^{-2} \epsilon^{-2} \right| + \left| \log \mathcal{R}_1 - \frac{1}{\alpha^2 + 3} b_-^{-2} \epsilon^{-2} \right| \leq C_K \log (\epsilon^{-1}),
\end{equation}
\begin{equation} \label{***.error}
|\mathfrak{E}_{X, 1}(\tau)| + |\mathfrak{E}_{R, 1}(\tau)| \leq C_K \epsilon^2 \cdot \left(\frac{\tau}{\tau(s_{K_1})}\right)^{\frac{2\beta}{\alpha^2 + 3}}.
\end{equation}
Hence, the spacetime corresponds to a Kasner-like spacetime, \blue{in the sense of \eqref{Kasner.like},} with Kasner exponents \blue{$p_1 = \frac{\alpha^2 - 1}{\alpha^2 + 3}$, $p_2 = p_3 = \frac{2}{\alpha^2 + 3}$}.
\end{theorem}

\begin{theorem}\label{thm.**}
Let $(r, \Omega^2, \phi, Q, \tilde{A})$ be a solution to the system (\ref{eq:raych})--(\ref{eq:phi_evol_2}), which by Proposition \ref{prop:k_ode} exists at least up to the value $s = s_i = \frac{b_-^{-2} \epsilon^{-2}}{4 |K_-| } + O(\log (\epsilon^{-1}))$ at which $r(s) = e^{- \delta_0 \epsilon^{-2}} r_-$. Suppose that for given $\eta > 0$ and \red{$0 < \sigma < 1$}, \eqref{**} is satisfied.

Then there exists some $\epsilon_0 (M, \mathbf{e}, \Lambda, m^2, q_0, \eta, \sigma) > 0$, such that, if $0 < |\epsilon| < \epsilon_0$, \blue{then} there exists some $s_{\infty} > s_i$ such that the solution of (\ref{eq:raych})--(\ref{eq:phi_evol_2}) exists for $s \in (- \infty, s_{\infty})$ with $\lim_{s \to s_{\infty}} r(s) = 0$.
We further single out two different regions with Kasner-like asymptotics, between \blue{which} there is an \emph{intermediate region} where the Kasner \blue{bounce} occurs. Letting $s_{K_1}$ be such that $r(s_{K_1}) = 2 |B| \mathfrak{W} \epsilon^2 r_-$, we define the following three regions:
\begin{equation*}
\mathcal{K}_1 = \{ s_{K_1} \leq s \leq s_{in} \}, \; \mathcal{K}_{\blue{bo}} = \{ s_{in} \leq s \leq s_{out} \}, \; \mathcal{K}_2 = \{ s_{out} \leq s \leq s_{\infty} \}.
\end{equation*}
We will describe Kasner-like asymptotics for the two regions $\mathcal{K}_1$ and $\mathcal{K}_2$.

In the region $\mathcal{K}_1$, one writes the metric in the following form, for $\alpha$ as in Proposition \ref{prop:k_ode}:
\begin{equation}
g = - d \tau^2 + \mathcal{X}_1 \cdot ( 1 + \mathfrak{E}_{X, 1}(\tau)) \,(\tau-\tau_0)^{\frac{2 (\alpha^2 -1)}{\alpha^2 + 3}} \, dt^2 + \mathcal{R}_1 \cdot ( 1 + \mathfrak{E}_{R, 1}(\tau)) \, r_-^2 (\tau-\tau_0)^{ \frac{4}{\alpha^2 + 3} } \, d \sigma_{\mathbb{S}^2} .
\end{equation}
Here  $\tau_0>0$, $\mathcal{X}_1$ and $\mathcal{R}_1$ are constants, and $\mathfrak{E}_{X, 1}(\tau)$ and $\mathfrak{E}_{R, 1}(\tau)$ are functions of $\tau$ satisfying
\begin{equation}
\left| \log \mathcal{X}_1 + \frac{\alpha^2 + 1}{\alpha^2 + 3} b_-^{-2} \epsilon^{-2} \right| + \left| \log \mathcal{R}_1 - \frac{1}{\alpha^2 + 3} b_-^{-2} \epsilon^{-2} \right| \leq C_K \log (\epsilon^{-1}),
\end{equation}
\begin{equation}
|\mathfrak{E}_{X, 1}(\tau)| + |\mathfrak{E}_{R, 1}(\tau)| \leq C_K \epsilon^2.
\end{equation}

On the other hand, in the region $\mathcal{K}_2$, one instead has the following form for the metric
\begin{equation}
g = - d \tau^2 + \mathcal{X}_2 \cdot ( 1 + \mathfrak{E}_{X, 2}(\tau)) \, \tau^{\frac{2 (1 - \alpha^2)}{1 + 3 \alpha^2}} \, dt^2 + \mathcal{R}_2 \cdot ( 1 + \mathfrak{E}_{R, 2}(\tau)) \, r_-^2 \tau^{ \frac{4 \alpha^2}{1 + 3 \alpha^2} } \, d \sigma_{\mathbb{S}^2} .
\end{equation}
The constants $\mathcal{X}_2$ and $\mathcal{R}_2$ are constants, and the functions $\mathfrak{E}_{X, 2}(\tau)$ and $\mathfrak{E}_{R, 2}(\tau)$ now satisfying the following bounds for $\beta =\min \{\frac{1}{2},1-\alpha^2\}$
\begin{equation}
\left| \log \mathcal{X}_2 + \frac{1 + \alpha^{-2}}{1 + 3 \alpha^2} b_-^{-2} \epsilon^{-2} \right| + \left| \log \mathcal{R}_2 - \frac{1}{1 + 3\alpha^2} b_-^{-2} \epsilon^{-2} \right| \leq C_K \log (\epsilon^{-1}),
\end{equation}
\begin{equation} \label{**.error}
|\mathfrak{E}_{X, 2}(\tau)| + |\mathfrak{E}_{R, 2}(\tau)| \leq C_K \epsilon^2 \cdot \left(\frac{\tau}{\tau(s_{out})}\right)^{\frac{2\beta}{\alpha^{-2} + 3}}.
\end{equation}

One sees that the spacetime exhibits a Kasner bounce from the \blue{Kasner-like} region $\mathcal{K}_1$ \blue{(in the sense of \eqref{Kasner.like})} with Kasner exponents of \blue{$p_1 = \frac{\alpha^2 - 1}{\alpha^2 + 3}, p_2 = p_3 = \frac{2}{\alpha^2 + 3}$} to \blue{another Kasner-like} region $\mathcal{K}_2$ with exponents of \blue{$p_1 = \frac{1 - \alpha^2}{1 + 3 \alpha^2}, p_2 = p_3 = \frac{2\alpha^2}{1 + 3 \alpha^2}$}. We further \blue{provide} the following estimates regarding the proper time length of the regions $\mathcal{K}_2$ and $\mathcal{K}_{\blue{bo}}$. For $\mathcal{K}_2$, the proper time $\tau$ from the singularity varies between $0$ and $\tau(s_{out})$, obeying
\begin{equation} \label{eq:propertimek2}
\left| \log \tau(s_{out}) - \frac{1}{2} \frac{\alpha^{-2} + 1}{1 - \alpha^2} \cdot b_-^{-2} \epsilon^{-2} \right| \leq C_K \log(\epsilon^{-1}).
\end{equation}
On the other hand, we have the following (non-sharp) upper bound for the size of \blue{$\mathcal{K}_{bo} = [s_{in}, s_{out}]$, where the proper time variable $\tau = \tau(s)$ satisfies $\frac{d \tau}{ds} = - \frac{\Omega(s)}{2}$.}
\begin{equation} \label{eq:propertimekinv}
0 \leq \tau(s_{in}) - \tau(s_{out}) \leq \exp( - \frac{b_-^{-2} \ep^{-2}}{1 - \alpha^2} ) \cdot \exp( C_K \log (\epsilon^{-1})).
\end{equation}
\end{theorem}

\subsection{Asymptotics for \texorpdfstring{$\Psi$}{Ψ} near the \texorpdfstring{$\{r = 0\}$}{{r=0}} singularity}

The first step in proving that the $\{ r = 0 \}$ singularity has Kasner-like asymptotics relies on showing that $\Psi$ tends to the appropriate constant $\alpha$ or $\alpha^{-1}$ sufficiently quickly near the singularity. The aim will be to show that $\Psi = \alpha + O( (\frac{r}{r_-})^{\beta})$ or $\Psi = \alpha^{-1} + O( (\frac{r}{r_-})^{\beta})$ for some positive exponent $\beta > 0$.

After the usual change of coordinates $r = e^{-R }\, r_-$, this translates to showing that $\Psi$ decays exponentially to $\alpha$ or $\alpha^{-1}$ in the variable $R$. For this purpose, we shall need to use the ODE (\ref{eq:k_ode_main}) from Proposition \ref{prop:k_ode}. By standard dynamical systems theory, one expects that $\Psi$ tends towards its stable fixed point at an exponential rate. We quantify this in the following lemma:

\begin{lemma} \label{lem:ode_decay}
Fix constants $0 < \sigma, \eta < \frac{1}{2}$. For all $\alpha \in \R$ satisfying $1 + \sigma \leq |\alpha| \leq \eta^{-1}$, consider the following ODE for the function $\Psi = \Psi(R)$:
\begin{equation} \label{eq:ode_general}
\frac{d \Psi}{d R} = - \Psi ( \Psi - \alpha ) ( \Psi - \alpha^{-1} ) + \mathcal{F}(R), \qquad |\mathcal{F}(R)| \leq e^{-R}.
\end{equation}
Define $\beta = \min\{ \frac{1}{2}, \alpha^2 - 1\} \geq \sigma$. Then there exists some $\nu_0 = \nu_0(\sigma, \eta) > 0$ such that, for all $|\nu| < \nu_0$ and $R_*$ satisfying both
\begin{equation} \label{eq:ode_general_init}
|\Psi(R_*) - \alpha| \leq \nu^2 \quad \textbf{ and } \quad e^{-R_*} \leq \nu^2,
\end{equation}
one finds that, for $R \geq R_*$, $\Psi(R)$ decays to $\alpha$ at the following exponential rate
\begin{equation} \label{eq:ode_general_decay}
|\Psi(R) - \alpha | \leq 8 \nu^2 e^{- \beta( R - R_* )}.
\end{equation}
\end{lemma}

The above lemma will be applied with \blue{$\nu \sim \ep^n$ or $\nu \sim e^{- \frac{\delta_0}{2} \ep^{-2}}$} in what follows.

\begin{proof}
We use a bootstrap argument along with Gr\"onwall's inequality. We take the bootstrap assumption to be the desired estimate (\ref{eq:ode_general_decay}). Assuming this holds in some bootstrap region, we have
\begin{equation*}
\Psi ( \Psi - \alpha^{-1} )  \geq \alpha^2 - 1 - | 2 \alpha - \alpha^{-1} | \cdot 8 \nu^2 e^{-\beta(R - R_*)}.
\end{equation*}

Therefore, for $R$ and $\tilde{R}$ in the bootstrap region, one computes that, for $\nu$ small enough,
\begin{align*}
- \int_{\tilde{R}}^R \Psi(R') ( \Psi(R') - \alpha^{-1} ) \, dR'
&\leq - (\alpha^2 - 1) (R - \tilde{R}) + \beta^{-1} |2 \alpha - \alpha^{-1}| \cdot 8 \nu^2 e^{- \beta(\tilde{R} - R_*)} \\
&\leq - \beta (R - \tilde{R}) + \log 2.
\end{align*}

Hence, after finding from (\ref{eq:ode_general}) the differential inequality
\begin{equation*}
\frac{d}{dR} | \Psi - \alpha | \leq - \Psi ( \Psi - \alpha^{-1} ) |\Psi - \alpha| + e^{-R},
\end{equation*}
one uses Gr\"onwall to deduce that
\begin{equation*}
|\Psi(R) - \alpha| \leq 2 \nu^2 e^{- \beta(R - R_*)} + \int^R_{R_*} 2 e^{-\tilde{R}} e^{- \beta(R - \tilde{R})} \, d\tilde{R} \leq 2(\nu^2 + (1 - \beta)^{-1} e^{-R_*}) e^{-\beta(R - R_*)}.
\end{equation*}
So taking into account the second assumption of (\ref{eq:ode_general_init}) and $\beta \leq 1/2$, we improve the bootstrap assumption (\ref{eq:ode_general_decay}). So this estimate is true for all $R \geq R_*$.
\end{proof}


\subsection{First case: absence of Kasner bounce}\label{no.inv.section}

In this subsection, assuming that (\ref{***}) holds, we prove the quantitative Kasner asymptotics in the region $\mathcal{K}_1 = \{ s_{K_1} \leq s < s_{\infty} \}$. The essential ingredient is the following lemma:

\begin{lemma} \label{lem:noinv_lapse}
Consider a solution $(r, \Omega^2, \phi, Q, \tilde{A})$ to the system (\ref{eq:raych})--(\ref{eq:phi_evol_2}) as in Proposition \ref{prop:k_ode}. Assuming also (\ref{***}), one finds that there exists some $\mathcal{Y}_1>0$ satisfying, for $s_{K_1} \leq s < s_{\infty}$:
\begin{equation} \label{eq:noinv_lapse}
\left| \log \left( \frac{\Omega^2}{- \dot{r}} \left( \frac{r}{r_-} \right)^{-\alpha^2} \right) (s) - \log \mathcal{Y}_1 \right| \leq D_1 
\epsilon^2 \left( \frac{r(s)}{r_-} \right)^{\beta},
\end{equation} with $\mathcal{Y}_1$ satisfying \begin{equation}
\left| \log \mathcal{Y}_1 + \frac{1}{2} b_-^{-2} \ep^{-2}\right|\leq D_1 \cdot \log(\ep^{-1}).
\end{equation}
Furthermore, one may find a constant $\mathcal{Z}_1>0$ such that 
\begin{equation} \label{eq:noinv_r}
\left| - r \dot{r} (s) - \mathcal{Z}_1 \right| \leq D_1 \epsilon^4 \left( \frac{r (s)}{r_-} \right)^{\beta},
\end{equation}  with $\mathcal{Z}_1$ satisfying \begin{equation}
\left| \mathcal{Z}_1  -  4 |B|^2 \mathfrak{W}^2 \omega_{RN} r_-^2 \ep^2 \right|\leq D_1 \cdot  \ep^4\log(\ep^{-1}).
\end{equation}
\end{lemma}

\begin{proof}
We will use the Raychaudhuri equation (\ref{eq:raych_2}). In light of Lemma \ref{lem:ode_decay}, it is preferable to change variables once again, now to the $R$-coordinate:
\begin{equation} \label{eq:raych_3}
\frac{d}{dR} \log \left( \frac{\Omega^2}{- \dot{r}} \right) = - \Psi^2 - \frac{r^4 q_0^2 \tilde{A}^2 \phi^2}{ (- r\dot{r})^2}.
\end{equation}

Proposition~\ref{prop:k_ode} \blue{and Lemma~\ref{lem:k_prelim_1}} tell us that the second term on the right hand side of this expression is $O( \epsilon^{-4} e^{-4 R} R^2)$. Adding $\alpha^2$ to both sides and using that $\Psi$ is bounded, we have
\begin{equation} \label{eq:noinv_lapse_der}
\left| \frac{d}{dR} \log \left( \frac{\Omega^2}{- \dot{r}} \right) + \alpha^2 \right| \lesssim |\Psi - \alpha| + \epsilon^{-4} e^{-4 R} R^2.
\end{equation}
The crucial observation is that the right-hand side will be integrable (and with small integral) in the region $\mathcal{K}_1$. Indeed, we claim that if $R_{K_1} = - \log( \frac{r_{K_1}} {r_-}) = - \log(2 |B| \mathfrak{W} \epsilon^2)$, then for $R_{K_1} \leq R < + \infty$, we have
\begin{equation} \label{eq:noinv_lapse_int}
\int_R^{+\infty} |\Psi(\tilde{R}) - \alpha| \, d \tilde{R}  + \int_R^{+\infty} \epsilon^{-4} e^{-4 \tilde{R}} \tilde{R}^2 \, d\tilde{R}\lesssim \epsilon^2 e^{- \beta R}.
\end{equation}

The bound for the latter integral follows by straightforward calculus, using also $\beta < 1/2$ and $e^{-R} \leq e^{- R_{K_1}} \sim \epsilon^2$ to get the correct dependence on $\epsilon$. For the former integral, we proceed in two steps; first we use Lemma \ref{lem:ode_decay} to deal with the region $\mathcal{K}$, then use Corollary \ref{cor:protokasner} to handle the remaining part $\mathcal{PK} \cap \mathcal{K}_1$.

Note that, for the region $\mathcal{K}$, where we have access to the ODE (\ref{eq:k_ode_main}), we have the bound $|\Psi(R_i) - \alpha| \lesssim e^{-\delta_0 \epsilon^{-2}}$ from Proposition \ref{prop:k_ode}. Hence, by Lemma~\ref{lem:ode_decay} with $\nu^2 \blue{\sim} e^{-\delta_0 \epsilon^{-2}}$, we know that, for $R \geq R_i$, we have $|\Psi - \alpha| \lesssim e^{-\delta_0 \epsilon^{-2}} e^{-\beta(R - R_i)}$. Thus, for $R \geq R_i$, one has
\begin{equation*}
\int_R^{+\infty} |\Psi(\tilde{R}) - \alpha | \, d\tilde{R} \lesssim e^{-\delta_0 \epsilon^{-2}} e^{- \beta(R - R_i)} \lesssim e^{- (1 - \beta) \delta_0 \epsilon^{-2}} e^{- \beta R}.
\end{equation*}
The last line here follows from the definition $R_i = \delta_0 \epsilon^{-2}$. The smallness of the expression $e^{- (1-\beta) \delta_0 \epsilon^{-2}}$ thus proves (\ref{eq:noinv_lapse_int}) for $R \geq R_i$. 

For the remaining portion $R \in [R_{K_1}, R_i]$, we use Corollary \ref{cor:protokasner}. The point is that, in this region, we have $|\Psi - \alpha| \leq |\Psi - \Psi_i| + |\Psi_i - \alpha| \lesssim e^{-2 R} \log(\epsilon^{-1}) + e^{- \delta_0 \epsilon^{-2}}$. Therefore, one finds that
\begin{equation*}
\int_R^{R_i} |\Psi(\tilde{R}) - \alpha| \, \blue{d \tilde{R}} \lesssim e^{-2 R} \log(\epsilon^{-1}) + e^{-\delta_0 \epsilon^{-2}} R_i \lesssim \epsilon^2 e^{- \beta R}.
\end{equation*}
\blue{(We also note here that \eqref{eq:no.bounce} follows straightforwardly from the $|\Psi - \alpha|$ estimates here.)}

So we have proved the estimate (\ref{eq:noinv_lapse_int}). From (\ref{eq:noinv_lapse_der}), this shows that the expression
\begin{equation}
\log \left( \frac{\Omega^2}{- \dot{r}} \right) + \alpha^2 R = \log \left( \frac{\Omega^2}{- \dot{r}} \left( \frac{r}{r_-} \right)^{-\alpha^2} \right)
\end{equation}
indeed has a finite limit $\log \mathcal{Y}_1$ as $R \to \infty$, and moreover satisfies the estimate (\ref{eq:noinv_lapse}).

We next estimate the constant $\log \mathcal{Y}_1$. To do this, we evaluate (\ref{eq:noinv_lapse}) at $r = r_{K_1}$ i.e.\ $s = s_{K_1}$ and use the estimate (\ref{eq:pk1_lapse}) from Corollary \ref{cor:protokasner}. As $|\log r(s_{K_1})|, |\log (- \dot{r}(s_{K_1}))| = O(\log (\epsilon^{-1}))$ here, we find that
\begin{equation} \label{eq:noinv_y}
\left| \log \mathcal{Y}_1 + \frac{1}{2} b_-^{-2} \epsilon^{-2} \right| \lesssim \log (\epsilon^{-1}).
\end{equation}

Finally, we show the estimate (\ref{eq:noinv_r}) by considering the equation (\ref{eq:r_evol_2}) after changing variables to $r$; Propositions \ref{prop:oscillation+} and \ref{prop:k_ode} and then the now-known (\ref{eq:noinv_lapse}) tell us that for $s \in \mathcal{K}_1$, we have
\begin{equation*}
\left| \frac{d}{dr} (- r \dot{r}) \right| \,\lesssim\, \frac{\Omega^2}{- r^2 \dot{r}} \, \lesssim \, \mathcal{Y}_1 \cdot \left( \frac{r}{r_-} \right)^{\alpha^2 - 2}.
\end{equation*}
Integrating this expression then yields (\ref{eq:noinv_r}) -- of course it is essential here that $\alpha^2 - 1 \geq \frac{\sigma }{2} > 0$, and we have the smallness (\ref{eq:noinv_y}) for the expression $\mathcal{Y}_1$.
\end{proof}

\subsection{Second case: presence of a Kasner bounce}

Assuming instead that (\ref{**}) holds, Proposition \ref{prop:k_ode} and Lemma \ref{lem:inv_interval} (for $n=2$) show that the spacetime must exhibit a Kasner \blue{bounce}. We proceed to show that if we define $\mathcal{K}_1$ and $\mathcal{K}_2$ as in \eqref{K1.def.inv}, \eqref{K2.def.inv}, there are Kasner-like asymptotics in both regimes. As in Section \ref{no.inv.section}, we first find the precise asymptotics for the lapse $\Omega^2$, as well as $- \dot{r}$.

\subsubsection{The pre-bounce Kasner regime}

We start to look at the region $\mathcal{K}_1$ (pre-\blue{bounce} regime).

\begin{lemma} \label{lem:preinv_lapse}
Consider a solution $(r, \Omega^2, \phi, Q, \tilde{A})$ to the system (\ref{eq:raych})--(\ref{eq:r_evol_2}) as in Proposition \ref{prop:k_ode}. Assuming now the condition (\ref{**}) on the value of $\Psi$ at $s = s_i$, there will exist some $\mathcal{Y}_1$ satisfying, for $s_{K_1} \leq s \leq s_{in}$, 
\begin{equation} \label{eq:preinv_lapse}
\left| \log \left( \frac{\Omega^2}{- \dot{r}} \left( \frac{r}{r_-} \right)^{-\alpha^2} \right) (s) - \log \mathcal{Y}_1 \right| \leq D_1 
\epsilon^2,
\end{equation} with $\mathcal{Y}_1$ satisfying 	\begin{equation} \label{eq:inv_y}
\left| \log \mathcal{Y}_1 + \frac{1}{2} b_-^{-2} \epsilon^{-2} \right| \lesssim \log (\epsilon^{-1}).
\end{equation}
Furthermore, one may find a constant $\mathcal{Z}_1>0$ such that 
\begin{equation} \label{eq:preinv_r}
\left| - r \dot{r} (s) - \mathcal{Z}_1 \right| \leq D_1 \epsilon^4,
\end{equation} with $\mathcal{Z}_1$ satisfying 	\begin{equation} \label{eq:inv_z}
\left|  \mathcal{Z}_1  -4|B|^2 \mathfrak{W}^2 \omega_{RN}  r_-^2 \ep^2 \right| \lesssim  \ep^4 \log (\epsilon^{-1}).
\end{equation}
\end{lemma}

\begin{proof}
We follow a similar template to the proof of Lemma \ref{lem:noinv_lapse}, though unlike \blue{in} that case, we do not need $(\frac{r}{r_-})^{\beta}$ decay rates. In particular, we are able to integrate from $R = R_{K_1}$ rather than backwards from $R = + \infty$. Using the same Raychaudhuri equation (\ref{eq:raych_3}), we need to prove the analogue of (\ref{eq:noinv_lapse_int}), i.e.~\blue{for $R \geq R_{K_1}$,}
\begin{equation} \label{eq:preinv_lapse_int}
\int_{R_{K_1}}^R |\Psi(\tilde{R}) - \alpha| \, d\tilde{R} + \int_{R_{K_1}}^R \epsilon^{-4} e^{-4 \tilde{R}} \tilde{R}^2 \, d\tilde{R} \lesssim \epsilon^2.
\end{equation}

As in Lemma \ref{lem:noinv_lapse}, the latter integral is straightforward, and we focus on the former. Following Corollary \ref{cor:protokasner} and Proposition \ref{prop:k_ode}, we know that \blue{for $s \in \mathcal{PK}' = \mathcal{PK} \cap \mathcal{K}_1$}, where $R \in [R_{K_1}, R_i]$, we have $|\Psi(R) - \alpha| \leq |\Psi(R) - \Psi_i| + |\Psi_i - \alpha| \lesssim e^{-2R} \log(\epsilon^{-1}) + e^{-\delta_0 \epsilon^{-2}}$. So we know as before that
\begin{equation} \label{eq:preinv_est1}
\int_{R_{K_1}}^{R_i} |\Psi(\tilde{R}) - \alpha| \, d \tilde{R} \lesssim \epsilon^2.
\end{equation}

On the other hand, when integrating between $R_i$ and $R_{in}$, $\Psi$ is growing away from $\alpha$ exponentially as opposed to exponentially decaying, so we need a new tactic. The key observation is that following the definition of $R_{in}$ from Lemma \ref{lem:inv_interval} for $n=2$, we know a priori that $|\Psi(R) - \alpha| \leq \epsilon^2$ when $R \in [R_i, R_{in}]$. The same lemma also tells us that $|R_i - R_{K_1}| \lesssim \epsilon^{-2}$, so these two observations combined tell us that $\int^{R_{in}}_{R_i} |\Psi(\tilde{R}) - \alpha| \, d\tilde{R} \lesssim 1$.

This is not quite the claimed estimate (\ref{eq:preinv_lapse_int}). To improve the $O(1)$ bound to $O(\epsilon^2)$, we apply Lemma \ref{lem:inv_interval} with $n=4$. Defining $R_{in}'$ to instead be the unique $R > R_i$ with $|\Psi(R_{in}') - \alpha| = \epsilon^4$, we find that, for $R \in [R_{in}', R_{in}]$, we have $|\Psi(R) - \alpha| \leq e^{\frac{\eta\sigma}{2}(R - R_{in})} \epsilon^2$. 
This is because, within this region, \red{\eqref{eq:k_unstable} holds and may be ``integrated backwards'' from $R = R_{in}$}.
Furthermore, for $R \in [R_i, R_{in}']$, we know a priori that $|\Psi(R) - \alpha| \leq \epsilon^4$. So
\begin{equation}
\int_{R_i}^{R_{in}} |\Psi(\tilde{R}) - \alpha| \, d \tilde{R} \leq \int_{R_i}^{R_{in}'} \epsilon^4 \, d \tilde{R} + \int_{R_{in}'}^{R_{in}} e^{\frac{\eta \sigma}{2} (\tilde{R} - R_{in})} \epsilon^2 \, d \tilde{R} \lesssim \epsilon^2.
\end{equation}

Combining this with (\ref{eq:preinv_est1}) yields the desired (\ref{eq:preinv_lapse_int}). The estimate (\ref{eq:preinv_lapse}) then follows exactly as in the proof of (\ref{eq:noinv_lapse}), as does the estimate (\ref{eq:inv_y}) on $\mathcal{Y}_1$.

For the estimate (\ref{eq:preinv_r}), let us define $\mathcal{Z}_1$ as $- r \dot{r}(s_i)$. By Proposition \ref{prop:oscillation+}, it is easy to integrate (\ref{eq:r_evol_2}) and deduce (\ref{eq:preinv_r}) in the region $s \in [s_{K_1}, s_i]$. The difficulty lies in showing (\ref{eq:preinv_r}) in the region $s \in [s_i, s_{in}]$, not least because $- r \dot{r}$ is expected to change value during the inversion process.

What helps us here is that $- r \dot{r}$ is approximately inversely proportional to the quantity $\Psi$, given that $r^2 \dot{\phi} = - r \dot{r} \cdot \Psi$ is approximately constant. To be quantitatively precise, we combine the estimate (\ref{eq:k_phid}) with the trivial a priori observation that $\left|\frac{\Psi(s)}{\Psi_i} - 1\right| \lesssim \epsilon^2$ for $ s \in [s_i, s_{in}]$ to get
\begin{equation}
\left| \frac{- r \dot{r}(s)}{- r \dot{r}(s_i)} - 1 \right| = \left| \frac{\Psi(s_i)}{\Psi(s)} \cdot \frac{r^2 \dot{\phi}(s)}{r^2 \dot{\phi}(s_i)} - 1 \right| \lesssim \epsilon^2.
\end{equation}
Since $-r\dot{r}(s_i) \sim \epsilon^2$, the estimate (\ref{eq:preinv_r}) follows immediately.
\end{proof}

\subsubsection{The post-bounce Kasner regime}

Finally, we consider the post-\blue{bounce} regime, i.e.\ the region $\mathcal{K}_2 = \{ s_{out} \leq s < s_{\infty}\}$.

\begin{lemma} \label{lem:postinv_lapse}
Consider a solution $(r, \Omega^2, \phi, Q, \tilde{A})$ to the system (\ref{eq:raych})--(\ref{eq:r_evol_2}) as in Proposition \ref{prop:k_ode}. Assuming the condition (\ref{**}) on the value of $\Psi$ at $s = s_i$, then for $s_{out} \leq s \leq s_{\infty}$, there exists some constant $\mathcal{Y}_2$ such that
\begin{equation} \label{eq:postinv_lapse}
\left| \log \left( \frac{\Omega^2}{- \dot{r}} \left( \frac{r}{r_-} \right)^{-\alpha^{-2}} \right) (s) - \log \mathcal{Y}_2 \right| \leq D_2
\epsilon^2 \cdot \left( \frac{r(s)}{r(s_{out})} \right)^{\beta}, 
\end{equation} with  $ \mathcal{Y}_2 >0$ satisfying \begin{equation}\label{Y2.inv}
\left|\log\mathcal{Y}_2 - \frac{1}{2} \alpha^{-2} b_-^{-2} \ep^{-2} \right|\leq D_2 \cdot \log(\ep^{-1}).
\end{equation} 
Here $\beta = \min \{ \alpha^{-2} - 1, \frac{1}{2} \}$. One may also find a constant $\mathcal{Z}_2>0$ with $|\mathcal{Z}_2| \sim \ep^2$  such that, in the same region,
\begin{equation} \label{eq:postinv_r}
\left| - r \dot{r} (s) - \mathcal{Z}_2 \right| \leq D_2 \epsilon^4 \cdot \left( \frac{r(s)}{r(s_{out})} \right)^{\beta}.
\end{equation}
\end{lemma}

\begin{proof}
We again begin the proof of (\ref{eq:postinv_lapse}) by using the Raychaudhuri equation (\ref{eq:raych_3}). In this case, this equation alongside Proposition \ref{prop:k_ode} will tell us that
\begin{equation} \label{eq:postinv_lapse_der}
\left| \frac{d}{dR} \log \left( \frac{\Omega^2}{- \dot{r}} \right) - \alpha^{-2} \right| \lesssim
|\Psi(R) - \alpha^{-1}| + \epsilon^{-4} e^{-4R} R^2.
\end{equation}

\blue{To estimate $|\Psi - \alpha^{-1}|$ in \eqref{eq:postinv_lapse_der}, }we apply Lemma \ref{lem:ode_decay} with $R_* = R_{out}$, $\nu=\ep$ and $\alpha$ replaced by $\alpha^{-1}$ [note that $\alpha^{-1}> 1+\sigma>0$ assuming \eqref{**}]. The lemma thus tells us that, for $\epsilon$ chosen sufficiently small, we have $|\Psi - \alpha^{-1}| \leq 8 \epsilon^2 e^{- \beta (R - R_{out})}$ for all $R\geq R_{out}$. Hence\blue{, using that $R \geq R_{out} \gtrsim \ep^{-2}$ which in turn yields $\ep^{-4} e^{-4 R} R^2 \lesssim e^{- \beta R} \ep^2$, the inequality} (\ref{eq:postinv_lapse_der}) implies that
\begin{equation*}
\left| \frac{d}{dR} \log \left( \frac{\Omega^2}{- \dot{r}} \left( \frac{r}{r_-} \right)^{\alpha^{-2}} \right) \right| \lesssim \epsilon^2 e^{- \beta(R - R_{out})}.
\end{equation*}
As the right hand side is integrable, $\frac{\Omega^2}{-\dot{r}} \left( \frac{r}{r_-} \right)^{\alpha^{-2}}$ has a well defined limit $\mathcal{Y}_2$ as $R \to + \infty$, and we obtain (\ref{eq:postinv_lapse}).

For (\ref{eq:postinv_r}), we will not estimate $- r \dot{r}(s)$ directly but instead use $\Psi$ and $r^2 \dot{\phi} = - r \dot{r} \cdot \Psi$. We already have $|\Psi - \alpha^{-1}| \lesssim \epsilon^2 e^{- \beta(R - R_{out})}$, while we integrate (\ref{eq:phi_evol_2}) backwards from $R = + \infty$ to find that
\begin{equation*}
| r^2 \dot{\phi} (s) - \lim_{\tilde{s} \to s_{\infty}} r^2 \dot{\phi}(\tilde{s})| \lesssim \int_s^{s_{\infty}} r^2 |\phi|(\tilde{s}) + \Omega^2 r^2 |\phi|(\tilde{s}) \, d \tilde{s} \lesssim \int_{0}^{r(s)} \frac{r^3 |\phi|}{- r \dot{r}} + \frac{\Omega^2}{- \dot{r}} r^2 |\phi| \, dr.
\end{equation*}

Using Proposition \ref{prop:k_ode} and that $\frac{\Omega^2}{-\dot{r}}$ is monotonically decreasing in $s$ and thus uniformly small (see e.g.\ (\ref{eq:k_lapse_upper})), it is straightforward to find that defining $\mathcal{P}_2 = \lim_{s \to s_{\infty}} r^2 \dot{\phi}(s) \sim \epsilon^2$, one (easily, since the RHS is $O(r^{4-})$) has
\begin{equation}
|r^2 \dot{\phi}(s) - \mathcal{P}_2| \lesssim \epsilon^4 e^{- \beta(R - R_{out})}.
\end{equation}
Combining this with the aforementioned $|\Psi - \alpha^{-1}| \lesssim \epsilon^2 e^{- \beta(R - R_{out})}$, we deduce (\ref{eq:postinv_r}) for $\mathcal{Z}_2 = \alpha \mathcal{P}_2$.

Finally, we provide the estimate (\ref{Y2.inv}) for $\log \mathcal{Y}_2$. Note that (\ref{eq:postinv_lapse_der}) is valid in the whole region $\mathcal{K} = \{ R_i \leq R < + \infty \}$, so that, in particular, we may integrate in the interval $R \in [R_{in}, + \infty)$ to find 
\begin{equation}
\left| \log \mathcal{Y}_2 - \log \left( \frac{\Omega^2}{- \dot{r}} \left( \frac{r}{r_-} \right)^{- \alpha^{-2}} \right) (R_{in})\right| \lesssim \epsilon^2 + \int_{R_{in}}^{R_{out}} |\Psi(R) - \alpha^{-1}| \, dR.
\end{equation}

We now compare this to the estimate (\ref{eq:preinv_lapse}) evaluated at $R = R_{in}$. As we are changing the exponent from $\alpha^{-2}$ to $\alpha^2$, we generate an extra term on the left hand side:
\begin{equation}
\left| \log \mathcal{Y}_2 - \log \mathcal{Y}_1 - (\alpha^{-2} - \alpha^2) R_{in} \right| \lesssim \epsilon^2 + \int_{R_{in}}^{R_{out}} |\Psi(R) - \alpha^{-1}| \, dR.
\end{equation}
Now we appeal to Lemma \ref{lem:inv_interval}. By this lemma, we know that $|R_{out} - R_{in}| \lesssim \log (\epsilon^{-1})$, so the integral on the RHS is $O(\log (\epsilon^{-1}))$. Using also (\ref{eq:inv_interval}) to estimate $R_{in}$ and the estimate (\ref{eq:noinv_y}) for $\log \mathcal{Y}_1$, we find that
\begin{equation}
\left| \log \mathcal{Y}_2 - \frac{1}{2} \alpha^{-2} b_-^{-2} \epsilon^{-2} \right| \lesssim \log(\epsilon^{-1}).
\end{equation}
This completes the proof of the lemma.
\end{proof}

\subsection{Kasner-like asymptotics in synchronous coordinates in both cases}

To complete the proofs of Theorems \ref{thm.***} and \ref{thm.**}, we simply need to use the estimates of Lemmas \ref{lem:noinv_lapse}, \ref{lem:preinv_lapse} and \ref{lem:postinv_lapse} to put the metric $g$ in the Kasner-like form stated. The following lemma will justify the change of coordinates.
\begin{lemma} \label{lem:synchronous}
Let $(M, g)$ be a spherically symmetric spacetime with metric (\ref{eq:metric}). Defining the coordinates $s = u + v$, $t = v  - u$, suppose further that $T = \partial_t = \frac{1}{2} (\partial_v - \partial_u)$ is a Killing field for the metric, i.e.\ $r(s)$ and $\Omega^2(s)$ are functions of $s$, and that we are in a trapped region with $\dot{r}(s) < 0$. 

Suppose that for some interval $\red{J} \subset \R$, there exist constants $\mathcal{Y}, \mathcal{Z} > 0$ and an exponent $\gamma \geq 0$, as well as sufficiently small ``lower-order terms'' $\mathfrak{E}(s)$, such that we have the following asymptotics for the expressions $\frac{\Omega^2}{- \dot{r}}$ and $- r \dot{r}$ when $s \in \red{J}$:
\begin{equation} \label{eq:synch_y}
\frac{\Omega^2}{- \dot{r}} (s) = \mathcal{Y} \cdot \left( \frac{r(s)}{r_-} \right)^{\gamma} \cdot \left( 1 + \mathfrak{E}(s) \right),
\end{equation}
\begin{equation} \label{eq:synch_z}
- r \dot{r} (s) = \mathcal{Z} \cdot (1 + \mathfrak{E}(s)).
\end{equation}
We quantify the required smallness of $\mathfrak{E}$ in the following way: there exists $\epsilon_* > 0$ small and a non-increasing function $\bar{\mathfrak{E}}(s)$ such that $|\mathfrak{E}(s)| \leq \bar{\mathfrak{E}}(s) \leq \epsilon_*$. 

Upon defining (up to translation) the past-directed \blue{proper time} coordinate $\tau$ such that
\begin{equation} \label{eq:tau}
\frac{d \tau}{d s} = -\frac{ \Omega(s) }{2},
\end{equation}
then there exist constants $\mathcal{X}$ and $\mathcal{R}$ depending on $\mathcal{Y}, \mathcal{Z}, \gamma, r_-$, and $\tau_0 \in \R$, such that \blue{one has} the following asymptotics for $\Omega^2$ and $r^2$ with respect to $\tau$:
\begin{equation} \label{eq:synch_x}
\Omega^2(\tau) = 4 \mathcal{X} \cdot (\tau - \tau_0)^{\frac{2(\gamma - 1)}{\gamma + 3}} \cdot (1 + \mathfrak{E}_{X}(\tau)),
\end{equation}
\begin{equation} \label{eq:synch_r}
r^2(\tau) = \mathcal{R} \cdot r_-^2 \cdot (\tau - \tau_0)^{\frac{4}{\gamma + 3}} \cdot (1 + \mathfrak{E}_{R}(\tau)).
\end{equation}
Here, there exists some $C = C(\gamma) > 0$ such that $|\mathfrak{E}_X(\tau)|, |\mathfrak{E}_R(\tau)| \leq C \bar{\mathfrak{E}}(s(\tau))$.

In particular, the metric $g$ can be written in the form
\begin{equation} \label{eq:metric_synch}
g = - d \tau^2 + \mathcal{X} \cdot (\tau - \tau_0)^{\frac{2(\gamma - 1)}{\gamma + 3}} \cdot (1 + \mathfrak{E}_X(\tau)) \, dt^2 + \mathcal{R} \cdot (\tau - \tau_0)^{\frac{4}{\gamma + 3}} \cdot ( 1 + \mathfrak{E}_R(\tau)) \cdot r_-^2 \, d \sigma_{\mathbb{S}^2}.
\end{equation}

Finally, one can find the following relationships between $\mathcal{X}, \mathcal{R}$ and $\mathcal{Y}$:
\begin{equation} \label{eq:xy}
\left| \log \mathcal{X} - \frac{2 (\gamma + 1)}{\gamma + 3}\log \mathcal{Y} \right| \lesssim 1 + | \log \mathcal{Z} |,
\end{equation}
\begin{equation} \label{eq:ry}
\left| \log \mathcal{R} + \frac{2}{\gamma + 3}\log \mathcal{Y} \right| \lesssim 1 + | \log \mathcal{Z} |.
\end{equation}
\end{lemma}

\begin{proof}
In this proof we schematically write all lower-order terms as $\mathfrak{E}$, leaving the precise claims on these errors to the reader. The crucial estimate is to find an expression for $\tau$ in terms of $r$. Using (\ref{eq:tau}), (\ref{eq:synch_y}) and (\ref{eq:synch_z}), we find
\begin{equation} \label{eq:taur}
\frac{d \tau}{dr} = \frac{\Omega}{- 2 \dot{r}} = \frac{ \mathcal{Y}^{1/2} r_-^{1/2}}{ 2 \mathcal{Z}^{1/2}} \cdot \left( \frac{r}{r_-} \right)^{\frac{\gamma+1}{2}} \cdot ( 1 + \mathfrak{E}).
\end{equation}

\blue{Although $\inf_{s \in \red{J}} r(s)$ is not necessarily $0$, we may artificially extend the interval $\red{J}$ to an interval $\red{\tilde{J}}$, and the functions $r(s), \tau(s)$ on $\red{\tilde{J}}$ so that $\red{\inf_{s \in \red{\tilde{J}}} r(s)} = 0$ and \eqref{eq:taur} holds for all $s \in \red{\tilde{J}}$. After doing so, we may normalize $\tau_0$ so that $r(\tau_0) = 0$,} and find that, for $s \in \red{\tilde{J}}$,
\begin{equation} \label{eq:taur2}
\tau - \tau_0 = \frac{\mathcal{Y}^{1/2} r_-^{3/2}}{(\gamma + 3) \mathcal{Z}^{1/2}} \cdot \left( \frac{r}{r_-} \right)^{\frac{\gamma + 3}{2}} \cdot (1 + \mathfrak{E}).
\end{equation}
The remaining assertions are immediate upon combining (\ref{eq:taur2}) with the assumptions (\ref{eq:synch_y}) and (\ref{eq:synch_z}).
\end{proof}

\textit{Conclusion of the proof of Theorem~\ref{thm.***} and \ref{thm.**}}. \blue{Theorem}~\ref{thm.***} follows immediately from combining  Lemmas \ref{lem:noinv_lapse} and \ref{lem:synchronous}, taking $\tau_0=0$ (note indeed that, in the \blue{No Kasner bounce} case \eqref{***}, \eqref{eq:taur} is true up to $\{r=0\}$, so $\tau_0=0$). Here $\gamma = \alpha^2$, and in order to obtain the final error estimate (\ref{***.error}), we also need \eqref{eq:taur} to change variables from $r$ to $\tau$.

For the proof of Theorem~\ref{thm.**} in case \eqref{**}, we first run the proof in $\mathcal{K}_1$ and apply Lemma~\ref{lem:preinv_lapse}  to show that \eqref{eq:synch_y} and \eqref{eq:synch_z}  with $\mathcal{Y}= \mathcal{Y}_1$, $\mathcal{Z}= \mathcal{Z}_1$, $\gamma=\alpha^2$ are true \red{for $J = \mathcal{K}_1$}. Unlike in case \eqref{***}, we are not able to continue these estimates all the way up to $r = 0$, so we take $\tau_0 \neq 0$ to account for the fact that the proper time variable with respect to which $\mathcal{K}_1$ is Kasner-like must be modified. Thus, we can apply Lemma~\ref{lem:synchronous} to deduce the Kasner-like behavior of $\mathcal{K}_1$ claimed in Theorem~\ref{thm.**}. On the other hand, Lemma~\ref{lem:postinv_lapse} and Lemma~\ref{lem:synchronous}, applied to $\mathcal{Y} = \mathcal{Y}_2$, $\mathcal{Z} = \mathcal{Z}_2$, $\gamma = \alpha^{-2}$ determine the Kasner-like behavior of $\red{J} = \mathcal{K}_2$ (where here we take $\tau = 0$).

Finally, we prove the proper time estimates \eqref{eq:propertimek2} and \eqref{eq:propertimekinv}. For the former, taking the logarithm of \eqref{eq:taur2} in the context of the region $\mathcal{K}_2$ (where $\tau_0 = 0$) yields:
\begin{equation*}
\log \tau = \frac{1}{2} \log \mathcal{Y}_2 - \frac{1}{2} \log \mathcal{Z}_2 + \frac{\alpha^{-2} + 3}{2} \log r + O(1).
\end{equation*}
In particular, evaluating this at $s = s_{out}$, then using Lemma~\ref{lem:inv_interval} to estimate $r(s_{out}) = r_{out}$ yields
\begin{equation*}
\log \tau^{-1}(s_{out}) = \frac{1}{4} \alpha^{-2} b_-^{-2} \epsilon^{-2} +\frac{\alpha^{-2} + 3}{4 (1 - \alpha^2)} b_-^{-2} \epsilon^{-2} + O(\log (\epsilon^{-1})) =  b_-^{-2}\frac{\alpha^{-2}+1}{2(1 - \alpha^{2})} + O(\log (\epsilon^{-1})).
\end{equation*}

For the final estimate \eqref{eq:propertimekinv}, we shall use the fact that $\frac{d \tau}{dr} = \frac{\Omega}{-2 \dot{r}}$ to find that
\begin{equation*}
|\tau(s_{in}) - \tau(s_{out})|
\leq (r_{in} - r_{out}) \cdot \max_{s \in \mathcal{K}_{\blue{bo}}} \frac{\Omega}{-2 \dot{r}} \leq r_{in}^{3/2} \cdot \max_{s \in {\mathcal{K}_{\blue{bo}}}} \left( \frac{\Omega^2}{-4 \dot{r}} \frac{1}{- r \dot{r}} \right)^{1/2}.
\end{equation*}
Now applying the estimate \eqref{monot} and then \eqref{eq:inv_interval} again (as $- r \dot{r} \sim \ep^2$ it is absorbed into the $\exp(D \log(\ep^{-1}))$ term in the next line), \blue{one obtains}
\begin{align*}
|\tau(s_{in}) - \tau(s_{out})|
&\leq (\frac{r_{in}}{r_-})^{\frac{3 + (\alpha - \ep^2)^2}{2}} \cdot \exp( - \frac{1}{4} b_-^{-2} \ep^{-2}) \cdot \exp(D \log (\ep^{-1})), \\[0.5em]
&\leq \exp(- \frac{1}{4} b_-^{-2} \ep^{-2} \left[\frac{3+ \alpha^2}{1 - \alpha^2} + 1\right]) \cdot \exp(D \log (\ep^{-1})),
\end{align*}
yielding \eqref{eq:propertimekinv} as required. This completes the proof of Theorem~\ref{thm.**}, and also Theorem~\ref{maintheorem2}.

\appendix

\section{Bessel functions} \label{sec:appendix_A}

\blue{In this appendix, we record several basic} facts about Bessel functions, which are used widely in Section \ref{sec:oscillations}. For further details, refer to Chapter 10 of \cite{NIST:DLMF}. 
We start by recalling \textit{Bessel's equation of order $\nu \in \CC$}:
\begin{equation} \label{eq:bessel_ode}
z^2 \frac{d^2 f}{d z^2} + z \frac{df}{dz} + (z^2 - \nu^2) f = 0.
\end{equation}
Solutions to this equation are known as \textit{Bessel functions of order $\nu$}. \blue{Within the appendix}, we shall always assume $\nu$ to be a nonnegative integer, \blue{and our applications will always be in} the cases $\nu = 0, 1$.

As (\ref{eq:bessel_ode}) is a second order linear homogeneous ODE, there are two linearly independent solutions, denoted $J_{\nu}(z)$ and $Y_{\nu}(z)$. These \blue{have the following properties}.

\begin{fact} \label{fact:bessel1_taylor}
The function $J_{\nu}(z)$, known as the \textit{Bessel function of the first kind}, is an entire function given by the following Taylor expansion:
\begin{equation} \label{eq:bessel1_taylor}
J_{\nu}(z) = \left( \tfrac{1}{2} z \right)^{\nu} \sum_{k=0}^{\infty} \frac{ \left( - \tfrac{1}{4} z^2 \right)^k }{k! \, \Gamma(\nu + k + 1)}.
\end{equation}
In particular, as $z \to 0$, $J_0(z) \to 1$ and $z^{-1} J_1(z) \to 1/2$. \blue{Furthermore, for $0 < z \leq 1$, one has $z |J_0(z)| \leq 1$.}
\end{fact}

\begin{fact} \label{fact:bessel2_taylor}
The function $Y_{\nu}(z)$, known as the \textit{Bessel function of the second kind}, can be defined via the following Fuchsian asymptotic expansion for $\nu = n \in \N \cup \{0\}$ and $z \in \R^+$:
\begin{multline} \label{eq:bessel2_taylor}
Y_{n}\left(z\right)=-\frac{(\tfrac{1}{2}z)^{-n}}{\pi}\sum_{k=0}^{n-1}\frac{(n-
k-1)!}{k!}\left(\tfrac{1}{4}z^{2}\right)^{k}+\frac{2}{\pi}\ln\left(\tfrac{1}{2%
}z\right)J_{n}\left(z\right)\\ -\frac{(\tfrac{1}{2}z)^{n}}{\pi}\sum_{k=0}^{\infty%
}(\psi\left(k+1\right)+\psi\left(n+k+1\right))\frac{(-\tfrac{1}{4}z^{2})^{k}}{%
k!(n+k)!},
\end{multline}
where $\psi(z) = \Gamma'(x) / \Gamma(x)$ and \blue{$\gamma \approx 0.577$} is the Euler-Mascheroni constant. In particular, one has the following asymptotics for $Y_0(z)$ and $Y_1(z)$:
\begin{equation} \label{eq:bessel2_smallasymp}
\left| Y_0(z) - \frac{2}{\pi} \left( \log \left( \tfrac{1}{2}z \right) + \gamma \right) \right| \lesssim z^2, \hspace{5pt} 
\left| Y_1(z) + \frac{2}{\pi z} \right| \lesssim |z \log z |.
\end{equation}
\blue{Furthermore, for $0 < z \leq 1$, one has the quantitative estimate $z |Y_0(z)| \leq 1$.}
\end{fact}

\begin{fact} \label{fact:besselrelation}
By studying the expansions (\ref{eq:bessel1_taylor}) and (\ref{eq:bessel2_taylor}),
\begin{equation} \label{eq:besselrelation}
\blue{J_0'(z) = - J_1 (z) \quad \text{ and } \quad Y_0'(z) = - Y_1 (z).}
\end{equation}
\end{fact}

\begin{fact} \label{fact:bessel_bigasymp}
As $z \to \infty$, we have the following asymptotics for $J_{\nu}(z)$ and $Y_{\nu}(z)$:
\begin{equation} \label{eq:bessel1_bigasymp}
J_{\nu}(z) = \sqrt{ \frac{2}{\pi z} } \left( \cos \left( z - \frac{\nu \pi}{2} - \frac{\pi}{4} \right) + O( z^{-1} ) \right),
\end{equation}
\begin{equation} \label{eq:bessel2_bigasymp}
Y_{\nu}(z) = \sqrt{ \frac{2}{\pi z} } \left( \sin \left( z - \frac{\nu \pi}{2} - \frac{\pi}{4} \right) + O( z^{-1} ) \right).
\end{equation}
\end{fact}

We finish this appendix with two lemmas concerning Bessel's equation (\ref{eq:bessel_ode}). These will be useful when studying Bessel-type equations with inhomogeneous error terms, as in Section \ref{sec:oscillations}. 

\begin{lemma} \label{lem:bessel_wronskian}
Consider Bessel's equation (\ref{eq:bessel_ode}), \blue{where $\nu$ is a} nonnegative integer. We define the usual \textit{Wronskian}, normalized by the Bessel functions of the first and second kind, by
\begin{equation}
\mathcal{W}_{\nu}(z) \coloneqq \det \left| \begin{matrix} J_{\nu}(z) & Y_{\nu}(z) \\ J'_{\nu}(z) & Y'_{\nu}(z) \end{matrix} \right| = J_{\nu}(z) Y'_{\nu}(z) - Y_{\nu}(z) J'_{\nu}(z).
\end{equation}
Then 
\begin{equation}
\mathcal{W}_{\nu}(z) = \frac{2}{\pi z}.
\end{equation}
\end{lemma}

\begin{proof}
A standard manipulation of the second order ODE (\ref{eq:bessel_ode}) gives
\begin{equation*}
\frac{d}{dz} \mathcal{W}_{\nu}(z) = - \frac{1}{z} \mathcal{W}_{\nu}(z).
\end{equation*}
Integrating for $z \in \R^+$ therefore yields $\mathcal{W}_{\nu}(z) = C z^{-1}$ for some constant of integration $C$.

To determine $C$, simply use the asymptotics of Facts \ref{fact:bessel1_taylor} and \ref{fact:bessel2_taylor}. This yields $C = \frac{2}{\pi}$, as required. (This computation is more straightforward in our case $\nu = 0$.)
\end{proof}

For the next lemma, we consider Bessel's equation of order $0$ as a first-order system in $(f_1, f_2) = (f, \frac{df}{dz})$:
\begin{equation} \label{eq:bessel_ode_v2}
\frac{d}{dz} \begin{bmatrix} f_1 \\ f_2 \end{bmatrix}
=
\begin{bmatrix} f_2 \\ - f_1 -  \tfrac{1}{z} f_2 \end{bmatrix}
=
\begin{bmatrix} 0 & 1 \\ -1 & - \tfrac{1}{z} \end{bmatrix} \begin{bmatrix} f_1 \\ f_2 \end{bmatrix}.
\end{equation}
So for any $z_0, z_1 > 0$, we may define the solution operator\footnote{
This can be viewed as a ``non-autonomous'' semigroup, \blue{also commonly known as} a propagator. In particular, \blue{it follows that} $\mathbf{S}(z_2; z_1) \circ \mathbf{S}(z_1; z_0) = \mathbf{S}(z_2; z_0)$.
}
$\mathbf{S}(z_1; z_0)$ as a linear operator $\mathbf{S}(z_1; z_0): \R^2 \to \R^2$ in the following manner: if $(f_1(z_0), f_2(z_0))$ is considered as data for the linear ODE (\ref{eq:bessel_ode_v2}), then $\mathbf{S}(z_1; z_0)$ maps this data $(f_1(z_0), f_2(z_0))$ to the value of the solution at $z=z_1$, namely $(f_1(z_1), f_2(z_1))$. As an equation:
\begin{equation} \label{eq:bessel_homog}
\begin{bmatrix} f_1(z_1) \\ f_2(z_1) \end{bmatrix}
=
\mathbf{S}(z_1; z_0) \begin{bmatrix} f_1(z_0) \\ f_2(z_0) \end{bmatrix}.
\end{equation}

The following lemma then asserts how to use this linear operator in solving Bessel's equation with inhomogeneous terms, as well as an explicit formula and estimates for the operator $\mathbf{S}(z_1; z_0)$.

\begin{lemma} \label{lem:bessel_solution}
Consider Bessel's equation (\ref{eq:bessel_ode}) with $\nu = 0$, but with an inhomogeneous term $F(z)$, i.e.\
\begin{equation}
\frac{d^2 f}{dz^2} + \frac{1}{z} \frac{df}{dz} + f = F.
\end{equation}
Then, for $z_0, z_1 \in \R^+$ and the solution operator $\mathbf{S}(z; w)$ defined as before, we have the following expression:
\begin{equation} \label{eq:bessel_inhomog}
\begin{bmatrix} f(z_1) \\ \frac{df}{dz}(z_1) \end{bmatrix}
=
\mathbf{S}(z_1; z_0) \begin{bmatrix} f(z_0) \\ \frac{df}{dz}(z_0) \end{bmatrix}
+
\bigintsss_{z_0}^{z_1} \mathbf{S}(z_1; \tilde{z}) \begin{bmatrix} 0 \\ F(\tilde{z}) \end{bmatrix} \, d \tilde{z} .
\end{equation}

Furthermore, we have the following results regarding the linear operator $\mathbf{S}(z_1; z_0)$:
\begin{enumerate}[(1)]
\item \label{bessel_uno}
In terms of the Bessel functions $J_{\nu}(z)$, $Y_{\nu}(z)$, one may write
\begin{align} \label{eq:bessel_explicit}
\mathbf{S}(z_1; z_0) &= 
\begin{bmatrix} J_0(z_1) & Y_0(z_1) \\ - J_1(z_1) & - Y_1(z_1) \end{bmatrix} \cdot
\begin{bmatrix} J_0(z_0) & Y_0(z_0) \\ - J_1(z_0) & - Y_1(z_0) \end{bmatrix}^{-1}, \\[1em] &=
\frac{\pi z_0}{2} \begin{bmatrix} J_0(z_1) & Y_0(z_1) \\ - J_1(z_1) & - Y_1(z_1) \end{bmatrix} \cdot
\begin{bmatrix} -Y_1(z_0) & -Y_0(z_0) \\ J_1(z_0) & J_0(z_0) \end{bmatrix}. \label{eq:bessel_explicit2}
\end{align}
\item \label{bessel_dos}
Using the usual $l^2$ norm on $\R^2$, defined as $\| (x,y)\|_{l^2}= \sqrt{x^2+y^2}$, and defining $\| \mathbf{A} \|_{l^2 \to l^2} \coloneqq \sup_{x \in \R^2, \|x\|_{l^2} = 1} \|\mathbf{A} x \|_{l^2} $ to be the $l^2$ operator norm of a linear map $\mathbf{A}$, \blue{we have that, for all $z_0, z_1 > 0$,}
\begin{equation} \label{eq:bessel_operator_norm}
\| \mathbf{S}(z_1; z_0) \|_{l^2 \to l^2} \leq \max \left \{ \frac{z_0}{z_1} , 1 \right\}.
\end{equation}
\end{enumerate}
\end{lemma}

\begin{proof}
Once again letting $(f_1, f_2) = (f, \frac{df}{dz})$, we have now the first-order system:
\begin{equation*}
\frac{d}{dz} \begin{bmatrix} f_1(z) \\ f_2(z) \end{bmatrix}
=
\begin{bmatrix} 0 & 1 \\ -1 & - \tfrac{1}{z} \end{bmatrix} \begin{bmatrix} f_1(z) \\ f_2(z) \end{bmatrix} + \begin{bmatrix} 0 \\ F(z) \end{bmatrix}.
\end{equation*}
Then the expression (\ref{eq:bessel_inhomog}) follows from the standard theory of first-order systems and Duhamel's principle.

To get (\ref{bessel_uno}), we recall that in light of Fact \ref{fact:besselrelation}, the most general solution to the first-order system with no inhomogeneous term is:
\begin{equation}
\begin{bmatrix} f_1(z) \\ f_2(z) \end{bmatrix}
=
\begin{bmatrix} c_J J_0(z) + c_Y Y_0(z) \\ - c_J J_1(z) - c_Y Y_1(z) \end{bmatrix}
=
\begin{bmatrix} J_0(z) & Y_0(z) \\ - J_1(z) & - Y_1(z) \end{bmatrix} \begin{bmatrix} c_J \\ c_Y \end{bmatrix},
\end{equation}
where $c_J, c_Y$ are coefficients in $\R$. Plugging this into (\ref{eq:bessel_homog}) and allowing $c_J, c_Y$ to vary will yield (\ref{eq:bessel_explicit}). The expression (\ref{eq:bessel_explicit2}) then follows by using Lemma \ref{lem:bessel_wronskian} to compute the inverse.

Finally, for (\ref{bessel_dos}), suppose $f(z)$ is a solution to the homogeneous Bessel's equation of order $0$. Then we directly compute
\begin{equation*}
\frac{d}{dz} ( f(z)^2 + f'(z)^2 ) = - \frac{2}{z} f'(z)^2.
\end{equation*}
Therefore, if $z_1 \geq z_0$, then it is clear that $f(z_1)^2 + f'(z_1)^2 \leq f(z_0)^2 + f'(z_0)^2$, while if $z_1 < z_0$, then Gr\"onwall's inequality gives
\begin{equation*}
f(z_1)^2 + f'(z_1)^2 \leq \exp( \int_{z_0}^{z_1} - \frac{2}{z} \, dz ) ( f(z_0)^2 + f'(z_0)^2 ) = \left( \frac{z_0}{z_1} \right)^2 ( f(z_0)^2 + f'(z_0)^2 ).
\end{equation*}
This yields exactly the required estimate on $\| \mathbf{S} (z_1; z_0) \|_{l^2 \to l^2}$.
\end{proof}

In this paper, we will often encounter rescaled versions of Bessel's equation: let $\chi > 0$ be a constant, then consider the equation:
\begin{equation} \label{eq:bessel_rescale}
\frac{d^2 f}{dx^2} + \frac{1}{x} \frac{df}{dx} + \chi^2 f = F(x).
\end{equation}
Using Lemma \ref{lem:bessel_solution}, it is then straightforward to deduce the following corollary.
\begin{corollary} \label{cor:bessel_solution_rescale}
Let $f = f(x)$ be a solution to (\ref{eq:bessel_rescale}). Then, if we denote by $\mathbf{Q}_{\chi}$ the scaling matrix
\begin{equation*}
\mathbf{Q}_{\chi} \coloneqq \begin{bmatrix} 1 & 0 \\ 0 & \chi \end{bmatrix},
\end{equation*}
and define $\mathbf{S}_{\chi}(x_1; x_0) = \mathbf{Q}_{\chi} \circ \mathbf{S}(\chi x_1; \chi x_0) \circ \mathbf{Q}_{\chi}^{-1}$, then for any $x_0, x_1 > 0$ we have
\begin{equation} 
\begin{bmatrix} f(x_1) \\ \frac{df}{dx}(x_1) \end{bmatrix}
=
\mathbf{S}_{\chi}(x_1; x_0) \begin{bmatrix} f(x_0) \\ \frac{df}{dx}(x_0) \end{bmatrix}
+
\bigintsss_{x_0}^{x_1} \mathbf{S}_{\chi}(x_1; \tilde{x}) \begin{bmatrix} 0 \\ F(\tilde{x}) \end{bmatrix} \, d \tilde{x} .
\end{equation}

Furthermore, we have the following results regarding $\mathbf{S}_{\chi}(x_1; x_0)$:
\begin{enumerate}[(1)]
\item \label{bessel_uno'}
An explicit formula is given by:
\begin{align} 
\mathbf{S}_{\chi}(x_1; x_0) &= 
\begin{bmatrix} J_0(\chi x_1) & Y_0(\chi x_1) \\ - \chi J_1(\chi x_1) & - \chi Y_1(x_1) \end{bmatrix} \cdot
\begin{bmatrix} J_0(\chi x_0) & Y_0(\chi x_0) \\ - \chi J_1(\chi x_0) & - \chi Y_1(\chi x_0) \end{bmatrix}^{-1}, \\[1em] &=
\frac{\pi x_0}{2} \begin{bmatrix} J_0(\chi x_1) & Y_0(\chi x_1) \\ - \chi J_0(\chi x_1) & - \chi Y_1(\chi x_1) \end{bmatrix} \cdot
\begin{bmatrix} - \chi Y_1(\chi x_0) & - Y_0(\chi x_0) \\ \chi J_1(\chi x_0) & J_0(\chi x_0) \end{bmatrix}.
\end{align}
\item \label{bessel_dos'}
The following estimate holds on the operator norm of $\mathbf{S}_{\chi}(x_1; x_0)$:
\begin{equation}
\| \mathbf{S}_{\chi}(x_1; x_0) \|_{l^2 \to l^2} \leq \max \{ \chi, \chi^{-1} \} \cdot \max \left \{ 1, \frac{x_0}{x_1} \right \}.
\end{equation}
\end{enumerate}
\end{corollary}

\begin{proof}
Simply use the substitution $z = \chi x$ in (\ref{eq:bessel_rescale}), and apply Lemma \ref{lem:bessel_solution}. Details are left to the reader.
\end{proof}

\printbibliography

\end{document}